%% file: main.tex
\documentclass[10pt,journal,letterpaper,compsoc,dvipsnames]{IEEEtran}
\input{init} 
\date{} 

\title{Carbon: Scaling Trusted Payments with Untrusted~Machines}

\begin{document}

\author{%
    Martina~Camaioni, Rachid~Guerraoui, Jovan~Komatovic, Matteo~Monti, Pierre-Louis~Roman, Manuel~Vidigueira, and Gauthier~Voron%
    \IEEEcompsocitemizethanks{%
        \IEEEcompsocthanksitem The authors are with Ecole Polytechnique F\'{e}d\'{e}rale de Lausanne (EPFL), Lausanne CH-1015, Switzerland. E-mails: 
        martina.camaioni@epfl.ch, rachid.guerraoui@epfl.ch, jovan.komatovic@epfl.ch, matteo.monti@epfl.ch, pierre-louis.roman@epfl.ch, manuel.ribeirovidigueira@epfl.ch, gauthier.voron@epfl.ch.
    }
    \thanks{This paper extends the one with DOI \href{https://doi.org/10.1109/TDSC.2024.3428617}{10.1109/TDSC.2024.3428617}, appearing at TDSC 2024, with formal definitions, pseudocode, and proofs in appendices.}
}

\IEEEtitleabstractindextext{
    \redefineacronyms 
    \begin{abstract}
    	\input{sections/0-abstract}
    \end{abstract}
    \redefineacronyms 

    \begin{IEEEkeywords}
        Distributed system, Payment system, Cryptocurrency, Asynchrony, Reliable broadcast
    \end{IEEEkeywords}
}

\maketitle
\IEEEdisplaynontitleabstractindextext

\input{sections/1-introduction}

\input{sections/2-background-related}

\input{sections/3-model-broadcast}

\input{sections/4-currency}

\input{sections/6-evaluation}

\input{sections/8-conclusion}
\input{sections/9-acknowledgments}

\bibliographystyle{IEEEtran}
\bibliography{bibliography.bib}

\clearpage
\onecolumn

\tableofcontents

\newpage
\appendices
\input{appendix/introduction}
\input{appendix/model}
\input{appendix/problem_definition}
\input{appendix/view_generator}
\input{appendix/storage_module}
\input{appendix/reconfiguration}
\input{appendix/transactions}
\input{appendix/voting}

\end{document}

%% file: init.tex
\usepackage[utf8]{inputenc}
\usepackage[T1]{fontenc}

\usepackage{algpseudocode}

\usepackage{amsmath}
\usepackage{amsthm}
\usepackage{amssymb}

\usepackage{arydshln} 
\usepackage{booktabs} 
\usepackage{dblfloatfix}
\usepackage{enumerate} 
\usepackage{float}
\usepackage{graphicx}
\usepackage{hyperref}
\usepackage{multirow} 
\usepackage{booktabs} 
\usepackage{paralist}
\usepackage{stmaryrd}
\usepackage[font=footnotesize]{subfig} 
\usepackage{tabularx} 
\usepackage[para]{threeparttable} 
\usepackage[normalem]{ulem} 
\usepackage{url}
\usepackage[dvipsnames]{xcolor}
\usepackage{xspace}

\usepackage[shortcuts]{glossaries-extra}
\glsdisablehyper

\newabbr{bft}{BFT}{Byzantine-fault tolerant}
\newabbr{brb}{BRB}{Byzantine reliable broadcast}
\newabbr%
    [longplural=central bank digital currencies]
    {cbdc}{CBDC}{central bank digital currency}
\newabbr{kv}{KV}{key-value}
\newabbr{mt}{MTree}{Merkle tree}
\newabbr{mpt}{MPTree}{Merkle Patricia tree}
\newabbr{mp}{MProof}{Merkle proof}
\newabbr{mr}{MRoot}{Merkle tree root}
\newabbr{pki}{PKI}{public key infrastructure}
\newabbr{pow}{PoW}{proof of work}
\newabbr{varint}{VarInt}{variable-length integer}
\newabbr{vdf}{VDF}{verifiable delay function}

\newcommand{\redefineacronyms}{%
    \glsresetall
    \glsunset{pki}
    \glsunset{varint}
}

\newcommand{\bcastname}{Draft\xspace}
\newcommand{\assignname}{Dibs\xspace}
\newcommand{\sysname}{Carbon\xspace}

\newcommand{\eg}{e.g.\xspace}

\newcommand{\ie}{i.e.\xspace}

\newcommand{\para}[1]{\vspace{.5\baselineskip}\noindent\textbf{#1.}}

\usepackage{listings}
\usepackage{bold-extra} 

\definecolor{mygreen}{rgb}{0.254,0.572,0.294}
\definecolor{mygray}{rgb}{0.5,0.5,0.5}
\definecolor{myorange}{rgb}{1,0.35,0}
\definecolor{mymauve}{rgb}{0.58,0,0.82}
\definecolor{myblue}{rgb}{0.2,0.4,0.6}

\definecolor{rakos4orange}{RGB}{255,165,0}
\definecolor{rakos4blue}{RGB}{14,48,173}
\definecolor{rakos4lblue}{RGB}{92,172,238}
\definecolor{rakos4dgray}{RGB}{77,77,77}
\definecolor{plainred}{RGB}{211,63,63}
\definecolor{plainorange}{RGB}{221,105,41}

\lstdefinelanguage{Golang}%
  {morekeywords=[1]{package,import,struct,defer,panic,%
     recover,select,var,const,iota, class},%
   morekeywords=[2]{string,uint,uint8,uint16,uint32,uint64,int,int8,int16,%
     int32,int64,bool,float32,float64,complex64,complex128,byte,rune,uintptr,%
     error,interface,node},%
   morekeywords=[3]{map,slice,make,new,nil,len,cap,copy,close,%
     delete,append,real,imag,complex,chan,},%
   morekeywords=[4]{break,continue,goto,switch,case,fallthrough,%
    default,},%
   morekeywords=[5]{Println,Printf,Error,Send},%
   sensitive=true,%
   morecomment=[l]{//},%
   morecomment=[s]{/*}{*/},%
   morestring=[b]",%
   morestring=[s]{`}{`},%
   }

\newcommand{\code}[1]{\lstinline{#1}}

\theoremstyle{definition}
\newtheorem{theorem}{Theorem}
\theoremstyle{definition}
\newtheorem{lemma}{Lemma}
\theoremstyle{definition}

\theoremstyle{definition}
\newtheorem{invariant}{Invariant}
\theoremstyle{definition}
\newtheorem{assumption}{Assumption}
\theoremstyle{definition}
\newtheorem{definition}{Definition}
\theoremstyle{definition}
\newtheorem{corollary}{Corollary}

\usepackage[capitalize]{cleveref}
\renewcommand{\Cref}[1]{\cref{#1}} 
\renewcommand{\autoref}[1]{\cref{#1}} 

\newcommand{\secref}{$\S\mkern-4mu$}
\crefname{section}{\secref}{$\S\S\mkern-4mu$}
\crefname{figure}{Fig.}{Figs.}
\crefname{table}{Tab.}{Tabs.}
\crefname{appendix}{Appx.}{Apps.}
\crefname{assumption}{Assumption}{Assumptions}
\crefname{invariant}{Invariant}{Invariants}
\crefname{observation}{Observation}{Observations}
\crefname{lstlisting}{Listing}{Listings}
\crefname{code}{Line}{Lines}

\usepackage[tracking,kerning,spacing,factor=1100,stretch=10,shrink=10]{microtype}

\graphicspath{ {./}{./figs/}{./figs_matplotlib/figures/} }

\newcommand{\step}[1]{\##1}

\newcommand{\formatmsg}[1]{\texttt{#1}\xspace}

%% file: sections/0-abstract.tex
This paper introduces \sysname, a high-throughput system enabling asynchronous (safe) and consensus-free (efficient) payments and votes within a dynamic set of clients.
\sysname is operated by a dynamic set of validators that may be reconfigured asynchronously, offering its clients eclipse resistance as well as lightweight bootstrap.
\sysname offers clients the ability to select validators by voting them in and out of the system thanks to its novel asynchronous and stake-less voting mechanism.
\sysname relies on an asynchronous and deterministic implementation of \acl{brb} that uniquely leverages a permissionless set of untrusted servers, brokers, to slash the cost of client authentication inherent to Byzantine fault tolerant systems.
\sysname is able to sustain a throughput of one million payments per second in a geo-distributed environment, outperforming the state of the art by three orders of magnitude with equivalent latencies.

%% file: sections/1-introduction.tex
\section{Introduction}
\label{sec:introduction}

Since their emergence with Bitcoin~\cite{bitcoin}, cryptocurrencies have received tremendous attention from academia~\cite{bitcoin-ng_nsdi16,algorand_sosp17,teechain_sosp19,prism_ccs19,red_belly_sp21} and industry~\cite{smr_libra_report19,stellar_sosp19,quorum}.
These payment systems are managed by validators that process operations issued by clients which may be user-driven.
A vast body of work has been relying on consensus among validators~\cite{pbft_osdi99,hbbft_ccs16,hotstuff_podc19,dagrider_podc21,aardvark_nsdi09,raynal_podc14} leading to payments being totally-ordered across all clients within a single log~\cite{algorand_sosp17,prism_ccs19,stellar_sosp19,quorum,omniledger_sp18,dfinity-icc-podc22}.

However, payments do not require consensus~\cite{consensus_number_crypto_dn21}.
Instead, it suffices to use reliable broadcast~\cite{bracha_brb_87,cachin_verifiable_bcast_crypto01,scalable_brb_disc19,partially-connected_brb_icdcs21} which is both more stable than consensus, as it can run deterministically in an asynchronous environment unlike consensus~\cite{flp85}, and more efficient, as it reduces ordering constraints on client operations hence enables parallelism.
Recent approaches have used reliable (or consistent) broadcast as alternatives to consensus to boost the throughput of these systems~\cite{astro_dsn20,fastpay_aft20,brick_fc21}.

\para{Insight}
Regardless of their theoretical efficiency, all these approaches still only exploit a fraction of the available resources in the system.
A key observation underlying our work is that a crucial strength of permissionless cryptocurrencies---the secure use of a large amount of untrusted resources to operate a secured system---has not been embraced by permissioned systems.
In this paper, we show how to leverage untrusted resources to scale the performance of trusted payment operations.
We propose a new hybrid model composed of  (a) clients, (b) validators---trusted, permissioned, stateful servers---and (c) \emph{brokers}---trustless, permissionless, stateless servers whose sole purpose is to assist validators.

Untrusted servers are typically used inefficiently in permissionless cryptocurrencies that makes them compete against one another via mining~\cite{bitcoin,bitcoin-ng_nsdi16,algorand_sosp17,wood_ethereum_2014,ouroboros_crypto17} to ensure Sybil resilience.
Instead, brokers in our system work in symbiosis with validators and are not forced to wastefully compete.
This allows us to efficiently tap into the vast pool of untrusted resources on the Internet to boost the performance of the proposed permissioned payment system.

\para{Goal}
We aim to implement a system that supports multiple types of operations, notably payments and votes.
Our system is composed of (1) a set of clients that issue operations, and of (2) the validators that commit said operations.
Each client operation is associated with (1) a client in whose name the operation was issued, and (2) a sequence number which, together with its client, uniquely identifies the operation.
The guarantees that such a system must ensure are:
\begin{itemize}
    \item \emph{Liveness:} If a correct client that never leaves the system issues an operation, that operation is eventually committed.
    
    \item \emph{Safety:} No two operations are committed for the same pair of client and sequence number.

    \item \emph{Integrity:} If an operation associated with a correct client is committed, then the client previously issued that operation.
\end{itemize}
For brevity, we relegate the formal definition of the problem solved in this paper to \cref{appx:problem_definition}.

\para{Carbon}
We propose the \sysname system that ensures the aforementioned guarantees for payments and votes even under asynchrony.
\sysname is the first system using reliable broadcast for efficiency, untrusted resources for performance, and user-driven reconfiguration for dynamic deployments.

\sysname is operated by a dynamic set of validators that may be reconfigured asynchronously~\cite{aguilera2011dynamic,alchieri2018efficient,async_byz_reconf_disc20,guerraoui2020dynamic} to allow for secure and long-lasting deployments.
Each reconfiguration is subject to client voting.
This voting procedure allows clients to both approve the joining of beneficial validators and initiate the removal of misbehaving validators, thus implementing user-driven accountability~\cite{peerreview_sosp07}.
Voting in \sysname is not based on staking~\cite{snow-white-fc19,async-pos-sss21} and thus does not require users to lock their funds away for that purpose.
Instead, \sysname's asynchronous voting uses clients' balance to determine their voting power in a Sybil-resilient fashion.

\sysname enables any client to register into the system to issue payment and vote operations.
Thanks to the asynchronous reconfiguration of validators, clients in \sysname can bootstrap their knowledge of the system in a secured manner that thwarts eclipse attacks~\cite{routing_attacks_cacm21,heilman_eclipse_bitcoin_usenixsec15,erebus_eclipse_sp20}.
Once bootstrapped, clients can register using a fitting Sybil-resilient scheme~\cite{dwork_pricing_crypto92,back_hashcash_2002,rem_mining_usenixsec_17,popersonhood_eurospw17} to obtain an account identifier used to reduce the size of the messages that carry client operations.
For scalability, \sysname ensures that each account incurs a storage footprint on validators that is bounded by a constant regardless of the number of operations related to that account.
Similarly, \sysname enables lightweight clients that only require minimal storage to function.

Brokers in \sysname act as middlemen between clients and validators and, hence, shield validators from clients' spams and DDoS attacks.
Additionally, brokers serve as caches for \sysname's state and for certificates emitted by validators---upon client registrations, payments and votes---to help the latter free resources.
Brokers also store the validators' reconfiguration information to help clients bootstrap.

Clients in \sysname use \bcastname, an asynchronous and deterministic \acl{brb}~\cite{oracular-brb-disc22}, to send their authenticated operations to the validators.
\bcastname boosts throughput by minimizing the processing time of client authentication on validators.
In effect, \bcastname implements an aggregate signature scheme~\cite{agg-bls-sig-eurocrypt03}, orchestrated by brokers, that transforms many client signatures into a single rapidly-verifiable signature.
This aggregation slashes both the CPU footprint and network bandwidth footprint of client authentication on validators.
\bcastname is further described in \cref{sec:bcast}.

\sysname only requires simple financial incentives to encourage useful behaviors.
In particular, the high throughput of \sysname simplifies the design of the payment fee policy as payment congestion is unlikely to happen.
Thanks to its throughput and simple incentive requirements, \sysname is a prime candidate for \acp{cbdc}~\cite{design_cbdc_nber_2020,chaum_cbdc_snb_2021} and global payment systems.

\para{Evaluation}
We have implemented \sysname in 26,000 Rust LOCs and evaluated it on a global AWS deployment.
Our evaluation focuses on
(1) exploring how \sysname performs under various workloads,
(2) dissecting the benefits of \sysname's internal mechanisms, and
(3) assessing the efficiency of \bcastname's aggregation scheme.

\sysname sustains a 1M~tx/s throughput with provable delivery latencies around 10~s.
\sysname's throughput outperforms that of the evaluated state of the art, namely Algorand, Quorum, and Fastpay, by several orders of magnitude.
We further prove the benefits of signature aggregation and parallel memory accesses on the throughput of \sysname.
Additionally, we show that \sysname's performance scales almost perfectly with the number of used validators, unlike classic consensus.
Finally, we demonstrate the efficiency of \bcastname's aggregate signature scheme compared to Ed25519~\cite{curve25519_pkc06} and BLS12-381 multi-signatures~\cite{bls_jcrypt04,bdn-asiacrypt18} schemes.

As a comparison, to the best of our knowledge, the only approach that reaches numbers similar to \sysname is the RLN system~\cite{rln} by Amazon AWS and SETL.
RLN aims to be a worldwide interbank exchange platform for \acp{cbdc} and achieves 1M~tx/s in a simulated environment.
Compared to \sysname, RLN only supports \emph{unsigned} payments.
We show with \sysname how to concretize the goal of RLN and secure it in the face of asynchrony and Byzantine failures.

\para{Roadmap}
We discuss background and related work in \cref{sec:background-related}.
We overview the model and the \bcastname and \assignname protocols used by \sysname in \cref{sec:model-bcast}.
We detail our main contribution \sysname in \cref{sec:currency}.
We overview our implementation and evaluation results in \cref{sec:evaluation}, then conclude in \cref{sec:conclusion}.
The appendices provide pseudocode and proofs of correctness for \sysname's mechanisms presented in \cref{sec:currency}.

%% file: sections/2-background-related.tex
\section{Background \& Related Work}
\label{sec:background-related}

This section presents a refresher of \sysname's underlying cryptographic primitives and discusses some related work.

\para{BLS multi-signatures}\label{sec:background-multisig}
\sysname relies on asymmetric cryptography for message authentication, particularly on BLS multi-signatures~\cite{bls_jcrypt04,bdn-asiacrypt18}.
In brief: a set of BLS signatures that sign the same statement can be quickly aggregated in a constant-sized multi-signature, which can be checked against the constant-sized aggregation of all signing public keys.
Once aggregated, a BLS multi-signature can be verified in constant time regardless of the number of signatures it aggregates~\cite{bdn-asiacrypt18}.
Whenever the set of signers of a statement is not fixed, metadata is required to identify the set of public keys signing the message (unlike threshold signatures).
The permissioned nature of \sysname allows us to do so efficiently: to each validator-generated multi-signature we attach an $n$-bit bitmask to identify which of the $n$ validators have contributed to the signature.

\para{Authenticated data structures}\label{sec:background-merkle-tree}
\sysname extensively uses authenticated data structures, namely \acp{mt} and \acp{mpt}. 

An \ac{mt} organizes a set of items in a binary hash tree.
The inclusion of an object in an \ac{mt} can be publicly proven against the corresponding \ac{mr} using a logarithmic-sized \ac{mp}.

\Acp{mpt} are a variant of \acp{mt} that also enables proofs of exclusion.
In essence: an \ac{mpt} stores elements only on its leaves; an element in an \ac{mpt} can only appear along the path determined by the bit-representation of its hash.
As such, the exclusion of an element can be proven by an \ac{mp} showing the absence of the element from the only leaf that could hold it.
Unlike \acp{mt}, \acp{mpt} are generally unbalanced, but can be optimized to minimize depth.

\begin{figure*}[tb]
    \centering
    \includegraphics[width=.9\textwidth]{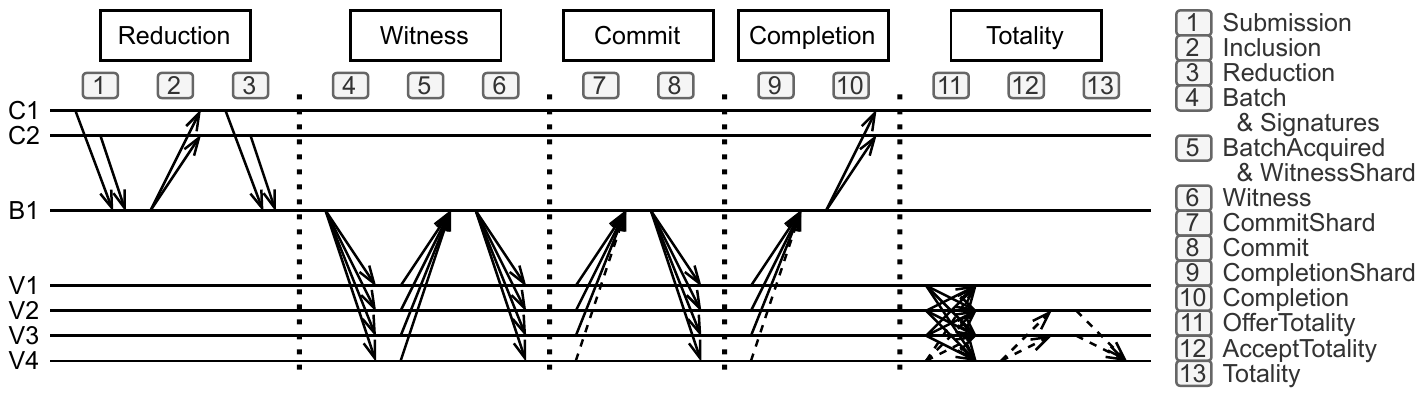}
    \caption{\textbf{Overview of the \bcastname \acs{brb} protocol~\cite{oracular-brb-disc22} totaling 13 steps (\step{1}--\step{13}) between 2 clients (C1--C2), 1 broker (B1) and 4 validators (V1--V4).}
    Clients broadcast a \sysname operation in \step{1} and receive confirmation of delivery in \step{10}.
    Compared to the original algorithm, \sysname merges two steps in both \step{4} and \step{5} for efficiency.
    Arrows of optional steps are dashed.}
    \label{fig:bcast-phases}
\end{figure*}

\para{Related work}\label{sec:related-work}
In the seminal Bitcoin paper~\cite{bitcoin}, Nakamoto defined the asset-transfer problem---the core of a payment system---and proposed a solution that relies on agreement of approved transactions secured by \ac{pow}.
Inspired by Bitcoin's approach, most cryptocurrencies rely on consensus~\cite{cachin_keynote_disc17}, whose throughput scales poorly~\cite{algorand_sosp17,hbbft_ccs16,hotstuff_podc19,ouroboros_crypto17}.
Recent approaches have shown that consensus throughput can be improved, to an extent, by running instances in parallel~\cite{red_belly_sp21,omniledger_sp18,ohie_sp20} or in a hierarchy~\cite{stellar_sosp19,steward_tdsc10}.

However, consensus is not necessary to perform payments~\cite{consensus_number_crypto_dn21}: partial order suffices.
This theoretical result has inspired a new line of research: \emph{consensus-less} payment systems, relying instead on \ac{brb}~\cite{bracha_brb_87}.
Intuitively, this communication primitive ensures that either 
(1) all processes in the system deliver the same message, or 
(2) no process delivers any message.
Crucially, reliable broadcast can be implemented deterministically in an asynchronous setting~\cite{cachin2011introduction}, unlike consensus~\cite{flp85}.

The first proposed consensus-less payment system was Astro~\cite{astro_dsn20}.
In contrast to \sysname, Astro assumes a static set of validators (dynamicity aspects are only sketched~\cite{astro-report}).
Moreover, each client in Astro has a single corresponding representative through which it issues payments; there is no guarantee that a client's payment will ever be processed if its representative is faulty.
Then, Pastro~\cite{kuznetsov2021permissionless} introduced an approach that enriches that of Astro by allowing reconfiguration of validators.
Compared to \sysname, Pastro allows a server to be promoted to a validator even if the server does not enjoy ``sufficiently strong'' support of clients: balance-based voting mechanism is not incorporated into Pastro.

FastPay~\cite{fastpay_aft20} and Brick~\cite{brick_fc21} rely on an even weaker primitive than \ac{brb}, \ie consistent broadcast, to disseminate payments asynchronously in a payment system.
However, both approaches are built atop an already-existing payment system, somehow leaving out of their scope the issues of validator replication and validator set dynamicity.

\para{Reconfiguration}
Reconfiguration is an important feature of long-lived distributed systems as it allows correct components to take the place of faulty ones.
Concretely, reconfiguration of synchronous or partially-synchronous~\cite{dls88} payment systems can be deterministically implemented using consensus~\cite{pbft_osdi99,bft-smart_dsn14,lamport2009vertical}: all validators simply agree on the next set of validators.
This deterministic approach is not possible in asynchrony~\cite{flp85}.
Instead, \sysname relies on asynchronous lattice agreement~\cite{attiya1995atomic,faleiro2012generalized,di2020byzantine,zheng_bla_async_opodis20} for a weaker, eventual notion of agreement for reconfiguration~\cite{aguilera2011dynamic,alchieri2018efficient,async_byz_reconf_disc20,guerraoui2020dynamic}.

%% file: sections/3-model-broadcast.tex
\section{Model \& Broadcast}
\label{sec:model-bcast}

This section presents the model of \sysname as well as the \bcastname and \assignname protocols~\cite{oracular-brb-disc22} that it uses.

\subsection{Model}
\label{sec:model}

We consider a system of asynchronous processes, each proceeding at its own arbitrary speed.
A process is either a \emph{client} or a \emph{server}.
Those processes that execute the protocol assigned to them are correct, those that do not are Byzantine and fail arbitrarily~\cite{byz_generals_toplas82}.
Processes communicate over an asynchronous network with unbounded (but finite) message delays.
We assume that messages between correct processes which never leave the system are eventually received to ensure forward progress, \ie liveness.
\sysname is asynchronous and dynamic, hence it cannot ensure a stronger guarantee which states that messages between any correct processes (even those which might eventually leave) are received.

The set of clients is dynamic.
We make no assumption on the number of faulty clients.
We assume the existence of a Sybil-resilient mechanism for clients as detailed in \cref{sec:client-registration}.
For voting safety, we assume that more than 50\% of the total money supply is held by correct clients as explained in \cref{sec:currency-votes}.

Servers are comprised of \emph{validators}, that process operations in \sysname, and \emph{brokers}, that assist validators.
Both sets of validators and brokers are dynamic.
Since only validators are critical to the safety of the system, we only detail the mechanism by which servers become validators, in \cref{sec:reconfiguration-validator}, and we do not restrict how servers become brokers (examples are given in \cref{sec:currency-incentives}).
As in other dynamic systems~\cite{async_byz_reconf_disc20,guerraoui2020dynamic}, the set of validators evolves in \emph{views}, with a fixed set of servers per view.
We assume that more than two-thirds of validators are correct in any given view, as required for safety~\cite{byz_generals_toplas82}.
We assume that at least one broker is correct, as required for liveness.
We assume a computationally bounded adversary unable to violate the correctness of cryptographic primitives.

\subsection{\bcastname for Byzantine Reliable Broadcast}
\label{sec:bcast}

A payment system can be safely replicated without total ordering, \ie without consensus~\cite{consensus_number_crypto_dn21}.
Instead, to safely and efficiently distribute payments among validators, it suffices to pair a \acf{brb}~\cite{bracha_brb_87} with \emph{source ordering} whereby each client orders its own messages using a simple sequence number.
\Ac{brb} ensures that a message broadcast by a correct client is eventually delivered by all correct validators~\cite[\secref~3.11]{cachin2011introduction} while source order ensures that messages from a correct client are delivered in the order that they were broadcast.
Compared to consensus, source ordered \ac{brb} offers deterministic guarantees even in an asynchronous environment~\cite{flp85} and validators can trivially deliver messages from different clients in parallel.

\para{Draft}
We build \sysname on top of the \bcastname \ac{brb}~\cite{oracular-brb-disc22} to maximize \sysname's throughput.
\bcastname alleviates the predominant bottleneck in high-throughput Byzantine broadcast protocols, namely the cost of client authentication on validators' CPU~\cite{mirbft-jsys22}, thanks to its novel aggregate signature scheme.
An aggregate signature scheme~\cite{agg-bls-sig-eurocrypt03} enables the compaction of any set of signatures on any set of messages.
The aggregate signature scheme of \bcastname combines \emph{rapidly-verifiable multi-signatures}~\cite{bdn-asiacrypt18} with an interactive protocol between brokers and clients to achieve \emph{rapidly-verifiable aggregate signatures}.
\Cref{fig:bcast-phases} overviews the \bcastname protocol composed of 5 phases (reduction, witness, commit, completion and totality), totaling 13 steps (noted \step{1}--\step{13}) as implemented in \sysname.

\para{Reduction phase}
In the first phase, brokers gather client messages into batches (\step{1}) and request clients to multi-sign the digest of the batch (\step{2}) using their BLS key.
In the best case, all clients are correct and promptly respond to the query (\step{3}).
Since the same digest is BLS multi-signed by all correct clients, the broker can aggregate all these multi-signatures into an aggregate signature that is rapidly verifiable.

However, the broker does not wait forever for the clients' answer since some clients may be slow or faulty.
The broker only waits for a given amount of time and aggregate the multi-signatures it receives after a timeout triggers.
For the remaining clients that did not respond in time, the broker uses the signatures that they sent in \step{1} as a fallback mechanism.
\sysname employs Ed25519 signatures~\cite{curve25519_pkc06} in \step{1}.

Thanks to this reduction phase, validators receive batches of authenticated client messages that each contains a single already-aggregated signature for correct clients.

\para{Witness phase}
The witness phase saves a majority of validators from performing the costly task of authenticating clients.
As proposed in Mir-BFT~\cite{mirbft-jsys22} and Red Belly~\cite{red_belly_sp21}, \bcastname requires the minimum number of validators, \ie $f+1$, to verify the signature(s) in batches.
The rationale is that if $f+1$ validators sign the same statement then at least one of these validators is correct since at most $f$ may fail.

Therefore, the broker in \bcastname sends the client messages to all validators (\formatmsg{Batch}, \step{4}) but the signature(s) to only $f+1$ validators (\formatmsg{Signatures}, \step{4}).
Correct validators acknowledge that they received the client messages (\formatmsg{BatchAcquired}, \step{5}), and, for the $f+1$ who were asked, verify the received client signature(s) and respond with a signed attestation that the clients have been correctly authenticated (\formatmsg{WitnessShard}, \step{5}).
The broker aggregates these $f+1$ signatures into a \emph{witness certificate} that it sends to all validators (\step{6}).
In turn, validators check the correctness of the witness certificate to not have to verify the client signature(s) themselves.

\para{Commit and totality phases}
The commit phase ensures the delivery of a batch by at least $2f+1$ validators such that all correct validators may deliver a batch either in the commit phase, or, later in the totality phase for slow validators.

To do so, all validators first ensure that the batch contains no \emph{equivocations}, \ie no conflicting client messages that violate the safety of the \ac{brb}.
Every correct client message is uniquely identifiable by a pair $(\textit{id}, \textit{seqnum})$ where $\textit{id}$ is the client identifier and $\textit{seqnum}$ is the sequence number used to ensure source order as described at the beginning of \cref{sec:bcast}.
A correct validator detects an equivocation when the same pair $(\textit{id}, \textit{seqnum})$ is associated to two different messages.
After checking for equivocations, validators send a signed attestation that they are ready to deliver the batch except for the equivocated messages (\step{7}).
The broker aggregates $2f+1$ of these attestations into a \emph{commit certificate} that it forwards back to all validators (\step{8}).
Upon verification of this certificate, a correct validator now trusts that a quorum of validators agree on the set of messages to deliver and proceeds to deliver these messages.

Correct validators may have missed some of the messages since the system may be asynchronous or since the orchestrating broker may be faulty.
To ensure the delivery of all correct messages by all correct validators, validators periodically communicate with each others in to: synchronize the digests of the batches that they have delivered (\step{11}), query the missing batches and certificates (\step{12}) such that they may eventually receive them and deliver them (\step{13}).

\para{Completion phase}
The completion phase is used to prove to clients the correct delivery of their messages, despite the asynchronous nature of the network, by a quorum of validators via the use of \emph{completion certificates}.
\sysname relies on \bcastname's completion certificates for forward progress.

A validator that delivers a batch---some messages may be excluded---also sends to the broker a signed attestation of the delivery (\step{9}).
The broker aggregates $2f+1$ of such attestations in order to send to each correct client a personalized completion certificate (\step{10}) that contains: (1) the digest of the delivered batch, (2) the indexes of the excluded messages in the batch, (3) the aggregated validator signatures, and (4) a \ac{mp} of inclusion of their message in the batch.

\subsection{\assignname for Asynchronous Identifier Assignment}
\label{sec:identifier-assignment}

\bcastname relies on \assignname~\cite{oracular-brb-disc22} to assign short identifiers to clients in an asynchronous and consensus-less manner.
Each client must execute \assignname before they may first broadcast a message using \bcastname: a client must include its \emph{assignment certificate} in the first message sent to a broker.
\assignname uses brokers to facilitate communication between clients and validators but we omit their trivial involvement in the following for clarity.
\assignname operates in two phases: an id must first be reserved before being assigned.
In short, (1) the reservation phase forces each validator to commit to (reserve) at most one id per client, while (2) the assignment phase forces the client to commit to (get assigned) at most one of the reserved ids.

Identifiers are composed of a \emph{domain} and an \emph{index}.
Each validator is responsible for a distinct domain.
A validator assigns a unique id by concatenating its domain with its index that it increments for each assignment request.

\para{Reservation phase}
A client initiates \assignname by sending its public keys to a random validator; the client retries with different validators until it receives a correct reply.
Upon receiving a request, a validator $v$ (1) reserves from its domain a new id $i$ to the requesting client, (2) associates the client's public keys to $i$, and (3) uses FIFO broadcast to inform all validators about $i$'s reservation.
The ordering of the FIFO broadcast matches that of the index.
Upon FIFO delivery of the reservation of $i$ by $v$, each validator replies back to the client with a BLS multi-signature attesting that $v$ reserved $i$ for that client.
The client gathers the first $f+1$ attestations into a \emph{reservation certificate} for id $i$.

\para{Assignment phase}
The client then broadcasts this reservation certificate to all validators to confirm its intent to be assigned this id.
In response, each correct validator checks that (1) this id is reserved but yet unassigned, that (2) the public keys have no other assigned id, and that (3) it did not sign a conflicting assignment attestation.
Correct validators then send to the client a BLS multi-signature attesting correct id assignment.
The aggregation of $2f+1$ of these signatures becomes an \emph{assignment certificate} that proves the unique assignment of this id.
Finally, \assignname completes by having the client broadcast this assignment certificate to all validators.

%% file: sections/4-currency.tex
\begin{table}[t]
    \small 
    \setlength{\tabcolsep}{4pt}
    \caption{\textbf{Field sizes of a payment received by a validator (in B).}
    \assignname and \bcastname can reduce a payment $8\times$ in size.}
    \label{tab:payment-msg-size}
    \begin{threeparttable}
    \begin{tabularx}{\columnwidth}{Xrrrrr}
        \toprule
        \bf Approach & \bf From & \bf To & \bf Amount & \bf Signature & \bf Total \\
        \midrule
        Strawman\tnote{a}                & 32 & 32 & 4--8 & 32 & 100--104 \\
        $+$ \assignname ids\tnote{b}   & \bf 4--8 & \bf 4--8 & 4--8 & 32 & 44--56 \\ 
        $+$ \bcastname BRB\tnote{c}      & \bf 4--8 & \bf 4--8 & 4--8 & \bf $\approx$0 & 12--24 \\
        \bottomrule
    \end{tabularx}
    \begin{tablenotes}
        \item[a] From and to are Ed25519 public keys~\cite{rfc8032_ed25519}, amount is a \ac{varint} of 4 to 8~B of length, signature uses Ed25519. \\
        \item[b] From and to are \acp{varint}, the first $2^{32}$ ids take 4~B (cf. \cref{sec:client-registration}). \\
        \item[c] The Ed25519 signature is replaced by a BLS12-381 multi-signature~\cite{rfc_bls_wip2020-09} that is aggregated before it reaches validators, hence the amortized signature footprint (cf. \cref{sec:bcast}).
    \end{tablenotes}
    \end{threeparttable}
\end{table}

\section{The \sysname Payment System}
\label{sec:currency}

In this section, we outline the main characteristics of our asynchronous and consensus-less payment system \sysname.
\sysname relies on \assignname to assign identifiers to clients (cf. \cref{sec:identifier-assignment}) and \bcastname to disseminate the client operations (cf. \cref{sec:bcast}).

\Cref{tab:payment-msg-size} summarizes the benefits of combining \assignname and \bcastname to greatly reduce the size of payment messages received by validators.
However, the payment format $\langle \text{from}, \text{to}, \text{amount}, \text{signature} \rangle$ as proposed in \cref{tab:payment-msg-size} is unfit for \sysname since such an operation modifies two accounts at once, which implies synchronization between user accounts, hence consensus.
We further detail how we adapt the payment format for an asynchronous consensus-less protocol.
\Cref{fig:system-overview} overviews the system and the payment protocol.

We first describe the asynchronous validator reconfiguration protocol that enables users to select which servers are validators (\cref{sec:reconfiguration-validator}).
We then detail how clients bootstrap their knowledge of the validator set in an efficient and eclipse-resilient manner (\cref{sec:client-bootstrap}) and how they register (\cref{sec:client-registration}) in order to perform operations.
Next, we detail asynchronous payment (\cref{sec:currency-payments}) and voting (\cref{sec:currency-votes}) protocols.
We finally discuss broker membership (\cref{sec:broker-membership}) and argue for \sysname's adaptivity regarding game-theoretical
incentives (\cref{sec:currency-incentives}).

The appendices contain correctness proofs of the reconfiguration (\cref{sec:reconfiguration-validator}), payment (\cref{sec:currency-payments}) and voting (\cref{sec:currency-votes}) modules in \cref{appx:reconfiguration_module}, \cref{appx:transaction_module} and \cref{appx:voting}, respectively.

\begin{figure}[t]
    \centering
    \includegraphics[width=\columnwidth]{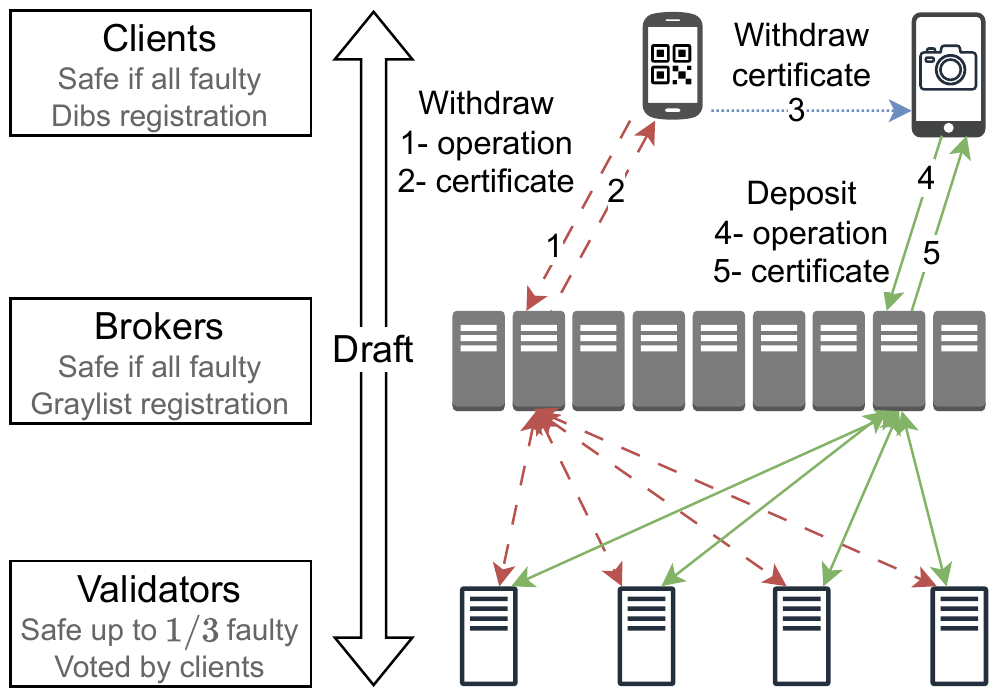}
    \caption{\textbf{\sysname system and payment protocol overview.}
    \sysname is composed of clients, validators, and brokers assisting validators, each with unique joining and resilience models.}
    \label{fig:system-overview}
\end{figure}

\subsection{User-Driven Validator Reconfiguration}
\label{sec:reconfiguration-validator}

\sysname lets users choose the set of validators running the system by pairing an asynchronous reconfiguration protocol with a novel asynchronous voting scheme.

\para{Voting for validators}
\sysname allows clients to select which validators get to run the system via its asynchronous voting mechanism (cf. \cref{sec:currency-votes}).
The product of a voting procedure is a \emph{voting certificate} proving that sufficiently many clients have supported the addition or removal of a validator.
Any change to the set of validators must be justified by an appropriate voting certificate.

\para{Problem definition}
Informally, \sysname's user-driven reconfiguration protocol ensures the following guarantees:
\begin{itemize}
    \item \emph{Join safety:} If a (correct or faulty) validator joins the system, the joining action is supported by a correct voting certificate.

    \item \emph{Leave safety:} If a correct validator leaves the system, then the validator requested to leave, or the leaving action is supported by a correct voting certificate.

    \item \emph{Join liveness:} If a correct validator requests to join with a correct voting certificate, the validator eventually joins.

    \item \emph{Leave liveness:} If a correct validator requests to leave, the validator eventually leaves.
    (Note that no voting certificate is required to leave.)

    \item \emph{Removal liveness:} If a correct process (client or server) obtains a voting certificate for the removal of a (correct or faulty) validator, the validator eventually leaves.
\end{itemize}
A correct validator joins (resp., leaves) the system \emph{only after} it has fully completed the joining (resp., leaving) subprotocol.
The fully-formal definition of the problem is relegated to \cref{appx:problem_definition}.
Besides the aforementioned reconfiguration-specific guarantees, \sysname also needs to ensure operation-specific guarantees (cf. \cref{sec:introduction}).

\para{Implementation overview}
We define a \emph{view} as a set of validators.
The genesis view is known by all processes, \eg it can be embedded in the code.
Each process captures in its \emph{current view} the current set of validators as seen from its perspective.
Due to the lack of a consensus-style agreement, current views of processes might differ.
Whenever a process learns about a new view, it updates its current view to reflect the newly obtained knowledge of the validator set.

\sysname orders all views into a \emph{totally-ordered} list: if a view $v_1$ precedes a view $v_2$ in the list, $v_2$ is more up-to-date than $v_1$.
Note that \sysname's reconfiguration protocol does not guarantee that all processes transit from one view to the same next view as such a behavior amounts to consensus.

\para{Operation liveness under dynamicity}
We assume that the list of all views is finitely long.\footnote{
    Even a (much simpler) regular register cannot be deterministically implemented in a reconfigurable manner if the membership changes infinitely often~\cite{baldoni2009implementing}.
}
Importantly, it is ensured that all processes, \ie both clients and servers, \emph{eventually converge} on the last view of the list.
Once the system stabilizes and all processes converge, every operation issued by a correct non-leaving client gets committed as the system's membership never changes again---there is ``enough time'' for protocol messages to be exchanged, leading to a commitment for each operation.
Note that \sysname cannot guarantee a commitment of operations issued by clients that leave as these clients might leave before ``reaching'' the final view of the list, which is the only one guaranteeing operation-commitments.

\para{Operation safety under dynamicity}
\sysname ensures safety both (1) within a \emph{single} view, and (2) across \emph{multiple} views.

To ensure safety within a view, \bcastname in \sysname (cf. \cref{sec:bcast}) relies on a standard technique of distributed computing: quorum intersection.
Namely, for an operation to be committed, a quorum of at least $2 / 3$ of the validators of a view must approve the operation.
Therefore, no two conflicting operations can be committed in the same view as the quorum intersection technique would imply that a correct validator has approved conflicting transactions, which a correct validator never does.

\sysname ``instantiates'' a more up-to-date view from a less up-to-date view.
In other words, if there exists a view $v_2$ in the list, a quorum of validators of some less up-to-date view $v_1$ has previously approved $v_2$.
To ensure safety across views, once a correct validator of $v_1$ approves $v_2$, the validator vouches to stop processing any operations intended for $v_1$.
Therefore, once any more up-to-date view is ``instantiated'' in \sysname, no operation intended for any less up-to-date view is ever committed.
\Cref{fig:reconfiguration_example} depicts an example of how safety across views is satisfied.

\Cref{appx:reconfiguration_module} contains the full details on the liveness and safety of operations under dynamicity.

\begin{figure}[h]
    \centering
    \includegraphics[width=\columnwidth]{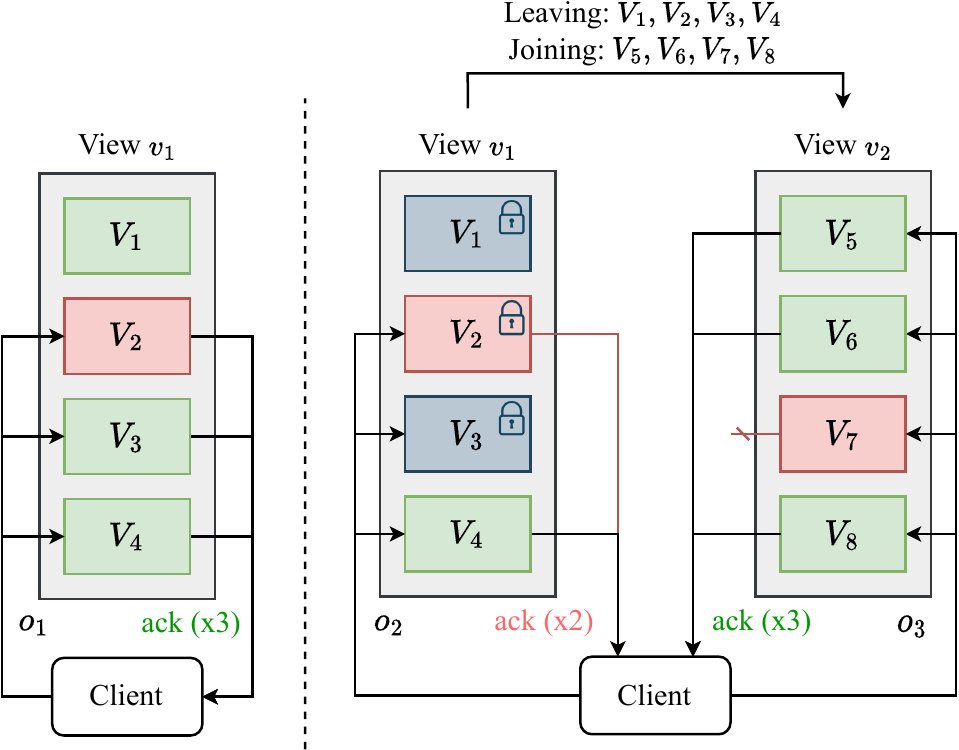}
    \caption{\textbf{Example of operation safety during reconfiguration.}
    Consider two views: (1) the genesis view $v_1$, with validators $V_1, V_2$, $V_3$ and $V_4$, and (2) a view $v_2$, with validators $V_5$, $V_6$, $V_7$ and $V_8$.
    All validators are correct, except for $V_2$ and $V_7$.
    Initially, a faulty client successfully commits its operation $o_1$ intended for $v_1$ (the left part of the figure).
    Then, the system reconfigures and view $v_2$ is ``produced'' (the right part of the figure).
    At this point, the client tries to commit two conflicting operations: $o_2$ intended for $v_1$, and $o_3$ intended for $v_2$.
    However, as a quorum of three validators of $v_1$ stopped processing operations intended for $v_1$, the client collects at most $2$ approvals for $o_2$: one from the faulty validator $V_2$, and one from the correct validator $V_4$ which did not approve the system's transition to $v_2$.
    As at least three approvals are needed, $o_2$ does not get committed, which preserves safety.}
    \label{fig:reconfiguration_example}
\end{figure}

\begin{table}[t]
    \centering
    \small 
    \setlength{\tabcolsep}{2.9pt}
    \caption{\textbf{\sysname client operations for payments and votes.}}
    \label{tab:client-operations}
    \begin{tabularx}{\columnwidth}{ccX}
        \toprule
        \bf Operation & \bf Parameter & \bf Description \\
        \midrule
        \formatmsg{Withdraw} & amount, to, epoch & Withdraw coins \\
        \formatmsg{Deposit} & withdraw certificate, & Deposit coins \\
            & deposit \ac{mp}, & \\
            & collect flag & \\
        \midrule
        \formatmsg{Support} & motion hash & Support motion \\
        \formatmsg{Abandon} & motion hash & Abandon supported motion \\
        \bottomrule
    \end{tabularx}
\end{table}

\subsection{Client Bootstrap}
\label{sec:client-bootstrap}

As discussed in \cref{sec:reconfiguration-validator}, the set of validators in \sysname is subject to churn.
As a result, before anything else, clients must learn the current validator set.
Doing so in \sysname is simple and secured since the genesis view is hard-coded in \sysname's client software, and thus known by every client, and since each subsequent view is signed by a quorum of a previous view.
A client verifies the authenticity of a learned view by verifying a chain of validations starting from the genesis.
This chain of validations of asynchronously-selected views ensures that a client operates free from eclipse attacks.

Since the authenticity of a view is publicly verifiable, a client can discover new views using numerous mechanisms.
For instance, views can be discovered via gossip, publicly known but untrusted discovery servers, or out of band.

\subsection{Client Registration}
\label{sec:client-registration}

Traditional cryptocurrencies require clients to attach (hashes of) their public keys to the operations they issue.
Instead, before \sysname clients can issue the operations listed in \cref{tab:client-operations} with \bcastname, they are first required to register their public keys into the system: a BLS12-381 key~\cite{rfc_bls_wip2020-09} for \bcastname's aggregation of multi-signatures, and an Ed25519 public key~\cite{rfc8032_ed25519} for \bcastname's fallback mechanism for slow clients.
BLS keys are challenged during registration to prevent rogue key attacks~\cite{proofs-possession-eurocrypt07}.
Upon registration, a client receives a 4--8~B numeric account id that it uses in its operations in lieu of its public key, thus greatly reducing the size of its messages.

\para{Protocol}
A client first executes \assignname (cf. \cref{sec:identifier-assignment}) to obtain an assignment certificate for a specific id, of 4 to 8~B of length, in an asynchronous manner.
Clients broadcast their assignment certificate to all validators in the last step of \assignname.

Upon receipt of a correct assignment certificate, a validator locally allocates an account for the given id.
As shown in \cref{tab:account-state}, the size of the state of an account is bounded by a constant for better scalability.
We describe the storage impact of payments and votes in their respective sections.

\para{Sybil resilience}
Attackers can create Sybil clients~\cite{douceur_sybil_iptps02} to pollute validators' storage and to frontrun the registration of correct clients.
Various schemes can defend against such attacks, such as \ac{pow}~\cite{dwork_pricing_crypto92,back_hashcash_2002,rem_mining_usenixsec_17,borisov_computational_p2p06}, \aclp{vdf}~\cite{vdf_crypto18}, or proof of personhood~\cite{popersonhood_eurospw17}.

\sysname's design is agnostic to the chosen Sybil-resilience mechanism.
For instance, users may be required to reveal their identities upon registration, \eg for a \ac{cbdc} deployment.
Alternatively, clients may be asked to perform a short \ac{pow} challenge to be allowed to register.
Even a challenge taking only 1~s to complete requires $10^{11}$ years of computation to fill the $2^{64}$ bits of account id space.

\begin{table}[t]
    \centering
    \small 
    \setlength{\tabcolsep}{1.3pt}
    \caption{\textbf{State stored by a validator per client account (in~B).}
    The size of an account state is bounded by a constant.}
    \label{tab:account-state}
    \begin{tabularx}{\columnwidth}{lcl}
        \toprule
        \bf Name & \bf Size & \bf Description \\
        \midrule
        Identifier  & 4--8 & Unique account id given by \assignname (\ac{varint}) \\
        Seq. num.   & 4--8 & Monotonic counter for \bcastname (\ac{varint}) \\
        \midrule
        EdDSA key  & 32 & Signing Ed25519~\cite{rfc8032_ed25519} public key \\
        BLS key & 48 & Multi-signing BLS12-381~\cite{rfc_bls_wip2020-09} public key \\
        \midrule
        Balance  & 4--8 & Tally of spendable coins (\ac{varint}) \\
        Deposits & $\langle \text{4--8}, \text{32} \rangle$ & Deposit epoch ($\langle \text{\ac{varint}}, \text{\ac{mr}} \rangle$) (cf. \cref{sec:currency-payments}) \\
        Votes & $\{ 32 \}$ & Set of supported motion hashes \\
            & & ~~~(using BLAKE3~\cite{blake3_report2021-11}) (cf. \cref{sec:currency-votes}) \\
        \bottomrule
    \end{tabularx}
\end{table}

\subsection{Payment}
\label{sec:currency-payments}

The asynchronous and consensus-less nature of the broadcast poses unique challenges when employed for payments.
\sysname splits each payment in two operations: a \formatmsg{Withdraw} issued by the payer, and a \formatmsg{Deposit} by the payee.
In short, a \formatmsg{Withdraw} is a commitment performed by the payer to transfer coins to the payee, while a \formatmsg{Deposit} makes the committed coins usable by the payee.

\para{Protocol}
As depicted in \cref{fig:system-overview}, (1) the payer client first issues a \formatmsg{Withdraw} operation to the validators, via \bcastname, to signal its intent to transfer some coins to the payee client.
(2) The payer obtains in return a \emph{withdrawal certificate} attesting that the payer withdrew coins from its balance and designated them for the payee. 
(3) The payer then forwards the withdrawal certificate to the payee out-of-band, \eg via remote messaging, bluetooth, QR code. 
Finally, (4) the payee uses \bcastname to send a \formatmsg{Deposit} operation to the validators alongside the withdrawal certificate and (5) receives in return a \emph{deposit certificate} that finalizes the payment.

\para{Payment safety}
The safety of a payment is violated if a faulty client double spends a coin~\cite{bitcoin}, \ie uses the same coin in two withdrawals or two deposits.

Since \bcastname ensures a source order of operations via sequence numbers, a double withdrawal amounts to two \formatmsg{Withdraw} operations using the same sequence number.
\bcastname's consistency property prevents such equivocations.

Preventing double deposits amounts to preventing clients from using the same withdrawal certificate twice.
This constraint is similar to Bitcoin's way of preventing double spends whereby validators are expected to store all spendable coins, \ie UTXOs.
However, doing so results in an unbounded storage footprint for validators.
Instead, to bound by a constant the account state held by validators (cf. \cref{tab:account-state}), \sysname clients provide a proof in their \formatmsg{Deposit} operation that a withdrawal certificate has never been used, as explained further.

\para{Storage efficiency for validators}\label{sec:currency-gc-payments}
Each client maintains an \ac{mpt} of its used withdrawal certificates.
Validators store and update the root of each client's \ac{mpt} to check the suitability of the withdrawal certificates.
The storage footprint per account for validators is thus bounded by a constant as only one \ac{mr} is stored per account (cf. \cref{tab:account-state}).

To deposit a coin, a client must attach a fitting \ac{mp} of exclusion (cf. \cref{sec:background-merkle-tree}), in the \code{deposit MProof} parameter, from its deposits to prove that the sent withdrawal certificate has never been used before.
Validators verify the received \ac{mp} against the client's \ac{mr} that they store.
If the \ac{mp} is correct, the stored \ac{mr} is updated by adding the hash of the withdrawal certificate as a new leaf of the partial \ac{mpt} encapsulated in the received \ac{mp} of exclusion.
Similarly, the client adds the withdrawal certificate to its own \ac{mpt} and updates the associated \ac{mr}.
Thanks to this proof system, a client cannot reuse the same withdrawal certificate for multiple deposits since only a single \ac{mp} of exclusion can be issued for a given certificate.

\para{Storage efficiency for clients}
Optionally, \sysname enables clients to garbage collect old withdrawal certificates and reset the \ac{mpt} used to store the withdrawal certificates deposited on their account.
This feature saves space on the client's device but is not safe in asynchrony as it may lead to the loss of coins for the client---it never leads to double spends.
A client should therefore make sure it is aware of all incoming payments before garbage collecting.

As such, deposits are organized in \emph{epochs}.
A payer must specify the payee's \code{epoch} in the \formatmsg{Withdraw} operation (cf. \cref{tab:client-operations}).
In turn, a payee can only deposit withdrawal certificates with an epoch that matches the current epoch of the payee as stored by the validators (cf. \cref{tab:account-state}).
This check prevents double deposits across epochs but may lead to withdrawal certificates that cannot be deposited if their epoch is incorrect, which may happen because of asynchrony.

A client sets the \code{collect flag} of a \formatmsg{Deposit} operation to inform validators of the garbage collection; validators in turn increment the stored epoch counter and empty the associated deposit \ac{mr}.
Infrequent users of \sysname can garbage collect their \ac{mpt} after each deposit. 
This garbage collection also reduces the network complexity that would be otherwise required to convey larger \acp{mp} of exclusion.

\begin{figure}[t]
    \centering
    \includegraphics[width=\columnwidth]{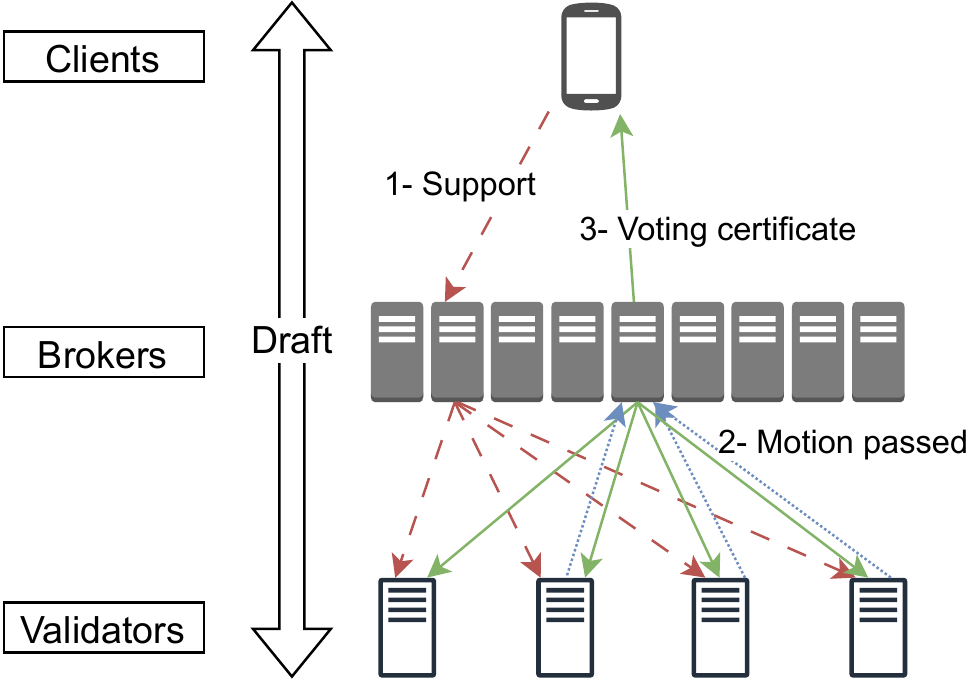}
    \caption{\textbf{\sysname's voting protocol.}
    Clients issue \formatmsg{Support} operations to indicate a vote for a motion.
    A motion eventually passes when the summed balance of all supporting clients' accounts reaches more than 50\% of the total supply.}
    \label{fig:voting}
\end{figure}

\subsection{Voting}
\label{sec:currency-votes}

\sysname clients, and \emph{only} clients, vote for \emph{motions} that represent any enforceable resolution that can be voted upon.
For instance, a motion may express the wish to update the incentive policy, akin to Tezos~\cite{tezos-whitepaper}, or to update the validator set.
Clients vote for a motion by issuing a \formatmsg{Support} operation and can retract a previous vote via an \formatmsg{Abandon} operation at any time.
These operations are issued and disseminated to the validators via \bcastname (in a ``closed-box'' manner).
Validators store the set of supported motions for each account (cf. \cref{tab:account-state}) to determine when a motion has passed.
To bound the storage footprint of each account by a constant, a client can only support a limited number of motions at any given time, \eg 5 in our implementation.

Since \sysname operates in an asynchronous setting, a motion is either \emph{``passed''} or \emph{``not yet passed''} but it never ``fails to pass''.
The latter requires a terminating dissemination primitive which is incompatible with \bcastname's asynchronous and consensus-less dissemination.
A motion passes when its tally crosses a given threshold, after which a \emph{voting certificate} is signed by a quorum of validators.

In our implementation, the weight of a vote is proportional to the balance of the voter.
A motion \emph{may pass}, from the perspective of an individual validator, when the sum of the balances of all the accounts supporting the motion crosses 50\% of the total money supply. 
This mechanism contrasts with stake-based approaches~\cite{algorand_sosp17,ouroboros_crypto17,snow-white-fc19,async-pos-sss21} that require clients to lock away the money with which they vote, forcing users to choose between payments and voting power.
\sysname's voting can easily adapt to a democratic personhood-based voting given a fitting Sybil-resilience mechanism upon client registration~\cite{popersonhood_eurospw17}.

\para{Protocol}
As depicted in \cref{fig:voting}, (1) a client first uses \bcastname to issue a \formatmsg{Support} operation for a motion $m$ to the validators.
Upon delivery, \ie commitment, each validator (2a) adds $m$'s hash to the set of supported motions by the client's account if space permits.
Each validator then 
(2b) updates $m$'s tally with the account balance and (2c) checks whether the tally crosses the 50\% threshold.
If so, each validator (2d) emits a signed \formatmsg{motion-passed} message for $m$.
(3) A voting certificate for a motion $m$ is aggregated
by any process from $2f+1$ such messages, and is broadcast to all processes.

Upon delivery of an \formatmsg{Abandon} for $m$, validators revert the effect of the related \formatmsg{Support} on their state if they have not yet signed a \formatmsg{motion-passed} message for $m$.
Upon delivery of a \formatmsg{Withdraw} or \formatmsg{Deposit}, steps (b)--(d) are repeated for all of the client's supported motions.

\para{Voting guarantees}
Regarding liveness, a motion is only guaranteed to pass if it is supported by more than 50\% of the total money supply \emph{forever}.
Due to asynchronous nature of the network, there is no guarantee that correct validators notice that a motion has sufficient support at a time unless it is supported forever.
In practice, a motion passes if it is supported for a sufficiently long period of time to ensure that a quorum of correct validators observe its support.

As for safety, a motion passes only if more than 50\% of the total money supply supported the motion \emph{at some point in time} after the motion was issued.
In the same way that payment safety must prevent the same coin from being used in two conflicting \formatmsg{Withdraw} or \formatmsg{Deposit} operations, \sysname must also prevent the same coin from being used in two conflicting \formatmsg{Support} operations.
Consider the example where
(1) Alice supports motion $m$,
(2) withdraws her coin for Bob, who, in turn,
(3) deposits the coin, and
(4) supports $m$. 
In this example, to guarantee safety, the voting power of that coin towards $m$'s support must not increase despite the change of hands.
\sysname prevents such scenarios by ensuring that, before a \formatmsg{Deposit} operation is delivered, a validator always delivers its corresponding \formatmsg{Withdraw} operation using the withdraw certificate in the deposit.
In the example scenario, a validator first decrements $m$'s tally (Alice's withdrawal) and then increments it (Bob's deposit), thus ensuring voting safety.
\Cref{appx:voting} contains full details on \sysname's voting mechanism.

\subsection{Broker Membership}
\label{sec:broker-membership}

Brokers do not impact the safety of \sysname.
As such, they do not have to be trusted by neither validators nor clients but can instead remain untrusted as any process outside of the system.
This lack of trust allows brokers to easily join and leave the system without requiring a well-defined protocol unlike validators (\cref{sec:reconfiguration-validator}) and clients (\cref{sec:client-registration}).
Therefore, brokers are never ``whitelisted'' by validators but instead remain on a \emph{``graylist''} and can be removed from it at any time in case of misbehavior, \eg spamming validators, not completing a \bcastname broadcast, building inefficient batches in \bcastname.

In this paper, we pair a simple fixed fee policy to reward correct brokers with a staking mechanism to render brokers accountable should they harm \sysname's liveness.
Alternatively, one could design a broker reputation system among validators or enforce that a joining validator must add brokers to the system with tied reputation.

\subsection{Incentives}
\label{sec:currency-incentives}

Thanks to its model, \sysname untangles the design of possibly complex game-theoretic incentives~\cite{sok_game_theory_arxiv20,roughgarden_fee_ec21,pos_compounding_fc19} from the inner workings of its protocols.
As a result, \sysname can easily be adapted to a wide-range of incentive policies and thus to equally numerous usecases, \eg \acp{cbdc}, global payment systems akin to Libra.
\sysname even enables users to vote for the incentive policy of their choice.

We note that the incentive policy may be simplified thanks to \sysname's performance.
In classic low-throughput payment systems, high-fee payments are naturally prioritized over those with low fees, which leads to a rapid fee escalation in the face of congestion.
On the other hand, we do not expect \sysname to face such congestion since it can sustain 1~M withdraws/s (cf. \cref{sec:eval-rq1}), thus congestion should likely not be accounted for in the fee design.

%% file: sections/6-evaluation.tex
\section{Evaluation}
\label{sec:evaluation}

We evaluate \sysname focusing on the following questions:

\noindent
(\cref{sec:eval-rq1}) What workload can \sysname sustain?

\noindent
(\cref{sec:eval-rq2}) How much are \sysname's internal components contributing to its scalability?

\noindent
(\cref{sec:eval-rq3}) How costly is client authentication in \sysname?

\noindent
We first detail our implementation and experimental setup.

\subsection{Implementation Aspects}
\label{sec:implementation}

We implemented \sysname and its components in 26,000 Rust LOCs~\cite{code_tokei}.

\para{Cryptography}
Our implementation uses the following Rust libraries for secure and efficient cryptographic primitives:
\newline\noindent (i)
BLAKE3 hash function~\cite{blake3_report2021-11} official library~\cite{code_blake3};
\newline\noindent (ii)
Ed25519 signature scheme~\cite{curve25519_pkc06,rfc8032_ed25519} libraries~\cite{code_ed25519,code_x25519};
\newline\noindent (iii)
BLS12-381 multi-signature~\cite{bls_jcrypt04,rfc_bls_wip2020-09} C bindings~\cite{code_bls};
\newline\noindent (iv)
ChaCha20Poly1305 authenticated encryption~\cite{poly1305_fse05,chacha_sacs08,rfc8439_chacha20poly1305} library~\cite{code_chacha20poly1305} for authenticated channels between validators (for \bcastname's totality phase and \assignname's FIFO broadcast).

\subsection{Experimental Setup}
\label{sec:eval-setup}

We compare \sysname's performance with Algorand~\cite{algorand_sosp17}, Quorum~\cite{quorum}, and FastPay~\cite{fastpay_aft20} on global AWS deployments.

\para{Bucketed storage}
\sysname uses an authenticated in-memory \ac{kv} store with multi-threaded parallel memory accesses.
The store partitions the key space into buckets that contain mutually exclusive subsets of the \ac{kv} pairs.
As such, buckets may be accessed concurrently by different threads for increased throughput.

The results reported for \sysname in RQ1 and RQ3 use as many buckets as there are hardware threads.
RQ2 in \cref{sec:eval-rq2} studies the impact of this store on throughput.
In the future, we plan to rely on RainBlock~\cite{rainblock_atc21} as a fully-fledged \ac{kv} store.

\para{Baselines}
Algorand~\cite{algorand_sosp17} is a permissionless blockchain that supports a dynamic set of validators using proof of stake and weighted sortition.
Blocks are disseminated by gossip and agreed upon using the \textit{BA$\star$} consensus.
We uniformly distribute stake among a static set of validators to obtain a uniformly random sortition for fair comparison with other systems.
Algorand claims a throughput of 1,000--46,000~tx/s with a lower block latency of 2.5--4.5~s.

Quorum~\cite{quorum} is a permissioned blockchain derived from Ethereum~\cite{wood_ethereum_2014}.
Quorum uses IBFT~\cite{ibft_spec,ibft_blockchain19}, based on PBFT~\cite{pbft_osdi99}, for block agreement among a dynamic set of validators chosen via consensus.
We use a static set of validators encoded in the chain's first block.
IBFT achieves 200~tx/s for a latency of 8--10~s on 16 LAN nodes~\cite{quorumperf_blockchain21}.

FastPay~\cite{fastpay_aft20} is a payment system based on consistent broadcast with reduced latency and increased throughput compared to a blockchain.
However, FastPay is not a stand-alone system.
FastPay is meant to complement a less-performant underlying system that provides totality to its payments, \eg via \ac{brb} as in \sysname and Astro~\cite{astro_dsn20}, or via consensus as in Algorand and Quorum.
Despite not being equivalent to \sysname, we evaluate FastPay as it is a state-of-the-art consensus-less payment system easy to deploy.
We use a static set of validators, TCP due to geo-distribution, and 64 FastPay shards per instance for parallelization.

\para{Hardware \& software}
We use a mix of \code{c5a.16xlarge} and \code{c5ad.16xlarge} AWS EC2 KVM virtual machines each with: an AMD EPYC 7R32 CPU with 64 hardware threads (32 cores) running at 2.8~GHz, 128~GB of RAM, and a 20~Gb/s network card.
We measured 4.5~Gb/s of effective bandwidth with \code{iperf}~\cite{iperf} between two intra-region instances.
Instances run an Ubuntu 20.04 LTS distribution with a Linux 5.11 kernel.
\sysname and FastPay are compiled with \code{rustc} 1.59.0 and 1.65.0, respectively, both in release mode.
Algorand and Quorum are compiled with \code{go} 1.14.3 and 1.13.8, respectively, and are deployed using Diablo~\cite{diablo-eurosys23}.

\para{Geo-distribution}
We use up to 64 instances for validators spread among 20 regions.
We set:
(i) 2 instances in Paris and Tokyo;
(ii) 3 instances in Ohio, Oregon, North Virginia, North California, S\~ao Paulo, Ireland, Stockholm, Milan, Frankfurt, Sydney, Singapore and Seoul;
(iii) 4 instances in Cape Town, Hong Kong, Mumbai, Canada, London and Bahrain.

\para{Algorand \& Quorum deployments}
We add an instance per region to emulate clients generating workload.
Clients in a region access the validators of that region for optimal performances.
Workloads are pre-generated and pre-signed by clients to not bottleneck the system in-benchmark.

\para{FastPay deployments}
We add a \code{c6i.4xlarge} instance in Paris and in Tokyo to emulate clients generating workload.

\begin{figure}[tb]
    \centering
    \includegraphics[width=\columnwidth]{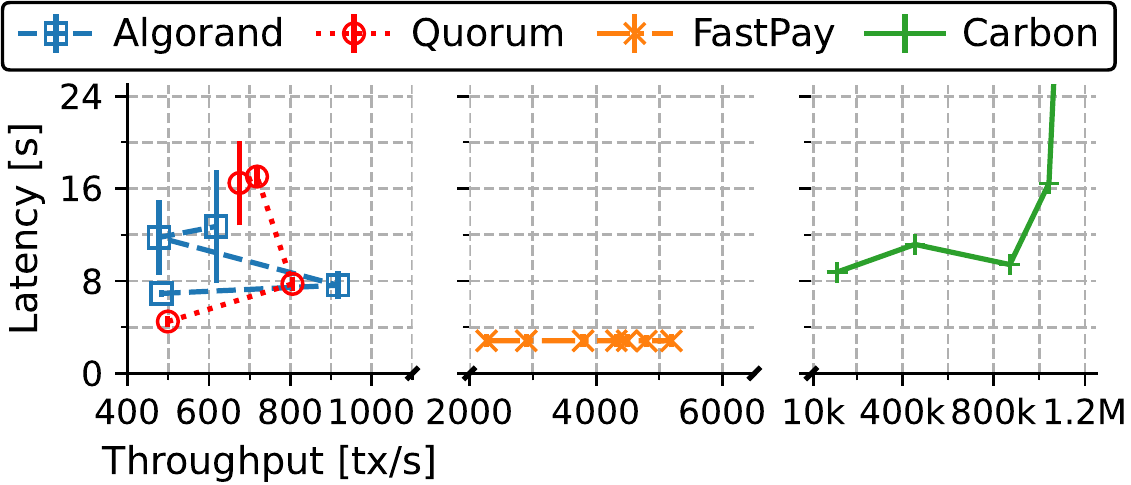}
    \caption{\textbf{Load handling of \sysname and its baselines.}
    \sysname sustains the delivery of 1M~tx/s without contention.
    With an equivalent latency of 8--15~s, the throughput of Algorand and Quorum remains below 1,000~tx/s, while FastPay reaches up to 5,200~tx/s with a latency consistently below 3~s.}
    \label{fig:eval-throughput-latency}
\end{figure}

\para{\sysname deployments}
In addition to validators, we also deploy brokers for \sysname.
We deploy two kinds of brokers to fulfill the throughput potential of \sysname: regular brokers and load brokers.
Regular brokers interact normally, per specification, with emulated clients on separated instances.
Regular brokers are used to accurately measure payment latencies.
To avoid the cost of deploying hundreds of thousands of clients, we use load brokers that pre-generate batches of payments, and hence single-handedly simulate many clients.
Load brokers are only used for throughput and do not bias the reported latencies.
Since brokers are permissionless, the use of load brokers lets us cheaply simulate many brokers and clients with no impact on the failure model.

We deploy an equal number of load brokers and validators, each on their dedicated instance, on each region listed above.
Two regions differ:
(1) Frankfurt contains 1 regular broker and 4 client emulators, for the purpose of latency measurements, instead of 3 load brokers;
(2) Stockholm contains 5 load brokers instead of 3 due to AWS restrictions.

We parameterize \bcastname as follows.
A broker waits for \formatmsg{Submission} messages for 1~s (\step{1} in \cref{fig:bcast-phases}), \ie a broker builds a batch every second.
Once a broker has sent all \formatmsg{Inclusion} messages (\step{2}), it waits for the corresponding \formatmsg{Reduction} messages for 1~s (\step{3}).
Batches contain 50,000 operations.

\para{Plots}
Each run lasts at least 120~s once warm-up and cool-down times are removed.
Since a withdrawal in \sysname effectively transfers coins from the payer to the payee (cf. \cref{sec:currency-payments}), we measure the throughput of \sysname as the number of withdrawal operations processed per time unit by the servers.
The latency reported for all systems represents the end-to-end latency from operation issuance by a client to the receipt of the acknowledgment, \eg the withdrawal certificate in \sysname (steps 1 and 2 in \cref{fig:system-overview}).

\begin{table}[t]
    \centering
    \small 
    \setlength{\tabcolsep}{3.5pt}
    \caption{\textbf{Load handling of Algorand and Quorum.}
    Each cell shows mean $\pm$ standard deviation of 5 runs.
    The ideal throughput is the closest to its requested workload.
    Both approaches can efficiently sustain a workload of 500~req/s but only handle less than 50\% of the requested workload for 1,500 and 2,000~req/s.}
    \label{tab:eval-baselines}
    \begin{tabular}{ccccc}
        \toprule
          & \multicolumn{4}{c}{\bf Input workload [req/s]} \\
          & \bf 500 & \bf 1,000 & \bf 1,500  & \bf 2,000 \\
        \midrule
          & \multicolumn{4}{c}{\bf Throughput [tx/s]} \\
        Algorand     & 483\,$\pm$\,32.0 & 917\,$\pm$\,36.4 & 477\,$\pm$\,170.5 & 618\,$\pm$\,54.2 \\
        Quorum       & 499\,$\pm$\,0.5 & 805\,$\pm$\,44.3 & 718\,$\pm$\,18.4 & 675\,$\pm$\,174.3 \\
        \midrule
          & \multicolumn{4}{c}{\bf Latency [s]} \\
        Algorand     & 6.92\,$\pm$\,0.34 & 7.63\,$\pm$\,1.21 & 11.77\,$\pm$\,3.21 & 12.73\,$\pm$\,4.89 \\
        Quorum       & 4.49\,$\pm$\,0.46 & 7.75\,$\pm$\,0.63 & 17.04\,$\pm$\,0.95 & 16.49\,$\pm$\,3.67 \\
        \bottomrule
    \end{tabular}
\end{table}

\subsection{RQ1 -- Load Handling}
\label{sec:eval-rq1}

We expose \sysname, Algorand, Quorum and FastPay to growing workloads until contention is reached and performances are degraded.
We ran Algorand, Quorum and FastPay five times and \sysname once, due to deployment cost.
\cref{fig:eval-throughput-latency} depicts the mean throughput and mean latency for each system under different workloads, while \cref{tab:eval-baselines} exhibits exact values for Algorand and Quorum.

\sysname maintains a throughput of 1M~tx/s before reaching contention.
\sysname's latency in steady state is capped at 11.1~s, making it a suitable daily payment system.

On the other hand, neither Algorand nor Quorum cross 1k tx/s; 
while FastPay manages to reach 5,200 tx/s with consistent latencies of 2.8~s since it drops requests when congested.
\sysname's throughput dwarfs theirs by 3 orders of magnitude for equivalent latencies.
We explain the relatively poor performance of Algorand and Quorum by
(1) the need for most validators to verify all client signatures, and
(2) the use of unscalable consensus protocols, $BA\star$ and IBFT.

\subsection{RQ2 -- \sysname's Design for Scalability}
\label{sec:eval-rq2}

We investigate the performance impact of \sysname's internal designs and what makes \sysname scale to greater number of validators for greater resilience thresholds.
We compare FastPay's scalability to \sysname's by deploying FastPay and three versions of \sysname with key features disabled:

\noindent
(i) \code{All} is the default version of \sysname;

\noindent
(ii) \code{No-Agg} bypasses and nullifies the gains of aggregate signatures in \bcastname, every client uses only Ed25519 signatures that validators have to verify without aggregation, as would happen if brokers were to only send non-reduced batches;

\noindent
(iii) \code{No-Buckets} is \code{No-Agg} with a single memory bucket in the \ac{kv} store which cancels any I/O parallelism benefit.

We deploy these versions of \sysname on a varying number $n$ of validators and brokers, with $n \in \{16, 32, 64\}$, each on its dedicated instance.
For $n = 16$, compared to the list in \cref{sec:eval-setup}, we deploy a pair of validator and broker in each region except Cape Town, S\~ao Paulo, Seoul and Tokyo.
For $n = 32$, we deploy: (i) a pair of validator and broker in Tokyo, Seoul, Singapore, Sydney, Frankfurt, Ireland, Paris and Milan; (ii) two pairs in Ohio, North Virginia, North California, Oregon, Cape Town, Hong Kong, Mumbai, Canada, London, Stockholm, Bahrain and S\~ao Paulo.

\begin{figure}[tb]
    \centering
    \includegraphics[width=\columnwidth]{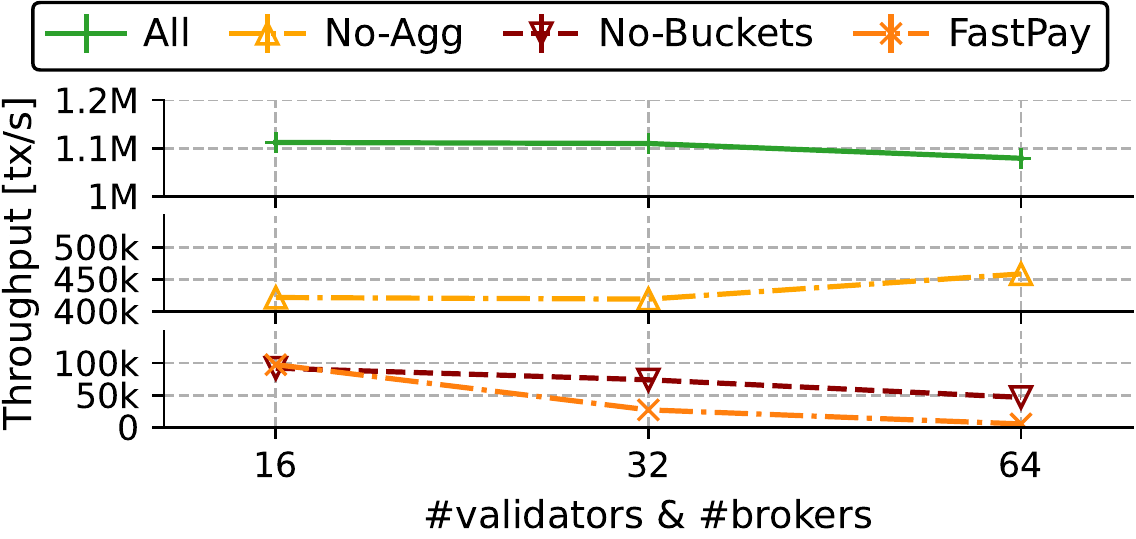}
    \caption{\textbf{Scalability of \sysname with various components deactivated, and of FastPay.}
    \sysname performs far better with (\code{All}) than without \bcastname's aggregation of signatures (\code{No-Agg}) or without multi-threaded memory accesses (\code{No-Buckets}).
    While \sysname's throughput remains constant from 16 to 64 validators, FastPay's throughput drops as the system scales.}
    \label{fig:eval-throughput-components}
    \vspace{-0.3em}
\end{figure}

\Cref{fig:eval-throughput-components} depicts the throughput observed in a run for each version of \sysname.
We first observe that both signature aggregation and parallel memory accesses hugely impact the throughput of the system.
Compared to the default version of \sysname, disabling signature aggregation reduces the throughput $2.6\times$ and disabling multi-threaded memory accesses reduces throughput yet $9.8\times$ more.
These performance gaps demonstrates the need to optimize both CPU and memory accesses for high-throughput payment systems.

Regarding \sysname's scalability, even with different number of validators and brokers, the throughput of \code{All} and \code{No-Agg} remains stable ($\pm10\%$) at 1.1M~tx/s and 420k~tx/s, respectively.
The throughput of \code{No-Buckets} is halved from a setup with $n=16$ (92k~tx/s) to $n=64$ (47k~tx/s).

The advantages of \sysname's design compared to that of the consensus-less FastPay are best observed with $n=16$.
With $n=16$, FastPay performs as well as \sysname with signature aggregation and I/O parallelization deactivated (\code{No-Buckets}).
The throughput of FastPay drops $3.6\times$ from a system with $n=16$ (97k~tx/s) to $n=32$ (27k~tx/s), and another $5.2\times$ from a system with $n=32$ to $n=64$ (5k~tx/s).

\begin{figure}[tb]
  \centering
  \includegraphics[width=\columnwidth]{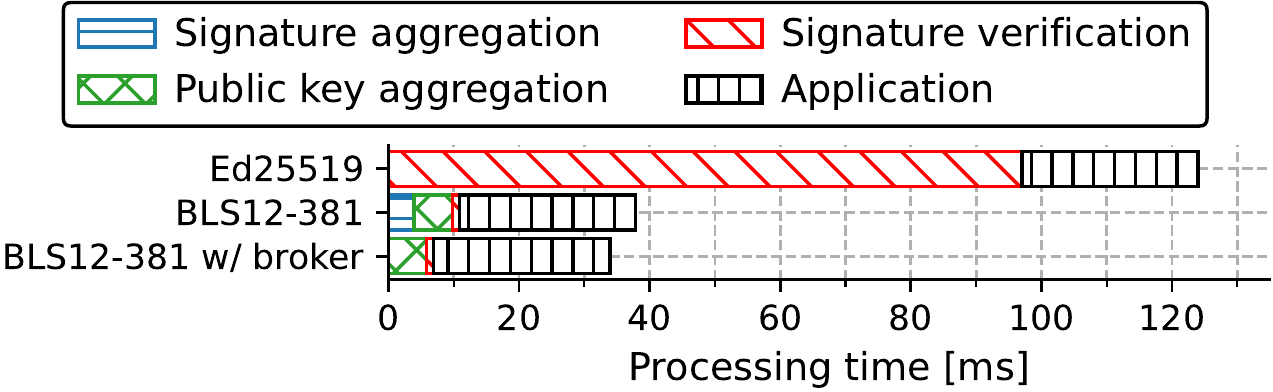}
  \caption{\textbf{Validator processing time of a large payment batch.}
  Verifying Ed25519 signatures is the main CPU utilization bottleneck.
  With BLS12-381 multi-signatures, most of \sysname's runtime is spent on application-level processing, \eg payment validity checks, memory accesses.}
  \label{fig:message-processing-time}
\end{figure}

\subsection{RQ3 -- Client Authentication Overhead}
\label{sec:eval-rq3}

Lastly, we evaluate the benefits of \bcastname's signature aggregation on the processing time experienced by validators in \sysname.
We compare \bcastname's BLS multi-signature scheme, with and without signature aggregation done by the brokers, with its fallback mechanism utilizing the efficient but non-aggregatable Ed25519 scheme.

\Cref{fig:message-processing-time} depicts the runtimes required for a validator to process a batch of 50,000 withdrawals, averaged over 40 runs.
Using Ed25519, validators spend the majority (78\%) of their processing time for the purpose of client authentication.
This number is reduced to 29\% for the BLS scheme and 20\% for the BLS scheme with signature aggregation performed by the brokers as in \bcastname.
\bcastname's scheme therefore reduces by $3.7\times$ the cost of client authentication for \sysname validators.

%% file: sections/8-conclusion.tex
\section{Conclusion}
\label{sec:conclusion}

\sysname is an asynchronous and consensus-less payment system operated by a dynamic set of validators.
Any client may register into the system to issue payments and votes to the validators.
Clients are protected against eclipse attacks thanks to the asynchronous reconfiguration mechanism used for validators.
\sysname is designed to incur minimal storage footprint on both validators and clients.
As conveyed by our evaluation, \sysname exhibits a throughput of 1M tx/s with latencies of 11~s.
To reach such performances, \sysname leverages a \acl{brb} that is able to uniquely exploit permissionless and untrusted brokers to greatly reduce the processing time due to client authentication on validators.

In the future, we plan to shard \bcastname to execute parallel protocol instances, increasing throughput further.
We also plan to apply the concept of trustless and permissionless brokers to other 
Byzantine total order broadcast with the aim of accelerating general purpose smart contracts.
Finally we plan to explore broker-assisted privacy-preserving techniques for \sysname operations~\cite{zerocash_sp14,zkledger_nsdi18,prifi_popets20,dandelion_journal18}. 

%% file: sections/9-acknowledgments.tex
\section*{Acknowledgments}

We thank the anonymous reviewers and Antoine Murat for their helpful feedback, as well as Athanasios Xygkis and Alberto Sonnino for their deployment scripts.
This work has been supported in part by the Interchain Foundation, the Hasler Foundation (21084), and Innosuisse (46752.1 IP-ICT).

%% file: appendix/introduction.tex
\section{Overview} \label{appx:introduction}

This appendix defines the problem \sysname solves, its implementation and formal proofs of correctness.

\para{Problem overview}
In its core, \sysname solves the asset-transfer problem~\cite{bitcoin,consensus_number_crypto_dn21,astro_dsn20,wood_ethereum_2014,kuznetsov2021permissionless,auvolat2020money} among its users.
A user of \sysname is able to issue a payment transaction transferring its money to another user, thus decreasing its balance and increasing the balance of the receiving user.
In a nutshell, \sysname guarantees the following two properties:
\begin{compactitem}
    \item \emph{Liveness:} A payment issued by a user is eventually processed, thus decreasing the balance of the issuer and increasing the balance of the receiver.
    
    \item \emph{Safety:} No user can successfully issue multiple payments using the ``same'' money.
\end{compactitem}
The violation of the safety property is traditionally known as \emph{double-spending}~\cite{bitcoin} and represents the core problem in implementing a cryptocurrency.

Since the goal of \sysname is to serve millions of users in an efficient way, users themselves cannot be responsible for processing payments.
To this end, we introduce \emph{servers}: members of the system that are actually responsible for processing payments issued by \sysname's users.
Since \sysname is designed to be a long-lived BFT system, set of servers that actually process transactions must change over time (e.g., failed servers should be replaced with new ones).
Thus, \sysname supports \emph{reconfiguration}: the feature of changing the set of servers running the system while still processing payments.

Importantly, the set of servers running the system is \emph{chosen} by users of \sysname.
Specifically, users select which servers run \sysname through an asynchronous, balance-based \emph{voting} mechanism: a user can support ``candidacy'' of a server for a spot in the ``running set'' by issuing a special vote transaction.
A server eventually joins the set of running servers if money owned by users that have supported its candidacy accounts for more than half of the total amount of money in \sysname.

\para{Inspiration}
\sysname is influenced by a variety of previously developed ideas.
This paragraph briefly presents the work which has inspired the development of \sysname.

The problem of a cryptocurrency rose to prominence with Bitcoin~\cite{bitcoin}.
Bitcoin solves the problem by ensuring that its users agree on a sequence of processed transactions (i.e., payments), i.e., all transactions are totally-ordered.
Since then, many protocols have followed the same approach in solving the problem~\cite{bitcoin-ng_nsdi16,algorand_sosp17,stellar_sosp19,omniledger_sp18,wood_ethereum_2014,buchman2018latest}.

Only recently has it been shown that total order of transactions is not necessary to solve the problem of asset-transfer~\cite{consensus_number_crypto_dn21}, which lies at the core of any cryptocurrency.
Consequently, the reliable broadcast primitive~\cite{cachin2011introduction}, which can be implemented in a completely asynchronous distributed system, suffices for asset-transfer~\cite{astro_dsn20,auvolat2020money}.
\sysname builds upon this approach by adopting the reliable broadcast primitive for processing transactions.

On the other hand, reconfiguration of distributed systems have been studied in both crash-stop~\cite{aguilera2011dynamic,alchieri2018efficient,jehl2015smartmerge,gafni2015elastic} and Byzantine~\cite{async_byz_reconf_disc20,guerraoui2020dynamic} failure model.
Since \sysname solves the cryptocurrency problem, which implies that \sysname assumes the Byzantine failure model, the reconfiguration logic of \sysname relies on~\cite{async_byz_reconf_disc20,guerraoui2020dynamic}.
More concretely, \sysname's reconfiguration mechanism is identical to the one of \textsc{dbrb}~\cite{guerraoui2020dynamic}.

\para{Roadmap}
We start by introducing the system model and preliminaries in \Cref{appx:model}.
\Cref{appx:problem_definition} provides the formal definition of the problem \sysname solves.
Then, we give the implementation of \sysname.
First, we present the view generator primitive, which is the crucial building block of the reconfiguration mechanism of \sysname (\Cref{appx:view_generator}).
\Cref{appx:storage_module} is devoted to the storage module of a server: the module used for storing information.
In \Cref{appx:reconfiguration_module}, we introduce the reconfiguration module of a server: the module that contains the reconfiguration logic of a server.
Next, we introduce the implementation of a \sysname user in \Cref{appx:client}.
\Cref{appx:transaction_module} presents the transaction module of a server: the module used for processing payments issued by users of \sysname.
Finally, \Cref{appx:voting} is devoted to the voting module of a server: the module responsible for ensuring the validity of \sysname's voting mechanism.

%% file: appendix/model.tex
\section{System Model \& Preliminaries} \label{appx:model}

\para{Clients \& servers}
We consider a system of \emph{asynchronous} processes: a process proceeds at its own arbitrary (and non-deterministic) speed.
Each process is either a \emph{client} or a \emph{server}.
We denote by $\mathcal{C}$ the (possibly infinite) set of clients and by $\mathcal{R}$ the (possibly infinite) set of servers.
Each process in the system is assigned its \emph{protocol} to follow.

A process is in one of the four states at all times:
\begin{compactitem}
    \item \emph{inactive:} The process has not yet performed a computational step.
    
    \item \emph{obedient:} The process has performed a computational step and followed its protocol.
    
    \item \emph{disobedient:} The process has performed a computational step and (at some point) deviated from its protocol. 
    
    \item \emph{halted:} The process has stopped executing computational steps.
\end{compactitem}
Every process starts its execution in the inactive state.
Once the special $\mathtt{start}$ event triggers at the process (e.g., the process turns on), the process transits to the obedient state.
The process remains in the obedient state as long as it respects the assigned protocol and does not execute the special $\mathtt{stop}$ command.
Once the process deviates from the assigned protocol, the process transits to the disobedient state and stays there permanently.
Finally, once the special $\mathtt{stop}$ command is executed by the process (e.g., the process shuts down) which is in the obedient state, the process halts and stays in the halted state.
If the process performs a computational step afterwards, it transits to the faulty state (and remains in that state forever).
The state diagram of a process is given in \Cref{fig:state_diagram}.

\begin{figure}[th]
    \centering
    \includegraphics[width = 8cm]{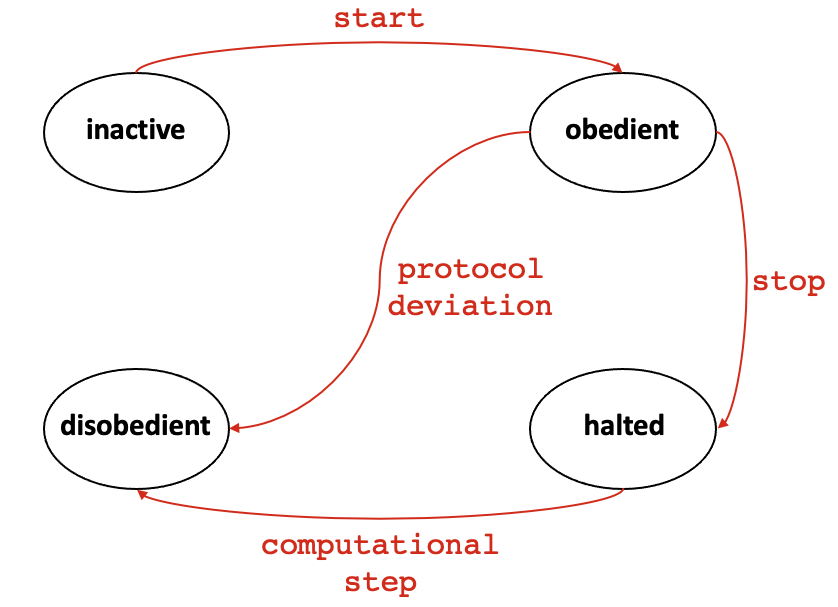}
    \caption{State diagram of a process.}
    \label{fig:state_diagram}
\end{figure}

Once the $\mathtt{start}$ event triggers at a process, we say that the process \emph{starts}.
We assume that only finitely many servers start in any execution.
On the other hand, infinitely many clients are allowed to start.
Any process that starts and does not transit to the disobedient state is said to be \emph{correct}.
Any correct process that never halts is \emph{forever-correct}.
If, at any point, a process transits to the disobedient state, the process is \emph{faulty}.

\para{Cryptographic primitives}
We make use of standard cryptographic primitives.
Specifically, we assume an idealized \emph{public-key infrastructure} (PKI): each process is associated with its own public/private key pair that is used to sign messages and verify signatures of other processes.
Each message has its sender; we denote by $m.\mathtt{sender}$ the sender of message $m$.
All messages are signed by their senders.
We assume that faulty processes cannot forge signatures of correct processes.
Moreover, a message that is not properly signed is immediately discarded by a correct process.
We assume that processes can forward messages to other processes and include messages in other messages; forwarded and included messages can also be authenticated, i.e., their signatures can be verified.\footnote{If an included or forwarded message is not properly signed, the entire message is discarded by a correct process.}

\para{Execution paradigm}
We assume that each process is single-threaded, i.e., the execution of the protocol assigned to each process is single-threaded.
The protocol of a process consists of a set of $\mathtt{upon}$ rules; if a rule guarding a part of the protocol is permanently active at a correct process, we assume that the part of the protocol guarded by that rule is eventually executed.

A discrete global clock is assumed and the range of clock's ticks is the set $\{0\} \cup \mathbb{N}$. 
No process has access to the discrete global clock.
Computation proceeds in atomic steps.
Each step happens at a clock's tick and a process may: (1) send a message, (2) receive a message, (3) get an external input (e.g., trigger of the special $\mathtt{start}$ event), or (4) perform an internal computation.

\para{Communication}
The processes are connected pair-wise by \emph{asynchronous} authenticated channels~\cite{cachin2011introduction}: there does not exist an upper bound on message delays.
The only requirement is that message delays are finite.

We assume the existence of the following communication primitives:
\begin{compactenum}
    \item \emph{Perfect links:} A process can directly send a message $m$ to a process $q$; we write ``\textcolor{blue}{$\mathtt{send}$} $m \text{ } \mathtt{to} \text{ } q$''.
    Perfect links provide the following properties:
    \begin{compactitem}
        \item \emph{Validity:} If a forever-correct process sends a message to a forever-correct process, then the message is eventually received.
        
        \item \emph{Integrity:} A message is received at most once by a correct process and, if the sender is correct, only if the message was sent by the sender.
    \end{compactitem}
    Note that perfect links present a generalization of the perfect links primitive~\cite{cachin2011introduction} to a dynamic environment in which correct processes can halt.
    
    \item \emph{Best-effort broadcast:} This primitive allows a process to send the same message $m$ to a fixed set of processes; we write ``\textcolor{blue}{$\mathtt{broadcast}$} $m \text{ } \mathtt{to} \text{ } S$'', where $S$ is a set of processes.
    The best-effort broadcast ensures:
    \begin{compactitem}
        \item \emph{Validity:} If a forever-correct process broadcasts a message to a set $S$ of processes and a process $q \in S$ is forever-correct, then $q$ eventually receives the message.
        
        \item \emph{Integrity:} A message is received at most once by a correct process and, if the sender is correct, only if the message was broadcast by the sender.
    \end{compactitem}
    The best-effort primitive represents the generalization of the primitive~\cite{cachin2011introduction} to dynamic environments.
    It consists of sending the broadcast message to all processes from the set of receivers using the perfect links.
    
    \item \emph{Gossip:}
    This primitive~\cite{kermarrec2007gossiping} allows a process to disseminate a message $m$ to all processes without specifying the particular set of receivers; we write ``\textcolor{blue}{$\mathtt{gossip}$} $m$''.
    The gossip primitive ensures the following properties:
    \begin{compactitem}
        \item \emph{Validity:} If a forever-correct process gossips a message, then every forever-correct process receives the message infinitely many times.
        
        \item \emph{Integrity:} If a message is received by a correct process and the sender of the message is correct, then the message was previously gossiped by the sender.
    \end{compactitem}
\end{compactenum}

\para{Preliminaries}
We now introduce the concepts of changes, views and sequences.
A \emph{change} $c \in \{+, -\} \times \mathcal{R}$ is a tuple that expresses an intent of a server either to \emph{join} the set of servers running \sysname or to \emph{leave} it (see \Cref{appx:introduction}).
For instance, $c = (+, r)$ denotes the intent of server $r \in \mathcal{R}$ to join, whereas $c' = (-, r')$ denotes the intent of server $r' \in \mathcal{R}$ to leave.

\begin{lstlisting}[
  caption={Change},
  label={lst:change},
  escapechar=?]
?\textbf{Change:}?
    instance ?$c =(s, r)$?, with ?$s \in \{+, -\}$? and ?$r$? is Server
\end{lstlisting}

A \emph{view} is a set of changes.
We associate three functions with any view $v$:
\begin{compactitem}
    \item $\mathtt{members}$: returns \emph{members} of the view $v$, i.e., returns all the servers $r \in \mathcal{R}$ such that $(+, r) \in v$ and $(-, r) \notin v$.
    
    \item $\mathtt{quorum}$: returns the \emph{quorum size} of the view $v$, i.e., it returns $n - \lfloor \frac{n - 1}{3} \rfloor$, where $n = |v.\mathtt{members}()|$.
    
    \item $\mathtt{plurality}$: returns the \emph{plurality} of the view $v$, i.e., it returns $\lfloor \frac{n - 1}{3} \rfloor + 1$, where $n = |v.\mathtt{members}()|$.
\end{compactitem}
We say that views $v_1$ and $v_2$ are \emph{comparable} if and only if $v_1 \subseteq v_2$ or $v_1 \supset v_2$.
If $v_1 \subset v_2$, we say that $v_1$ is \emph{smaller than} $v_2$ and $v_2$ is \emph{greater than} $v_1$.
Finally, we render some views as \emph{valid} (we give the formal definition in \Cref{definition:valid_view}).
 
\begin{lstlisting}[
  caption={View},
  label={lst:view},
  escapechar = ?]
?\textbf{View:}?
    instance ?$v$?, with ?$v$? is Set(Change)

    ?\textbf{function}? (View ?$v$?).members():
        ?\textbf{return}? ?$\{r$? ?$|$? ?$(+, r) \in v$? and ?$(-, r) \notin v\}$?
        
    ?\textbf{function}? (View ?$v$?).quorum():
        ?$n = |v$?.members()?$|$?
        ?\textbf{return}? ?$n - \lfloor \frac{n - 1}{3} \rfloor$?
    
    ?\textbf{function}? (View ?$v$?).plurality():
        ?$n = |v$?.members()?$|$?
        ?\textbf{return}? ?$\lfloor \frac{n - 1}{3} \rfloor + 1$?
        
    ?\textbf{function}? comparable(View ?$v$?, View ?$w$?):
        ?\textbf{return}? ?$v \subseteq w$? or ?$v \supset w$?
        
    ?\textbf{function}? valid(View ?$v$?) ?$ \to \{\top, \bot\}$?
\end{lstlisting}

Lastly, a \emph{sequence} is a set of views with any two views of the sequence being comparable.
We say that a sequence $\mathit{seq}$ \emph{follows} a view $v$ if all views contained in $\mathit{seq}$ are greater than $v$.
Moreover, we define the \emph{first} (resp., \emph{last}) view of a sequence as the smallest (resp., greatest) view that belongs to the sequence.

\begin{lstlisting}[
  caption={Sequence},
  label={lst:sequence},
  escapechar = ?]
?\textbf{Sequence:}?
    instance ?$\mathit{seq}$?, with ?$\mathit{seq}$? is Set(View) such that comparable(?$v, w$?) ?$= \top$?, for all ?$v, w \in \mathit{seq}$?
    
    ?\textbf{function}? (Sequence ?$\mathit{seq}$?).follows(View ?$v$?):
        ?\textbf{return}? for every ?$w \in \mathit{seq}$?: ?$w \supset v$?
    
    ?\textbf{function}? (Sequence ?$\mathit{seq}$?).first():
        ?\textbf{return}? ?$w \in \mathit{seq}$? such that ?$w \subseteq v$?, for every ?$v \in \mathit{seq}$?
        
    ?\textbf{function}? (Sequence ?$\mathit{seq}$?).last():
        ?\textbf{return}? ?$w \in \mathit{seq}$? such that ?$w \supseteq v$?, for every ?$v \in \mathit{seq}$?
\end{lstlisting}
If $\mathit{seq} = \{v_1, ..., v_x\}$ is a sequence, we write $\mathit{seq} = v_1 \to ... \to v_x$, where $v_i \subset v_{i+1}$, for every $1 \leq i < x$.

In summary, a change is a tuple $c = (\mathit{sym} \in \{+, -\}, r \in \mathcal{R})$ and it symbolizes the intent of server $r$ to either join (if $\mathit{sym} = +$) or leave (if $\mathit{sym} = -$).
A view is a set of changes, whereas a sequence is a set of comparable views.

\para{Failure model}
We assume that at most $\lfloor \frac{n - 1}{3} \rfloor$ members of $v$ are faulty, where $n = |v.\mathtt{members()}|$, for every valid view $v$.\footnote{Recall that a view $v$ is valid if and only if $\mathtt{valid}(v) = \top$ (see \Cref{lst:view}).}

\para{Certificates}
Throughout the implementation of \sysname, we often use \emph{certificates}.
A certificate is a construct used to prove that a value is indeed ``produced'' by a specific instance of a distributed primitive.
Therefore, certificates prevent faulty processes to ``lie'' about obtained values.
If $\mathtt{verify\_output}(\mathit{value}, \mathit{instance}, \omega) = \top$, then we say that $\mathit{value}$ is \emph{produced} by $\mathit{instance}$, where $\mathit{instance}$ is an instance of a distributed primitive.

\begin{lstlisting}[
  caption={The $\mathtt{verify\_output}$ function},
  label={lst:verify_output},
  escapechar = ?]
?\textbf{function}? verify_output(Value ?$\mathit{value}$?, Distributed_Primitive_Instance ?$\mathit{instance}$?, Certificate ?$\omega$?) ?$\to \{\top, \bot\}$?
\end{lstlisting}

\para{Constants}
We assume the existence of a specific view known by all processes; we denote this view by $\mathit{genesis}$.
The $\mathit{genesis}$ view does not contain $(-, r)$ changes, i.e., $(-, r) \notin \mathit{genesis}$, for any server $r$.
Moreover, $\mathit{genesis}$ has only finitely many members.
Finally, we assume that the $\mathtt{start}$ event eventually triggers at all members of the $\mathit{genesis}$ view.

\begin{lstlisting}[
  caption={Constants},
  label={lst:constants},
  escapechar=?]
?\textbf{Constants:}?
    View ?$\mathit{genesis}$?
\end{lstlisting}

%% file: appendix/problem_definition.tex
\section{Problem Definition} \label{appx:problem_definition}

As briefly mentioned in \Cref{appx:introduction}, the core problem solved by \sysname is the asset-transfer problem~\cite{bitcoin,consensus_number_crypto_dn21,astro_dsn20,wood_ethereum_2014,kuznetsov2021permissionless,auvolat2020money} among clients (i.e., users of \sysname).
We present the formal definition of the asset-transfer problem in \Cref{subsection:core_problem}.

Since our goal is to allow \sysname to efficiently serve millions of clients, the burden of transaction processing must be taken away from clients.
We put this burden on a set of servers - \emph{validators}.
In order to allow clients to select validators that actually process payments in \sysname, we develop \emph{asynchronous balance-based voting}: the mechanism that allows clients to vote a server in or out of the set of validators.
We formally define the asynchronous balance-based voting problem in \Cref{subsection:balance_based_voting}.

Finally, the definitive formulation of the problem consists of the asset-transfer problem (\Cref{subsection:core_problem}) enriched by the problem of changing the set of validators based on our voting mechanism (\Cref{subsection:balance_based_voting}).
We merge these in \Cref{subsection:carbon_properties}, thus painting the full picture of the problem \sysname solves.

\subsection{Core Problem: Asset-Transfer} \label{subsection:core_problem}

The asset-transfer problem allows clients to issue transactions in order to (1) transfer some of their money to another clients, thus decreasing their balance, and (2) claim money transferred to them by another clients, thus increasing their balance.
We proceed to formally define the asset-transfer problem.

\para{Transactions}
Clients can issue transactions of the following three types:
\begin{compactitem}
    \item \emph{withdrawal:} A client issues a withdrawal transaction once it aims to transfer some of its money to another client.
    
    \item \emph{deposit:} A client issues a deposit transaction in order to claim money transferred to it by another client.
    
    \item \emph{minting:} A client issues a minting transaction to increase the amount of money owned by it.
\end{compactitem}

\begin{center}
\hfill
\begin{minipage}{.25\textwidth}
\begin{lstlisting}[
  label={lst:withdraw_transaction},
  escapechar=?]
?\textbf{Withdrawal\_Transaction:}?
    Client ?$\mathit{issuer}$?
    Client ?$\mathit{receiver}$?
    Integer ?$\mathit{amount}$?
    Integer ?$\mathit{sn}$?
\end{lstlisting}
\end{minipage}
\hfill
\begin{minipage}{.4\textwidth}
\begin{lstlisting}[
  label={lst:deposit_transaction},
  escapechar=?]
?\textbf{Deposit\_Transaction:}?
    Client ?$\mathit{issuer}$?
    Withdrawal_Transaction? $\mathit{withdrawal}$?
    Integer ?$\mathit{sn}$?
\end{lstlisting}
\end{minipage}
\hfill
\begin{minipage}{.25\textwidth}
\begin{lstlisting}[
  label={lst:mint_transaction},
  escapechar=?]
?\textbf{Minting\_Transaction:}?
    Client ?$\mathit{issuer}$?
    Integer ?$\mathit{amount}$?
    Integer ?$\mathit{sn}$?
\end{lstlisting}
\end{minipage}
\hfill
\hfill
\hfill
\end{center}

Each transaction is parameterized with its issuer (i.e., the client issuing the transaction) and its sequence number.
Withdrawal transactions specify the client that receives the transferred money and the amount of the money being transferred.
Each deposit transaction specifies a corresponding withdrawal transaction, i.e., the transaction that allows the issuer to claim the money.
A minting transaction specifies the amount of money the issuer obtains.
For the sake of simplicity, we assume that sequence numbers of all transactions are greater than $0$ and, if $\mathit{tx}$ is a deposit transaction, then $\mathit{tx}.\mathit{widthdrawal}.\mathit{receiver} = \mathit{tx}.\mathit{issuer}$.
We denote by $\mathcal{T}$ the set of all possible transactions.

\para{Client's balance}
Each client has its \emph{balance}, i.e., the amount of money owned by the client at a given time.
The balance of a client is the sum of claimed (by deposit or minting transactions) money subtracted by the sum of transferred money.
The initial balance of a client $c$ is denoted by $\mathtt{initial\_balance}(c)$.

\para{Commitment proofs}
We denote by $\Sigma_c$ the set of \emph{commitment proofs}.
Intuitively, a commitment proof shows that a specific transaction is succesfully processed in \sysname.
Formally, we define the $\mathtt{verify\_commit}\text{: } \mathcal{T} \times \Sigma_c \to \{\top, \bot\}$ function that maps a transaction and a commitment proof into a boolean value.
A transaction $\mathit{tx} \in \mathcal{T}$ is \emph{commited} if and only if a (correct or faulty) process obtains (i.e., stores in its local memory) a commitment proof $\sigma_c \in \Sigma_c$ such that $\mathtt{verify\_commit}(\mathit{tx}, \sigma_c) = \top$; if the first attainment of such commitment proof occurs at time $t$, we say that $\mathit{tx}$ is committed at time $t' \geq t$.

\para{Logs}
A \emph{log} is a set of transactions.
We say that a log $\mathit{log}$ is \emph{admissible} if and only if:
\begin{compactitem}
    \item for every deposit transaction $\mathit{tx} \in \mathit{log}$, $\mathit{tx}.\mathit{withdrawal} \in \mathit{log}$, i.e., a deposit transaction is ``accompanied'' by the corresponding withdrawal transaction, and
    
    \item for each client $c$, let $\mathit{log}_c = \{\mathit{tx} \,|\, \mathit{tx} \in \mathit{log} \text{ and } \mathit{tx}.\mathit{issuer} = c\}$.
    Then, the following holds:
    \begin{compactitem}
        \item $\mathit{log}_c$ does not contain conflicting transactions, i.e., different transactions with the same sequence number, and
        
        \item all transactions in $\mathit{log}_c$ must have adjacent sequence numbers, i.e., no ``gaps'' can exist, and
        
        \item a single withdrawal transaction is referenced by at most one deposit transaction, i.e., the client cannot claim the ``same'' money multiple times, and
        
        \item the client has enough money for every issued withdrawal transaction.
    \end{compactitem}
\end{compactitem}
More formally, a log $\mathit{log}$ is admissible if and only if $\mathtt{admissible}(\mathit{log}) = \top$, where the $\mathtt{admissible}$ function is defined in \Cref{lst:admissible_transaction_set}.

Before introducing the $\mathtt{admissible}$ function, we define the notion of \emph{precedence} in a set of transactions.
Let $X$ be a set of transactions.
We say that $\mathit{tx}_1 \in X$ precedes $\mathit{tx}_2 \in X$ in $X$ if :
\begin{compactitem}
    \item $\mathit{tx}_1.\mathit{issuer} = \mathit{tx}_2.\mathit{issuer}$ and $\mathit{tx}_1.\mathit{sn} < \mathit{tx}_2.\mathit{sn}$, or
    
    \item $\mathit{tx}_1$ is a withdrawal transaction, $\mathit{tx}_2$ is a deposit transaction and $\mathit{tx}_2.\mathit{withdrawal} = \mathit{tx}_1$, or
    
    \item there exists a transaction $\mathit{tx}$ such that $\mathit{tx}_1$ precedes $\mathit{tx}$ in $X$ and $\mathit{tx}$ precedes $\mathit{tx}_2$ in $X$.
\end{compactitem}

Finally, we give the definition of the $\mathtt{admissible}$ function.
\begin{lstlisting}[
  caption={Log admissibility},
  label={lst:admissible_transaction_set},
  escapechar=?]
?\textbf{function}? admissible(Log ?$\mathit{log}$?):
    // check whether there are cycles
    if ?exists? Transaction ?$\mathit{tx} \in \mathit{log}$? such that ?$\mathit{tx}$? precedes ?$\mathit{tx}$? ?\textcolor{black}{in}? ?$\mathit{log}$? ?\label{line:no_cycles}?
        ?\textbf{return}? ?$\bot$?

    // check whether the withdrawals of all deposit transactions are in ?\textcolor{gray}{$\mathit{log}$}?
    for each Deposit_Transaction ?$\mathit{tx} \in \mathit{log}$?:
        if ?$\mathit{tx}.\mathit{withdrawal} \notin \mathit{log}$?: ?\label{line:check_withdrawal_admissible}?
            ?\textbf{return}? ?$\bot$?

    for each Client ?$c$?:
        Log ?$\mathit{log}_c = \{\mathit{tx} \,|\, \mathit{tx} \in \mathit{log} \text{ and } \mathit{tx}.\mathit{issuer} = c\}$? 
        if admissible_client_log(?$c$?, ?$\mathit{log}_c$?) ?$ = \bot$?: // the client log is not admissible ?\label{line:admissible_client_log}?
            ?\textbf{return}? ?$\bot$?
    ?\textbf{return}? ?$\top$?
  
?\textbf{function}? admissible_client_log(Client ?$c$?, Log ?$\mathit{log}_c$?):
    // if ?\textcolor{gray}{$\mathit{log}_c$}? is empty, then return ?\textcolor{gray}{$\top$}?
    if ?$\mathit{log}_c = \emptyset$?:
        ?\textbf{return}? ?$\top$?
    
    // if there exist conflicting transactions, return ?\textcolor{gray}{$\bot$}?
    if ?exist? Transaction ?$\mathit{tx}_1, \mathit{tx}_2 \in \mathit{log}_c$? such that ?$\mathit{tx}_1 \neq \mathit{tx}_2$? and ?$\mathit{tx}_1.\mathit{sn} = \mathit{tx}_2.\mathit{sn}$?: ?\label{line:check_conflicting_well_formed}?
        ?\textbf{return}? ?$\bot$?
    
    // sequence numbers of transactions must be adjacent
    if ?exists? Transaction ?$\mathit{tx} \in \mathit{log}_c$? and ?exists? Integer ?$\mathit{num}$?, where ?$1 \leq \mathit{num} < \mathit{tx}.\mathit{sn}$?, such that does not exist Transaction ?$\mathit{tx}' \in \mathit{log}_c$? with ?$\mathit{tx}'.\mathit{sn} = \mathit{num}$?: ?\label{line:no_gaps_admissible}?
        ?\textbf{return}? ?$\bot$?
        
    // the client log contains non-conflicting transactions and transactions have adjacent sequence numbers
    // sort the log by the sequence numbers of the transactions in the ascending order
    Array ?$\mathit{array\_log}_c = \mathit{log}_c$?.sort_by(?$\mathit{sn}$?)
    
    // introduce the balance of the client before executing any transaction from its log
    Integer ?$\mathit{balance} = $? initial_balance(?$c$?)
    Set(Withdrawal_Transaction) ?$\mathit{withdrawals} = \emptyset$? // observed withdrawal transactions
    
    // iterate through all transactions in the increasing order of their sequence numbers
    for each Transaction ?$\mathit{tx} \in \mathit{array\_log}_c$?:
        if ?$\mathit{tx}$? is Withdrawal_Transaction: // withdrawal transaction ?\label{line:check_admissible_client_start}?
            if ?$\mathit{tx}.\mathit{amount} > \mathit{balance}$?: // not enough money
                ?$\textbf{return}$? ?$\bot$?
            ?$\mathit{balance} = \mathit{balance} - \mathit{tx}.\mathit{amount}$?
        else if ?$\mathit{tx}$? is Deposit_Transaction: // deposit transaction
            if ?$\mathit{tx}.\mathit{withdrawal} \in \mathit{withdrawals}$?: // withdrawal already used ?\label{line:check_admissible_client_stop}?
                ?\textbf{return}? ?$\bot$?
            else:
                ?$\mathit{withdrawals} = \mathit{withdrawals} \cup \{\mathit{tx}.\mathit{withdrawal}\}$?
            ?$\mathit{balance} = \mathit{balance} + \mathit{tx}.\mathit{withdrawal}.\mathit{amount}$?
        else: // minting transaction
            ?$\mathit{balance} = \mathit{balance} + \mathit{tx}.\mathit{amount}$?
                
    ?\textbf{return}? ?$\top$?
\end{lstlisting}

\para{Interface}
Interface of a client consists of:
\begin{compactitem}
    \item invocation $\mathtt{issue} \text{ } \mathit{tx}$: the client issues a transaction $\mathit{tx}$,
    
    \item invocation $\mathtt{query} \text{ } \mathtt{market}$: the client requests to learn the total amount of money in \sysname,
    
    \item indication $\mathtt{committed} \text{ } \mathit{tx}$: the client learns that a transaction $\mathit{tx}$ is committed, and
    
    \item indication $\mathtt{total} \text{ } \mathtt{money} \text{ } \mathit{money}$: the client learns that the total amount of money in \sysname is $\mathit{money}$.
\end{compactitem}

\para{Rules}
The following is assumed about a correct client $c$:
\begin{compactitem}
    \item The client does not invoke an operation (i.e., does not issue a transaction or does not query about the total amount of money) before the previous invocation is completed.

    \item If a transaction $\mathit{tx}$ is issued by $c$, then $\mathit{tx}.\mathit{issuer} = c$.
    
    \item Client $c$ does not issue a transaction before it learns that all previously issued transactions by $c$ are committed.
    
    \item Let $\mathit{tx}$ be the $i$-th transaction issued by $c$.
    Then, $\mathit{tx}.\mathit{sn} = i$.
    
    \item Client $c$ does not issue a withdrawal transaction $\mathit{tx}$ unless the balance of $c$ at the moment of issuing $\mathit{tx}$ is at least $\mathit{tx}.\mathit{amount}$.
    
    \item Client $c$ does not issue a deposit transaction $\mathit{tx}$ before it learns that $\mathit{tx}.\mathit{withdrawal}$ is committed.
    
    \item Client $c$ does not issue conflicting transactions, i.e., different transactions with the same sequence number.
    
    \item Client $c$ does not issue multiple deposit transactions that reference the same withdrawal transaction.
\end{compactitem}

\para{Properties}
Finally, we introduce the properties of the asset-transfer problem:
\begin{compactitem}
    \item \emph{Commitment Validity:} If a forever-correct client issues a transaction $\mathit{tx}$, then $\mathit{tx}$ is eventually committed.
    
    \item \emph{Commitment Integrity:} If a transaction $\mathit{tx}$ is committed and $\mathit{tx}.\mathit{issuer}$ is correct, then $\mathit{tx}$ was issued by $\mathit{tx}.\mathit{issuer}$.
    
    \item \emph{Commitment Learning:} Let a transaction $\mathit{tx}$ be committed.
    If $\mathit{tx}.\mathit{issuer}$ is forever-correct, then $\mathit{tx}$ eventually learns that $\mathit{tx}$ is committed.
    If $\mathit{tx}$ is a withdrawal transaction and $\mathit{tx}.\mathit{receiver}$ is forever-correct, then $\mathit{tx}.\mathit{receiver}$ eventually learns that $\mathit{tx}$ is committed.
    
    \item \emph{Commitment Admissibility:} Let $\mathit{log}_t$ denote the set of all committed transactions at time $t$.
    For all times $t$, $\mathit{log}_{t}$ is admissible.
    
    \item \emph{Query Validity:} If a forever-correct client requests to learn the total amount of money in \sysname, the client eventually learns.
    
    \item \emph{Query Safety:} Let $\mathit{mints}_{\infty}$ denote the set of all committed minting transactions.
    If a correct client learns that the total amount of money is $\mathit{money}$, then $\mathit{money} \leq \sum\limits_{\mathit{tx} \in \mathit{mints}_{\infty}} \mathit{tx}.\mathit{amount}$.
    
    \item \emph{Query Liveness:} Let $\mathit{mints}$ denote a finite set of committed minting transactions.
    If a forever-correct client requests infinitely many times to learn the total amount of money in \sysname, the client eventually learns $\mathit{money} \geq \sum\limits_{\mathit{tx} \in \mathit{mints}} \mathit{tx}.\mathit{amount}$.
\end{compactitem}

\subsection{Asynchronous Balance-Based Voting} \label{subsection:balance_based_voting}

The asynchronous balance-based voting mechanism revolves around the clients being able to express their support for a specific \emph{motion} by issuing \emph{vote} transactions.
Intuitively, a motion \emph{passes} if ``enough'' clients (i.e., clients that own ``enough'' money) have expressed their support for the motion.
We formally define the problem below.

\para{Vote transactions}
As withdrawal and deposit transactions, each vote transaction is parametrized with its issuer and its sequence number.
Moreover, each vote transaction specifies the motion for which the issuer votes by issuing the transaction.

\begin{center}
\begin{minipage}{.25\textwidth}
\begin{lstlisting}[
  label={lst:vote_transaction},
  escapechar=?]
?\textbf{Vote\_Transaction:}?
    Client ?$\mathit{issuer}$?
    Motion ?$\mathit{motion}$?
    Integer ?$\mathit{sn}$?
\end{lstlisting}
\end{minipage}
\end{center}

\para{Motions}
A motion is any statement that can be voted for by clients.
We say that a motion $\mathit{mot}$ is \emph{proposed} at time $t$ if and only if a correct process (either client or a server) obtains a vote transaction for $\mathit{mot}$ at time $t$ and no correct process has obtained a vote transaction for $\mathit{mot}$ before $t$.
Intuitively, a motion is proposed at time $t$ if and only if the first correct process to observe a client voting for the motion observes such voting action at time $t$.

A log \emph{supports} a motion if a ``majority'' of the money votes for the motion.
In order to formally define when a log supports a motion, we first define how the balance of each client is calculated in a log.
Moreover, we define the total amount of money in a log.
Finally, we define the amount of money in a log ``voting'' for a motion.

\begin{lstlisting}[
  caption={Log},
  label={lst:log},
  escapechar = ?]
?\textbf{Log:}?
    instance ?$\mathit{log}$?, with ?$\mathit{log}$? is Set(Transaction)
    
    ?\textbf{function}? (Log ?$\mathit{log}$?).?balance?(Client ?$c$?):
        Log ?$\mathit{log}_c = \{\mathit{tx} \,|\, \mathit{tx} \in \mathit{log} \text{ and } \mathit{tx}.\mathit{issuer} = c\}$?
        ?$\mathit{balance} = $? initial_balance(?$c$?)
        Array ?$\mathit{array\_log}_c = \mathit{log}_c$?.sort_by(?$\mathit{sn}$?) // sort the log by the sequence numbers
        for each Transaction ?$\mathit{tx} \in \mathit{array\_log}_c$?:
            if ?$\mathit{tx}$? is Withdrawal_Transaction:
                ?$\mathit{balance} = \mathit{balance} - \mathit{tx}.\mathit{amount}$?
            else if ?$\mathit{tx}$? is Deposit_Transaction:
                ?$\mathit{balance} = \mathit{balance} + \mathit{tx}.\mathit{witdrawal}.\mathit{amount}$?
            else if ?$\mathit{tx}$? is Minting_Transaction:
                ?$\mathit{balance} = \mathit{balance} + \mathit{tx}.\mathit{amount}$?
        ?\textbf{return}? ?$\mathit{balance}$?
        
    ?\textbf{function}? (Log ?$\mathit{log}$?).total_money():
        Set(Transaction) ?$\mathit{mints} = \{\mathit{tx} \,|\, \mathit{tx} \in \mathit{log} \text{ and } \mathit{tx} \text{ is Minting\_Transaction}\}$?
        ?\textbf{return}? ?$\sum\limits_{\mathit{tx} \in \mathit{mints}} \mathit{tx}.\mathit{amount}$? // the sum of all minting transactions
        
    ?\textbf{function}? (Log ?$\mathit{log}$?).voted_for(Motion ?$\mathit{mot}$?):
        Set(Client) ?$\mathit{voted\_for} = \{c \,|\, \mathit{tx} \in \mathit{log} \text{ and } \mathit{tx} \text{ is Vote\_Transaction and } \mathit{tx}.\mathit{motion} = \mathit{mot} \text{ and } \mathit{tx}.\mathit{issuer} = c\}$?
        ?\textbf{return}? ?$\sum\limits_{c \in \mathit{voted\_for}} \mathit{log}\text{.balance(}c\text{)}$?
\end{lstlisting}

Finally, we are ready to formally define when a log supports a motion.
A log $\mathit{log}$ supports a motion $\mathit{mot}$ if and only if $\mathtt{support}(\mathit{log}, \mathit{mot}) = \top$, where the $\mathtt{support}$ function is defined in the following manner:

\begin{lstlisting}[
  caption={Motion support},
  label={lst:log_supports_motion},
  escapechar=?]
?\textbf{function}? support(Log ?$\mathit{log}$?, Motion ?$\mathit{mot}$?):
    ?\textbf{return}? ?$\mathit{log}$?.voted_for(?$\mathit{mot}$?)?$> 0$? and ?$\mathit{log}\text{.voted\_for(}\mathit{mot}\text{)} \geq \frac{\mathit{log}\text{.total\_money()}}{2}$?
\end{lstlisting}

\para{Voting proofs}
We denote by $\Sigma_v$ the set of \emph{voting proofs}.
Intuitively, a voting proof shows that a specific motion passes.
Formally, a motion $\mathit{mot}$ \emph{passes} if and only if a (correct or faulty) process obtains a voting proof $\sigma_v \in \Sigma_v$ such that\\ $\mathtt{verify\_voting}(\mathit{mot}, \sigma_v) = \top$.
We use voting proofs in \Cref{subsection:carbon_properties} in order to abstract away the asynchronous balance-based voting problem (for the sake of simplicity).

\para{Rules}
We assume that a correct client does not vote for the same motion more than once.

\para{Properties}
Finally, we define the properties of the asynchronous balance-based voting primitive:
\begin{compactitem}
    \item \emph{Voting Safety:} Let a motion $\mathit{mot}$ be proposed at time $t$ and let $\mathit{mot}$ pass.
    Let a correct client learn at time $t$ that the total amount of money in \sysname is $\mathit{money}$ and let $\mathit{log}_{\infty}$ denote the set of all committed transactions.
    There exists an admissible log $\mathit{log}_{\mathit{pass}} \subseteq \mathit{log}_{\infty}$ such that $\mathit{log}_{\mathit{pass}}.\mathtt{voted\_for}(\mathit{mot}) \geq \frac{money}{2}$.
    
    \item \emph{Voting Liveness:} Let a motion $\mathit{mot}$ be proposed and let $\mathit{log}_{\infty}$ denote the set of all committed transactions.
    If there exists an admissible log $\mathit{log}_{\mathit{pass}} \subseteq \mathit{log}_{\infty}$ such that $\mathit{log}_{\mathit{greater}}$ supports $\mathit{mot}$, for every admissible log $\mathit{log}_{\mathit{greater}}$, where $\mathit{log}_{\mathit{pass}} \subseteq \mathit{log}_{\mathit{greater}} \subseteq \mathit{log}_{\infty}$, then $\mathit{mot}$ passes.
\end{compactitem}
Intuitively, the voting safety property guarantees that, if a motion passes, more than half of the money held in the system ``at the moment'' of proposing the motion has ``voted'' for the motion to pass.
Voting liveness guarantees that a motion passes if, after some time, the motion is ``forever-supported''. 

\subsection{\sysname = Asset-Transfer + Voting-Driven Reconfiguration} \label{subsection:carbon_properties}

Finally, we formally define the ``complete'' problem solved by \sysname.

\para{Validators}
A correct server \emph{joins} once the server triggers the special $\mathtt{joined}$ event; if the special $\mathtt{joined}$ event is triggered at time $t$, the server joins at time $t$.
Importantly, a correct server joins (i.e., triggers the special $\mathtt{joined}$ event) \emph{only after} it has fully executed the joining subprotocol (see line~\ref{line:joined} of \Cref{lst:reconfiguration_view_transition}) initiated by the invocation of the $\mathtt{join} \text{ } \mathtt{with} \text{ } \mathtt{voting} \text{ } \mathtt{proof} \text{ } \sigma_v$ operation (see the ``Server's interface'' paragraph below).
(If a correct server is a member of the $\mathit{genesis}$ view, the server triggers joins upon starting Carbon at line~\ref{line:join_init} of \Cref{lst:reconfiguration_initialization}.)
A faulty server $r$ joins if a correct server triggers the special $r \text{ } \mathtt{joined}$ event and no correct server has previously triggered the special $r \text{ } \mathtt{left}$ event; if the first $r \text{ } \mathtt{joined}$ event that satisfies the aforementioned condition is triggered at time $t$, the server joins at time $t$.

A correct server \emph{leaves} once the server triggers the special $\mathtt{left}$ event; if the special $\mathtt{left}$ event is triggered at time $t$, the server leaves at time $t$.
Importantly, a correct server leaves (i.e., triggers the special $\mathtt{left}$ event) \emph{only after} it has fully executed the leaving subprotocol (see line~\ref{line:left} of \Cref{lst:reconfiguration_view_transition}) initiated by the invocation of the $\mathtt{leave}$ operation (see the ``Server's interface'' paragraph below).

A faulty server $r$ leaves if a correct server triggers the special $r \text{ } \mathtt{left}$ event; if the first $r \text{ } \mathtt{left}$ event triggered by correct servers is triggered at time $t$, the server leaves at time $t$.
In summary, a correct server joins (resp., leaves) once the server triggers the special $\mathtt{joined}$ (resp., $\mathtt{left}$) event; a faulty server $r$ joins (resp., leaves) once a correct server triggers the special $r \text{ } \mathtt{joined}$ (resp., $r \text{ } \mathtt{left}$) event.

A server $r$ is a \emph{validator} at time $t$ if and only if the server joins by time $t$ and does not leave by time $t$.
Moreover, a server is a \emph{forever-validator} if and only if it joins and never leaves.

\para{Commitment proofs - revisited}

Since the notion of a validator is introduced, we aim to ensure that transactions are ``processed'' only by validators of \sysname.
Intuitively, if a transaction $\mathit{tx}$ is issued at time $t$ and a server $r$ left before time $t$, our goal is to guarantee that $r$ cannot be a server that helps $\mathit{tx}$ to be processed.
In order to formally define such property of \sysname, we introduce a set of \emph{signers} of a commitment proof: the set of servers that have ``created'' a commitment proof.
We denote by $\sigma_c.\mathtt{signers}$ the set of signers of a commitment proof $\sigma_c$.

\para{Server's interface}
Interface of a server consists of two invocations:
\begin{compactitem}
    \item Invocation ``$\mathtt{join} \text{ } \mathtt{with} \text{ } \mathtt{voting} \text{ } \mathtt{proof} \text{ } \sigma_v$'':
    The server requests to join (i.e., to become a validator) and provides the voting proof $\sigma_v$ such that $\mathtt{verify\_voting}(\text{``add server } r\text{''}, \sigma_v) = \top$.
    
    \item Invocation ``$\mathtt{leave}$'':
    The server requests to leave (i.e., to stop being a validator).
\end{compactitem}

\para{Server's rules}
The following is assumed about a correct server $r$:
\begin{compactenum}
    \item If $r \in \mathit{genesis}.\mathtt{members()}$, then $r$ does not request to join.

    \item Server $r$ requests to join at most once.

    \item If $r$ requests to join and provides the voting proof $\sigma_v$, then $\mathtt{verify\_voting}(\text{``add server } r\text{''}, \sigma_v) = \top$.
    
    \item If $r$ requests to join, then $r$ does not halt before it leaves.

    \item If $r$ requests to leave, $r$ has previously joined.
    
    \item If $r$ leaves, then $r$ immediately halts.
\end{compactenum}

\para{Properties of \sysname}
Finally, we introduce the complete list of properties satisfied by \sysname.
We say that a transaction $\mathit{tx}$ is \emph{issued} at time $t$ if and only if:
\begin{compactitem}
    \item if $\mathit{tx}.\mathit{issuer}$ is correct, then $\mathtt{issue} \text{ } \mathit{tx}$ was invoked at time $t$ (see \Cref{subsection:core_problem}),
    
    \item if $\mathit{tx}.\mathit{issuer}$ is not correct, then the first correct process that obtains $\mathit{tx}$ obtains $\mathit{tx}$ at time $t$.
\end{compactitem}
If the motion ``add server $r$'' passes, we say that $r$ is \emph{voted in}.
Similarly, if the motion ``remove server $r$'' passes, we say that $r$ is \emph{voted out}.
We are ready to define the properties of \sysname:
\begin{compactitem}
    \item \emph{Commitment Validity:} If a forever-correct client issues a transaction $\mathit{tx}$, then $\mathit{tx}$ is eventually committed.
    
    \item \emph{Commitment Integrity:} If a transaction $\mathit{tx}$ is committed and $\mathit{tx}.\mathit{issuer}$ is correct, then $\mathit{tx}$ was issued by $\mathit{tx}.\mathit{issuer}$.
    
    \item \emph{Commitment Learning:} Let a transaction $\mathit{tx}$ be committed.
    If $\mathit{tx}.\mathit{issuer}$ is forever-correct, then $\mathit{tx}$ eventually learns that $\mathit{tx}$ is committed.
    If $\mathit{tx}$ is a withdrawal transaction and $\mathit{tx}.\mathit{receiver}$ is forever-correct, then $\mathit{tx}.\mathit{receiver}$ eventually learns that $\mathit{tx}$ is committed.
    
    \item \emph{Commitment Admissibility:} Let $\mathit{log}_t$ denote the set of all committed transactions at time $t$.
    At all times $t$, $\mathit{log}_t$ is admissible.    
    
    \item \emph{Commitment Signing:} Let a transaction $\mathit{tx}$ be issued at time $t$.
    Let a commitment proof $\sigma_c$ be obtained at time $t'$, where $\mathtt{verify\_commit}(\mathit{tx}, \sigma_c) = \top$.
    If $a \in \sigma_c.\mathtt{signers}$, then $a$ is a validator at time $t_{\mathit{val}}$, where $t \leq t_{\mathit{val}} \leq t'$.
    
    \item \emph{Query Validity:} If a forever-correct client requests to learn the total amount of money in \sysname, the client eventually learns.
    
    \item \emph{Query Safety:} Let $\mathit{mints}_{\infty}$ denote the set of all committed minting transactions.
    If a correct client learns that the total amount of money is $\mathit{money}$, then $\mathit{money} \leq \sum\limits_{\mathit{tx} \in \mathit{mints}_{\infty}} \mathit{tx}.\mathit{amount}$.
    
    \item \emph{Query Liveness:} Let $\mathit{mints}$ denote a finite set of committed minting transactions.
    If a forever-correct client requests infinitely many times to learn the total amount of money in \sysname, the client eventually learns $\mathit{money} \geq \sum\limits_{\mathit{tx} \in \mathit{mints}} \mathit{tx}.\mathit{amount}$.
    
    \item \emph{Join Safety:} If a server $r$ joins, then $r \in \mathit{genesis}.\mathtt{members()}$ or $r$ is voted in.
        
    \item \emph{Leave Safety:} If a correct server leaves, then the server requested to leave or is voted out.
     
    \item \emph{Join Liveness:} If a correct server requests to join, the server eventually joins (i.e., becomes a validator).
        
    \item \emph{Leave Liveness:} If a correct server requests to leave, the server eventually leaves (i.e., stops being a validator).
        
    \item \emph{Removal Liveness:} If a correct forever-validator obtains a voting proof $\sigma_v$ such that $\mathtt{verify\_voting}(\text{``remove server } r\text{''}, \sigma_v) = \top$ and $r$ is a validator at some time (i.e., $r$ has joined), then $r$ eventually leaves (i.e., stops being a validator).
\end{compactitem}

%% file: appendix/view_generator.tex
\section{View Generator} \label{appx:view_generator}

The first primitive we present, which is used in the implementation of \sysname, is the \emph{view generator} primitive~\cite{alchieri2018efficient}.
The view generator primitive is used by the servers exclusively.

Each instance of the view generator is parameterized with a view.
We denote by $\mathit{vg}(v)$ the instance of the view generator primitive parameterized with view $v$.
Servers can (1) start $\mathit{vg}(v)$ with or without a proposal, (2) stop $\mathit{vg}(v)$, (3) propose a set of views (along with some additional information) to $\mathit{vg}(v)$, and (4) decide a set of views from $\mathit{vg}(v)$.
Finally, we assume that (at least) a quorum of members of view $v$ are correct.

We assume an existence of \emph{evidences}, which are (for now) abstract constructs.
The input of a view generator instance is a tuple $(v, \mathit{set}, \epsilon)$, where $v$ is a view, $\mathit{set}$ is a set of views and and $\epsilon$ is a proof.
For every tuple $\mathit{input} = (v, \mathit{set}, \epsilon)$ and every view $v'$, we define $\mathtt{valid}(\mathit{input}, v')$ function that returns either $\top$ or $\bot$.
Again, the concrete implementation of the $\mathtt{valid}$ function is not given for now.

\begin{lstlisting}[
  caption={(View, Set, Evidence) tuple},
  label={lst:view_set_proof},
  escapechar = ?]
?\textbf{View, Set, Evidence:}?
    instance ?$(v, \mathit{set}, \epsilon)$?, with ?$v$? is View, ?$\mathit{set}$? is Set(View) and ?$\epsilon$? is Evidence
    
    ?\textbf{function}? valid(View, Set, Evidence ?$(v, \mathit{set}, \epsilon)$?, View ?$v'$?) ?$\to \{\top, \bot\}$?
\end{lstlisting}

We are now ready to introduce the interface of an instance of the view generator primitive.

\begin{lstlisting}[
  caption={View generator - interface},
  label={lst:view_generator_interface},
  escapechar = ?]
?\textbf{View Generator:}?
    instance ?$\mathit{vg}(v)$?, with ?$v$? is View
    
    Interface?\textcolor{plainorange}{:}?
        ?\textcolor{blue}{Requests:}?
            ?$\bullet$? <?$\mathit{vg}(v)$?.Start | View, Set, Evidence ?$(v', \mathit{set}', \epsilon')$?>: starts ?$\mathit{vg}(v)$?; if ?$(v', \mathit{set}', \epsilon') = \bot$?, then ?$\mathit{vg}(v)$? is started without a proposal. If ?$(v', \mathit{set}', \epsilon') \neq \bot$?, then the server has proposed to ?$\mathit{vg}(v)$?.
            
            ?$\bullet$? <?$\mathit{vg}(v)$?.Stop>: stops ?$\mathit{vg}(v)$?.
            
            ?$\bullet$? <?$\mathit{vg}(v)$?.Propose | View, Set, Evidence ?$(v', \mathit{set}', \epsilon')$?>: proposes ?$(v', \mathit{set}', \epsilon')$? to ?$\mathit{vg}(v)$?.
            
        ?\textcolor{blue}{Indications:}?
            ?$\bullet$? <?$\mathit{vg}(v)$?.Decide | Set(View) ?$\mathit{set}'$?, Certificate ?$\omega'$?>: indicates that ?$\mathit{set}'$? is decided with the certificate ?$\omega'$?.
            
        ?\textcolor{blue}{Rules:}?
            1) No correct server ?$r \notin v$?.members() invokes any request to ?$\mathit{vg}(v)$?.
            2) Every correct server invokes each request at most once.
            3) If a correct server invokes <?$\mathit{vg}(v)$?.Start | ?$(v', \mathit{set}', \epsilon') \neq \bot$?>, then the server does not invoke <?$\mathit{vg}(v)$?.Propose>.
            4) No correct server invokes <?$\mathit{vg}(v)$?.Stop> or <?$\mathit{vg}(v)$?.Propose> unless it has already invoked?\\?<?$\mathit{vg}(v)$?.Start>.
            5) A correct server invokes <?$\mathit{vg}(v)$?.Stop> before halting.
            6) If a correct server invokes <?$\mathit{vg}(v)$?.Start | ?$(v', \mathit{set}', \epsilon') \neq \bot$?> or <?$\mathit{vg}(v)$?.Propose | ?$(v', \mathit{set}', \epsilon')$?>, then valid(?$(v', \mathit{set}', \epsilon')$?, ?$v$?) ?$ = \top$?. ?\label{line:rule_seq_sequence}?
\end{lstlisting}

We denote by $E_t^v$ the set of evidences such that, for every $\epsilon \in E_t^v$, the following holds: (1) $\epsilon$ is obtained (i.e., stored in the local memory) by time $t$ by a (correct or faulty) process\footnote{Once an evidence is obtained by a (correct or faulty) process, the evidence is obtained by any time in the future. In other words, an attainment of an evidence is \emph{irrevocable}.}, and (2) there exists a view $v'$ and a set of views $\mathit{set}'$ that satisfy $\mathtt{valid}((v', \mathit{set}', \epsilon), v) = \top$.
By definition, $E_t^v \subseteq E_{t'}^v$, for any times $t, t'$ with $t' > t$.

Now, let $\Lambda_t^v = \{(v', \mathit{set}', \epsilon') \,|\, \mathtt{valid}((v', \mathit{set}', \epsilon'), v) = \top \text{ and } \epsilon' \in E_t^v\}$.
Due to the fact that $E_t^v \subseteq E_{t'}^{v}$, for any times $t, t'$ with $t' > t$, we have that $\Lambda_t^v \subseteq \Lambda_{t'}^v$.

Moreover, we define the $\mathit{preconditions}(v, t)$ function such that $\mathit{preconditions}(v, t) = \top$ if and only if:
\begin{compactitem}
    \item for every $(v', \mathit{set}', \epsilon') \in \Lambda_t^v$, $\mathit{set}'$ is a sequence, and 
    
    \item for every $(v', \mathit{set}' = \emptyset, \epsilon') \in \Lambda_t^v$, $v \subset v'$, and
    
    \item for every $(v', \mathit{set}' \neq \emptyset, \epsilon') \in \Lambda_t^v$, $\mathit{set}'.\mathtt{follows}(v) = \top$ and $v' \in \mathit{set}'$, and
    
    \item if $(+, r) \notin v$, for some server $r$, then, for every $(v', \mathit{set}', \epsilon') \in \Lambda_t^v$, $(-, r) \notin v'$ and, if $\mathit{set}' \neq \emptyset$, $(-, r) \notin \mathit{set}'.\mathtt{last}()$, and
    
    \item for any pair $(v', \mathit{set}', \epsilon'), (v'', \mathit{set}'', \epsilon'') \in \Lambda_t^v$, either $\mathit{set}' \subseteq \mathit{set}''$ or $\mathit{set}' \supset \mathit{set}''$.
\end{compactitem}
Since $\Lambda_t^v \subseteq \Lambda_{t'}^v$, for any times $t, t'$ with $t' > t$, we have that $\mathit{preconditions}(v, t') = \top \implies \mathit{preconditions}(v, t) = \top$.

Finally, if a (correct or faulty) process obtains (i.e., stores in its local memory) a certificate $\omega$, where $\mathtt{verify\_output}(\mathit{set}, \mathit{vg}(v), \omega) = \top$, for some set $\mathit{set}$ of views, then we say that $\mathit{set}$ is \emph{committed} by $\mathit{vg}(v)$.
Moreover, $\mathit{set}$ is committed at time $t$ if and only if the certificate attainment happens at time $t' \leq t$.
Note that we assume that an attainment of a certificate is irrevocable; hence, if $\mathit{set}$ is committed at time $t$, then $\mathit{set}$ is committed at any time $t' > t$.

We are now ready to define the properties of $\mathit{vg}(v)$.

\begin{lstlisting}[
  caption={View generator - properties},
  label={lst:view_generator_properties},
  escapechar = ?]
?\textbf{View Generator:}?
    instance ?$\mathit{vg}(v)$?, with ?$v$? is View
    
    ?\textcolor{plainorange}{Properties:}?
        ?\textbf{- Integrity:}? No correct server ?$r \notin v$?.members() receives any indication from ?$\mathit{vg}(v)$?.
        ?\textbf{- Comparability:}? If ?$\mathit{set}_1$? and ?$\mathit{set}_2$? are committed by ?$\mathit{vg}(v)$?, then either ?$\mathit{set}_1 \subseteq \mathit{set}_2$? or ?$\mathit{set}_1 \supset \mathit{set}_2$?.
        ?\textbf{- Validity:}? If ?$\mathit{set}$? is committed by ?$\mathit{vg}(v)$? at time ?$t$? and ?$\mathit{preconditions}(v, t - 1) = \top$?, then (1) ?$\mathit{set} \neq \emptyset$?, (2) ?$\mathit{set}$? is a sequence, and (3) ?$\mathit{set}$?.follows(?$v$?) = ?$\top$?.
        ?\textbf{- Membership Validity:}? If ?$\mathit{set}$? is committed by ?$\mathit{vg}(v)$? at time ?$t$?, ?$(+, r) \notin v$?, for some server ?$r$?, and ?$\mathit{preconditions}(v, t - 1) = \top$?, then ?$(-, r) \notin \mathit{set}$?.last().
        ?\textbf{- Safety:}? Let ?$\mathit{set}$? be committed by ?$\mathit{vg}(v)$? at time ?$t$? and ?let? ?$\mathit{preconditions}(v, t - 1) = \top$?. Then, ?$(v^*, \emptyset, \epsilon^*) \in \Lambda_{t-1}^v$? or ?$(v^*, \mathit{set}, \epsilon^*) \in \Lambda_{t-1}^v$?, for some view ?$v^*$? and some evidence ?$\epsilon^*$?.
        ?\textbf{- Decision Certification:}? If a correct server decides ?$\mathit{set}$? with a certificate ?$\omega$? from ?$\mathit{vg}(v)$?, then verify_output(?$\mathit{set}$?, ?$\mathit{vg}(v)$?, ?$\omega$?) ?$= \top$?.
        ?\textbf{- Decision Permission:}? If a set of views is committed by ?$\mathit{vg}(v)$?, then at least ?$v$?.plurality() of correct members of ?$v$? have previously started ?$\mathit{vg}(v)$?.
        ?\textbf{- Bounded Decisions:}? Let ?$\mathit{SET} = \{\mathit{set} \,|\, \mathit{set} \text{ is committed by } \mathit{vg}(v)\}$?. Then, ?$|\mathit{SET}| < \infty$?.
        ?\textbf{- Liveness:}? If all correct servers ?$r \in v$?.members() start ?$\mathit{vg}(v)$?, at least one correct server proposes to ?$\mathit{vg}(v)$? and no correct server stops ?$\mathit{vg}(v)$?, then all correct servers ?$r \in v$?.members() eventually decide from ?$\mathit{vg}(v)$?.
\end{lstlisting}

In order to present the view generator primitive, we start by introducing \emph{reconfiguration lattice agreement} (RLA).

\subsection{Reconfiguration Lattice Agreement (RLA)}

The RLA primitive is extremely similar to the Byzantine lattice agreement~\cite{di2020byzantine,di2020synchronous,zheng2019byzantine} primitive.
In the Byzantine lattice agreement (BLA) primitive, each process $p$ starts with its proposal $\mathit{pro}_p \in \mathcal{V}$ and eventually decides $\mathit{dec}_p \in \mathcal{V}$.
Values from the $\mathcal{V}$ set form a join semi-lattice $L = (\mathcal{V}, \oplus)$ for a commutative join operation $\oplus$.
That is, for any two $v_1, v_2 \in \mathcal{V}$, $v_1 \leq v_2$ if and only if $v_1 \oplus v_2 = v_2$.
If $V = \{v_1, v_2, ..., v_n\} \subseteq \mathcal{V}$, $\bigoplus(V) = v_1 \oplus v_2 \oplus ... \oplus v_n$.
BLA that tolerates up to $f$ Byzantine failures ensures:
\begin{compactitem}
    \item \emph{Liveness:} Each correct process eventually outputs its decision value $\mathit{dec} \in \mathcal{V}$.
    
    \item \emph{Stability:} Each correct process outputs a unique decision value.
    
    \item \emph{Comparability:} Given any two correct processes $p, q$, either $\mathit{dec}_p \leq \mathit{dec}_q$ or $\mathit{dec}_q \leq \mathit{dec}_p$.
    
    \item \emph{Inclusivity:} Given any correct process $p$, we have that $\mathit{pro}_p \leq \mathit{dec}_p$.
    
    \item \emph{Non-Triviality:} Given any correct process $p$, we have that $\mathit{dec}_p \leq \bigoplus(X \cup B)$, where $X$ is the set of proposed values of all correct processes ($X = \{\mathit{pro}_q \,|\, q \text{ is correct}\}$) and $B \subseteq \mathcal{V}$ satisfies $|B| \leq f$.
\end{compactitem}

RLA considers a specific semi-lattice over sets, i.e., $\mathcal{V}$ is a set of sets in RLA and $\oplus$ is the union operation.
More specifically, $\mathcal{V}$ is a set of sets of objects of abstract type $\mathtt{X}$, where the function $\mathtt{valid}$ is defined for each object of type $X$ and each view.
Finally, in contrast to BLA, not all correct processes are required to propose in RLA.

Each instance of the RLA primitive is associated with a single view.
However, an RLA instance is not completely defined by its view (in contrast to instances of the view generator primitive), i.e., there can exist multiple instances of the RLA primitive associated with the same view.
Therefore, each instance of the RLA primitive is also associated with a unique identifier.
In summary, one instance of the RLA primitive is completely defined by its view and its identifier; the instance associated with view $v$ and identifier $\mathit{id}$ is denoted by $\mathit{rla}(v, \mathit{id})$.
As for the failure model assumed by RLA, (at least) a quorum of members of view $v$ are assumed to be correct.

We now introduce the interface of an RLA instance.

\begin{lstlisting}[
  caption={RLA - interface},
  label={lst:rla_interface},
  escapechar = ?]
?\textbf{Reconfiguration Lattice Agreement:}?
    instance ?$\mathit{rla}(v, \mathit{id})$?, with ?$v$? is View and ?$\mathit{id}$? is Integer
    
    ?\textcolor{plainorange}{Interface:}?
        ?\textcolor{blue}{Requests:}?
            ?$\bullet$? <?$\mathit{rla}(v, \mathit{id})$?.Start | Set(X) ?$\mathit{pro}$?>: starts ?$\mathit{rla}(v, \mathit{id})$? with proposal ?$\mathit{pro}$?; if ?$\mathit{pro} = \bot$?, then ?$\mathit{rla}(v, \mathit{id})$? is started without a proposal. If ?$\mathit{pro} \neq \bot$?, then the server has proposed to ?$\mathit{rla}(v, \mathit{id})$?.            
            ?$\bullet$? <?$\mathit{rla}(v, \mathit{id})$?.Stop>: stops ?$\mathit{rla}(v, \mathit{id})$?.
            ?$\bullet$? <?$\mathit{rla}(v, \mathit{id})$?.Propose | Set(X) ?$\mathit{pro}$?>: proposes ?$\mathit{pro}$? to ?$\mathit{rla}(v, \mathit{id})$?.
            
        ?\textcolor{blue}{Indications:}?
            ?$\bullet$? <?$\mathit{rla}(v, \mathit{id})$?.Decide | Set(X) ?$\mathit{dec}$?, Certificate ?$\omega$?>: indicates that ?$\mathit{dec}$? is decided with the certificate ?$\omega$?.
            
        ?\textcolor{blue}{Rules:}?
            1) No correct server ?$r \notin v$?.members() invokes any request to ?$\mathit{rla}(v, \mathit{id})$?. ?\label{line:rule_requests_from_members_only}?
            2) Every correct server invokes each request at most once. ?\label{line:rule_request_at_most_once}?
            3) If a correct server invokes <?$\mathit{rla}(v, \mathit{id})$?.Start | ?$\mathit{pro} \neq \bot$?>, then the server does not invoke?\\?<?$\mathit{rla}(v, \mathit{id})$?.Propose>. ?\label{line:rule_no_propose_after_start_with_proposal}?
            4) No correct server invokes <?$\mathit{rla}(v, \mathit{id})$?.Stop> or <?$\mathit{rla}(v, \mathit{id})$?.Propose> unless it has already invoked <?$\mathit{rla}(v, \mathit{id})$?.Start>.
            5) A correct server invokes <?$\mathit{rla}(v, \mathit{id})$?.Stop> before halting.
            6) If a correct server invokes  <?$\mathit{rla}(v, \mathit{id})$?.Start | ?$\mathit{pro} \neq \bot$?> or <?$\mathit{rla}(v, \mathit{id})$?.Propose | ?$\mathit{pro}$?>, then (1) ?$\mathit{pro} = \{x\}$\footnote{The proposal contains a single element.}?, and (2) valid(?$x, v$?) = ?$\top$?. ?\label{line:rule_valid_proposals}?
\end{lstlisting}

If a (correct or faulty) process obtains (i.e., stores in its local memory) a certificate $\omega$, where $\mathtt{verify\_output}(\mathit{dec}, \mathit{rla}(v, \mathit{id}), \omega) = \top$, for some set $\mathit{dec}$ of objects of type $\mathtt{X}$, then we say that $\mathit{dec}$ is \emph{committed} by $\mathit{rla}(v, \mathit{id})$.
Moreover, $\mathit{dec}$ is committed at time $t$ if and only if the certificate attainment happens at time $t' \leq t$.
Note that we assume that an attainment of a certificate is irrevocable; hence, if $\mathit{dec}$ is committed at time $t$, then $\mathit{dec}$ is committed at any time $t' \geq t$.
Finally, we present the properties of an RLA instance.

\begin{lstlisting}[
  caption={RLA - properties},
  label={lst:rla_properties},
  escapechar = ?]
?\textbf{Reconfiguration Lattice Agreement:}?
    instance ?$\mathit{rla}(v, \mathit{id})$?, with ?$v$? is View and ?$\mathit{id}$? is Integer

    Properties?\textcolor{plainorange}{:}?
        ?\textbf{- Integrity:}? No correct server ?$r \notin v$?.members() receives any indication from ?$\mathit{rla}(v, \mathit{id})$?.
        ?\textbf{- Decision Certification:}? If a correct server decides ?$\mathit{dec}$? with a certificate ?$\omega$? from ?$\mathit{rla}(v, \mathit{id})$?, then verify_output(?$\mathit{dec}$?, ?$\mathit{rla}(v, \mathit{id})$?, ?$\omega$?) ?$= \top$?.
        ?\textbf{- Comparability:}? If ?$\mathit{dec}$? and ?$\mathit{dec}'$? are committed by ?$\mathit{rla}(v, \mathit{id})$?, then either ?$\mathit{dec} \subseteq \mathit{dec}'$? or ?$\mathit{dec} \supset \mathit{dec}'$?.
        ?\textbf{-}? ?\textbf{Validity:}? If ?$\mathit{dec}$? is committed by ?$\mathit{rla}(v, \mathit{id})$?, ?$\mathit{dec} \neq \emptyset$?, and, for every ?$x \in \mathit{dec}$?, valid(?$x, v$?) = ?$\top$?.
        ?\textbf{- Decision Permission:}? If a set is committed by ?$\mathit{rla}(v, \mathit{id})$?, then at least ?$v$?.plurality() of correct members of ?$v$? have previously started ?$\mathit{rla}(v, \mathit{id})$?.
        ?\textbf{- Bounded Decisions:}? Let ?$\mathit{DEC} = \{\mathit{dec} \,|\, \mathit{dec} \text{ is committed by } \mathit{rla}(v, \mathit{id})\}$?. Then, ?$|\mathit{DEC}| < \infty$?.
        ?\textbf{- Liveness:}? If all correct servers ?$r \in v$?.members() start ?$\mathit{rla}(v, \mathit{id})$?, at least one correct server proposes to ?$\mathit{rla}(v, \mathit{id})$? and no correct server stops ?$\mathit{rla}(v, \mathit{id})$?, then all correct servers ?$r \in v$?.members() eventually decide from ?$\mathit{rla}(v, \mathit{id})$?.
\end{lstlisting}

\vspace{-\baselineskip} 
\para{Implementation}
We now give an implementation of the $\mathit{rla}(v, \mathit{id})$ instance of the RLA primitive.
Recall that the failure model of $\mathit{rla}(v, \mathit{id})$ assumes (at least) a quorum of correct members of view $v$.
Since $\mathit{rla}(v, \mathit{id})$ is associated with view $v$, a correct server $r \in v.\mathtt{members()}$ discards all messages sent by processes that are not members of $v$; for brevity, this check is omitted from \cref{lst:rla_proposer,lst:rla_acceptor}.

Moreover, we assume that there exists the reliable broadcast primitive that allows processes to ``reliably'' broadcast a message $m$ to a fixed set of processes; we write ``\textcolor{blue}{$\mathtt{reliably} \text{ } \mathtt{broadcast}$} $m \text{ } \mathtt{to} \text{ } S$'', where $S$ is a set of processes.
The reliable broadcast primitive can be implemented in static systems~\cite{cachin2011introduction}, i.e., in systems in which correct processes never halt.
However, any implementation of the primitive in a static environment implements the primitive in a dynamic system (in which correct processes might halt) if we assume that the properties of the primitive need to be ensured only if no correct process halts.

The presented implementation is highly inspired by the protocol given in~\cite{di2020byzantine}; the changes we introduce account for subtle differences between RLA and BLA.
For the presentational purposes, we distinguish two roles of servers implementing $\mathit{rla}(v, \mathit{id})$: \emph{proposer} and \emph{acceptor}.
However, the following implementation assumes that each server takes \emph{both} roles.

\vspace{-\baselineskip} 
\begin{lstlisting}[
  caption={RLA - proposer implementation},
  label={lst:rla_proposer},
  escapechar = ?]
?\textbf{Reconfiguration Lattice Agreement:}?
    instance ?$\mathit{rla}(v, \mathit{id})$?, with ?$v$? is View and ?$\mathit{id}$? is Integer

    ?\textcolor{plainorange}{Proposer Implementation:}?
        upon <?$\mathit{rla}(v, \mathit{id})$?.Init>: // initialization of ?\textcolor{gray}{$\mathit{rla}(v, \mathit{id})$}? for the proposer
            Set(X) ?$\mathit{proposed\_value} = \bot$?
            Integer ?$\mathit{init\_counter} = 0$?
            Set(X) ?$\mathit{Proposed\_Set} = \emptyset$?
            Set(Server) ?$\mathit{Ack\_Set} = \emptyset$?
            Set(Message) ?$\omega = \emptyset$? ?\label{line:init_omega}?
            Set(X) ?$\mathit{Safe\_Set} = \emptyset$?
            Set(Message) ?$\mathit{Waiting\_Msgs} = \emptyset$?
            String ?$\mathit{state} = $? "disclosing"
        
        // Disclosure Phase
        upon <?$\mathit{rla}(v, \mathit{id})$?.Start | Set(X) ?$\mathit{proposal}$?>:
            if ?$\mathit{proposal} \neq \bot$?:
                ?$\mathit{proposed\_value} = \mathit{proposal}$? ?\label{line:set_proposed_value_start_rla}? ?\label{line:success_propose_1}?
                ?$\mathit{Safe\_Set} = \mathit{Safe\_Set} \cup \mathit{proposed\_value}$? ?\label{line:modify_svs_start_rla}?
                ?$\mathit{Proposed\_Set} = \mathit{Proposed\_Set} \cup \mathit{proposed\_value}$? ?\label{line:modify_proposed_set_start_rla}?
                ?\textcolor{blue}{reliably broadcast}? ?$[$?DISCLOSURE(?$v$?, ?$\mathit{id}$?), ?$\mathit{proposed\_value}]$? to ?$v$?.members() ?\label{line:reliable_broadcast_start_rla}?
            start processing protocol messages of ?$\mathit{rla}(v, \mathit{id})$? ?\label{line:start_processing_rla}?
            
        upon <?$\mathit{rla}(v, \mathit{id})$?.Propose | Set(X) ?$\mathit{proposal}$?>:
            if ?$\mathit{proposed\_value} = \bot$?:
                ?$\mathit{proposed\_value} = \mathit{proposal}$? ?\label{line:set_proposed_value_propose_rla}?
                ?$\mathit{Safe\_Set} = \mathit{Safe\_Set} \cup \mathit{proposed\_value}$? ?\label{line:modify_svs_propose_rla}?
                ?$\mathit{Proposed\_Set} = \mathit{Proposed\_Set} \cup \mathit{proposed\_value}$? ?\label{line:modify_proposed_set_propose_rla}?
                ?\textcolor{blue}{reliably broadcast}? ?$[$?DISCLOSURE(?$v$?, ?$\mathit{id}$?), ?$\mathit{proposed\_value}]$? to ?$v$?.members() ?\label{line:reliable_broadcast_propose_rla}?
                
        upon <?$\mathit{rla}(v, \mathit{id})$?.Stop>:
            stop processing protocol messages of ?$\mathit{rla}(v, \mathit{id})$?
                
        upon reliably deliver Message ?$m =[$?DISCLOSURE(?$v$?, ?$\mathit{id}$?), Set(X) ?$\mathit{proposal}]$?: ?\label{line:reliable_deliver_rla}?
            if ?$|\mathit{proposal}| = 1$? and for all ?$x \in \mathit{proposal}: \mathtt{valid}(x, v) = \top$?: ?\label{line:check_reliable_deliver_rla}?
                ?$\mathit{Safe\_Set} = \mathit{Safe\_Set} \cup \mathit{proposal}$? ?\label{line:modify_svs_reliable_deliver_rla}?
                ?$\mathit{init\_counter} = \mathit{init\_counter} + 1$?
                if ?$\mathit{state} =$? "disclosing": ?\label{line:check_state_disclosing_la1}?
                    if ?$\mathit{proposed\_value} = \bot$?: ?\label{line:check_proposed_value_bot}?
                        ?$\mathit{proposed\_value} = \mathit{proposal}$? ?\label{line:set_proposed_value_reliable_deliver_rla}?
                        ?\textcolor{blue}{reliably broadcast}? ?$[$?DISCLOSURE(?$v$?, ?$\mathit{id}$?), ?$\mathit{proposed\_value}]$? to ?$v$?.members()
                    ?$\mathit{Proposed\_Set} = \mathit{Proposed\_Set} \cup \mathit{proposal}$? ?\label{line:modify_proposed_set_reliable_deliver_rla}?
        
        // Deciding Phase
        upon ?$\mathit{init\_counter} \geq v$?.quorum() and ?$\mathit{state} =$? "disclosing": ?\label{line:quorum_rla}?
            ?$\mathit{state} =$? "proposing"
            broadcast ?$[$?ACK_REQ(?$v$?, ?$\mathit{id}$?), ?$\mathit{Proposed\_Set}]$? to ?$v$?.members() ?\label{line:broadcast_ack_req_1}?
        
        upon receipt of Message ?$m$? such that ?$m =[$?ACK_CON(?$v$?, ?$\mathit{id}$?), Set(X) ?$\mathit{proposed}]$? or ?$m =[$?NACK(?$v$?, ?$\mathit{id}$?), Set(X) ?$\mathit{update}$?, Set(X) ?$\mathit{proposed}]$?:
            ?$\mathit{Waiting\_Msgs} = \mathit{Waiting\_Msgs} \cup \{m\}$?

        upon ?exists? Message ?$m \in \mathit{Waiting\_Msgs}$? such that Safe(?$m$?) ?$ = \top$? and ?$\mathit{state} =$? "proposing" and ?$m =[$?ACK_CON(?$v$?, ?$\mathit{id}$?), Set(X) ?$\mathit{proposed}]$? and ?$\mathit{proposed} = \mathit{Proposed\_Set}$?: ?\label{line:process_ack_con_rla}?
            ?$\mathit{Waiting\_Msgs} = \mathit{Waiting\_Msgs} \setminus \{m\}$?
            if ?$m$?.sender ?$\notin \mathit{Ack\_Set}$?:
                ?$\mathit{Ack\_Set} = \mathit{Ack\_Set} \cup \{m$?.sender?$\}$?
                ?$\omega = \omega \cup \{m\}$?
            
        upon ?exists? Message ?$m \in \mathit{Waiting\_Msgs}$? such that Safe(?$m$?) ?$= \top$? and ?$\mathit{state} =$? "proposing" and ?$m =[$?NACK(?$v$?, ?$\mathit{id}$?), Set(X) ?$\mathit{update}$?, Set(X) ?$\mathit{proposed}]$? and ?$\mathit{proposed} = \mathit{Proposed\_Set}$?: ?\label{line:nack_rule}?
            ?$\mathit{Waiting\_Msgs} = \mathit{Waiting\_Msgs} \setminus \{m\}$?
            if ?$\mathit{update} \cup \mathit{Proposed\_Set} \neq \mathit{Proposed\_Set}$?: ?\label{line:check_nack_2_la1}?
                ?$\mathit{Proposed\_Set} = \mathit{Proposed\_Set} \cup \mathit{update}$? ?\label{line:modify_proposed_set_nack_rla}?
                ?$\mathit{Ack\_Set} = \emptyset$?; ?$\omega = \emptyset$? ?\label{line:flush}?
                broadcast ?$[$?ACK_REQ(?$v$?, ?$\mathit{id}$?), ?$\mathit{Proposed\_Set}]$? to ?$v$?.members() ?\label{line:broadcast_ack_req}?
                    
          upon ?$|\mathit{Ack\_Set}| \geq v$?.quorum() and ?$\mathit{state} =$? "proposing": ?\label{line:decide_rule_la1}?
            ?$\mathit{state} =$? "decided" ?\label{line:modify_state_decided_la1}?
            ?\textbf{trigger}? <?$\mathit{rla}(v, \mathit{id})$?.Decide | ?$\mathit{Proposed\_Set}$?, ?$\omega$?> ?\label{line:decide_la1}?
        
        ?\textbf{function}? Safe(Message ?$m$?):
            if ?$m =[$?ACK_CON(?$v$?, ?$\mathit{id}$?), Set(X) ?$\mathit{proposed}]$?:
                ?\textbf{return}? ?$\mathit{proposed} \subseteq \mathit{Safe\_Set}$?
            if ?$m =[$?NACK(?$v$?, ?$\mathit{id}$?), Set(X) ?$\mathit{update}$?, Set(X) ?$\mathit{proposed}]$?:
                ?\textbf{return}? ?$\mathit{update} \subseteq \mathit{Safe\_Set}$? ?\label{line:nack_update_safe_set}?
            ?\textbf{return}? ?$\bot$?
\end{lstlisting}

Next, we provide the acceptor implementation.

\vspace{-\baselineskip} 
\begin{lstlisting}[
  caption={RLA - acceptor implementation},
  label={lst:rla_acceptor},
  escapechar = ?]
?\textbf{Reconfiguration Lattice Agreement:}?
    instance ?$\mathit{rla}(v, \mathit{id})$?, with ?$v$? is View and ?$\mathit{id}$? is Integer

    ?\textcolor{plainorange}{Acceptor Implementation:}?
        upon <?$\mathit{rla}(v, \mathit{id})$?.Init>: // initialization of ?\textcolor{gray}{$\mathit{rla}(v, \mathit{id})$}? for the acceptor
            Set(X) ?$\mathit{Accepted\_Set} = \emptyset$?
            Set(Message) ?$\mathit{Waiting\_Msgs} = \emptyset$?
            Set(X) ?$\mathit{Safe\_Set}$? // reference to ?\textcolor{gray}{$\mathit{Safe\_Set}$}? in the corresponding proposer
            
        upon receipt of Message ?$m$? such that ?$m =[$?ACK_REQ(?$v$?, ?$\mathit{id}$?), Set(X) ?$\mathit{proposal}]$?:
            ?$\mathit{Waiting\_Msgs} = \mathit{Waiting\_Msgs} \cup \{m\}$?
        
        upon ?exists? Message ?$m \in \mathit{Waiting\_Msgs}$? such that Safe(?$m$?) ?$= \top$? and ?$m =[$?ACK_REQ(?$v$?, ?$\mathit{id}$?), Set(X) ?$\mathit{proposal}]$? and ?$\mathit{proposal} \neq \emptyset$?: ?\label{line:ack_req_check}?
            ?$\mathit{Waiting\_Msgs} = \mathit{Waiting\_Msgs} \setminus \{m\}$?
            if ?$\mathit{Accepted\_Set} \subseteq \mathit{proposal}$?: ?\label{line:check_accepted_set_la1}?
                ?$\mathit{Accepted\_Set} = \mathit{proposal}$? ?\label{line:modify_accepted_set_1_rla}?
                send ?$[$?ACK_CON(?$v$?, ?$\mathit{id}$?), ?$\mathit{Accepted\_Set}]$? to ?$m$?.sender ?\label{line:send_ack_con}?
            else:
                send ?$[$?NACK(?$v$?, ?$\mathit{id}$?), ?$\mathit{Accepted\_Set}$?, ?$\mathit{proposal}]$? to ?$m$?.sender ?\label{line:send_nack}?
                ?$\mathit{Accepted\_Set} = \mathit{Accepted\_Set} \cup \mathit{proposal}$? ?\label{line:modify_accepted_set_2_rla}?
        
        ?\textbf{function}? Safe(Message ?$m$?): ?\label{line:safe_function_acceptor}?
            if ?$m =[$?ACK_REQ(?$v$?, ?$\mathit{id}$?), Set(X) ?$\mathit{proposal}]$?:
                ?\textbf{return}? ?$\mathit{proposal} \subseteq \mathit{Safe\_Set}$? ?\label{line:return_safe_function_acceptor}?
            ?\textbf{return}? ?$\bot$?
\end{lstlisting}

In order to conclude the implementation of $\mathit{rla}(v, \mathit{id})$, we need to define when $\mathtt{verify\_output}(\mathit{dec}, \mathit{rla}(v, \mathit{id}), \omega) = \top$.

\begin{lstlisting}[
  caption={The $\mathtt{verify\_output}$ function for $\mathit{rla}(v, \mathit{id})$},
  label={lst:verify_output_rla},
  escapechar = ?]
?\textbf{function}? verify_output(Set(X) ?$\mathit{dec}$?, Distributed_Primitive_Instance ?$\mathit{rla}(v, \mathit{id})$?, Certificate ?$\omega$?):
    if ?$\omega$? is not Set(Message): ?\textbf{return}? ?$\bot$?
        
    Integer ?$\mathit{senders\_num} = |\{r\,|\, m \in \omega$? and ?$m$?.sender ?$ = r \text{ and } r \in v\text{.members()}\}|$?
    if ?$\mathit{senders\_num} < $? ?$v$?.quorum(): ?\textbf{return}? ?$\bot$?
    
    ?\textbf{return}? for every Message ?$m \in \omega$?, ?$m = [$?ACK_CON(?$v$?, ?$\mathit{id}$?), ?$\mathit{dec}]$?
\end{lstlisting}

\noindent \textit{Proof of correctness.}
We are now ready to prove the correctness of the implementation given in \cref{lst:rla_proposer,lst:rla_acceptor}.
We start by proving the safety properties of $\mathit{rla}(v, \mathit{id})$.

First, we prove the integrity property.

\begin{theorem} [Integrity]
Algorithm given in \cref{lst:rla_proposer,lst:rla_acceptor} satisfies integrity.
\end{theorem}
\begin{proof}
Since a server $r \notin v.\mathtt{members()}$ never starts $\mathit{rla}(v, \mathit{id})$, it never starts processing protocol messages of $\mathit{rla}(v, \mathit{id})$ (i.e., $r$ never executes line~\ref{line:start_processing_rla} of \Cref{lst:rla_proposer}).
Hence, the rule at line~\ref{line:decide_rule_la1} of \Cref{lst:rla_proposer} is never active at server $r$.
Thus, $r$ never decides from $\mathit{rla}(v, \mathit{id})$ and the theorem holds.
\end{proof}

Next, we prove decision certification.

\begin{theorem} [Decision Certification]
Algorithm given in \cref{lst:rla_proposer,lst:rla_acceptor} satisfies decision certification.
\end{theorem}
\begin{proof}
Let a correct server $r$ decide $\mathit{dec}$ with a certificate $\omega$.
We know that $\omega$ is a set of messages (by line~\ref{line:init_omega} of \Cref{lst:rla_proposer}) and that the number of distinct senders ``contained'' in $\omega$ is greater than or equal to $v.\mathtt{quorum()}$ (by line~\ref{line:decide_rule_la1} of \Cref{lst:rla_proposer}).
Moreover, all messages that belong to $\omega$ are for $\mathit{Proposed\_Set} = \mathit{dec}$ due to the check at line~\ref{line:process_ack_con_rla} of \Cref{lst:rla_proposer} and due to the ``flush'' of the $\mathit{Ack\_Set}$ variable (at line~\ref{line:flush} of \Cref{lst:rla_proposer}) once the value of $\mathit{Proposed\_Set}$ is updated.
\end{proof}

We now prove comparability.

\begin{theorem} [Comparability]
Algorithm given in \cref{lst:rla_proposer,lst:rla_acceptor} satisfies comparability.
\end{theorem}
\begin{proof}
Let $\mathit{dec}_1$ and let $\mathit{dec}_2$ be committed by $\mathit{rla}(v, \mathit{id})$.
Hence, at least $v.\mathtt{quorum()}$ acceptors have sent the $\mathtt{ACK\_CON}(v, \mathit{id})$ message for $\mathit{dec}_1$.
Similarly, $v.\mathtt{quorum()}$ acceptors have sent the $\mathtt{ACK\_CON}(v, \mathit{id})$ message for $\mathit{dec}_2$.
Hence, there exists at least a single correct acceptor that has sent both messages (due to the quorum intersection and the fact that at least $v.\mathtt{quorum()}$ members of $v$ are correct); let that acceptor be $k$.

Without loss of generality, we assume that $k$ has first sent the $\mathtt{ACK\_CON}(v, \mathit{id})$ message for $\mathit{dec}_1$.
At the moment $t$ of sending the $\mathtt{ACK\_CON}(v, \mathit{id})$ message for $\mathit{dec}_1$, we have that $\mathit{Accepted\_Set} = \mathit{dec}_1$ at server $k$ (line~\ref{line:modify_accepted_set_1_rla} of \Cref{lst:rla_acceptor}).
On the other hand, at the moment $t' \geq t$ of sending the $\mathtt{ACK\_CON}(v, \mathit{id})$ message for $\mathit{dec}_2$, we have that $\mathit{Accepted\_Set} = \mathit{dec}_2$ at $k$ (line~\ref{line:modify_accepted_set_1_rla} of \Cref{lst:rla_acceptor}).
Hence, we investigate how $\mathit{Accepted\_Set}$ changes after time $t$ at acceptor $k$.

Initially, at time $t$, $\mathit{Accepted\_Set} = \mathit{dec}_1$.
Hence, we set the induction hypothesis $\mathit{dec}_1 \subseteq \mathit{Accepted\_Set}$.
We investigate all possibilities for $k$ to modify its $\mathit{Accepted\_Set}$ variable:
\begin{compactitem}
    \item line~\ref{line:modify_accepted_set_1_rla} of \Cref{lst:rla_acceptor}: 
    The hypothesis is preserved because of the check at line~\ref{line:check_accepted_set_la1} of \Cref{lst:rla_acceptor}.
    
    \item line~\ref{line:modify_accepted_set_2_rla} of \Cref{lst:rla_acceptor}:
    The hypothesis is preserved because of the fact that the new value of $\mathit{Accepted\_Set}$ contains the old value of $\mathit{Accepted\_Set}$ and $\mathit{dec}_1$ is included in the old value of $\mathit{Accepted\_Set}$.
\end{compactitem}
Therefore, all values of the $\mathit{Accepted\_Set}$ variable at acceptor $k$ after time $t$ contain $\mathit{dec}_1$.
Thus, we have that $\mathit{dec}_1 \subseteq \mathit{dec}_2$, which concludes the theorem.
\end{proof}

The following lemma proves that, if $x \in \mathit{Safe\_Set}$ at a correct server, then $\mathtt{valid}(x, v) = \top$.

\begin{lemma} \label{lemma:safe_set_valid}
Let $r$ be a correct server.
If $x \in \mathit{Safe\_Set}$ at $r$, then $\mathtt{valid}(x, v) = \top$.
\end{lemma}
\begin{proof}
We examine all possible ways for a correct server to modify its $\mathit{Safe\_Set}$ variable:
\begin{compactitem}
    \item line~\ref{line:modify_svs_start_rla} of \Cref{lst:rla_proposer}: By rule 6 of $\mathit{rla}(v, \mathit{id})$ (line~\ref{line:rule_valid_proposals} of \Cref{lst:rla_interface}), a single element $x$ is added to $\mathit{Safe\_Set}$ and $\mathtt{valid}(x, v) = \top$.
    
    \item line~\ref{line:modify_svs_propose_rla} of \Cref{lst:rla_proposer}: 
    As in the previous case, a single element $x$ is added to $\mathit{Safe\_Set}$ and $\mathtt{valid}(x, v) = \top$ (by rule 6 of $\mathit{rla}(v, \mathit{id})$).
    
    \item line~\ref{line:modify_svs_reliable_deliver_rla} of \Cref{lst:rla_proposer}:
    A single element $x$ is added to $\mathit{Safe\_Set}$ and $\mathtt{valid}(x, v) = \top$ (by the check at line~\ref{line:check_reliable_deliver_rla} of \Cref{lst:rla_proposer}).
\end{compactitem}
The fact that, at all times, for every $x \in \mathit{Safe\_Set}$, $\mathtt{valid}(x, v) = \top$ concludes the theorem.
\end{proof}

The following theorem proves validity.

\begin{theorem} [Validity]
Algorithm given in \cref{lst:rla_proposer,lst:rla_acceptor} satisfies validity.
\end{theorem}
\begin{proof}
Let $\mathit{dec}$ be committed by $\mathit{rla}(v, \mathit{id})$.
Since no correct acceptor ever sends the $[\mathtt{ACK\_CON}(v, \mathit{id}), \emptyset]$ message (due to the check at line~\ref{line:ack_req_check} of \Cref{lst:rla_acceptor}), $\mathit{dec} \neq \emptyset$.

Since $\mathit{dec}$ is committed by $\mathit{rla}(v, \mathit{id})$, there exists a correct acceptor $k \in v.\mathtt{members}()$ that has sent the $\mathtt{ACK\_CON}(v, \mathit{id})$ message for $\mathit{dec}$.
At that time, we know that $\mathit{dec} \subseteq \mathit{Safe\_Set}$ at server $k$ (by line~\ref{line:ack_req_check} of \Cref{lst:rla_acceptor}).
Therefore, the theorem conclusion follows from \Cref{lemma:safe_set_valid}.
\end{proof}

The next safety property we have to prove is the bounded decisions property.
In order to do so, we first prove that $\mathit{Accepted\_Set} \subseteq \mathit{Safe\_Set}$ at a correct server at all times.

\begin{lemma} \label{lemma:accepted_set_svs}
At all times, $\mathit{Accepted\_Set} \subseteq \mathit{Safe\_Set}$ at a correct acceptor.
\end{lemma}
\begin{proof}
Initially, $\mathit{Accepted\_Set} = \mathit{Safe\_Set} = \emptyset$.
Thus, initially the invariant holds.

We now introduce the induction hypothesis $\mathit{Accepted\_Set} \subseteq \mathit{Safe\_Set}$.
Let us investigate all places at which the $\mathit{Accepted\_Set}$ variable or the $\mathit{Safe\_Set}$ variable is changed:
\begin{compactitem}
    \item line~\ref{line:modify_svs_start_rla} of \Cref{lst:rla_proposer}:
    Holds because of the induction hypothesis.
    
    \item line~\ref{line:modify_svs_propose_rla} of \Cref{lst:rla_proposer}:
    Holds because of the induction hypothesis.
    
    \item line~\ref{line:modify_svs_reliable_deliver_rla} of \Cref{lst:rla_proposer}:
    Holds because of the induction hypothesis.

    \item line~\ref{line:modify_accepted_set_1_rla} of \Cref{lst:rla_acceptor}: 
    Holds because $\mathit{proposal} \subseteq \mathit{Safe\_Set}$ (because of the $\mathtt{Safe}()$ function; line~\ref{line:return_safe_function_acceptor} of \Cref{lst:rla_acceptor}).
    
    \item line~\ref{line:modify_accepted_set_2_rla} of \Cref{lst:rla_acceptor}:
    Holds because of the induction hypothesis and the fact that $\mathit{proposal} \subseteq \mathit{Safe\_Set}$ (because of the $\mathtt{Safe}()$ function).
\end{compactitem}
Since the induction hypothesis is always preserved, the lemma holds.
\end{proof}

We are now ready to prove the bounded decisions property.

\begin{theorem} [Bounded Decisions] \label{theorem: bounded_decision_rla}
Algorithm given in \cref{lst:rla_proposer,lst:rla_acceptor} satisfies bounded decisions.
\end{theorem}
\begin{proof}
Let $\mathit{dec}$ be committed by $\mathit{rla}(v, \mathit{id})$.
We have that $\mathit{Accepted\_Set} = \mathit{dec}$ at some correct server $r$ (line~\ref{line:modify_accepted_set_1_rla} of \Cref{lst:rla_acceptor}).
By \Cref{lemma:accepted_set_svs}, we know that $\mathit{dec} \subseteq \mathit{Safe\_Set}$ at server $r$ at that time.

Let $\mathit{proposals} = \{\mathit{proposal} \,|\, \mathit{proposal} \subseteq \mathit{Safe\_Set} \text{ at a correct server at some time and } |\mathit{proposal}| = 1\}$.
We know that $|\mathit{proposals}| \leq |v.\mathtt{members}()|$ due to the fact that each $\mathit{proposal} \in \mathit{proposals}$ is broadcast using the reliable broadcast primitive.
Now, let $\mathit{SAFE\_SET} = \bigcup\limits_{\mathit{proposal} \in \mathit{proposals}} \mathit{proposal}$.
Hence, at any point in time, we have that $\mathit{Safe\_Set} \subseteq \mathit{SAFE\_SET}$ at a correct server.
That means that $\mathit{dec} \subseteq \mathit{SAFE\_SET}$, for any committed $\mathit{dec}$.
Since $\mathit{SAFE\_SET}$ is finite, it has finitely many subsets.
Therefore, the theorem holds.
\end{proof}

The last safety property we need to prove is decision permission.

\begin{theorem} [Decision Permission]
Algorithm given in \cref{lst:rla_proposer,lst:rla_acceptor} satisfies decision permission.
\end{theorem}
\begin{proof}
Assume that $\mathit{dec}$ is committed by $\mathit{rla}(v, \mathit{id})$.
Hence, at least $v.\mathtt{plurality()}$ of correct members of $v$ have sent the $[\mathtt{ACK\_CON}(v, \mathit{id}), \mathit{dec}]$ message (by \Cref{lst:verify_output_rla}).
Since a correct server does not send any $\mathtt{ACK\_CON}(v, \mathit{id})$ message before starting $\mathit{rla}(v, \mathit{id})$ (ensured by the fact that a correct server starts processing protocol messages of $\mathit{rla}(v, \mathit{id})$ only once it starts $\mathit{rla}(v, \mathit{id})$; line~\ref{line:start_processing_rla} of \Cref{lst:rla_proposer}), the decision permission property is ensured.
\end{proof}

Lastly, we prove liveness.
Recall that liveness is guaranteed only if all correct servers start $\mathit{rla}(v, \mathit{id})$, a correct server proposes to $\mathit{rla}(v, \mathit{id})$ and no correct server stops $\mathit{rla}(v, \mathit{id})$.
We implicitly assume that these hold in the following lemmas.

First, we show that $\mathit{Safe\_Set}$ variables of correct servers are eventually identical.
This is ensured because of the properties of the reliable broadcast primitive.

\begin{lemma} \label{lemma:svs_identical}
Eventually, $\mathit{Safe\_Set}_r = \mathit{Safe\_Set}_{r'}$, where $r$ and $r'$ are correct servers, $\mathit{Safe\_Set}_r$ denotes the value of $\mathit{Safe\_Set}$ at server $r$ and $\mathit{Safe\_Set}_{r'}$ denotes the value of $\mathit{Safe\_Set}$ at server $r'$.
\end{lemma}
\begin{proof}
A correct server updates its $\mathit{Safe\_Set}$ variable at following places in \Cref{lst:rla_proposer}:
\begin{compactitem}
    \item line~\ref{line:modify_svs_start_rla}: The update is reliably broadcast at line~\ref{line:reliable_broadcast_start_rla} of \Cref{lst:rla_proposer}.
    By rule 6 of $\mathit{rla}(v, \mathit{id})$, we know that $\mathit{proposed\_value} = \{x\}$ and $\mathtt{valid}(x, v) = \top$.
    Hence, whenever another correct server receives the $\mathtt{DISCLOSURE}(v, \mathit{id})$ message for $\mathit{proposed\_value}$, the check at line~\ref{line:check_reliable_deliver_rla} of \Cref{lst:rla_proposer} passes and $\mathit{proposed\_value}$ is included in $\mathit{Safe\_Set}$.
    
    \item line~\ref{line:modify_svs_propose_rla}: The update is reliably broadcast at line~\ref{line:reliable_broadcast_propose_rla} of \Cref{lst:rla_proposer}.
    By rule 6 of $\mathit{rla}(v, \mathit{id})$, we know that $\mathit{proposed\_value} = \{x\}$ and $\mathtt{valid}(x, v) = \top$.
    Hence, whenever another correct server receives the $\mathtt{DISCLOSURE}(v, \mathit{id})$ message for $\mathit{proposed\_value}$, the check at line~\ref{line:check_reliable_deliver_rla} of \Cref{lst:rla_proposer} passes and $\mathit{proposed\_value}$ is included in $\mathit{Safe\_Set}$.
    
    \item line~\ref{line:modify_svs_reliable_deliver_rla}: The update is reliably delivered (at line~\ref{line:reliable_deliver_rla} of \Cref{lst:rla_proposer}). 
    Due to the check at line~\ref{line:check_reliable_deliver_rla} of \Cref{lst:rla_proposer}, we know that $\mathit{proposal} = \{x\}$ and $\mathtt{valid}(x, v) = \top$.
    Hence, whenever another correct server receives the $\mathtt{DISCLOSURE}(v, \mathit{id})$ message for $\mathit{proposal}$ (which happens because of the totality and agreement properties of the reliable broadcast primitive), the check at line~\ref{line:check_reliable_deliver_rla} of \Cref{lst:rla_proposer} passes and $\mathit{proposal}$ is included in $\mathit{Safe\_Set}$.
\end{compactitem}
Thus, every update to $\mathit{Safe\_Set}$ of a correct server eventually reaches every other correct server due to the properties of the reliable broadcast primitive.
The lemma holds.
\end{proof}

Next, we prove that $\mathit{Proposed\_Set} \subseteq \mathit{Safe\_Set}$ at a correct server.

\begin{lemma} \label{lemma:proposed_set_svs}
At all times, $\mathit{Proposed\_Set} \subseteq \mathit{Safe\_Set}$ at a correct proposer.
\end{lemma}
\begin{proof}
Initially, $\mathit{Proposed\_Set} = \mathit{Safe\_Set} = \emptyset$.
Thus, initially the invariant holds.

We now introduce the induction hypothesis $\mathit{Proposed\_Set} \subseteq \mathit{Safe\_Set}$.
Let us investigate all places at which either $\mathit{Proposed\_Set}$ or $\mathit{Safe\_Set}$ variables are changed in \Cref{lst:rla_proposer}:
\begin{compactitem}
    \item line~\ref{line:modify_svs_start_rla}:
    Holds because of the induction hypothesis.
    
    \item line~\ref{line:modify_proposed_set_start_rla}:
    Holds because of the induction hypothesis and the fact that $\mathit{proposed\_value} \subseteq \mathit{Safe\_Set}$.
    
    \item line~\ref{line:modify_svs_propose_rla}:
    Holds because of the induction hypothesis.
    
    \item line~\ref{line:modify_proposed_set_propose_rla}:
    Holds because of the induction hypothesis and the fact that $\mathit{proposed\_value} \subseteq \mathit{Safe\_Set}$.
    
    \item line~\ref{line:modify_svs_reliable_deliver_rla}:
    Holds because of the induction hypothesis.
    
    \item line~\ref{line:modify_proposed_set_reliable_deliver_rla}:
    Holds because of the induction hypothesis and the fact that $\mathit{proposal} \subseteq \mathit{Safe\_Set}$.
    
    \item line~\ref{line:modify_proposed_set_nack_rla}:
    Holds because of the induction hypothesis and the fact that $\mathit{update} \subseteq \mathit{Safe\_Set}$ (because of the $\mathtt{Safe}()$ function; line~\ref{line:nack_update_safe_set} of \Cref{lst:rla_proposer}).
\end{compactitem}
Since the induction hypothesis is always preserved, the lemma holds.
\end{proof}

Now, we show that every message sent by a correct server is eventually ``safe'' at all other correct servers.

\begin{lemma} \label{lemma:sent_message_received}
Any $\mathtt{ACK\_CON}(v, \mathit{id})$, $\mathtt{NACK}(v, \mathit{id})$ or $\mathtt{ACK\_REQ}(v, \mathit{id})$ message sent by a correct server is eventually safe (i.e., the $\mathtt{Safe}()$ function returns $\top$) for any other correct server.
\end{lemma}
\begin{proof}
Let us first consider a message $m = [\mathtt{ACK\_CON}(v, \mathit{id}), \mathit{set}]$ sent by a correct server $r$ to a correct server $r'$.
Since $r$ has previously received the $\mathtt{ACK\_REQ}$ message for $\mathit{set}$ from $r'$, we know that $\mathit{Proposed\_Set} = \mathit{set}$ at server $r'$ at the moment of broadcasting this message (line~\ref{line:broadcast_ack_req_1} or line~\ref{line:broadcast_ack_req} of \Cref{lst:rla_proposer}).
Hence, $\mathit{set} \subseteq \mathit{Safe\_Set}$ at server $r'$ (by \Cref{lemma:proposed_set_svs}), which ensures that the lemma holds in this case.

Now, let $m = [\mathtt{NACK}(v, \mathit{id}), \mathit{set}, \mathit{proposed}]$ be sent by a correct server $r$.
At the moment of sending $m$, we know that $\mathit{Accepted\_Set} = \mathit{set}$ at server $r$ (by line~\ref{line:send_nack} of \Cref{lst:rla_acceptor}).
By \Cref{lemma:accepted_set_svs}, we know that $\mathit{set} \subseteq \mathit{Safe\_Set}$ at server $r$.
Eventually, by \Cref{lemma:svs_identical}, we know that $\mathit{set} \subseteq \mathit{Safe\_Set}$ at any other correct server.
Hence, $m$ eventually becomes safe.

Finally, let $m = [\mathtt{ACK\_REQ}(v, \mathit{id}), \mathit{set}, \mathit{number}]$ be broadcast by a correct server $r$.
We know that $\mathit{Proposed\_Set} = \mathit{set}$ (by lines~\ref{line:broadcast_ack_req_1} and~\ref{line:broadcast_ack_req} of \Cref{lst:rla_proposer}).
By \Cref{lemma:proposed_set_svs}, we know that $\mathit{set} \subseteq \mathit{Safe\_Set}$ at server $r$.
Eventually, by \Cref{lemma:svs_identical}, we know that $\mathit{set} \subseteq \mathit{Safe\_Set}$ at any other correct server.
Hence, $m$ eventually becomes safe, which concludes the lemma.
\end{proof}

Now, we show that a correct server modifies its proposal at most $\lfloor \frac{n - 1}{3} \rfloor$ times, where $n = |v.\mathtt{members}()|$.

\begin{lemma} \label{lemma:refine_f}
A correct server $r$ refines its proposal (i.e., executes line~\ref{line:broadcast_ack_req} of \Cref{lst:rla_proposer}) at most $\lfloor \frac{n - 1}{3} \rfloor$ times, where $n = |v.\mathtt{members}()|$.
\end{lemma}
\begin{proof}
The first proposal of server $r$ includes proposals from (at least) $v.\mathtt{quorum()}$ servers (lines~\ref{line:quorum_rla} and~\ref{line:broadcast_ack_req_1} of \Cref{lst:rla_proposer}).
Since $|\mathit{Safe\_Set}|$ at $r$ is at most $n$ (due to rule 6 and the check at line~\ref{line:check_reliable_deliver_rla} of \Cref{lst:rla_proposer}) and $n - v.\mathtt{quorum}() = \lfloor \frac{n - 1}{3} \rfloor$, $r$ can execute line~\ref{line:broadcast_ack_req} of \Cref{lst:rla_proposer} at most $\lfloor \frac{n - 1}{3} \rfloor$ times.
\end{proof}

The next lemma shows that a correct proposer never sends an $\mathtt{ACK\_REQ}$ message for $\emptyset$.

\begin{lemma} \label{lemma:no_empty_ack_req}
Let a correct server $r$ send the $[\mathtt{ACK\_REQ}(v, \mathit{id}), \mathit{proposal}]$ message.
Then, $\mathit{proposal} \neq \emptyset$.
\end{lemma}
\begin{proof}
The first $\mathtt{ACK\_REQ}$ message is sent for $\mathit{proposal}$ at line~\ref{line:broadcast_ack_req_1} of \Cref{lst:rla_proposer}, where $\mathit{proposal} = \mathit{Proposed\_Set}$.
By the check at line~\ref{line:check_reliable_deliver_rla} of \Cref{lst:rla_proposer}, we know that $\mathit{Proposed\_Set} \neq \emptyset$.

After the first sent $\mathtt{ACK\_REQ}$, other $\mathtt{ACK\_REQ}$ messages are sent at line~\ref{line:broadcast_ack_req} of \Cref{lst:rla_proposer} for $\mathit{Proposed\_Set}$.
Since $\mathit{Proposed\_Set}$ is never updated to $\emptyset$ and elements of $\mathit{Proposed\_Set}$ are never removed, we conclude the lemma.
\end{proof}

The next lemma shows that, if a correct proposer stops refining its proposal, then the proposer eventually decides.

\begin{lemma} \label{lemma:no_refine_decide}
If there exists a time $t$ after which a correct server $r$, which is in the $\mathtt{proposing}$ state, does not execute line~\ref{line:broadcast_ack_req} of \Cref{lst:rla_proposer}, then $r$ eventually decides.
\end{lemma}
\begin{proof}
Let $\mathit{proposal}$ be the last proposal for which $r$ sends an $\mathtt{ACK\_REQ}(v, \mathit{id})$ message in an execution.
Since $r$ does not execute line~\ref{line:broadcast_ack_req} of \Cref{lst:rla_proposer} after this message has been sent, either $r$ does not receive any $\mathtt{NACK}(v, \mathit{id})$ message or no received $\mathtt{NACK}(v, \mathit{id})$ message allows $r$ to pass the check at line~\ref{line:check_nack_2_la1} of \Cref{lst:rla_proposer}.
Since $r$ is correct, its $[\mathtt{ACK\_REQ}(v, \mathit{id}), \mathit{proposal}]$ message eventually reaches every correct acceptor (by \cref{lemma:sent_message_received,lemma:no_empty_ack_req}).
Then, each correct acceptor sends the $[\mathtt{ACK\_CON}(v, \mathit{id}), \mathit{proposal}]$ message to $r$.
By contradiction, suppose that a correct acceptor sends $[\mathtt{NACK}(v, \mathit{id}), \mathit{update}, \mathit{proposed}]$ to $r$, then:
\begin{compactitem}
    \item $\mathit{update} \not\subseteq \mathit{proposal}$ (by line~\ref{line:check_accepted_set_la1} of \Cref{lst:rla_acceptor}), and

    \item $m$ would eventually be ``safe'' at $r$ (due to \Cref{lemma:sent_message_received}) and, since $r$ does not update its $\mathit{Proposed\_Set}$ variable, line~\ref{line:nack_rule} of \Cref{lst:rla_proposer} eventually passes at $r$, and
        
    \item the check at line~\ref{line:check_nack_2_la1} of \Cref{lst:rla_proposer} passes at $r$ since $\mathit{update} \not\subseteq \mathit{proposal}$.
\end{compactitem}
Therefore, $r$ would refine its proposal, which contradicts the fact that $r$ never refines $\mathit{proposal}$.

Hence, once $r$ receives $\mathtt{ACK\_CON}(v, \mathit{id})$ messages from all correct servers, these messages are safe for $r$ (by \Cref{lemma:sent_message_received}) and, since $\mathit{Proposed\_Set}$ at server $r$ does not change, the rule at line~\ref{line:decide_rule_la1} of \Cref{lst:rla_proposer} is active, which implies that $r$ decides.
\end{proof}

Finally, we prove the liveness property of $\mathit{rla}(v, \mathit{id})$.

\begin{theorem}[Liveness]
Algorithm given in \Cref{lst:rla_proposer,lst:rla_acceptor} satisfies liveness.
\end{theorem}
\begin{proof}
Let a correct server $\mathit{proposer}$ propose to $\mathit{rla}(v, \mathit{id})$.
This means that every correct server $r$ eventually obtains proposals of $v.\mathtt{quorum()}$ servers (line~\ref{line:quorum_rla} of \Cref{lst:rla_proposer}) and sends an $\mathtt{ACK\_REQ}$ message (line~\ref{line:broadcast_ack_req_1} of \Cref{lst:rla_proposer}).
By \Cref{lemma:refine_f}, $r$ refines its proposal only finitely many times.
Once $r$ stops refining its proposal, $r$ decides (by \Cref{lemma:no_refine_decide}).
\end{proof}

\subsection{View Generator - Implementation} \label{subsubsection:view_generator_implementation}

Now that we have introduced RLA, we are ready to present our implementation of $\mathit{vg}(v)$.
Recall that the failure model of $\mathit{vg}(v)$ assumes (at least) a quorum of correct members of view $v$.

Our implementation of $\mathit{vg}(v)$ consists of two RLA instances: $\mathit{rla}(v, 1)$ and $\mathit{rla}(v, 2)$.
We provide the descriptions of both instances below:
\begin{compactenum}
    \item $\mathit{rla}(v, 1)$ description:
    Sets of tuples $(v', \mathit{set}', \epsilon')$, where $v'$ is a view, $\mathit{set}'$ is a set of views and $\epsilon'$ is an evidence, are processed by $\mathit{rla}(v, 1)$.
    Recall that such tuples are given as inputs to $\mathit{vg}(v)$.

    \item $\mathit{rla}(v, 2)$ description:
    Sets of tuples $(\mathit{set}', \epsilon')$ are processed by $\mathit{rla}(v, 2)$, where $\mathit{set}'$ is a set of views and $\epsilon'$ is an evidence (see \Cref{lst:set_proof}).

\begin{lstlisting}[
  caption={(Set, Evidence) tuple},
  label={lst:set_proof},
  escapechar = ?]
?\textbf{Set, Evidence:}?
    instance ?$(\mathit{set}', \epsilon')$?, with ?$\mathit{set}'$? is Set(View) and ?$\epsilon'$? is Evidence
    
    ?\textbf{function}? construct(Set(View, Set, Evidence) ?$\mathit{set}$?):
        Set(Set(View)) ?$\mathit{SET} = \{\mathit{set}^* \,|\, (v^*, \mathit{set}^*, \epsilon^*) \in \mathit{set} \text{ and } \mathit{set}^* \neq \emptyset\}$? ?\label{line:SEQ}?
        Set(View) ?$\mathit{VIEW} = \{v^* \,|\, (v^*, \mathit{set}^*, \epsilon^*) \in \mathit{set} \text{ and } \mathit{set}^* = \emptyset\}$? ?\label{line_VIEW}?
        Set(View) ?$\mathit{result\_set} = \bot$?
    
        if ?$|\mathit{VIEW}| = 0$?: // no ``stand-alone'' views
            ?$\mathit{result\_set} = $? max_cardinality?\footnote{We assume that the $\mathtt{max\_cardinality(\mathit{set})}$ deterministic function, where $\mathit{set}$ is a set of sets, returns the set with the greatest cardinality that belongs to $\mathit{set}$; if there are multiple sets that satisfy the condition, then any such set is returned.}?(?$\mathit{SET}$?) // get the set with the greatest cardinality ?\label{line:return_set_no_standalone_views}?
        else:
            Set(View) ?$\mathit{greatest\_set} = $? max_cardinality(?$\mathit{SET}$?) // get the set with the greatest cardinality
            View ?$\mathit{greatest\_view} = \emptyset$?
            if ?$\mathit{greatest\_set} \neq \bot$?:
                ?$\mathit{greatest\_view} =$? max_cardinality(?$\mathit{greatest\_set}$?) // get the view with the greatest cardinality?\footnote{Recall that a view is a set of changes.}? ?\label{line:view_from_greatest}?
            View ?$\mathit{union\_view} = \bigcup\limits_{v^* \in \mathit{VIEW}} v^*$?
            ?$\mathit{result\_set} = \{\mathit{greatest\_view} \cup \mathit{union\_view}\}$? ?\label{line:result_set_2}?
        ?\textbf{return}? ?$\mathit{result\_set}$?
    
    // we need to define the ?\textcolor{gray}{$\mathtt{valid}$}? function
    ?\textbf{function}? valid(Set, Evidence ?$(\mathit{set}', \epsilon')$?, View ?$v$?):
        if does not ?exist? Set(View, Set, Evidence) ?$\mathit{set}$? such that verify_output(?$\mathit{set}$?, ?$\mathit{rla}(v, 1)$?, ?$\epsilon'$?) ?$ = \top$?:
            ?\textbf{return}? ?$\bot$?
        ?\textbf{return}? construct(?$\mathit{set}$?) ?$= \mathit{set}'$?
\end{lstlisting}
\end{compactenum}

Finally, we give our implementation of $\mathit{vg}(v)$.

\begin{lstlisting}[
  caption={View generator - implementation},
  label={lst:view_generator_implementation},
  escapechar = ?]
?\textbf{View Generator:}?
    instance ?$\mathit{vg}(v)$?, with ?$v$? is View
    
    ?\textcolor{plainorange}{Implementation:}?
        upon <?$\mathit{vg}(v)$?.Init>: // initialization of ?\textcolor{gray}{$\mathit{vg}(v)$}?
            Bool ?$\mathit{started1} = \bot$?
            Bool ?$\mathit{started2} = \bot$?
    
        upon <?$\mathit{vg}(v)$?.Start | View, Set, Evidence ?$(v', \mathit{set}', \epsilon')$?>:
            ?$\mathit{started1} = \top$?
            if ?$(v', \mathit{set}', \epsilon') \neq \bot$?:
                ?\textbf{trigger}? <?$\mathit{rla}(v, 1)$?.Start | ?$\{(v', \mathit{set}', \epsilon')\}$?> ?\label{line:start_rla_1_1}?
            else:
                ?\textbf{trigger}? <?$\mathit{rla}(v, 1)$?.Start | ?$\bot$?> ?\label{line:start_rla_1_2}?
        
        upon <?$\mathit{vg}(v)$?.Propose | View, Set, Evidence ?$(v', \mathit{set}', \epsilon')$?>:
            ?\textbf{trigger}? <?$\mathit{rla}(v, 1)$?.Propose | ?$\{(v', \mathit{set}', \epsilon')\}$?>
        
        upon <?$\mathit{rla}(v, 1)$?.Decide | Set(View, Set, Evidence) ?$\mathit{dec}$?, Certificate ?$\omega$?>:
            Set(View) ?$\mathit{proposal} = $? construct(?$\mathit{dec}$?) // the function from ?\textcolor{gray}{\Cref{lst:set_proof}}?
            ?$\mathit{started2} = \top$?
            ?\textbf{trigger}? <?$\mathit{rla}(v, 2)$?.Start | ?$\{(\mathit{proposal}, \omega)\}$?> ?\label{line:start_rla_2}?
        
        upon <?$\mathit{rla}(v, 2)$?.Decide | Set(Set, Evidence) ?$\mathit{dec}$?, Certificate ?$\omega$?>: ?\label{line:vg_rule_decide}?
            Set(Set(View)) ?$\mathit{SET} = \{\mathit{set}^* \,|\, (\mathit{set}^*, \epsilon^*) \in \mathit{dec}\}$? ?\label{line: set}?
            ?\textbf{trigger}? <?$\mathit{vg}(v)$?.Decide | ?$\bigcup\limits_{\mathit{set}^* \in \mathit{SET}} \mathit{set}^*$?, ?$\omega$?> ?\label{line:union}?
        
        upon <?$\mathit{vg}(v)$?.Stop>:
            if ?$\mathit{started1} = \top$?:
                ?\textbf{trigger}? <?$\mathit{rla}(v, 1)$?.Stop>
            if ?$\mathit{started2} = \top$?:
                ?\textbf{trigger}? <?$\mathit{rla}(v, 2)$?.Stop>
\end{lstlisting}

In order to conclude the implementation of $\mathit{vg}(v)$, we need to define when $\mathtt{verify\_output}(\mathit{set}_{\mathit{res}}, \mathit{vg}(v), \omega) = \top$.

\begin{lstlisting}[
  caption={The $\mathtt{verify\_output}$ function for $\mathit{vg}(v)$},
  label={lst:verify_output_vg},
  escapechar = ?]
?\textbf{function}? verify_output(Set(View) ?$\mathit{set}_{\mathit{res}}$?, Distributed_Primitive_Instance ?$\mathit{vg}(v)$?, Certificate ?$\omega$?):
    if ?$\omega$? is not Set(Message): ?\textbf{return}? ?$\bot$? ?\label{line:check_omega_vg}?
    
    if does not ?exist? Set(Set, Evidence) ?$\mathit{set}$? such that verify_output(?$\mathit{set}$?, ?$\mathit{rla}(v, 2)$?, ?$\omega$?) ?$ = \top$?: ?\label{line:check_verify_vg}?
        ?\textbf{return}? ?$\bot$?
        
    Set(Set(View)) ?$\mathit{SET} = \{\mathit{set}^* \,|\, (\mathit{set}^*, \epsilon^*) \in \mathit{set}\}$?
    ?\textbf{return}? ?$\mathit{set}_{\mathit{res}} = \bigcup\limits_{\mathit{set}^* \in \mathit{SET}} \mathit{set}^*$?
\end{lstlisting}

\noindent \textit{Proof of correctness.}
Finally, we prove the correctness of $\mathit{vg}(v)$.
We start by proving the integrity property.

\begin{theorem} [Integrity]
Algorithm given in \Cref{lst:view_generator_implementation} satisfies integrity.
\end{theorem}
\begin{proof}
By integrity of $\mathit{rla}(v, 2)$, a correct server $r \notin v.\mathtt{members}()$ never decides from $\mathit{rla}(v, 2)$.
Hence, the rule at line~\ref{line:vg_rule_decide} of \Cref{lst:view_generator_implementation} is never active at server $r$, which means that $r$ never decides from $\mathit{vg}(v)$.
Therefore, integrity is satisfied.
\end{proof}

Next, we prove the comparability property.

\begin{theorem} [Comparability]
Algorithm given in \Cref{lst:view_generator_implementation} satisfies comparability.
\end{theorem}
\begin{proof}
Let $\mathit{set}_1$ and $\mathit{set}_2$ be committed by $\mathit{vg}(v)$.
Hence, a set $\mathit{set}_1^*$, where (1) $\mathit{set}_1^*$ is a set of tuples of the $\mathtt{Set, Evidence}$ type, and (2) $\mathit{set}_1 = \bigcup\limits_{(\mathit{set}^*, \epsilon^*) \in \mathit{set}_1^*} \mathit{set}^*$, is committed by $\mathit{rla}(v, 2)$ (by \Cref{lst:verify_output_vg}).
Similarly, a set $\mathit{set}_2^*$, where (1) $\mathit{set}_2^*$ is a set of tuples of the $\mathtt{Set, Evidence}$ type, and (2) $\mathit{set}_2 = \bigcup\limits_{(\mathit{set}^*, \epsilon^*) \in \mathit{set}_2^*} \mathit{set}^*$, is committed by $\mathit{rla}(v, 2)$ (by \Cref{lst:verify_output_vg}).
By comparability of $\mathit{rla}(v, 2)$, either $\mathit{set}_1^* \subseteq \mathit{set}_2^*$ or $\mathit{set}_1^* \supset \mathit{set}_2^*$.
We investigate all cases:
\begin{compactitem}
    \item $\mathit{set}_1^* = \mathit{set}_2^*$: In this case, $\mathit{set}_1 = \mathit{set}_2$ and the theorem holds.
    
    \item $\mathit{set}_1^* \subset \mathit{set}_2^*$: Therefore, every view $v_1 \in \mathit{set}_1$ belongs to $\mathit{set}_2$.
    Hence, the theorem holds.
    
    \item $\mathit{set}_1^* \supset \mathit{set}_2^*$: The case is symmetrical to the previous one.
\end{compactitem}
The theorem is satisfied in all possible cases, which concludes the proof.
\end{proof}

Next, we prove validity.
In order to do so, we prove that all sets committed by $\mathit{rla}(v, 1)$ at time $t$ are subsets of $\Lambda_{t-1}^v$.

\begin{lemma} \label{lemma:committed_sigma}
Let $\mathit{dec}$ be committed at $\mathit{rla}(v, 1)$ at time $t$.
Then, $\mathit{dec} \subseteq \Lambda_{t - 1}^v$.
\end{lemma}
\begin{proof}
Suppose that there exists a tuple $(v', \mathit{set}', \epsilon') \in \mathit{dec}$ such that $\mathtt{valid}((v', \mathit{set}', \epsilon'), v) = \bot$.
In this case, the validity property of $\mathit{rla}(v, 1)$ is violated.
Hence, $\mathtt{valid}((v', \mathit{set}', \epsilon'), v) = \top$, for every $(v', \mathit{set}', \epsilon') \in \mathit{dec}$.
Let $E_{\mathit{dec}} = \{\epsilon' \,|\, (v', \mathit{set}', \epsilon') \in \mathit{dec}\}$.

Since $\mathit{dec}$ is committed at time $t$, there exists a (correct or faulty) process that obtains $\omega_{\mathit{dec}}$ at time $t' \leq t$ such that $\mathtt{verify\_output}(\mathit{dec}, \mathit{rla}(v, 1), \omega_{\mathit{dec}}) = \top$.
We know that $\omega_{\mathit{dec}}$ is a set of $[\mathtt{ACK\_CON}(v, 1), \mathit{dec}]$ messages with $|\omega_{\mathit{dec}}| \geq v.\mathtt{quorum()}$ (by \Cref{lst:verify_output_rla}).
Hence, a correct server $r_{\mathit{dec}} \in v.\mathtt{members()}$ sends its $\mathtt{ACK\_CON}$ message by time $t$ (since $\omega_{\mathit{dec}}$ is obtained by time $t$) at line~\ref{line:send_ack_con} of \Cref{lst:rla_acceptor}.
Therefore, $\mathit{dec}$ is obtained before time $t$ (at line~\ref{line:modify_accepted_set_1_rla} of \Cref{lst:rla_acceptor}), which implies that every $\epsilon' \in E_{\mathit{dec}}$ is obtained before time $t$.
Thus, the lemma holds.
\end{proof}

Since $\Lambda_{t - 1}^v \subseteq \Lambda_{t'}^v$, for any time $t' > t - 1$, we conclude that $\mathit{dec} \subseteq \Lambda_{t'}^v$, if $\mathit{dec}$ is committed by $\mathit{rla}(v, 1)$ at time $t$.
Next, we prove that all valid inputs of $\mathit{rla}(v, 2)$ ``contain'' sequences if the $\mathit{preconditions}$ predicate is satisfied.

\begin{lemma} \label{lemma:all_valid_views_input_comparable_1}
Let a (correct or faulty) process obtain $\epsilon$ at time $t$, such that $o = (\mathit{set}, \epsilon)$ is a tuple of the $\mathtt{Set, Evidence}$ type and $\mathtt{valid}(o, v) = \top$.
Let $\mathit{preconditions}(v, t - 1) = \top$.
Then, $\mathit{set}$ is a sequence.
\end{lemma}
\begin{proof}
Since $\mathtt{valid}(o, v) = \top$, we conclude that a set $\mathit{set}^*$, where $\mathit{set}^*$ is a set of tuples of the $\mathtt{View, Set, Evidence}$ type and $\mathit{set} = \mathtt{construct}(\mathit{set}^*)$, is committed by $\mathit{rla}(v, 1)$ at time $t$ (by \Cref{lst:set_proof}).
Hence, $\mathit{set}^* \subseteq \Lambda_{t - 1}^v$ (by \Cref{lemma:committed_sigma}).
Let us take a closer look at the $\mathtt{construct}(\mathit{set}^*)$ function:
\begin{compactitem}
    \item Let $|\mathit{VIEW}| = 0$.
    Since $\mathit{preconditions}(v, t - 1) = \top$ and $\mathit{set}^* \subseteq \Lambda_{t - 1}^v$, we know that all sets from $\mathit{SET}$ are sequences.
    Hence, $\mathit{set}$ is a sequence.
    
    \item Let $|\mathit{VIEW}| \neq 0$.
    In this case, we have that $|\mathit{set}| = 1$.
    Hence, $\mathit{set}$ is trivially a sequence.
\end{compactitem}
The lemma holds.
\end{proof}

The next lemma builds upon \Cref{lemma:all_valid_views_input_comparable_1} by showing that
all views that belong to two valid inputs of $\mathit{rla}(v, 2)$ are comparable.

\begin{lemma} \label{lemma:all_valid_views_input_comparable}
Let a (correct or faulty) process obtain $\epsilon_1$ at time $t_1$, such that $o_1 = (\mathit{set}_1, \epsilon_1)$ is a tuple of the $\mathtt{Set, Evidence}$ type and $\mathtt{valid}(o_1, v) = \top$.
Moreover, let a (correct or faulty) process obtain $\epsilon_2$ at time $t_2 \geq t_1$, such that $o_2 = (\mathit{set}_2, \epsilon_2)$ is a tuple of the $\mathtt{Set, Evidence}$ type and $\mathtt{valid}(o_2, v) = \top$.
Let $v_1 \in \mathit{set}_1$, $v_2 \in \mathit{set}_2$ and $\mathit{preconditions}(v, t_2 - 1) = \top$.
Then, either $v_1 \subseteq v_2$ or $v_1 \supset v_2$.
\end{lemma}
\begin{proof}
Since $\mathtt{valid}(o_1, v) = \top$, we conclude that a set $\mathit{set}_1^*$, where $\mathit{set}_1^*$ is a set of tuples of the $\mathtt{View, Set, Evidence}$ type and $\mathit{set}_1 = \mathtt{construct}(\mathit{set}_1^*)$, is committed by $\mathit{rla}(v, 1)$ at time $t_1$ (by \Cref{lst:set_proof}).
Similarly, a set $\mathit{set}_2^*$, where $\mathit{set}_2^*$ is a set of tuples of the $\mathtt{View, Set, Evidence}$ type and $\mathit{set}_2 = \mathtt{construct}(\mathit{set}_2^*)$, is committed by $\mathit{rla}(v, 1)$ at time $t_2$ (by \Cref{lst:set_proof}).
By comparability of $\mathit{rla}(v, 1)$, we know that either $\mathit{set}_1^* \subseteq \mathit{set}_2^*$ or $\mathit{set}_1^* \supset \mathit{set}_2^*$.
Furthermore, we know that $\mathit{set}_1^* \subseteq \Lambda_{t_2 - 1}^v$ and $\mathit{set}_2^* \subseteq \Lambda_{t_2 - 1}^v$ (by \Cref{lemma:committed_sigma}).
We investigate all three possibilities:
\begin{compactitem}
    \item $\mathit{set}_1^* = \mathit{set}_2^*$:
    In this case, $\mathit{set}_1 = \mathit{set}_2$.
    Now, if $|\mathit{set}_1| = |\mathit{set}_2| = 1$, the lemma is satisfied.
    However, if $|\mathit{set}_1| = |\mathit{set}_2| \neq 1$, we need to show that $\mathit{set}_1 = \mathit{set}_2$ is a sequence.
    
    If $|\mathit{set}_1| = |\mathit{set}_2| \neq 1$, we know that $\mathit{set}_1 = \mathit{set}_2 = \mathtt{max\_cardinality}(\mathit{SET})$ in the $\mathtt{construct}$ function (line~\ref{line:return_set_no_standalone_views} of \Cref{lst:set_proof}).
    Since $\mathit{preconditions}(v, t_2 - 1) = \top$, we know that all sets from $\mathit{SET}$ are sequences.
    Therefore, $\mathit{set}_1 = \mathit{set}_2$ is a sequence.
    Thus, the lemma holds in this case.
    
    \item $\mathit{set}_1^* \subset \mathit{set}_2^*$:
    We further distinguish three cases:
    \begin{compactitem}
        \item $|\mathit{VIEW}_1| = |\mathit{VIEW}_2| = 0$, where $|\mathit{VIEW}_1|$ is the number of ``stand-alone'' views in $\mathit{set}_1^*$ (line~\ref{line_VIEW} of \Cref{lst:set_proof}) and $|\mathit{VIEW}_2|$ is the number of ``stand-alone'' views in $\mathit{set}_2^*$ (line~\ref{line_VIEW} of \Cref{lst:set_proof}):
        We know that $\mathit{set}_1^*, \mathit{set}_2^* \subseteq \Lambda_{t_2 - 1}^v$.
        Since $\mathit{preconditions}(v, t_2 - 1) = \top$, we know that $\mathit{set}_1$ is the greatest set of $\mathit{SET}_1$ with respect to the containment relation (line~\ref{line:SEQ} of \Cref{lst:set_proof}), i.e., $\mathit{set} \subseteq \mathit{set}_1$, for every $\mathit{set} \in \mathit{SET}_1$.
        Similarly, we know that $\mathit{set}_2$ is the greatest set of $\mathit{SET}_2$ with respect to the containment relation (line~\ref{line:SEQ} of \Cref{lst:set_proof}), i.e., $\mathit{set} \subseteq \mathit{set}_2$, for every $\mathit{set} \in \mathit{SET}_2$.
        Since $\mathit{preconditions}(v, t_2 - 1) = \top$ and $\mathit{set}_1^* \subset \mathit{set}_2^*$, we know that (1) $\mathit{set}_2$ is a sequence, and (2) $\mathit{set}_1 \subseteq \mathit{set}_2$.
        Finally, the fact that $v_1 \in \mathit{set}_2$ ensures that $\mathtt{comparable}(v_1, v_2) = \top$, which concludes the lemma in this case.

        \item $|\mathit{VIEW}_1| = 0$ and $|\mathit{VIEW}_2| \neq 0$:
        In this case, $\mathit{set}_2 = \{v_2\}$.
        Since $\mathit{set}_1^* \subseteq \Lambda_{t_2 - 1}^v$ and $\mathit{preconditions}(v, t_2 - 1) = \top$, we know that $\mathit{set}_1$ is the greatest set of $\mathit{SET}_1$ with respect to the containment relation (line~\ref{line:SEQ} of \Cref{lst:set_proof}), i.e., $\mathit{set} \subseteq \mathit{set}_1$, for every $\mathit{set} \in \mathit{SET}_1$.
        Moreover, we know that $\mathit{set}_1$ is a sequence.
        
        Since $\mathit{set}_1^* \subset \mathit{set}_2^*$, we know that $\mathit{set}_1 \in \mathit{SET}_2$.
        Let $\mathit{set}_2^{\mathit{max}}$ be the greatest set of $\mathit{SET}_2$ with respect to the containment relation (line~\ref{line:SEQ} of \Cref{lst:set_proof}), i.e., $\mathit{set} \subseteq \mathit{set}_2^{\mathit{max}}$, for every $\mathit{set} \in \mathit{SET}_2$; because $\mathit{set}_2^* \subseteq \Lambda_{t_2 - 1}^v$ and $\mathit{preconditions}(v, t_2 - 1) = \top$, the greatest set $\mathit{set}_2^{\mathit{max}}$ among $\mathit{SET}_2$ is indeed well defined.
        Moreover, $\mathit{set}_2^{\mathit{max}}$ is a sequence.
        We conclude that $\mathit{set}_1 \subseteq \mathit{set}_2^{\mathit{max}}$.
        Therefore, $\mathit{set}_1.\mathtt{last}() \subseteq \mathit{greatest\_view}$, where $\mathit{greatest\_view} = \mathtt{max\_cardinality}(\mathit{set}_2^{\mathit{max}}) = \mathit{set}_2^{\mathit{max}}.\mathtt{last()}$ (line~\ref{line:view_from_greatest} of \Cref{lst:set_proof}).
        This implies, since $\mathit{set}_1$ is a sequence, that $v_1 \subseteq v_2$, which concludes the lemma in this case.

        \item $|\mathit{VIEW}_1| \neq 0$ and $|\mathit{VIEW}_2| \neq 0$:
        Let $\mathit{max}_1$ be the greatest set of $\mathit{SET}_1$ with respect to the containment relation, i.e., $\mathit{set} \subseteq \mathit{max}_1$, for every $\mathit{set} \in \mathit{SET}_1$; since $\mathit{set}_1^* \subseteq \Lambda_{t_2 - 1}^v$ and $\mathit{preconditions}(v, t_2 - 1) = \top$, the greatest set $\mathit{max}_1$ among $\mathit{SET}_1$ is well defined. 
        Similarly, let $\mathit{max}_2$ be the greatest set of $\mathit{SET}_2$ with respect to the containment relation, i.e., $\mathit{set} \subseteq \mathit{max}_2$, for every $\mathit{set} \in \mathit{SET}_2$; since $\mathit{set}_2^* \subseteq \Lambda_{t_2 - 1}^v$ and $\mathit{preconditions}(v, t_2 - 1) = \top$, the greatest set $\mathit{max}_2$ among $\mathit{SET}_2$ is well defined. 
        
        If $\mathit{max}_1 \neq \bot$ and $\mathit{max}_2 \neq \bot$, we know that $\mathit{max}_1 \subseteq \mathit{max}_2$ (since $\mathit{set}_1^* \subset \mathit{set}_2^*$).
        Note that it is impossible that $\mathit{max}_1 \neq \bot$ and $\mathit{max}_2 = \bot$ due to the fact that $\mathit{set}_1^* \subset \mathit{set}_2^*$.
        Hence, $\mathit{greatest\_view}_1 \subseteq \mathit{greatest\_view}_2$, where $\mathit{greatest\_view}_1$ (resp., $\mathit{greatest\_view}_2$) is the value of the $\mathit{greatest\_view}$ variable at line~\ref{line:result_set_2} of \Cref{lst:set_proof} of the $\mathtt{construct}(\mathit{set}_1^*)$ (resp., $\mathtt{construct}(\mathit{set}_2^*)$) function.

        Similarly, we know that $\bigcup\limits_{v^* \in \mathit{VIEW}_1} v^* \subseteq \bigcup\limits_{v^{**} \in \mathit{VIEW}_2} v^{**}$ (by comparability of $\mathit{rla}(v, 1)$).
        Finally,\\$v_1 = \mathit{greatest\_view}_1 \cup \bigcup\limits_{v^* \in \mathit{VIEW}_1} v^*$ and $v_2 = \mathit{greatest\_view}_2 \cup \bigcup\limits_{v^{**} \in \mathit{VIEW}_2} v^{**}$, which proves that $v_1 \subseteq v_2$.
    \end{compactitem}
    
    \item $\mathit{set}_1^* \supset \mathit{set}_2^*$:
    This case is symmetrical to the previous one.
\end{compactitem}
The lemma holds.
\end{proof}

The last lemma we prove prior to proving the validity property shows that a set $\mathit{set}$ given as an input to $\mathit{rla}(v, 2)$ must satisfy $\mathit{set}.\mathtt{follows}(v) = \top$.

\begin{lemma} \label{lemma:follows}
Let a (correct or faulty) process obtain $\epsilon$ at time $t$, such that $o = (\mathit{set}, \epsilon)$ is a tuple of the $\mathtt{Set, Evidence}$ type and $\mathtt{valid}(o, v) = \top$.
Let $\mathit{preconditions}(v, t - 1) = \top$.
Then, $\mathit{set}.\mathtt{follows}(v) = \top$.
\end{lemma}
\begin{proof}
Since $\mathtt{valid}(o, v) = \top$, we conclude that a set $\mathit{set}^*$, where $\mathit{set}^*$ is a set of tuples of the $\mathtt{View, Set, Evidence}$ type and $\mathit{set} = \mathtt{construct}(\mathit{set}^*)$, is committed by $\mathit{rla}(v, 1)$ at time $t$ (by \Cref{lst:set_proof}).

By validity of $\mathit{rla}(v, 1)$, we know that $\mathit{set}^* \neq \emptyset$.
Moreover, $\mathit{set}^* \subseteq \Lambda_{t - 1}^v$ (by \Cref{lemma:committed_sigma}).
Let us take a closer look at every tuple $(v', \mathit{set}', \epsilon') \in \mathit{set}^*$:
\begin{compactitem}
    \item If $\mathit{set}' = \emptyset$, we know that $v \subset v'$ (due to the fact that $\mathit{preconditions}(v, t - 1) = \top$).
    
    \item Otherwise, $\mathit{set}'.\mathtt{follows}(v) = \top$ and $v' \in \mathit{set}'$ (due to the fact that $\mathit{preconditions}(v, t - 1) = \top$).
\end{compactitem}

Finally, we consider the $\mathtt{construct}(\mathit{set}^*)$ function:
\begin{compactitem}
    \item Let $|\mathit{VIEW}| = 0$.
    Since $\mathit{preconditions}(v, t - 1) = \top$ and $\mathit{set}^* \subseteq \Lambda_{t - 1}^v$, we know that $\mathit{set}.\mathtt{follows}(v) = \top$, which concludes the lemma.
    
    \item Let $|\mathit{VIEW}| \neq 0$.
    In this case, $\mathit{set}$ contains a single view.
    Moreover, we have that $v \subset v'$, for every ``stand-alone'' view $v' \in \mathit{VIEW}$.
    Therefore, $\mathit{set}.\mathtt{follows}(v) = \top$.
\end{compactitem}
Hence, the lemma holds.
\end{proof}

Finally, we are able to prove the validity property of $\mathit{vg}(v)$.

\begin{theorem} [Validity]
Algorithm given in \Cref{lst:view_generator_implementation} satisfies validity.
\end{theorem}
\begin{proof}
Let $\mathit{set}$ be committed by $\mathit{vg}(v)$ at time $t$ and let $\mathit{preconditions}(v, t - 1) = \top$.
Since $\mathit{set}$ is committed at time $t$ by $\mathit{vg}(v)$, a set $\mathit{set}_{\mathit{rla}(v, 2)}$ is committed by $\mathit{rla}(v, 2)$ at time $t$, where (1) $\mathit{set}_{\mathit{rla}(v, 2)}$ is a set of tuples of the $\mathtt{Set, Evidence}$ type, and (2) $\mathit{set} = \bigcup\limits_{(\mathit{set}^*, \epsilon^*) \in \mathit{set}_{\mathit{rla}(v, 2)}} \mathit{set}^*$ (by \Cref{lst:verify_output_vg}).
By validity of $\mathit{rla}(v, 2)$, we know that $\mathit{set}_{\mathit{rla}(v, 2)} \neq \emptyset$.
Moreover, because of the $\mathtt{construct}$ function, we know that $\mathit{set} \neq \emptyset$.

Next, we show that $\mathit{set}$ is a sequence.
Suppose, by contradiction, that $v_1, v_2 \in \mathit{set}$ such that $\mathtt{comparable}(v_1, v_2) = \bot$.
Since $v_1 \in \mathit{set}$, we know that $v_1 \in \mathit{set}^*$, where $(\mathit{set}^*, \epsilon^*) \in \mathit{set}_{\mathit{rla}(v, 2)}$.
Similarly, we know that $v_2 \in \mathit{set}^{**}$, where $(\mathit{set}^{**}, \epsilon^{**}) \in \mathit{set}_{\mathit{rla}(v, 2)}$.
We distinguish two scenarios:
\begin{compactitem}
    \item Let $\mathit{set}^* = \mathit{set}^{**}$.
    We know that $\mathtt{valid}((\mathit{set}^*, \epsilon^*), v) = \top$ (by validity of $\mathit{rla}(v, 2)$).
    Moreover, a process obtains $\epsilon^*$ at time $t' \leq t$.
    Since $\mathit{preconditions}(v, t - 1) = \top$, we conclude that $\mathit{set}^* = \mathit{set}^{**}$ is a sequence (by \Cref{lemma:all_valid_views_input_comparable_1}).
    Thus, $\mathtt{comparable}(v_1, v_2)$ cannot return $\bot$.
    
    \item Let $\mathit{set}^* \neq \mathit{set}^{**}$.
    By validity of $\mathit{rla}(v, 2)$, we know that $\mathtt{valid}((\mathit{set}^*, \epsilon^*), v) = \mathtt{valid}((\mathit{set}^{**}, \epsilon^{**}), v) = \top$.
    Since $\mathit{set}$ is committed at time $t$, we conclude that a (correct or faulty) process obtains $\epsilon^*$ (resp., $\epsilon^{**}$) at time $t^* \leq t$ (resp., $t^{**} \leq t$).
    Finally, the fact that $\mathit{preconditions}(v, t - 1) = \top$ implies that it is impossible that $\mathtt{comparable}(v_1, v_2) = \bot$ (by \Cref{lemma:all_valid_views_input_comparable}).
\end{compactitem}
Therefore, $\mathit{set}$ is indeed a sequence.

Finally, we show that $\mathit{set}.\mathtt{follows}(v) = \top$.
Suppose, by contradiction, that there exists a view $v^* \in \mathit{set}$ such that $v \not\subset v^*$.
Hence, there exists $(\mathit{set}^*, \epsilon^*) \in \mathit{set}_{\mathit{rla}(v, 2)}$ such that (1) $v^* \in \mathit{set}^*$, and (2) $\mathtt{valid}((\mathit{set}^*, \epsilon^*), v) = \top$ (by validity of $\mathit{rla}(v, 2)$).
Moreover, $\epsilon^*$ is obtained at time $t^* \leq t$ by a server.
Since $\mathit{preconditions}(v, t - 1) = \top$, we know that $\mathit{set}^*.\mathtt{follows}(v) = \top$ (by \Cref{lemma:follows}).
Hence, $v^* \notin \mathit{set}^*$, which concludes the theorem.
\end{proof}

The next theorem proves membership validity of $\mathit{vg}(v)$.

\begin{theorem} [Membership Validity]
Algorithm given in \Cref{lst:view_generator_implementation} satisfies membership validity.
\end{theorem}
\begin{proof}
Let $\mathit{preconditions}(v, t - 1) = \top$ and let $(+, r) \notin v$, for some server $r$.
Let $\mathit{set}$ be committed by $\mathit{vg}(v)$ at time $t$.
By contradiction, suppose that $(-, r) \in \mathit{set}.\mathtt{last()}$; recall that $\mathit{set}$ is a sequence because of the validity property of $\mathit{vg}(v)$.
Let $\mathit{set}.\mathtt{last()} = v^*$.

Since $\mathit{set}$ is committed at time $t$ by $\mathit{vg}(v)$, a set $\mathit{set}_{\mathit{rla}(v, 2)}$ is committed by $\mathit{rla}(v, 2)$ at time $t$, where (1) $\mathit{set}_{\mathit{rla}(v, 2)}$ is a set of tuples of the $\mathtt{Set, Evidence}$ type, and (2) $\mathit{set} = \bigcup\limits_{(\mathit{set}^*, \epsilon^*) \in \mathit{set}_{\mathit{rla}(v, 2)}} \mathit{set}^*$ (by \Cref{lst:verify_output_vg}).
We conclude that there exists a tuple $(\mathit{set}^*, \epsilon^*) \in \mathit{set}_{\mathit{rla}_{v, 2}}$ such that $v^* \in \mathit{set}^*$.
Moreover, we know that $\mathtt{valid}((\mathit{set}^*, \epsilon^*), v) = \top$ (by the validity property of $\mathit{rla}(v, 2)$).

Since $\epsilon^*$ is obtained by time $t$, a set $\mathit{set}_{\mathit{rla}(v, 1)}$ is committed by $\mathit{rla}(v, 1)$ at time $t$, where (1) $\mathit{set}_{\mathit{rla}(v, 1)}$ is a set of tuples of the $\mathtt{View, Set, Evidence}$ type, and (2) $\mathit{set^*} = \mathtt{construct}(\mathit{set}_{\mathit{rla}(v, 1)})$ (by \Cref{lst:set_proof}).
Hence, $\mathit{set}_{\mathit{rla}(v, 1)} \subseteq \Lambda_{t - 1}^v$ (by \Cref{lemma:committed_sigma}).
Finally, we consider the $\mathtt{construct}(\mathit{set}_{\mathit{rla}(v, 1)})$ function:
\begin{compactitem}
    \item Let $|\mathit{VIEW}| = 0$.
    Since $\mathit{preconditions}(v, t - 1) = \top$ and $\mathit{set}_{\mathit{rla}(v, 1)} \subseteq \Lambda_{t - 1}^v$, we know that $(-, r) \notin \mathit{set}^*.\mathtt{last()}$.
    Since $\mathit{set}^*$ is a sequence, that means that $v^* \notin \mathit{set}^*$.
    
    \item Let $|\mathit{VIEW}| \neq 0$.
    In this case, $\mathit{set}$ contains a single view.
    We know that $(-, r) \notin \mathit{greatest\_view}$ at line~\ref{line:result_set_2} of \Cref{lst:set_proof} since $\mathit{preconditions}(v, t - 1) = \top$.
    Moreover, we have that $(-, r) \notin v'$, for every  ``stand-alone'' view $v' \in \mathit{VIEW}$ (since $\mathit{preconditions}(v, t - 1) = \top$).
    Therefore, it is impossible that $v^* \in \mathit{set}^*$.
\end{compactitem}
The theorem holds.
\end{proof}

Next, we prove the safety property of $\mathit{vg}(v)$.

\begin{theorem} [Safety]
Algorithm given in \Cref{lst:view_generator_implementation} satisfies safety.
\end{theorem}
\begin{proof}
Let $\mathit{set}$ be committed by $\mathit{vg}(v)$ at time $t$.
In the rest of the proof, let $(v', \emptyset, \epsilon') \notin \Lambda_{t - 1}^v$.
In order to prove the theorem, we prove that $(v_{\mathit{set}}, \mathit{set}, \epsilon_{\mathit{set}}) \in \Lambda_{t - 1}^v$ in this case, for some view $v_{\mathit{set}}$ and some proof $\epsilon_{\mathit{set}}$.

Since $\mathit{set}$ is committed by $\mathit{vg}(v)$ at time $t$, we know that $\mathit{set}_{\mathit{rla}(v, 2)}$ is committed by $\mathit{rla}(v, 2)$ at time $t$, where $\mathit{set} = \bigcup\limits_{(\mathit{set}^*, \epsilon^*) \in \mathit{set}_{\mathit{rla}(v, 2)}} \mathit{set}^*$ (by \Cref{lst:verify_output_vg}).

Let $(\mathit{set}^*, \epsilon^*) \in \mathit{set}_{\mathit{rla}(v, 2)}$.
We know that $\mathtt{valid}((\mathit{set}^*, \epsilon^*), v) = \top$ (by validity of $\mathit{rla}(v, 2)$).
Hence, a set $\mathit{set}_{\mathit{rla}(v, 1)}$ of tuples of the $\mathtt{View, Set, Evidence}$ type is committed by $\mathit{rla}(v, 1)$ at time $t$, where $\mathit{set}^* = \mathtt{construct}(\mathit{set}_{\mathit{rla}(v, 1)})$ (by \Cref{lst:set_proof}).
By \Cref{lemma:committed_sigma}, we know that $\mathit{set}_{\mathit{rla}(v, 1)} \subseteq \Lambda_{t - 1}^v$.

Let us now take a closer look at the $\mathtt{construct}(\mathit{set}_{\mathit{rla}(v, 1)})$ function.
Since $(v', \emptyset, \epsilon') \notin \Lambda_{t-1}^v$ and $\mathit{preconditions}(v, t - 1) = \top$, we conclude that $\mathit{set}^* \in \mathit{SET}$ (because $|\mathit{VIEW}| = 0$).
Hence, $(v^*, \mathit{set}^*, \epsilon^*) \in \Lambda_{t - 1}^v$.

Because $\mathit{preconditions}(v, t - 1) = \top$, either $\mathit{set}_1^* \subseteq \mathit{set}_2^*$ or $\mathit{set}_1^* \supset \mathit{set}_2^*$, where $(\mathit{set}_1^*, \epsilon_1^*), (\mathit{set}_2^*, \epsilon_2^*) \in \mathit{set}_{\mathit{rla}(v, 2)}$.
Hence, $\mathit{set} = \mathit{set}^{**}$, where $(\mathit{set}^{**}, \epsilon^{**}) \in \mathit{set}_{\mathit{rla}(v, 2)}$.
Therefore, $(v_{\mathit{set}}, \mathit{set}, \epsilon_{\mathit{set}}) \in \Lambda_{t-1}^v$, for some view $v_{\mathit{set}}$ and some evidence $\epsilon_{\mathit{set}}$, which concludes the proof of the theorem.
\end{proof}

The next theorem proves the decision certification property.

\begin{theorem} [Decision Certification]
Algorithm given in \Cref{lst:view_generator_implementation} satisfies decision certification.
\end{theorem}
\begin{proof}
The theorem follows from the decision certification property of $\mathit{rla}(v, 2)$ and the fact that a correct server decides the union of all sets (at line~\ref{line:union} of \Cref{lst:view_generator_implementation}).
\end{proof}

Next, we prove the decision permission property.

\begin{theorem} [Decision Permission]
Algorithm given in \Cref{lst:view_generator_implementation} satisfies decision permission.
\end{theorem}
\begin{proof}
Let $\mathit{set}$ be committed by $\mathit{vg}(v)$.
Hence, $\mathit{set}_{\mathit{rla}(v, 2)}$ is committed by $\mathit{rla}(v, 2)$ (by line~\ref{line:check_verify_vg} of \Cref{lst:verify_output_vg}).
Hence, at least $v.\mathtt{plurality()}$ of correct members of $v$ have previously started $\mathit{rla}(v, 2)$ (by decision permission of $\mathit{rla}(v, 2)$).
All those servers have started $\mathit{rla}(v, 2)$ at line~\ref{line:start_rla_2} of \Cref{lst:view_generator_implementation}.

Therefore, $\mathit{set}_{\mathit{rla}(v, 1)}$ is committed by $\mathit{rla}(v, 1)$ (by the decision certification of $\mathit{rla}(v, 1)$).
Again, at least $v.\mathtt{plurality()}$ of correct members of $v$ have previously started $\mathit{rla}(v, 1)$ (by decision permission of $\mathit{rla}(v, 1)$).
All those servers have started $\mathit{rla}(v, 1)$ at line~\ref{line:start_rla_1_1} or at line~\ref{line:start_rla_1_2} of \Cref{lst:view_generator_implementation}.
Thus, all those servers have started $\mathit{vg}(v)$, which concludes the theorem.
\end{proof}

Next, we prove the bounded decisions property of $\mathit{vg}(v)$.

\begin{theorem} [Bounded Decisions]
Algorithm given in \Cref{lst:view_generator_implementation} satisfies bounded decisions.
\end{theorem}
\begin{proof}
Follows directly from the bounded decisions property of $\mathit{rla}(v, 2)$ (by \Cref{lst:verify_output_vg}).
\end{proof}

Finally, we prove the liveness property of $\mathit{vg}(v)$.

\begin{theorem} [Liveness]
Algorithm given in \Cref{lst:view_generator_implementation} satisfies liveness.
\end{theorem}
\begin{proof}
If every correct server $r \in v.\mathtt{members}()$ starts $\mathit{vg}(v)$, a correct server proposes to $\mathit{vg}(v)$ and no correct server $r' \in v.\mathtt{members}()$ stops $\mathit{vg}(v)$, every correct server eventually decides from $\mathit{rla}(v, 1)$ (by liveness of $\mathit{rla}(v, 1)$).
Then, every correct server starts $\mathit{rla}(v, 2)$ with a proposal.
By liveness of $\mathit{rla}(v, 2)$, all correct servers eventually decide from $\mathit{rla}(v, 2)$.
Thus, all correct servers eventually decide from $\mathit{vg}(v)$ at line~\ref{line:union} of \Cref{lst:view_generator_implementation}.
\end{proof}

%% file: appendix/storage_module.tex
\section{Server's Modules: Storage Module} \label{appx:storage_module}

As servers carry the most of the logic of \sysname, we separate the protocol executed by servers into a multiple \emph{modules}.
Each module of a server contains its own logic.

The first module of a server we present is the \emph{storage module}.
The storage module, as the name suggests, has the responsibility of storing fundamental information used by other modules.
For example, the storage module keeps track of received messages.

\para{``Install'' messages}
Servers exchange different types of messages in \sysname.
The information about the evolution of the system membership is carried by $\mathtt{INSTALL}$ messages.
We define these messages below.

\begin{lstlisting}[
  caption={$\mathtt{INSTALL}$ message},
  label={lst:install_messages},
  escapechar=?]
?\textbf{Install\_Message:}?
    instance ?$[$?INSTALL, ?$\mathit{source}$?, ?$\mathit{set}$?, ?$\omega]$?, with ?$\mathit{source}$? is View, ?$\mathit{set}$? is Set(View) and ?$\omega$? is Certificate such that verify_output(?$\mathit{set}$?, ?$\mathit{vg}(\mathit{source})$?, ?$\omega$?) ?$ = \top$?
    
    ?\textbf{function}? (Install_Message ?$m$?).source():
        let ?$m = [$?INSTALL, ?$\mathit{source}$?, ?$\mathit{set}$?, ?$\omega]$?
        ?\textbf{return}? ?$\mathit{source}$?
        
    ?\textbf{function}? (Install_Message ?$m$?).destination():
        let ?$m = [$?INSTALL, ?$\mathit{source}$?, ?$\mathit{set}$?, ?$\omega]$?
        ?\textbf{return}? min_cardinality?\footnote{We assume that the $\mathtt{min\_cardinality(\mathit{set})}$ deterministic function, where $\mathit{set}$ is a set of sets, returns the set with the smallest cardinality that belongs to $\mathit{set}$; if there are multiple sets that satisfy the condition, then any such set is returned.}?(?$\mathit{set}$?) // get the view with the smallest cardinality?\footnote{Recall that a view is a set of changes.}?
        
    ?\textbf{function}? (Install_Message ?$m$?).tail():
        let ?$m = [$?INSTALL, ?$\mathit{source}$?, ?$\mathit{set}$?, ?$\omega]$?
        ?\textbf{return}? ?$\mathit{set} \setminus\{m$?.destination()?$\}$?
\end{lstlisting}

As it can be seen from \Cref{lst:install_messages}, $\mathit{set}$ is a set of views.
However, we prove in the following that $\mathit{set}$ is, actually, a sequence.
Until we prove this claim, we treat $\mathit{set}$ as a set of views (as written in \Cref{lst:install_messages}).

\para{View-paths}
Next, we introduce \emph{view-paths}: ordered sets of ``adjacent'' $\mathtt{INSTALL}$ messages.

\begin{lstlisting}[
  caption={View-path},
  label={lst:view_path},
  escapechar = ?]
?\textbf{View\_Path:}?
    instance ?$[m_1, m_2, ..., m_k]$? with ?$k \geq 1$? is Integer, ?$m_i$? is Install_Message, for every ?$i \in [1, k]$?,?\\??$m_1$?.source() ?$ = \mathit{genesis}$? and ?$m_j$?.source() ?$= m_{j - 1}$?.destination(), for every ?$j \in [2, k]$? ?\label{line:instance_view_path}?
    
    ?\textbf{function}? (View_Path ?$\mathit{path}$?).views():
        let ?$\mathit{path} = [m_1, ..., m_k]$?
        ?\textbf{return}? ?$\{v \,|\, v = \mathit{genesis} \text{ or } v = m_i$?.destination(), where ?$m_i \in \mathit{path}\}$? ?\label{line:views_path_views}?
    
    ?\textbf{function}? (View_Path ?$\mathit{path}$?).destination():
        let ?$\mathit{path} = [m_1, ..., m_k]$?
        ?\textbf{return}? ?$m_k$?.destination()
        
    ?\textbf{function}? (View_Path ?$\mathit{path}$?).tail():
        let ?$\mathit{path} = [m_1, ..., m_k]$?
        ?\textbf{return}? ?$m_k$?.tail()
\end{lstlisting}

By slightly abusing the notion, sometimes we write ``$\mathit{path} \subseteq \mathit{set}$'', where $\mathit{path}$ is a view-path and $\mathit{set}$ is a set of $\mathtt{INSTALL}$ messages, to mean that every message that belongs to $\mathit{path}$ belongs to $\mathit{set}$, as well.

\para{Storage module - implementation}
We now give the implementation of the storage module.

\begin{lstlisting}[
  caption={Storage module - initialization},
  label={lst:storage_initialization},
  escapechar=?]
?\textbf{Storage Module:}?
    ?\textcolor{plainorange}{Implementation:}?
        upon start: // initialization of the module; executed as soon as the ?\textcolor{gray}{$\mathtt{start}$}? event is triggered
            Set(Message) ?$\mathit{waiting\_messages} = \emptyset$? // all received messages that have not been processed
            
            Set(View) ?$\mathit{history} = \{\mathit{genesis}\}$? // we prove that ?\textcolor{gray}{$\mathit{history}$}? is actually a sequence; see ?\textcolor{gray}{\Cref{subsection:safety_properties}}? ?\label{line:init_history}?
            Set(Install_Message) ?$\mathit{install\_messages} = \emptyset$? ?\label{line:init_install_messages}?
            
            Map(View ?$\to$? View) ?$\mathit{source} = \{\mathit{genesis} \to \bot\}$?
            // as we show further, these are actually sequences and not ``just'' sets of views; see ?\textcolor{gray}{\Cref{subsection:safety_properties}}?
            Map(View ?$\to$? Set(View)) ?$\mathit{sequence} = \{\mathit{genesis} \to \bot\}$? ?\label{line:sequence_init}?
            Map(View ?$\to$? View_Path) ?$\mathit{view\_path} = \{\mathit{genesis} \to \emptyset\}$?
\end{lstlisting}

\begin{lstlisting}[
  caption={Storage module - message processing},
  label={lst:storage_processing},
  escapechar=?]
?\textbf{Storage Module:}?
    ?\textcolor{plainorange}{Implementation:}?
        upon receipt of Message ?$m$?:
            ?$\mathit{waiting\_messages} = \mathit{waiting\_messages} \cup \{m\}$?
            
        upon ?exists? Install_Message ?$m \in \mathit{waiting\_messages}$? such that ?$m = [$?INSTALL, ?$v$?, ?$\mathit{set}$?, ?$\omega]$? and ?$v \in \mathit{history}$? and ?$m \notin \mathit{install\_messages}$?:  ?\label{line:update_history_rule}?
            ?$\mathit{waiting\_messages} = \mathit{waiting\_messages} \setminus{\{m\}}$?
            ?$\mathit{install\_messages} = \mathit{install\_messages} \cup \{m\}$? ?\label{line:update_install_messages}?
            extract_requests_and_voting_proofs(?$m$?) // see ?\textcolor{gray}{\Cref{lst:reconfiguration_extract_requests_passing_proofs}}? ?\label{line:extract_invocation}?
            View ?$v' = \mathit{m}.$?destination()
            ?$\mathit{history} = \mathit{history} \cup \{v'\}$? ?\label{line:update_history}?
            gossip ?$m$? ?\label{line:gossip_install}? // gossip the message
            if ?$\mathit{source}[v'] = \bot$? or ?$\mathit{source}[v'] \supset v$?:
                ?$\mathit{source}[v'] = v$?
                ?$\mathit{sequence}[v'] = \mathit{set}$?
                ?$\mathit{view\_path}[v'] = \mathit{view\_path}[v] || m$? // ?\textcolor{gray}{$||$}? denotes the concatenation
\end{lstlisting}

%% file: appendix/reconfiguration.tex
\section{Server's Modules: Reconfiguration Module} \label{appx:reconfiguration_module}

This section is devoted to the \emph{reconfiguration module} of a server.
The reconfiguration module carries the reconfiguration logic of \sysname.

\para{Valid views}
Next, we give the formal definition of \emph{valid views} (see \Cref{lst:view}).

\begin{definition} [Valid View] \label{definition:valid_view}
Let $\mathcal{I}^*(\infty)$ denote the set of $\mathtt{INSTALL}$ messages ever obtained by (correct or faulty) processes.
A view $v$ is \emph{valid} if and only if:
\begin{compactitem}
    \item $v = \mathit{genesis}$, or
    
    \item there exists a view-path $\mathit{path}$ such that $\mathit{path} \subseteq \mathcal{I}^*(\infty)$ and $\mathit{path}.\mathtt{destination()} = v$.
\end{compactitem}
\end{definition}
Recall that we assume that at least a quorum of members of a valid view is correct (see \Cref{appx:model}, paragraph ``Failure model'').

\para{Forever-alive views}
Next important concept we introduce is a concept of \emph{forever-alive views}.
Intuitively, a view $v$ is forever-alive if there exists a forever-correct process that ``knows'' a view-path to $v$ and ``shares'' that information by gossiping it.
We give the formal definition below.

\begin{definition} [Forever-Alive View] \label{definition:forever_alive_view}
We say that a valid view $v$ is \emph{forever-alive} if and only if:
\begin{compactitem}
    \item $v = \mathit{genesis}$, or
    
    \item there exists a forever-correct process $p$ such that (1) $p$ obtains a view-path $\mathit{path}$ with $\mathit{path}.\mathtt{destination()} = v$, and (2) $p$ gossips all messages that belong to $\mathit{path}$.
\end{compactitem}
\end{definition}
Observe that every forever-alive view is valid.

\para{Properties}
We now introduce a few properties that are satisfied by the reconfiguration (and storage) module.
A correct process \emph{updates} its current view to a view $v$ if and only if it triggers the special $\mathtt{update} \text{ } \mathtt{current} \text{ } \mathtt{view} \text{ } \mathtt{to} \text{ } \mathtt{view} \text{ } v$ event.
Similarly, a correct process \emph{installs} a view $v$ if and only if it triggers the special $\mathtt{install} \text{ } \mathtt{view} \text{ } v$ event.
Finally, a correct process \emph{stops processing} in a view $v$ if and only if it triggers the special $\mathtt{stop} \text{ } \mathtt{processing} \text{ } \mathtt{in} \text{ } \mathtt{view} \text{ } v'$ event, where $v' \supseteq v$.

\begin{lstlisting}[
  caption={Reconfiguration and storage modules - properties},
  label={lst:reconfiguration_properties_new},
  escapechar = ?]
?\textbf{Reconfiguration and Storage Module:}?
    ?\textcolor{plainorange}{Properties:}?
        ?\textbf{- View Comparability:}? All valid views are comparable.
        ?\textbf{- Finitely Many Valid Views:}? There exist only finitely many valid views.
        ?\textbf{- Finality:}? Let ?$v_{\mathit{final}}$? be the greatest forever-alive view.?\footnote{Note that $v_{\mathit{final}}$ is well-defined due to the fact that $\mathit{genesis}$ is a forever-alive view, all forever-alive views are valid, there are finitely many valid views and all valid views are comparable (by the view comparability property).}? Then, (1) all correct members of ?$v_{\mathit{final}}$? update their current view to ?$v_{\mathit{final}}$?, (2) no correct member of ?$v_{\mathit{final}}$? updates its current view to any view?\\?after updating it to ?$v_{\mathit{final}}$?, (3) all correct members of ?$v_{\mathit{final}}$? install ?$v_{\mathit{final}}$?, (4) no correct member of ?$v_{\mathit{final}}$? leaves, and (5) no correct member of ?$v_{\mathit{final}}$? stops processing ?in? ?$v_{\mathit{final}}$?.
\end{lstlisting}

\subsection{Reconfiguration Module - Implementation} \label{subsection:reconfiguration_implementation}

We now give the implementation of the reconfiguration module.
We start by introducing the evidence used as the input to the view generator primitive (see \Cref{appx:view_generator}).

\begin{lstlisting}[
  caption={Evidence used as the input for the view generator primitive},
  label={lst:vg_evidence},
  escapechar = ?]
?\textbf{View\_Generator\_Evidence:}?
    instance ?$\epsilon_{\mathit{vg}} = $? {
        View_Path ?$\mathit{path}$?
        Map(Change ?$\to$? Message) ?$\mathit{requests}$?
        Map(Change ?$\to$? Voting_Proof) ?$\mathit{voting\_proofs}$?
    }
\end{lstlisting}

Next, we define when are inputs of the view generator primitive deemed valid.

\begin{lstlisting}[
  caption={(View, Set, Evidence) tuple - the $\mathtt{valid}$ function},
  label={lst:view_set_evidence_valid},
  escapechar = ?]
?\textbf{function}? valid(View, Set, Evidence ?$(v', \mathit{set}', \epsilon')$?, View ?$v$?):
    if ?$v = \mathit{genesis}$?:
        if ?$v \not\subset v'$? or ?$\mathit{set}' \neq \emptyset$? or ?$\epsilon'$? is not View_Generator_Evidence: ?\label{line:v_equal_genesis}?
            ?\textbf{return}? ?$\bot$?
            
        if ?$\epsilon'.\mathit{path} \neq \bot$?:
            ?\textbf{return}? ?$\bot$?
        
        if ?$(+, r) \notin v$? and ?$(-, r) \in v'$?, for some Server ?$r$?: ?\label{line:check_minus_r_1}?
            ?\textbf{return}? ?$\bot$?
        
        for each Change ?$c \in v'$?:
            if ?$c \notin \mathit{genesis}$?:
                if ?$c = (+, r)$?, for some Server ?$r$?:
                    Message ?$m = \mathit{requests}[c]$?
                    if ?$m \neq [$?JOIN, Voting_Proof ?$\sigma_v]$? or ?$m.$?sender ?$\neq r$? or?\\?verify_voting(?\textnormal{``add server $r$''}?, ?$\sigma_v$?) ?$\neq \top$?: ?\label{line:check_join_1}?
                        ?\textbf{return}? ?$\bot$?
                else:
                    let ?$c = (-, r)$?, for some Server ?$r$?
                    Message ?$m = \mathit{requests}[c]$?
                    Voting_Proof ?$\sigma_v = \mathit{voting\_proofs}[c]$?
                    if (?$m \neq [$?LEAVE?$]$? or ?$m.$?sender ?$\neq r$?) and verify_voting(?\textnormal{``remove server $r$''}?, ?$\sigma_v$?) ?$\neq \top$?: ?\label{line:check_leave_1}?
                        ?\textbf{return}? ?$\bot$?
        ?\textbf{return}? ?$\top$?
        
    // ?\textcolor{gray}{$v \neq \mathit{genesis}$}?
    if ?$\mathit{set}' \neq \emptyset$? and ?$v' \notin \mathit{set}'$?: ?\label{line:view_in_seq}?
        ?\textbf{return}? ?$\bot$?
        
    if ?$\mathit{set}' = \emptyset$? and ?$v \not\subset v'$?:
        ?\textbf{return}? ?$\bot$?
        
    if ?exists? View ?$v^* \in \mathit{set}'$? such that ?$v \not\subset v^*$?: ?\label{line:set_follows}?
        ?\textbf{return}? ?$\bot$?

    if ?$(+, r) \notin v$? and ?exists? View ?$v^* \in \mathit{set}'$? such that ?$(-, r) \in v^*$?, for some Server ?$r$?: ?\label{line:check_valid_view_set_proof}?
        ?\textbf{return}? ?$\bot$?
        
    if ?$(+, r) \notin v$? and ?$(-, r) \in v'$?, for some Server ?$r$?: ?\label{line:check_minus_r_2}?
        ?\textbf{return}? ?$\bot$?
    
    if ?$\epsilon'$? is not View_Generator_Evidence:
        ?\textbf{return}? ?$\bot$?
        
    if ?$\epsilon'.\mathit{path}$?.destination() ?$\neq v$? or ?$\epsilon'.\mathit{path}$?.tail() ?$\neq \mathit{set}'$?: ?\label{line:check_destination_tail}?
        ?\textbf{return}? ?$\bot$?
        
    for each Change ?$c \in v' \cup \bigcup\limits_{v^* \in \mathit{set}'} v^*$?:
            if ?$c \notin \mathit{genesis}$?:
                if ?$c = (+, r)$?, for some Server ?$r$?:
                    Message ?$m = \mathit{requests}[c]$?
                    if ?$m \neq [$?JOIN, Voting_Proof ?$\sigma_v]$? or ?$m.$?sender ?$\neq r$? or? ?verify_voting(?\textnormal{``add server $r$''}?, ?$\sigma_v$?) ?$\neq \top$?: ?\label{line:check_join_2}?
                        ?\textbf{return}? ?$\bot$?
                else:
                    let ?$c = (-, r)$?, for some Server ?$r$?
                    Message ?$m = \mathit{requests}[c]$?
                    Voting_Proof ?$\sigma_v = \mathit{voting\_proofs}[c]$?
                    if (?$m \neq [$?LEAVE?$]$? or ?$m.$?sender ?$\neq r$?) and verify_voting(?\textnormal{``remove server $r$''}?, ?$\sigma_v$?) ?$\neq \top$?: ?\label{line:check_leave_2}?
                        ?\textbf{return}? ?$\bot$?
    ?\textbf{return}? ?$\top$?
\end{lstlisting}

Finally, we give the implementation of the reconfiguration module.

\begin{lstlisting}[
  caption={Reconfiguration module - initialization},
  label={lst:reconfiguration_initialization},
  escapechar=?]
?\textbf{Reconfiguration Module:}?
    ?\textcolor{plainorange}{Implementation:}?
        upon start: // initialization of the module; executed as soon as the ?\textcolor{gray}{$\mathtt{start}$}? event is triggered
            ?$\mathit{current\_view} = $? {
                View ?$\mathit{view} = \bot$? ?\label{line:current_view_init}?
                Bool ?$\mathit{installed} = \bot$?
                Bool ?$\mathit{processing} = \bot$?
                View_Path ?$\mathit{proof} = \bot$? ?\label{line:current_view_proof_init}?
            }
            
            Bool ?$\mathit{joined} = \bot$? ?\label{line:joined_init}?
            View ?$\mathit{stop\_processing\_until} = \emptyset$?
            Bool ?$\mathit{proposed} = \bot$?
            
            ?$\mathit{reconfiguration} = $? {
                Set(Change) ?$\mathit{requested} = \emptyset$? 
                Map(Change ?$\to$? Message) ?$\mathit{requests} = \{c \to \bot\text{, for every Change } c\}$?
                Map(Change ?$\to$? Voting_Proof) ?$\mathit{voting\_proofs} = \{c \to \bot\text{, for every Change } c\}$?
            
                View ?$\mathit{source} = \bot$?
                View ?$\mathit{destination} = \bot$?
                Set(View) ?$\mathit{sequence} = \bot$? // we prove that this is actually a sequence; see ?\textcolor{gray}{\Cref{subsection:safety_properties}}? ?\label{line:reconfiguration_sequence_init}?
                
                Bool ?$\mathit{prepared} = \bot$?
                Bool ?$\mathit{discharged} = \bot$?
            }
            
            ?$\mathit{state\_transfer} = $? {
                Bool ?$\mathit{in\_progress} = \bot$?
                Map(View ?$\to$? Set(State)) ?$\mathit{states} = \{v \to \emptyset\text{, for every View } v\}$?
                Map(View ?$\to$? Set(Server)) ?$\mathit{states\_from} = \{v \to \emptyset\text{, for every View } v\}$?
            }
            
            ?$dischargement = $? {
                View ?$\mathit{dischargement\_view} = \bot$?
                Map(View ?$\to$? Set(Server)) ?$\mathit{dischargements} = \{v \to \emptyset\text{, for every View } v\}$?
            }
            
            if ?$r \in \mathit{genesis}$?.members():
                ?$\mathit{current\_view}.\mathit{view} = \mathit{genesis}$? ?\label{line:set_current_to_genesis}?
                ?\textbf{trigger}? update current view to ?$\mathit{genesis}$? // the server updates its current view to ?\textcolor{gray}{$\mathit{genesis}$}? ?\label{line:formal_update_view_1}?
                
                ?$\mathit{current\_view}.\mathit{installed} = \top$?
                ?\textbf{trigger}? install view ?$\mathit{genesis}$? // the server installs ?\textcolor{gray}{$\mathit{genesis}$}? ?\label{line:install_init}?
                installed(?$\mathit{genesis}$?) // see ?\textcolor{gray}{\Cref{lst:transaction_module}}?
                ?$\mathit{current\_view}.\mathit{processing} = \top$?
                
                ?$\mathit{joined} = \top$?
                ?\textbf{trigger}? joined // the server joins ?\label{line:join_init}?
                
                // triggering events for joined servers
                for each Server ?$r \in \mathit{genesis}$?.members():
                    ?\textbf{trigger}? ?$r$? joined // server ?\textcolor{gray}{$r$}? joins ?\label{line:trigger_for_others_init}?
                
                ?\textbf{trigger}? <?$\mathit{vg}(\mathit{genesis})$?.Start | ?$\bot$?>
\end{lstlisting}

\begin{lstlisting}[
  caption={Reconfiguration module - the $\mathtt{extract\_requests\_and\_voting\_proofs}$ function},
  label={lst:reconfiguration_extract_requests_passing_proofs},
  escapechar=?]
?\textbf{Reconfiguration Module:}?
    ?\textcolor{plainorange}{Implementation:}?
        ?\textbf{function}? extract_requests_and_voting_proofs(Install_Message ?$m$?):
            update ?$\mathit{reconfiguration}.\mathit{requests}$? with requests from ?$m$?
            update ?$\mathit{reconfiguration}.\mathit{voting\_proofs}$? with voting ?proofs? from ?$m$?
\end{lstlisting}

\begin{lstlisting}[
  caption={Reconfiguration module - joining},
  label={lst:reconfiguration_joining},
  escapechar=?]
?\textbf{Reconfiguration Module:}?
    ?\textcolor{plainorange}{Implementation:}?
        upon join with voting proof ?$\sigma_v$?: // the server requests to join
            gossip [JOIN, ?$\sigma_v$?] ?\label{line:gossip_join}?
            
    upon ?exists? Message ?$m \in \mathit{waiting\_messages}$? such that ?$m = [$?JOIN, Voting_Proof ?$\sigma_v]$? and?\\?verify_voting(?\textnormal{``add server $m$}.sender\textnormal{''}?, ?$\sigma_v$?) ?$ = \top$? and ?$(+, m$?.sender?$) \notin \mathit{reconfiguration}.\mathit{requested}$?: ?\label{line:join_message}?
            ?$\mathit{waiting\_messages} = \mathit{waiting\_messages} \setminus{\{m\}}$?
            ?$\mathit{reconfiguration}$?.?$\mathit{requested} = \mathit{reconfiguration}$?.?$\mathit{requested} \cup \{(+, m.$?sender?$)\}$?
            ?$\mathit{reconfiguration}.\mathit{requests}[(+, m.\text{sender})] = m$? ?\label{line:store_join_message}?
\end{lstlisting}

\begin{lstlisting}[
  caption={Reconfiguration module - leaving},
  label={lst:reconfiguration_leaving},
  escapechar=?]
?\textbf{Reconfiguration Module:}?
    ?\textcolor{plainorange}{Implementation:}?
        upon leave: // the server requests to leave
            gossip [LEAVE] ?\label{line:gossip_leave}?
            
        upon ?exists? Message ?$m \in \mathit{waiting\_messages}$? such that ?$m = [$?LEAVE?$]$? and?\\$(-, m$?.sender?$) \notin \mathit{reconfiguration}.\mathit{requested}$?: ?\label{line:leave_message}?
            ?$\mathit{waiting\_messages} = \mathit{waiting\_messages} \setminus{ \{m\}}$?
            ?$\mathit{reconfiguration}.\mathit{requested} = \mathit{reconfiguration}.\mathit{requested} \cup \{(-, m$?.sender?$)\}$?
            ?$\mathit{reconfiguration}.\mathit{requests}[(-, m.\text{sender})] = m$? ?\label{line:store_leave_message}?
            
        upon obtaining Voting_Proof ?$\sigma_v$? such that?\\?verify_voting(?\textnormal{``remove server $r$''}?, ?$\sigma_v$?) ?$ = \top$? and ?$(-, r) \notin \mathit{reconfiguration}.\mathit{requested}$?:
            ?$\mathit{reconfiguration}.\mathit{requested} = \mathit{reconfiguration}.\mathit{requested} \cup \{(-, r)\}$?
            ?$\mathit{reconfiguration}.\mathit{voting\_proofs}[(-, r)] = \sigma_v$? ?\label{line:store_passing_proof}?
\end{lstlisting}

\vspace{-\baselineskip} 
\begin{lstlisting}[
  caption={Reconfiguration module - view generation},
  label={lst:reconfiguration_view_generation},
  escapechar=?]
?\textbf{Reconfiguration Module:}?
    ?\textcolor{plainorange}{Implementation:}?
    ?\textbf{function}? to_propose(): ?\label{line:to_propose_function}?
        Set(Change) ?$\mathit{proposal} = \emptyset$? 
        for each Change ?$\mathit{change} \in \mathit{reconfiguration}.\mathit{requested}$? such that ?$\mathit{change} \notin \mathit{current\_view}.\mathit{view}$?: ?\label{line:iteration}?
            if ?$\mathit{change} = (+, r)$?, for some Server ?$r$?:
                ?$\mathit{proposal} = \mathit{proposal} \cup \{\mathit{change}\}$?
            else if ?$\mathit{change} = (-, r)$?, for some Server ?$r$?:
                if ?$(+, r) \in \mathit{current\_view}.\mathit{view}$?:
                    ?$\mathit{proposal} = \mathit{proposal} \cup \{\mathit{change}\}$?
        ?\textbf{return}? ?$\mathit{proposal}$?
    
    upon to_propose() ?$ = \mathit{proposal} \neq \emptyset$? and ?$\mathit{current\_view}.\mathit{installed} = \top$? and ?$\mathit{proposed} = \bot$?: ?\label{line:new_proposal}?
        View_Generator_Evidence ?$\epsilon_{\mathit{vg}}$?
        ?$\epsilon_{\mathit{vg}}.\mathit{path} = \mathit{current\_view}.\mathit{proof}$?
        ?$\epsilon_{\mathit{vg}}.\mathit{requests} = \mathit{reconfiguration}.\mathit{requests}$?
        ?$\epsilon_{\mathit{vg}}.\mathit{voting\_proofs} = \mathit{reconfiguration}.\mathit{voting\_proofs}$?
        ?\textbf{trigger}? <?$\mathit{vg}(\mathit{current\_view}.\mathit{view})$?.Propose | ?$(\mathit{current\_view}.\mathit{view} \cup \mathit{proposal}, \emptyset, \epsilon_{\mathit{vg}})$?> ?\label{line:propose_installed}?
        ?$\mathit{proposed} = \top$?
        
    upon <?$\mathit{vg}(\mathit{current\_view}.\mathit{view})$?.Decide | Set(View) ?$\mathit{set}$?, Certificate ?$\omega$?>: ?\label{line:decide_from_vg}?
        Install_Message ?$m = [$?INSTALL, ?$\mathit{current\_view}.\mathit{view}$?, ?$\mathit{set}$?, ?$\omega]$? // create the INSTALL message
        ?$\mathit{waiting\_messages} = \mathit{waiting\_messages} \cup \{m\}$?
\end{lstlisting}

\vspace{-\baselineskip} 
\begin{lstlisting}[
  caption={Reconfiguration module - view transition},
  label={lst:reconfiguration_view_transition},
  escapechar=?]
?\textbf{Reconfiguration Module:}?
    ?\textcolor{plainorange}{Implementation:}?
        upon ?$\{$?(?$\mathit{joined} = \bot$? and ?exists? View ?$v \in \mathit{history}$? such that ?$r \in v$?.members() and ?$r \notin \mathit{source}[v]$?.members()) or (?$\mathit{joined} = \top$? and ?exists? View ?$v \in \mathit{history}$? such that ?$\mathit{current\_view}.\mathit{view} \subset v$? and ?$\mathit{source}[v] \subseteq \mathit{current\_view}.\mathit{view}$?)?$\}$? and ?$\mathit{reconfiguration}.\mathit{destination} = \bot$?: ?\label{line:updated_discovery_exists}?
            ?$\mathit{reconfiguration}.\mathit{destination} = v$? ?\label{line:set_reconfiguration_destination}?
            ?$\mathit{reconfiguration}.\mathit{source} = \mathit{source}[v]$?
            ?$\mathit{reconfiguration}.\mathit{sequence} = \mathit{sequence}[v]$?
            View_Path ?$\mathit{path} = \mathit{view\_path}[v]$?
            
            if ?$r \notin \mathit{reconfiguration}.\mathit{destination}$?.members():
                ?$\mathit{reconfiguration}.\mathit{prepared} = \top$? ?\label{line:prepared_free}?
            else:
                ?$\mathit{state\_transfer}.\mathit{in\_progress} = \top$?
                gossip ?$[$?STATE-REQUEST, ?$\mathit{reconfiguration}.\mathit{source}$?, ?$\mathit{reconfiguration}.\mathit{sequence}]$? // ask for state ?\label{line:gossip_state_request}?
                
            if ?$r \in \mathit{reconfiguration}.\mathit{destination}$?.members():
                ?$\mathit{reconfiguration}.\mathit{discharged} = \top$?
            else:
                ?$\mathit{dischargement}.\mathit{dischargement\_view} = \mathit{reconfiguration}.\mathit{destination}$?    
                gossip ?$[$?DISCHARGEMENT-REQUEST, ?$\mathit{reconfiguration}.\mathit{destination}]$? ?\label{line:gossip_dischargement_request}?
            
            ?\textbf{wait until}? ?$\mathit{reconfiguration}.\mathit{prepared} = \top$? and ?$\mathit{reconfiguration}.\mathit{discharged} = \top$? ?\label{line:wait_for}?
            
            // stop the view generator of the previous view
            if ?$\mathit{joined} = \top$?:
                ?\textbf{trigger}? <?$\mathit{vg}(\mathit{current\_view}.\mathit{view})$?.Stop> ?\label{line:stop_vg_reconfiguration}?
            
            if ?$r \in \mathit{reconfiguration}.\mathit{destination}$?.members():
                ?$\mathit{current\_view}.\mathit{view} = \mathit{reconfiguration}.\mathit{destination}$? ?\label{line:update_current_view}?
                ?\textbf{trigger}? update current view to ?$\mathit{reconfiguration}.\mathit{destination}$? 
                ?$\mathit{current\_view}.\mathit{proof} = \mathit{path}$? ?\label{line:update_view_formal_2}?
            if ?$\mathit{joined} = \bot$?:
                ?$\mathit{joined} = \top$?
                ?\textbf{trigger}? joined // the server joins ?\label{line:joined}?
            
            ?$\mathit{reconfiguration}.\mathit{source} = \bot$?
            ?$\mathit{reconfiguration}.\mathit{destination} = \bot$?
            
            // reset state transfer
            ?$\mathit{state\_transfer}.\mathit{states}[v] = \emptyset$?, for every View ?$v$?
            ?$\mathit{state\_transfer}.\mathit{states\_from}[v] = \emptyset$?, for every View ?$v$?
            
            // reset dischargement
            ?$\mathit{dischargement}.\mathit{dischargement\_view} = \bot$?
            ?$\mathit{dischargement}.\mathit{dischargements}[v] = \emptyset$?, for every View ?$v$?
        
            // reset prepared and discharged
            ?$\mathit{reconfiguration}.\mathit{prepared} = \bot$?
            ?$\mathit{reconfiguration}.\mathit{discharged} = \bot$?
            
            // reset proposed
            ?$\mathit{proposed} = \bot$?
            
            if ?$r \notin \mathit{reconfiguration}.\mathit{destination}$?.members(): ?\label{line:check_members_current_view}?
                ?\textbf{trigger}? left // the server leaves ?\label{line:left}?
                ?\textbf{return}? // we assume that the server executes the special ?\textcolor{gray}{$\mathtt{stop}$}? command as soon as it leaves
                
            // triggering events for joined servers
            for each Server ?$r \in \mathit{reconfiguration}.\mathit{destination}$?.members()?$\setminus{\mathit{reconfiguration}.\mathit{source}}$?.members():
                ?\textbf{trigger}? ?$r$? joined // server ?\textcolor{gray}{$r$}? joins ?\label{line:trigger_for_others_2}?
                
            // triggering events for left servers
            for each Server ?$r \in \mathit{reconfiguration}.\mathit{source}$?.members()?$\setminus{\mathit{reconfiguration}.\mathit{destination}}$?.members():
                ?\textbf{trigger}? ?$r$? left // server ?\textcolor{gray}{$r$}? leaves ?\label{line:leave_others}?
            
            // ?\textcolor{gray}{$\mathit{reconfiguration}.\mathit{sequence}$}? has a single element
            if ?$\mathit{reconfiguration}.\mathit{sequence} = \{\mathit{current\_view}.\mathit{view}\}$?:
                ?$\mathit{current\_view}.\mathit{installed} = \top$?
                ?\textbf{trigger}? install view ?$\mathit{current\_view}.\mathit{view}$? // the server installs ?\textcolor{gray}{$\mathit{current\_view}.\mathit{view}$}? ?\label{line:install}?
                installed(?$\mathit{genesis}$?) // see ?\textcolor{gray}{\Cref{lst:transaction_module}}?
                
                if ?$\mathit{current\_view}.\mathit{view} \supset \mathit{stop\_processing\_until}$?:
                    ?$\mathit{stop\_processing\_until} = \emptyset$?
                    ?$\mathit{current\_view}.\mathit{processing} = \top$?
                else:  
                    ?$\mathit{current\_view}.\mathit{processing} = \bot$?
                ?\textbf{trigger}? <?$\mathit{vg}(\mathit{current\_view}.\mathit{view})$?.Start | ?$\bot$?>
            else:
                View_Generator_Evidence ?$\epsilon_{\mathit{vg}}$?
                ?$\epsilon_{\mathit{vg}}.\mathit{path} = \mathit{current\_view}.\mathit{proof}$?
                ?$\epsilon_{\mathit{vg}}.\mathit{requests} = \mathit{reconfiguration}.\mathit{requests}$?
                ?$\epsilon_{\mathit{vg}}.\mathit{voting\_proofs} = \mathit{reconfiguration}.\mathit{voting\_proofs}$?
                ?\textbf{trigger}? <?$\mathit{vg}(\mathit{current\_view}.\mathit{view})$?.Start | ?$($?min_cardinality(?$\mathit{current\_view}.\mathit{proof}$?.tail()), ?$ \mathit{current\_view}.\mathit{proof}$?.tail(), ?$\epsilon_{\mathit{vg}})$?> ?\label{line:start_with_proposal}?
            ?$\mathit{reconfiguration}.\mathit{sequence} = \bot$?
\end{lstlisting}

\vspace{-\baselineskip} 
\begin{lstlisting}[
  caption={Reconfiguration module - state transfer},
  label={lst:reconfiguration_state_transfer},
  escapechar=?]
?\textbf{Reconfiguration Module:}?
    ?\textcolor{plainorange}{Implementation:}?
        upon ?exists? Message ?$m \in \mathit{waiting\_messages}$? such that ?$m = [$?STATE-REQUEST, View ?$\mathit{source}'$?, Set(View) ?$\mathit{set}']$? such that ?$\mathit{source}' \in \mathit{history}$? and ?$[$?INSTALL, ?$\mathit{source}'$?, ?$\mathit{set}'$?, Certificate ?$\omega'] \in \mathit{install\_messages}$?: ?\label{line:state-request_active}?
            ?$\mathit{waiting\_messages} = \mathit{waiting\_messages} \setminus{\{
            m\}}$?
            let ?$m' = [$?INSTALL, ?$\mathit{source}'$?, ?$\mathit{set}'$?, ?$\omega']$?
            if ?$m$?.sender ?$\in m'$?.destination().members(): // otherwise, the message can safely be discarded
                // the server does not need to have joined
                if ?$\mathit{self} \in \mathit{source}'$?.members() and (?$\mathit{current\_view}.\mathit{view} \subseteq \mathit{source}'$? or ?$\mathit{current\_view}.\mathit{view} = \bot$?): ?\label{line:check_member_source}?
                    ?$\mathit{current\_view}.\mathit{processing} = \bot$?
                    ?\textbf{trigger}? stop processing ?in? view ?$\mathit{source}'$? // the server stops processing in ?\textcolor{gray}{$\mathit{source}'$}? ?\label{line:stop_processing}?
                    if ?$\mathit{source}' \supset \mathit{stop\_processing\_until}$?:
                        ?$\mathit{stop\_processing\_until} = \mathit{source}'$? ?\label{line:stop_processing_until}?
                // ?\textcolor{gray}{$\mathit{state}$}? is a variable defined in ?\textcolor{gray}{\Cref{lst:transaction_module}}?
                send ?$[$?STATE-UPDATE, ?$\mathit{source}'$?, ?$\mathit{set}'$?, ?$\mathit{state}$?, ?$\mathit{current\_view.\mathit{view}}]$? to ?$m$?.sender ?\label{line:send_state_update}? 
                
        upon ?exists? Message ?$m \in \mathit{waiting\_messages}$? such that ?$m = [$?STATE-UPDATE, View ?$\mathit{source}'$?, Set(View) ?$\mathit{set}'$?, State_Representation ?$\mathit{state}$?, View ?$\mathit{view}]$? such that ?$\mathit{source}' = \mathit{reconfiguration}.\mathit{source}$? and ?$\mathit{set}' = \mathit{reconfiguration}.\mathit{sequence}$? and ?$\mathit{state}$?.verify() ?$= \top$? and ?$\mathit{state\_transfer}.\mathit{in\_progress} = \top$?: ?\label{line:state_update_received}?
            ?$\mathit{waiting\_messages} = \mathit{waiting\_messages} \setminus{\{
            m\}}$?
            Server ?$r' = m$?.sender
            if ?$r' \notin \mathit{state\_transfer}.\mathit{states\_from}[\mathit{view}]$? and ?$r' \in \mathit{view}$?.members():
                ?$\mathit{state\_transfer}.\mathit{states}[\mathit{view}] = \mathit{state\_transfer}.\mathit{states}[\mathit{view}] \cup \{\mathit{state}\}$? // store the state
                ?$\mathit{state\_transfer}.\mathit{states\_messages}[\mathit{view}] = \mathit{state\_transfer}.\mathit{states\_messages}[\mathit{view}] \cup \{m\}$? // store the sender
                ?$\mathit{state\_transfer}.\mathit{states\_from}[\mathit{view}] = \mathit{state\_transfer}.\mathit{states\_from}[\mathit{view}] \cup \{r'\}$? // store the sender
                
            if ?$r' \notin \mathit{state\_transfer}.\mathit{states\_from}[\mathit{source}']$? and ?$r' \in \mathit{source}'$?.members():
                ?$\mathit{state\_transfer}.\mathit{states}[.\mathit{source}'] = \mathit{state\_transfer}.\mathit{states}[\mathit{source}'] \cup \{\mathit{state}\}$? // store the state
                ?$\mathit{state\_transfer}.\mathit{states\_messages}[\mathit{source}'] = \mathit{state\_transfer}.\mathit{states\_messages}[\mathit{source}'] \cup \{m\}$? // store the sender
                ?$\mathit{state\_transfer}.\mathit{states\_from}[\mathit{source}'] = \mathit{state\_transfer}.\mathit{states\_from}[\mathit{source}'] \cup \{r'\}$? // store the sender
                
        ?\textbf{function}? enough_states_received(View ?$v$?):
            Set(State_Representation) ?$\mathit{states} = \emptyset$?
            for each Server ?$\mathit{rep} \in v$?.members():
                if ?exists? View ?$\mathit{view}$? such that ?$\mathit{view} \supseteq v$? and ?$\mathit{rep} \in \mathit{state\_transfer}.\mathit{states\_from}[\mathit{view}]$?:
                    ?$\mathit{states} = \mathit{states} \cup \{\mathit{state}\}$?, where ?$m = [$?STATE-UPDATE, ?$\mathit{reconfiguration}.\mathit{source}$?, ?$\mathit{reconfiguration}.\mathit{sequence}$?, ?$\mathit{state}$?, ?$v] \in \mathit{state\_transfer}.\mathit{states}[v]$? and ?$m$?.sender ?$= \mathit{rep}$?
            if ?$|\mathit{states}| < v$?.plurality():
                ?\textbf{return}? ?$\bot$?
            else:
                ?\textbf{return}? ?$\mathit{states}$?
        
        // state received from a quorum of members of ?\textcolor{gray}{$\mathit{reconfiguration}.\mathit{source}$}?        
        upon ?$|\mathit{state\_transfer}.\mathit{states\_from}[\mathit{reconfiguration}.\mathit{source}]| \geq \mathit{reconfiguration}.\mathit{source}$?.quorum() and ?$\mathit{state\_transfer}.\mathit{in\_progress} = 
        \top$?: ?\label{line:prepared_from_source}?
            ?$\mathit{reconfiguration}.\mathit{prepared} = \top$? ?\label{line:prepared_1}?
            ?$\mathit{state\_transfer}.\mathit{in\_progress} = \bot$?
            refine_state(?$\mathit{state\_transfer}.\mathit{states}[\mathit{reconfiguration}.\mathit{source}]$?)
                
        // state received from at least one correct member of a view ``greater'' than ?\textcolor{gray}{$\mathit{reconfiguration}.\mathit{source}$}?
        upon ?exists? View ?$v \in \mathit{history}$? such that ?$\mathit{reconfiguration}.\mathit{source} \subset v$? and ?$\mathit{states} = $? enough_states_received(?$v$?) and ?$\mathit{states} \neq \bot$? and ?$\mathit{state\_transfer}.\mathit{in\_progress} = \top$?: ?\label{line:prepared_from_bigger}?
            ?$\mathit{reconfiguration}.\mathit{prepared} = \top$? ?\label{line:prepared_2}?
            ?$\mathit{state\_transfer}.\mathit{in\_progress} = \bot$?
            refine_state(?$\mathit{states}$?)
\end{lstlisting}

\vspace{-\baselineskip} 
\begin{lstlisting}[
  caption={Reconfiguration module - view dischargement},
  label={lst:reconfiguration_view_dischargement},
  escapechar=?]
?\textbf{Reconfiguration Module:}?
    ?\textcolor{plainorange}{Implementation:}?
        upon ?exists? Message ?$m \in \mathit{waiting\_messages}$? such that ?$m = [$?DISCHARGEMENT-REQUEST, View ?$v]$? and ?$\mathit{joined} = \top$? and ?$\mathit{current\_view}.\mathit{view} \supseteq v$?:
            ?$\mathit{waiting\_messages} = \mathit{waiting\_messages} \setminus{\{m\}}$?
            send ?$[$?DISCHARGEMENT-CONFIRM, ?$\mathit{current\_view}.\mathit{view}]$?
            
        upon ?exists? Message ?$m \in \mathit{waiting\_messages}$? such that ?$m = [$?DISCHARGEMENT-CONFIRM, View ?$v]$? and ?$v \in \mathit{history}$? and ?$m$?.sender ?$\in v$?.members():
            ?$\mathit{waiting\_messages} = \mathit{waiting\_messages} \setminus{\{m\}}$?
            ?$\mathit{dischargement}.\mathit{dischargements}[v] = \mathit{dischargement}.\mathit{dischargements}[v] \cup \{m$?.sender?$\}$?
            
        upon ?exists? View ?$v$? such that ?$\mathit{dischargement}.\mathit{dischargement\_view} \neq \bot$? and ?$v \supseteq \mathit{dischargement}.\mathit{dischargement\_view}$? and ?$|\mathit{dischargement}.\mathit{dischargements}[v]| \geq v$?.quorum(): ?\label{line:discharged_rule}?
            ?$\mathit{reconfiguration}.\mathit{discharged} = \top$?
\end{lstlisting}

\subsection{Proof of Correctness} \label{subsection:correctness_proof}

We now prove the properties presented in \Cref{lst:reconfiguration_properties_new}.
We start by proving some intermediate results that play the crucial role in the proof (\Cref{subsubsection:intermediate_results}).

\subsubsection{Intermediate Results} \label{subsubsection:intermediate_results}

First, we show that a correct server that updates its $\mathit{current\_view}.\mathit{view}$ variable to $v$ has $v$ in its $\mathit{history}$ variable.

\begin{lemma} \label{lemma:current_view_view_path}
Let $\mathit{current\_view}.\mathit{view} = v$ at a correct server $r$ at time $t$.
Then, $v \in \mathit{history}$ at server $r$ at time $t$.
\end{lemma}
\begin{proof}
If $v = \mathit{genesis}$, the lemma follows from line~\ref{line:init_history} of \Cref{lst:storage_initialization} and the fact that no view is ever removed from the $\mathit{history}$ variable of a correct server.

Let $v \neq \mathit{genesis}$.
Since $\mathit{current\_view}.\mathit{view} = v$ at server $r$ at time $t$ and $v \neq \mathit{genesis}$, we know that line~\ref{line:update_current_view} of \Cref{lst:reconfiguration_view_transition} is executed by $r$ at time $t_v \leq t$.
Hence, $\mathit{reconfiguration}.\mathit{destination} = v$ at time $t_v$.
Moreover, $v \in \mathit{history}$ at time $t_v$ at server $r$ (by line~\ref{line:updated_discovery_exists} of \Cref{lst:reconfiguration_view_transition} and the fact that no view is ever removed from the $\mathit{history}$ variable of a correct server).
Therefore, the lemma holds in this case, as well.
\end{proof}

Next, we show that, if $v \in \mathit{history}$ at a correct server, then the server has previously obtained a view-path to $v$.

\begin{lemma} \label{lemma:history_path}
Let $v \in \mathit{history}$ at a correct server $r$ at time $t$, where $v \neq \mathit{genesis}$.
Then, there exists a view-path $\mathit{path} = [m_1, ..., m_k]$ such that (1) $\mathit{path}.\mathtt{destination}() = v$, and (2) $m_i \in \mathit{install\_messages}$ at server $r$ at time $t$, for every $i \in [1, k]$. 
\end{lemma}
\begin{proof}
In order to prove the lemma, we show that, if, at any point in time, $v^* \in \mathit{history}$ at a correct server $r$, then $r$ has previously obtained a view-path to $v^*$ (and the view-path consists of messages from the $\mathit{install\_messages}$ variable).
Note that no message is ever removed from the $\mathit{install\_messages}$ variable of a correct server.
We set the following invariant to hold at time $t^* \geq 0$: If $v^* \in \mathit{history}$ at time $t^*$, then $r$ has previously included a view-path to $v^*$ in $\mathit{install\_messages}$.
Observe that the invariant holds initially (i.e., at the time the server starts) since $\mathit{history} = \{\mathit{genesis}\}$ (line~\ref{line:init_history} of \Cref{lst:storage_initialization}) and $\mathit{install\_messages} = \emptyset$ (line~\ref{line:init_install_messages} of \Cref{lst:storage_initialization}).

We now prove that the invariant is preserved once the $\mathit{history}$ variable at server $r$ is updated with a view $v^{**}$.
The only place in which the $\mathit{history}$ variable can be modified is line~\ref{line:update_history} of \Cref{lst:storage_processing}.
Since $v \in \mathit{history}$ (by line~\ref{line:update_history_rule} of \Cref{lst:storage_processing}), we know that there exists a view-path to $v$, which is previously obtained and included in $\mathit{install\_messages}$ (because of the invariant hypothesis).
Hence, there is a view-path to $v^{**}$ since the received $\mathtt{INSTALL}$ message is included in $\mathit{install\_messages}$ (line~\ref{line:update_install_messages} of \Cref{lst:storage_processing}).
Thus, the lemma holds.
\end{proof}

\Cref{lemma:history_valid} proves that any view $v$ that belongs to the $\mathit{history}$ variable of a correct server is valid.

\begin{lemma} \label{lemma:history_valid}
Let $v \in \mathit{history}$ at a correct server.
Then, $v$ is a valid view.
\end{lemma}
\begin{proof}
If $v = \mathit{genesis}$, \Cref{definition:valid_view} is satisfied for $v$.
Otherwise, \Cref{definition:valid_view} holds for $v$ because of \Cref{lemma:history_path}.
\end{proof}

Finally, we show that correct servers ``transit'' only to valid views, i.e., if $\mathit{current\_view}.\mathit{view} = v$ at a correct server, then $v$ is a valid view.

\begin{lemma} \label{lemma:current_view_valid}
Let $\mathit{current\_view}.\mathit{view} = v$ at a correct server $r$ at some time $t \geq 0$.
Then, $v$ is a valid view.
\end{lemma}
\begin{proof}
We know that $v \in \mathit{history}$ at server $r$ (by \Cref{lemma:current_view_view_path}).
Hence, $v$ is valid by \Cref{lemma:history_valid}.
\end{proof}

We now give a brief explanation of the upcoming steps in proving the properties from \Cref{lst:reconfiguration_properties_new}.

\para{Explanation}
We fix some time $t$ in an execution of \sysname.
At time $t$, we observe all $\mathtt{INSTALL}$ messages obtained by (correct or faulty) processes.
Then, we introduce a set of invariants concerned with all obtained $\mathtt{INSTALL}$ messages (this set of invariants is satisfied at time $0$).
The introduced invariants allow us to identify all possible ``new'' $\mathtt{INSTALL}$ messages that can be obtained after time $t$.
Finally, we prove that the invariants are preserved after such $\mathtt{INSTALL}$ messages are indeed obtained by processes after time $t$.

Fix a time $t \geq 0$ in an execution.
Let $\mathcal{I}_p(t)$ denote all $\mathtt{INSTALL}$ messages obtained by a (correct or faulty) process $p$ by time $t$; note that obtaining of an $\mathtt{INSTALL}$ message is irrevocable, i.e., $\mathcal{I}_p(t') \subseteq \mathcal{I}_p(t)$, for any time $t' \leq t$.
Moreover, $\mathcal{I}_p(0) = \emptyset$, for every process $p$.
Let $\mathcal{I}^*(t) = \bigcup\limits_{p \in \mathcal{C} \cup \mathcal{R}} \mathcal{I}_p(t)$.
Observe that $\mathcal{I}^*(t') \subseteq \mathcal{I}^*(t)$, for any time $t' < t$.

\noindent We say that a view $v$ is \emph{well-founded} at time $t$ if and only if:
\begin{compactitem}
    \item $v = \mathit{genesis}$, or
    
    \item there exists a view-path $\mathit{path}$ such that $\mathit{path} \subseteq \mathcal{I}^*(t)$ and $\mathit{path}.\mathtt{destination()} = v$.
\end{compactitem}
Now, we prove that, if a process has obtained an $\mathtt{INSTALL}$ message (see \Cref{lst:install_messages}) associated with a view $\mathit{source}$ by time $t$ and $\mathit{source}$ is not well-founded at time $t$, then $\mathit{source}$ is not a valid view.

\begin{lemma} \label{lemma:not_well_founded_not_valid}
Let a (correct or faulty) process obtain $m = [\mathtt{INSTALL}, \mathit{source}, \mathit{set}, \omega]$ at some time $t > 0$.
Moreover, let $\mathit{source}$ not be well-founded at time $t - 1$.
Then, $\mathit{source}$ is not a valid view.
\end{lemma}
\begin{proof}
Note that $\mathit{source} \neq \mathit{genesis}$, since $\mathit{genesis}$ is well-founded at any time.
We prove the lemma by contradiction.
Hence, let $\mathit{source}$ be a valid view.

Since $\mathit{source}$ is a valid view, that means that (at least) a quorum of members of $\mathit{source}$ are correct (by the failure model).
Moreover, we know that $\mathit{set}$ is committed by $\mathit{vg}(\mathit{source})$ (by \Cref{lst:install_messages}).
Because of the fact that $\omega$ contains messages from at least $\mathit{source}.\mathtt{quorum()}$ distinct members of $\mathit{source}$ (see \Cref{lst:verify_output_vg}), there exists a correct server $r$ that sets its $\mathit{current\_view}.\mathit{view}$ variable to $\mathit{source}$ before time $t$ at line~\ref{line:update_current_view} of \Cref{lst:reconfiguration_view_transition} (since $\mathit{source} \neq \mathit{genesis}$ and $\mathit{vg}(\mathit{source})$ is started by $r$, which follows from the decision permission property of $\mathit{vg}(\mathit{source})$).
By \Cref{lemma:current_view_view_path}, $\mathit{source} \in \mathit{history}$ before time $t$.
Furthermore, \Cref{lemma:history_path} shows that $r$ has obtained a view-path to $\mathit{source}$ before time $t$.
Therefore, $\mathit{source}$ is well-founded at time $t - 1$, which implies contradiction.
\end{proof}

Finally, we introduce a few concepts that describe the state of the system at some fixed time $t$.
First, we denote by $V(t)$ the set of well-founded views at time $t$, i.e., $V(t) = \{v \,|\, v \text{ is well-founded at time } t\}$.

For every view $v \in V(t)$, we define a logical predicate $\alpha_t(v)$ such that $\alpha_t(v) = \top$ if and only of: 
\begin{compactitem}
    \item $v = \mathit{genesis}$, or
    
    \item $[\mathtt{INSTALL}, v' , \{v\}, \omega] \in \mathcal{I}^*(t)$, where $v' \in V(t)$; in other words, $m \in \mathcal{I}^*(t)$, where $m.\mathtt{source()} \in V(t)$, $m.\mathtt{destination()} = v$ and $m.\mathtt{tail()} = \emptyset$.
\end{compactitem}
Otherwise, $\alpha_t(v) = \bot$.
If $\alpha_t(v) = \top$, we say that $v$ is \emph{installable} at time $t$.

Moreover, for every view $v \in V(t)$, we define $\beta_t(v)$ such that $\mathit{set} \in \beta_t(v)$ if and only if $m = [\mathtt{INSTALL}, v', \mathit{set}' = \{v\} \cup \mathit{set}, \omega] \in \mathcal{I}^*(t)$, where $v' \in V(t)$, $m.\mathtt{destination()} = v$, and $\mathit{set} \neq \emptyset$;
in other words, $m \in \mathcal{I}^*(t)$, where $m.\mathtt{source()} \in V(t)$, $m.\mathtt{destination()} = v$ and $m.\mathtt{tail()} = \mathit{set} \neq \emptyset$.
Otherwise, $\beta_t(v) = \emptyset$.

Finally, for every view $v \in V(t)$, we define $\rho_t(v)$ such that $\mathit{set} \in \rho_t(v)$ if and only if $[\mathtt{INSTALL}, v, \mathit{set}, \omega] \in \mathcal{I}^*(t)$.

First, we prove that every view $v \in V(t)$ is valid.

\begin{lemma} \label{invariant:valid}
Let $v \in V(t)$.
Then, $v$ is a valid view.
\end{lemma}
\begin{proof}
Follows from the definition of well-founded views and \Cref{definition:valid_view}.
\end{proof}

We now introduce the first invariant.

\begin{invariant} \label{invariant:rhos}
Let $v \in V(t)$.
If $\mathit{set} \in \rho_t(v)$, then (1) $\mathit{set} \neq \emptyset$, (2) $\mathit{set}$ is a sequence, and (3) $\mathit{set}.\mathtt{follows}(v) = \top$.
\end{invariant}

Since $\mathcal{I}^*(0) = \emptyset$, we conclude that $V(0) = \{\mathit{genesis}\}$ and \Cref{invariant:rhos} holds at time $0$.
Now, we show that any information obtained by a process by time $t$ is captured in $V(t)$, $\rho_t$ and $\beta_t$.

\begin{lemma} \label{lemma:knows_view_system}
Let a (correct or faulty) process $p$ obtain a view-path $\mathit{path} = [m_1, m_2, ..., m_k]$ by time $t$, where $k \geq 1$.
Then, the following holds:
\begin{compactitem}
    \item $\mathit{path}.\mathtt{views}() \subseteq V(t)$;
    
    \item Let $\mathit{destination} = \mathit{path}.\mathtt{destination}()$.
    If $\mathit{path}.\mathtt{tail}() \neq \emptyset$, then $\mathit{path}.\mathtt{tail}() \in \beta_t(\mathit{destination})$.
\end{compactitem}
\end{lemma}
\begin{proof}
Since $p$ obtains $\mathit{path}$ by time $t$, we conclude that $m_i \in \mathcal{I}^*(t)$, for every $i \in [1, k]$.
We know that $\mathit{genesis} \in V(t)$, by the definition of well-founded views.
All other views that belong to $\mathit{path}.\mathtt{views}()$, therefore, are included in $V(t)$ (including $\mathit{destination}$ and $m_k.\mathtt{source}()$).
Hence, if $\mathit{path}.\mathtt{tail}() \neq \emptyset$, then $\mathit{path}.\mathtt{tail}() \in \beta_t(\mathit{destination})$.
\end{proof}

Next, we show that any two sets that belong to $\rho_t(v)$, where $v \in V(t)$, are ``comparable''.

\begin{lemma} \label{lemma:weak_accuracy}
Let $v \in V(t)$ and let $\mathit{set}_1, \mathit{set}_2 \in \rho_t(v)$.
Then, either $\mathit{set}_1 \subseteq \mathit{set}_2$ or $\mathit{set}_1 \supset \mathit{set}_2$.
\end{lemma}
\begin{proof}
Since $v \in V(t)$, we know that $v$ is a valid view (by \Cref{invariant:valid}).
Moreover, we know that $\mathit{set}_1$ and $\mathit{set}_2$ are committed by $\mathit{vg}(v)$ (by \Cref{lst:install_messages}).
Hence, the lemma follows from the comparability property of $\mathit{vg}(v)$.
\end{proof}

We now define an invariant that explains how installable views are ``instantiated''.

\begin{invariant} [Creation of Installable Views] \label{invariant:creation_installable}
For every $v \in V(t)$ such that (1) $\alpha_t(v) = \top$, and (2) $v \neq \mathit{genesis}$, there exists a view $v' \in V(t)$ such that:
\begin{compactitem}
    \item $\alpha_t(v') = \top$, and
    \item $\mathit{seq} \in \rho_t(v')$, where $\mathit{seq} = v_1 \to ... \to v_x \to v, (x \geq 0)$,\footnote{Recall that all sets that belong to $\rho_t(v)$, for any view $v \in V(t)$, are sequences (by \Cref{invariant:rhos}).} and
    \item for every $v_i \in \{v_1, ..., v_x\}$, $v_i \in V(t)$ and $\alpha_t(v_i) = \bot$, and
    \item for every $v_i \in \{v_1, ..., v_x\}$, $\mathit{seq}_i \in \rho_t(v_i)$, where $\mathit{seq}_i = v_{i + 1} \to ... \to v_x \to v$.
\end{compactitem}
\end{invariant}

Note that \Cref{invariant:creation_installable} is satisfied at time $0$ since $V(0) = \{\mathit{genesis}\}$.
Next, we define what it means for a view to \emph{lead} to another view.

\begin{definition} [Leading to a View] \label{definition:leads}
Consider views $v, v' \in V(t)$ such that:
\begin{compactitem}
    \item $\alpha_t(v) = \alpha_t(v') = \top$, and
    \item $\mathit{seq} \in \rho_t(v')$, where $\mathit{seq} = v_1 \to ... \to v_x \to v, (x \geq 0)$, and
    \item for every $v_i \in \{v_1, ..., v_x\}$, $v_i \in V(t)$ and $\alpha_t(v_i) = \bot$, and
    \item for every $v_i \in \{v_1, ..., v_x\}$, $\mathit{seq}_i \in \rho_t(v_i)$, where $\mathit{seq}_i = v_{i + 1} \to ... \to v_x \to v$.
\end{compactitem}
We say that $v'$ \emph{leads to} $v$ at time $t$.
\end{definition}

We now show that $v' \subset v$ if $v'$ leads to $v$.

\begin{lemma} \label{lemma:leads_to_bigger}
Let $v' \in V(t)$ lead to $v \in V(t)$ at time $t$.
Then, $v' \subset v$.
\end{lemma}
\begin{proof}
By \Cref{definition:leads}, $\mathit{seq} \in \rho_t(v')$, where $\mathit{seq} = ... \to v$.
Since $\mathit{seq}.\mathtt{follows}(v') = \top$ (by \Cref{invariant:rhos}), the lemma holds.
\end{proof}

Now, we introduce another invariant that we assume holds at time $t$: $\mathit{genesis}$ is a subset of every other installable view.

\begin{invariant} \label{invariant:initial_view}
For every view $v \in V(t)$ such that (1) $\alpha_t(v) = \top$, and (2) $v \neq \mathit{genesis}$, $\mathit{genesis} \subset v$.
\end{invariant}

\Cref{invariant:initial_view} holds at time $0$ since $V(0) = \{\mathit{genesis}\}$.
Now, we show that, for every installable view different from $\mathit{genesis}$, there exists a view that leads to it at time $t$ (follows directly from \Cref{invariant:creation_installable}).

\begin{lemma} \label{lemma:exists_leads_to}
Let $v' \in V(t)$ such that (1) $\alpha_t(v') = \top$, and (2) $v' \neq \mathit{genesis}$.
Then, there exists a view $v \in V(t)$ such that $v$ leads to $v'$ at time $t$.
\end{lemma}
\begin{proof}
According to \Cref{invariant:creation_installable}, there exists a view $v \in V(t)$ such that:
\begin{compactitem}
    \item $\alpha_t(v) = \top$, and
    \item $\mathit{seq} \in \rho_t(v)$, where $\mathit{seq} = v_1 \to ... \to v_x \to v', (x \geq 0)$, and
    \item for every $v_i \in \{v_1, ..., v_x\}$, $v_i \in V(t)$ and $\alpha_t(v_i) = \bot$, and
    \item for every $v_i \in \{v_1, ..., v_x\}$, $\mathit{seq}_i \in \rho_t(v_i)$, where $\mathit{seq}_i = v_{i + 1} \to ... \to v_x \to v'$.
\end{compactitem}
By \Cref{definition:leads}, $v$ leads to $v'$ at time $t$, which concludes the proof.
\end{proof}

We now prove that no view leads to $\mathit{genesis}$ at time $t$.

\begin{lemma} \label{lemma:genesis_no_lead}
No view $v \in V(t)$ leads to $\mathit{genesis}$ at time $t$.
\end{lemma}
\begin{proof}
We prove the lemma by contradiction.
Let there exist a view $v \in V(t)$ such that $v$ leads to $\mathit{genesis}$ at time $t$.
By \Cref{lemma:leads_to_bigger}, $v \subset \mathit{genesis}$.
However, this contradicts \Cref{invariant:initial_view}, which concludes the proof.
\end{proof}

A view can lead to (at most) one other view.
The following lemma proves this statement.

\begin{lemma} \label{lemma:leads_to_one}
Let $v \in V(t)$ such that $\alpha_t(v) = \top$.
If $v$ leads to $v' \in V(t)$ at time $t$ and $v$ leads to $v'' \in V(t)$ at time $t$, then $v' = v''$.
\end{lemma}
\begin{proof}
According to \Cref{definition:leads}, $\mathit{seq}' \in \rho_t(v)$, where $\mathit{seq}'$ is a sequence (by \Cref{invariant:rhos}) and $\mathit{seq}' = ... \to v'$.
Similarly, $\mathit{seq}'' \in \rho_t(v)$, where $\mathit{seq}''$ is a sequence (by \Cref{invariant:rhos}) and $\mathit{seq}'' = ... \to v''$.
According to \Cref{lemma:weak_accuracy}, either $\mathit{seq}' \subseteq \mathit{seq}''$ or $\mathit{seq}' \supset \mathit{seq}''$.
Let us analyze two possible cases:
\begin{compactitem}
    \item $\mathit{seq}' = \mathit{seq}''$: In this case, we have that $v' = v''$ and the lemma holds.
    \item $\mathit{seq}' \neq \mathit{seq}''$: Without loss of generality, let $\mathit{seq}' \subset \mathit{seq}''$.
    Therefore, $v' \in \mathit{seq}''$ and $v' \subseteq v''$.
    If $v' \subset v''$, \Cref{definition:leads} is not satisfied for $v''$ (since $\alpha_t(v') = \top$).
    Hence, $v' = v''$.
\end{compactitem}
The lemma holds.
\end{proof}

Now, we prove that there exists exactly one view that does not lead to any view.
Recall that we assume that only finitely many valid views exist (see \Cref{appx:reconfiguration_module}, paragraph ``Failure model \& assumptions'').
Since all views that belong to $V(t)$ are valid (by \Cref{invariant:valid}), we conclude that $|V(t)| < \infty$.

\begin{lemma} \label{lemma:n-1}
There exists \emph{exactly} one view $v \in V(t)$ such that (1) $\alpha_t(v) = \top$, and (2) $v$ does not lead to any view at time $t$.
\end{lemma}
\begin{proof}
Let there be $x$ installable views at time $t$.
Because of \cref{lemma:exists_leads_to,lemma:genesis_no_lead}, we conclude that \emph{exactly} $x - 1$ installable views (all installable views except $\mathit{genesis}$; $\mathit{genesis}$ is installable at any time) have views that lead to them at time $t$.
According to \Cref{lemma:leads_to_one}, each installable view leads to at most one view at time $t$.
Hence, at least $x - 1$ views lead to a view at time $t$.

In order to prove the lemma, it is sufficient to show that it is impossible for all $x$ installable views to lead to some view at time $t$.
We prove this statement by contradiction.

Let $\mathcal{I}\mathcal{N}\mathcal{S}$ denote the set of all installable views at time $t$; note that $|\mathcal{I}\mathcal{N}\mathcal{S}| < \infty$ (since $|V(t)| < \infty$ and $\mathcal{I}\mathcal{N}\mathcal{S} \subseteq V(t)$).
Consider the following construction:
\begin{compactenum}
    \item Start with $V \gets \{\mathit{genesis}\}$, $\mathit{last} \gets \mathit{genesis}$ and $R \gets \mathcal{I}\mathcal{N}\mathcal{S} \setminus \{\mathit{genesis}\}$.
    \item Repeat until $V \neq \mathcal{I}\mathcal{N}\mathcal{S}$:
    \begin{compactenum}
        \item Select a view $v \in \mathcal{I}\mathcal{N}\mathcal{S}$ such that $\mathit{last}$ leads to $v$ at time $t$.
        \item Update $V \gets V \cup \{v\}$, $\mathit{last} \gets v$ and $R \gets R \setminus \{v\}$.
        \item If $V \neq \mathcal{I}\mathcal{N}\mathcal{S}$, go to step 2.
        \item Otherwise, select a view $v \in \mathcal{I}\mathcal{N}\mathcal{S}$ such that $\mathit{last}$ leads to $v$ at time $t$.
    \end{compactenum}
\end{compactenum}
    
Note that $V$ represents a sequence (because of the construction and \Cref{lemma:leads_to_bigger}).
Hence, $\mathit{last}$ represents the greatest (with respect to the containment relation) element of $V$.
Therefore, a view $v$, where $\mathit{last}$ leads to $v$ at time $t$, must belong to $R$ (otherwise, we would contradict \Cref{lemma:leads_to_bigger}).
Thus, the $V$ set is constantly ``growing'' and the $R$ set is constantly ``shrinking''.
    
Once $V = \mathcal{I}\mathcal{N}\mathcal{S}$ in step 2(c), we conclude that $R = \emptyset$.
Hence, $\mathit{last}$ cannot lead to any view.
Thus, we reach contradiction with the fact that $\mathit{last}$ leads to a view, which concludes the proof.
\end{proof}

The next lemma shows that at most one view can lead to a view $v$ at time $t$.

\begin{lemma} \label{lemma:leads_from_one}
Let $v \in V(t)$ such that $\alpha_t(v) = \top$.
If $v' \in V(t)$ leads to $v$ at time $t$ and $v'' \in V(t)$ leads to $v$ at time $t$, then $v' = v''$.
\end{lemma}
\begin{proof}
Let there be $x$ installable views at time $t$.
Because of \cref{lemma:exists_leads_to,lemma:genesis_no_lead}, we conclude that exactly $x - 1$ installable views (all installable views except $\mathit{genesis}$) have views that lead to them.
According to \Cref{lemma:leads_to_one}, each view leads to at most one view.
\Cref{lemma:n-1} shows that there exists an installable view that does not lead to any view at time $t$.
Hence, there are exactly $x - 1$ installable views and exactly $x - 1$ ``lead-to'' relations.
Thus, only one view leads to $v$ at time $t$ and the lemma holds.
\end{proof}

Finally, we show that all installable views are comparable.

\begin{lemma} \label{lemma:installable_sequence}
Let $v, v' \in V(t)$ such that $\alpha_t(v) = \alpha_t(v') = \top$.
Then, either $v \subseteq v'$ or $v \supset v'$.
\end{lemma}
\begin{proof}
Let $\mathcal{I}\mathcal{N}\mathcal{S}$ denote the set of all installable views at time $t$.
In order to prove the lemma, we prove that $\mathcal{I}\mathcal{N}\mathcal{S}$ is a sequence.
Let $x = |\mathcal{I}\mathcal{N}\mathcal{S}| < \infty$.
According to \Cref{lemma:n-1}, exactly $x - 1$ installable views lead to some view at time $t$.

Consider the construction similar to the one from the proof of \Cref{lemma:n-1}:
\begin{compactenum}
    \item Start with $V \gets \{\mathit{genesis}\}$ and $\mathit{last} \gets \mathit{genesis}$.
    \item Repeat until $\mathit{last}$ does not lead to any view at time $t$: 
    \begin{compactenum}
        \item Select a view $v \in \mathcal{I}\mathcal{N}\mathcal{S}$ such that $\mathit{last}$ leads to $v$ at time $t$.
        \item Update $V \gets V \cup \{v\}$ and $\mathit{last} \gets v$.
    \end{compactenum}
\end{compactenum}

Because of the construction and \Cref{lemma:leads_to_bigger}, $V$ is a sequence.
In order to conclude the lemma, it suffices to prove that $V = \mathcal{I}\mathcal{N}\mathcal{S}$.
We prove this statement by contradiction.

Suppose that $V \neq \mathcal{I}\mathcal{N}\mathcal{S}$.
Because of the construction, we know that $V \subset \mathcal{I}\mathcal{N}\mathcal{S}$.
Let $|V| = p < x$.
This means that there exist $x - p$ views in $R = \mathcal{I}\mathcal{N}\mathcal{S} \setminus{V}$.
The following holds for each view $v \in R$:
\begin{compactitem}
    \item $v$ leads to a view $v' \in \mathcal{I}\mathcal{N}\mathcal{S}$ at time $t$: 
    This statement holds because $x - 1$ installable views lead to some view at time $t$ (by \Cref{lemma:n-1}).
    By construction, we have that $p - 1$ views of $V$ lead to some view at time $t$.
    Given that $|R| = x - p$ and $x - 1 - (p - 1) = x - p$, the statement is true.
    
    \item $v$ leads to a view $v' \in R$ at time $t$:
    Suppose that $v$ leads to a view $v' \in V$ at time $t$.
    If $v' = \mathit{genesis}$, \Cref{lemma:genesis_no_lead} is contradicted. 
    If $v' \neq \mathit{genesis}$, \Cref{lemma:leads_from_one} is contradicted (since $v'$ has more than one view leading to it).
    Thus, the statement is correct.
\end{compactitem}
Finally, we model the set $R$ with a graph $G$ with $x - p$ vertices ($x - p$ installable views) and $x - p$ edges ($x - p$ ``lead-to'' relations).
Thus, $G$ has a cycle.
However, that is not possible given \Cref{lemma:leads_to_bigger}.
Hence, $V = \mathcal{I}\mathcal{N}\mathcal{S}$ and the lemma holds.
\end{proof}

The next lemma proves that only the greatest installable view at time $t$ does not lead to any view at time $t$.
Recall that all installable views are comparable (by \Cref{lemma:installable_sequence}).

\begin{lemma} \label{lemma:greatest_does_not_lead}
A view $v \in V(t)$, where $\alpha_t(v) = \top$, does not lead to any view at time $t$ if and only if $v$ is the greatest installable view at time $t$.
\end{lemma}
\begin{proof}
We prove the lemma by proving both directions of the statement.

Let $v$ be the view that does not lead to any view at time $t$; such view exists due to \Cref{lemma:n-1}.
By contradiction, suppose that $v$ is not the greatest installable view at time $t$.
Let $v_{\mathit{max}}$ be the greatest installable view at time $t$.
According to \Cref{lemma:n-1}, $v_{\mathit{max}}$ leads to some view $v'$ at time $t$.
By \Cref{lemma:leads_to_bigger}, we have that $v_{\mathit{max}} \subset v'$, which contradicts the fact that $v_{\mathit{max}}$ is the greatest installable view at time $t$.
Thus, the lemma holds in this direction.

Suppose that $v$ is the greatest installable view at time $t$ and that, by contradiction, it leads to a view $v' \in V(t)$ at time $t$.
According to \Cref{lemma:leads_to_bigger}, $v \subset v'$.
This is contradiction with the fact that $v$ is the greatest installable view at time $t$.
The lemma holds in this direction, as well.
\end{proof}

Now, we introduce an invariant that explains how non-installable views are ``instantiated''.

\begin{invariant} [Creation of Non-Installable Views] \label{invariant:creation_non_installable}
For every $v \in V(t)$ such that $\alpha_t(v) = \bot$, there exists a view $v' \in V(t)$ such that:
\begin{compactitem}
    \item $\alpha_t(v') = \top$, and
    \item $\mathit{seq} \in \rho_t(v')$, where $\mathit{seq} = v_1 \to ... \to v_x \to v \to v_1' \to ... \to v_y', (x \geq 0, y \geq 1)$, and
    \item for every $v_i \in \{v_1, ..., v_x\}$, $v_i \in V(t)$ and $\alpha_t(v_i) = \bot$, and
    \item for every $v_i \in \{v_1, ..., v_x\}$, $\mathit{seq}_i \in \rho_t(v_i)$, where $\mathit{seq}_i = v_{i + 1} \to ... \to v_x \to v \to v_1' \to ... \to v_y'$.
\end{compactitem}
\end{invariant}

\Cref{invariant:creation_non_installable} holds at time $0$ since $V(0) = \{\mathit{genesis}\}$ and $\alpha_0(\mathit{genesis}) = \top$.
Next, we define \emph{auxiliary views}.

\begin{definition} [Auxiliary View] \label{definition:auxiliary_view}
Consider views $v, v' \in V(t)$ such that:
\begin{compactitem}
    \item $\alpha_t(v') = \top$, and
    \item $\alpha_t(v) = \bot$, and
    \item $\mathit{seq} \in \rho_t(v')$, where $\mathit{seq} = v_1 \to ... \to v_x \to v \to v_1' \to ... \to v_y', (x \geq 0, y \geq 1)$, and
    \item for every $v_i \in \{v_1, ..., v_x\}$, $v_i \in V(t)$ and $\alpha_t(v_i) = \bot$, and
    \item for every $v_i \in \{v_1, ..., v_x\}$, $\mathit{seq}_i \in \rho_t(v_i)$, where $\mathit{seq}_i = v_{i + 1} \to ... \to v_x \to v \to v_1' \to ... \to v_y'$.
\end{compactitem}
We say that $v$ is an \emph{auxiliary view} for $v'$ at time $t$.
\end{definition}

The following lemma shows that any non-installable view is an auxiliary view for some installable view.

\begin{lemma} \label{lemma:auxiliary_for_one}
Let $v \in V(t)$ such that $\alpha_t(v) = \bot$.
Then, $v$ is an auxiliary view for a view $v' \in V(t)$ at time $t$.
\end{lemma}
\begin{proof}
The lemma follows from \Cref{invariant:creation_non_installable} and \Cref{definition:auxiliary_view}.
\end{proof}

If $v$ is an auxiliary view for $v'$ at time $t$, then $v' \subset v$.
The following lemma proves this claim.

\begin{lemma} \label{lemma:auxiliary_bigger}
Let $v \in V(t)$ be an auxiliary view for a view $v' \in V(t)$ at time $t$.
Then, $v' \subset v$.
\end{lemma}
\begin{proof}
According to \Cref{definition:auxiliary_view}, $\mathit{seq} \in \rho_t(v')$, where $\mathit{seq} = ... \to v \to v_1' \to ... \to v_y', (y \geq 1)$.
By \Cref{invariant:rhos}, we know that $\mathit{seq}.\mathtt{follows}(v') = \top$, which implies that $v' \subset v$.
\end{proof}

If $v$ is an auxiliary view for a view $v'$ and $v'$ leads to a view $v''$, then $v \subset v''$.

\begin{lemma} \label{lemma:auxiliary_smaller}
Let $v \in V(t)$ be an auxiliary view for a view $v' \in V(t)$ at time $t$.
Let $v'$ lead to $v'' \in V(t)$ at time $t$.
Then, $v \subset v''$.
\end{lemma}
\begin{proof}
According to \Cref{definition:auxiliary_view}, $\mathit{seq} \in \rho_t(v')$, where $\mathit{seq} = ... \to v \to v_1' \to ... \to v_y', (y \geq 1)$.
By \Cref{definition:leads}, $\mathit{seq}'' \in \rho_t(v')$, where $\mathit{seq}'' = ... \to v''$.
\Cref{invariant:rhos} ensures that (1) $\mathit{seq}$ and $\mathit{seq}''$ are sequences, and (2) $\mathit{seq}.\mathtt{follows}(v') = \mathit{seq}''.\mathtt{follows}(v') = \top$.
By \Cref{lemma:weak_accuracy}, there are three possible scenarios:
\begin{compactitem}
    \item $\mathit{seq} = \mathit{seq}''$: In this case, we have that $v \subset v''$.
    The lemma holds.
    
    \item $\mathit{seq} \subset \mathit{seq}''$: In this case, we know that $v \in \mathit{seq}''$.
    Thus, $v \subset v''$ and the lemma holds.
    
    \item $\mathit{seq} \supset \mathit{seq}''$: Hence, $v'' \in \mathit{seq}$.
    We conclude that $v \subset v''$ because, otherwise, either (1) $\alpha_t(v) = \top$ (if $v = v''$), or (2) $v$ is not an auxiliary view for $v'$ at time $t$ since $\alpha_t(v'') = \top$ (if $v'' \subset v$).
\end{compactitem}
The lemma holds since the claim is correct in all possible cases.
\end{proof}

The following lemma shows that if $v$ leads to $v'$, then there does not exist an installable view $v''$ such that $v \subset v'' \subset v'$.

\begin{lemma} \label{lemma:installable_between_leading}
Let $v \in V(t)$ lead to $v' \in V(t)$ at time $t$.
There does not exist a view $v'' \in V(t)$ such that (1) $\alpha_t(v'') = \top$, and (2) $v \subset v'' \subset v'$.
\end{lemma}
\begin{proof}
Let $\mathcal{I}\mathcal{N}\mathcal{S} = \{v \,|\, v \in V(t) \land \alpha_t(v) = \top\}$.
According to \Cref{lemma:installable_sequence}, $\mathcal{I}\mathcal{N}\mathcal{S} $ is a sequence.
Moreover, $|\mathcal{I}\mathcal{N}\mathcal{S}| < \infty$.

We prove the lemma by contradiction.
Hence, suppose that there exists a view $v''$ such that (1) $v'' \in \mathcal{I}\mathcal{N}\mathcal{S}$, and (2) $v \subset v'' \subset v'$.

Since $v$ leads to $v'$ at time $t$, we know that $v \subset v'$ (by \Cref{lemma:leads_to_bigger}).
We consider the following construction:
\begin{compactenum}
    \item Start with $V \gets \{v, v'\}$ and $\mathit{first} \gets v$.
    
    \item Repeat until $\mathit{first} = \mathit{genesis}$:
    \begin{compactenum}
        \item Select a view $v_i \in \mathcal{I}\mathcal{N}\mathcal{S}$ such that $v_i$ leads to $\mathit{first}$.
        \item Update $V \gets V \cup \{v_i\}$ and $\mathit{first} \gets v_i$.
    \end{compactenum}
\end{compactenum}
According to the construction and \Cref{lemma:leads_to_bigger}, $V$ is a sequence.
Moreover, $\mathit{first}$ represents the first (i.e., smallest) view of $V$.
Therefore, a view $v_i$, where $v_i$ leads to $\mathit{first}$, must belong to $\mathcal{I}\mathcal{N}\mathcal{S} \setminus{V}$ (otherwise, we would contradict \Cref{lemma:leads_to_bigger}).
Thus, $V$ is ``growing'' with each iteration.
Finally, an execution of the construction eventually terminates since (1) $|\mathcal{I}\mathcal{N}\mathcal{S}| < \infty$, (2) \Cref{lemma:exists_leads_to} holds, (3) \Cref{lemma:genesis_no_lead} stands, and (4) $\mathit{genesis} \in \mathcal{I}\mathcal{N}\mathcal{S}$.

Since $v \subset v''$, we conclude that $v'' \notin V$ and $v'' \neq \mathit{genesis}$.
Therefore, there exists a view that leads to $v''$ (by \Cref{lemma:exists_leads_to}).
Let $R = \mathcal{I}\mathcal{N}\mathcal{S} \setminus{V}$.
Hence, $v'' \in R$.
Now, we consider the following construction:
\begin{compactenum}
    \item Start with $V' \gets \{v''\}$ and $\mathit{first} \gets v''$.
    
    \item Select a view $v_i \in \mathcal{I}\mathcal{N}\mathcal{S}$ such that $v_i$ leads to $\mathit{first}$.
    
    \item If $v_i \in V$, then terminate.
    
    \item Update $V' \gets V' \cup \{v_i\}$ and $\mathit{first} \gets v_i$.
    
    \item Go to step 2.
\end{compactenum}
By construction and \Cref{lemma:leads_to_bigger}, $V'$ is a sequence.
Hence, for every view $\mathit{view} \in V'$, $\mathit{view} \subseteq v''$.
Importantly, an execution of this construction eventually terminates because (1) $|\mathcal{I}\mathcal{N}\mathcal{S}| < \infty$, (2) \Cref{lemma:exists_leads_to}, (3) \Cref{lemma:genesis_no_lead}, and (4) $\mathit{genesis} \in V$.
Thus, once the execution terminates, we either contradict \Cref{lemma:leads_to_bigger} (if $v_i = v'$; since $\mathit{view} \subseteq v''$, for every $\mathit{view} \in V'$) or \Cref{lemma:leads_to_one} (if $v_i \neq v'$).
Therefore, the lemma holds.
\end{proof}

Now, we prove that a view can be an auxiliary view for a single installable view.

\begin{lemma} \label{lemma:auxiliary_for_one_at_most}
Let $v \in V(t)$ be an auxiliary view for a view $v' \in V(t)$ at time $t$.
Moreover, let $v$ be an auxiliary view for a view $v'' \in V(t)$ at time $t$.
Then, $v' = v''$.
\end{lemma}
\begin{proof}
By contradiction, suppose that $v' \neq v''$.
\Cref{lemma:installable_sequence} shows that either $v' \subset v''$ or $v' \supset v''$.
Without loss of generality, let $v' \subset v''$.
Given \Cref{lemma:auxiliary_bigger}, we conclude that $v' \subset v$ and $v'' \subset v$.

Let $v'$ lead to a view $v_i \in V(t)$ at time $t$ (such view $v_i$ exists since $v' \subset v''$ and \cref{lemma:n-1,lemma:greatest_does_not_lead}).
We know, by \Cref{lemma:installable_sequence}, that either $v_i \subseteq v''$ or $v_i \supset v''$.
We conclude that $v_i \subseteq v''$ (otherwise, the statement of \Cref{lemma:installable_between_leading} would be violated).

\Cref{lemma:auxiliary_smaller} shows that $v \subset v_i$.
Given that $v_i \subseteq v''$, we conclude that $v \subset v''$.
This represents contradiction with $v'' \subset v$, which concludes the proof.
\end{proof}

The next lemma proves that, if a process obtains an evidence $\epsilon$ for $\mathit{set} \neq \emptyset$ and view $v'$ by time $t$, then $\mathit{set} \in \beta_t(v)$.

\begin{lemma} \label{lemma:subset_of_beta}
Let a (correct or faulty) process obtain an evidence $\epsilon'$ by time $t$ such that $\mathtt{valid}((v', \mathit{set}', \epsilon'), v) = \top$, where $\mathit{set}' \neq \emptyset$.
Then, $v \in V(t)$ and $\mathit{set}' \in \beta_t(v)$.
\end{lemma}
\begin{proof}
First, we conclude that $v \neq \mathit{genesis}$ since $\mathtt{valid}((v', \mathit{set}' \neq \emptyset, \epsilon'), \mathit{genesis})$ must return $\bot$ (because of the check at line~\ref{line:v_equal_genesis} of \Cref{lst:view_set_evidence_valid}).
Moreover, we know that $\epsilon'.\mathit{path}.\mathtt{destination()} = v$ and $\epsilon'.\mathit{path}.\mathtt{tail()} = \mathit{set}'$ (by the check at line~\ref{line:check_destination_tail} of \Cref{lst:view_set_evidence_valid}).
By \Cref{lemma:knows_view_system}, we know that $v \in V(t)$ and $\mathit{set}' \in \beta_t(v)$, which concludes the proof.
\end{proof}

Similarly, if a process obtains an evidence $\epsilon$ for $\mathit{set} = \emptyset$ by time $t$, then $\alpha_t(v) = \top$.

\begin{lemma} \label{lemma:subset_of_beta_2}
Let a (correct or faulty) process obtain an evidence $\epsilon'$ by time $t$ such that $\mathtt{valid}((v', \emptyset, \epsilon'), v) = \top$,
Then, $v \in V(t)$ and $\alpha_t(v) = \top$.
\end{lemma}
\begin{proof}
If $v = \mathit{genesis}$, the lemma trivially holds.
Otherwise, we know that $\epsilon'.\mathit{path}.\mathtt{destination()} = v$ and $\epsilon'.\mathit{path}.\mathtt{tail()} = \emptyset$ (by the check at line~\ref{line:check_destination_tail} of \Cref{lst:view_set_evidence_valid}).
By \Cref{lemma:knows_view_system}, $v \in V(t)$.
Let $m$ be the last message in the ordered set of $\epsilon'.\mathit{path}$.
We know that $m.\mathtt{tail()} = \emptyset$ and $m.\mathtt{destination()} = v$ (by \Cref{lst:view_path}).
Moreover, $m \in \mathcal{I}^*(t)$ and $m.\mathtt{source()} \in V(t)$ (by \Cref{lemma:knows_view_system}).
Hence, $\alpha_t(v) = \top$, the lemma holds.
\end{proof}

We now introduce the last invariant we assume holds at time $t$.

\begin{invariant} \label{invariant:accepted_produced_installable}
Let $v \in V(t)$.
If $\mathit{set} \in \beta_t(v)$, where $\mathit{set} = \{v_1', ..., v_y'\}$ and $y \geq 1$, then:
\begin{compactitem}
    \item $\mathit{seq}' \in \rho_t(v')$, where $v' \in V(t)$, $\alpha_t(v') = \top$ and $\mathit{seq}' = v_1 \to ... \to v_x \to v \to v_1' \to ... \to v_y', (x \geq 0)$, and
    \item for every $v_i \in \{v_1, ..., v_x\}$, $v_i \in V(t)$ and $\alpha_t(v_i) = \bot$, and
    \item for every $v_i \in \{v_1, ..., v_x\}$, $\mathit{seq}_i \in \rho_t(v_i)$, where $\mathit{seq}_i = v_{i + 1} \to ... \to v_x \to v \to v_1' \to ... \to v_y'$.
\end{compactitem}
\end{invariant}

Since $\mathcal{I}^*(0) = \emptyset$ at time $0$, we know that $V(0) = \{\mathit{genesis}\}$ and $\beta_0(\mathit{genesis}) = \emptyset$.
Hence, \Cref{invariant:accepted_produced_installable} is satisfied at time $0$.

\begin{lemma} \label{lemma:beta_genesis_empty}
$\beta_t(\mathit{genesis}) = \emptyset$.
\end{lemma}
\begin{proof}
By contradiction, let $\beta_t(\mathit{genesis}) \neq \emptyset$.
Let $\mathit{set} \in \beta_t(\mathit{genesis})$, where $\mathit{set} = \{v_1', ..., v_y'\}$ and $y \geq 1$.
According to \Cref{invariant:accepted_produced_installable}, there exists an installable view $v \in V(t)$ such that $\mathit{seq}' \in \rho_t(v)$, where $\mathit{seq}' = v_1 \to ... \to v_x \to \mathit{genesis} \to v_1' \to ... \to v_y'$ and $x \geq 0$.
By \Cref{invariant:rhos}, we know that $\mathit{seq}'.\mathtt{follows}(v) = \top$, which implies that $v \subset \mathit{genesis}$.
However, this statement contradicts \Cref{invariant:initial_view}, which concludes the proof.
\end{proof}

We prove that, if $\mathit{set} \in \beta_t(v)$, then $\mathit{set}$ is a sequence and $\mathit{set}.\mathtt{follows}(v) = \top$.

\begin{lemma} \label{lemma:beta_sequence_follows}
Let $v \in V(t)$.
If $\mathit{set} \in \beta_t(v)$, then (1) $\mathit{set}$ is a sequence, and (2) $\mathit{set}.\mathtt{follows}(v) = \top$.
\end{lemma}
\begin{proof}
Let $\mathit{set} = \{v_1', ..., v_y'\}$, where $y \geq 1$.
By \Cref{invariant:accepted_produced_installable}, there exists a view $v' \in V(t)$ such that (1) $\alpha_t(v') = \top$, and (2) $\mathit{seq}' \in \rho_t(v')$, where $\mathit{seq}' = v_1 \to ... \to v_x \to v \to v_1' \to ... \to v_y'$, where $x \geq 0$.
Therefore, $\mathit{set}$ is a sequence and $\mathit{set}.\mathtt{follows}(v) = \top$.
\end{proof}

The next lemma shows that, if $\mathit{seq} \in \beta_t(v)$, where $v$ is an installable view, and $\mathit{seq}$ ``comes from'' an installable view $v'$, then $v'$ leads to $v$.
Recall that all sets in $\beta_t(v)$ are sequences according to \Cref{lemma:beta_sequence_follows}.

\begin{lemma} \label{lemma:invariant_4_installable}
Let $\mathit{seq} \in \beta_t(v)$, where $\mathit{seq} = v_1' \to ... \to v_y', (y \geq 1)$ and $\alpha_t(v) = \top$. 
Then, $\mathit{seq}' \in \rho_t(v')$, where $\mathit{seq}' = v_1 \to ... \to v_x \to v \to v_1' \to ... \to v_y', (x \geq 0)$ and $v'$ leads to $v$ at time $t$. 
\end{lemma}
\begin{proof}
By \Cref{invariant:accepted_produced_installable}, we know that:
\begin{compactitem}
    \item $\mathit{seq}' \in \rho_t(v')$, where $\alpha_t(v') = \top$ and $\mathit{seq}' = v_1 \to ... \to v_x \to v \to v_1' \to ... \to v_y', (x \geq 0)$, and
    \item for every $v_i \in \{v_1, ..., v_x\}$, $v_i \in V(t)$ and $\alpha_t(v_i) = \bot$, and
    \item for every $v_i \in \{v_1, ..., v_x\}$, $\mathit{seq}_i \in \rho_t(v_i)$, where $\mathit{seq}_i = v_{i + 1} \to ... \to v_x \to v \to v_1' \to ... \to v_y'$.
\end{compactitem}
In order to conclude the proof, we need to prove that $v'$ indeed leads to $v$ at time $t$.

By \Cref{invariant:rhos}, we have that $v' \subset v$.
By contradiction, suppose that $v'$ leads to $v'' \neq v$ at time $t$ (such view $v''$ exists since $v' \subset v$ and \cref{lemma:n-1,lemma:greatest_does_not_lead}).
Therefore, by \Cref{definition:leads}, $\mathit{seq}'' \in \rho_t(v')$, where $\mathit{seq}'' = ... \to v''$.
Let us analyze all three possible cases according to \Cref{lemma:weak_accuracy}:
\begin{compactitem}
    \item $\mathit{seq}' = \mathit{seq}''$: Hence, $v \subset v''$ and \Cref{definition:leads} is not satisfied for $v''$ (since $\alpha_t(v) = \top$).
    
    \item $\mathit{seq}' \subset \mathit{seq}''$: Hence, $v \in \mathit{seq}''$.
    Thus, $v \subset v''$ and \Cref{definition:leads} is not satisfied for $v''$ (since $\alpha_t(v) = \top$).
    
    \item $\mathit{seq}' \supset \mathit{seq}''$: Hence, $v'' \in \mathit{seq}'$.
    If $v'' \subset v$, \Cref{invariant:accepted_produced_installable} is violated since $\alpha_t(v'') = \top$.
    If $v \subset v''$, then we have the following: (1) $v'$ leads to $v''$, and (2) $v' \subset v \subset v''$, where $\alpha_t(v) = \top$.
    Hence, \Cref{lemma:installable_between_leading} is contradicted.
\end{compactitem}
The lemma holds.
\end{proof}

The next lemma shows that, if $\mathit{seq} \in \beta_t(v)$, where $v$ is non-installable, and $\mathit{seq}$ ``comes from'' an installable view $v'$, then $v$ is an auxiliary view for $v'$.

\begin{lemma} \label{lemma:invariant_4_auxiliary}
Let $\mathit{seq} \in \beta_t(v)$, where $\mathit{seq} = v_1' \to ... \to v_y', (y \geq 1)$ and $\alpha_t(v) = \bot$. 
Then, $\mathit{seq}' \in \rho_t(v')$, where $\mathit{seq}' = v_1 \to ... \to v_x \to v \to v_1' \to ... \to v_y', (x \geq 0)$, views from the $\{v_1, ..., v_x\}$ set are not installable at time $t$ and $v$ is an auxiliary view for $v'$ at time $t$. 
\end{lemma}
\begin{proof}
By \Cref{invariant:accepted_produced_installable}, we know that:
\begin{compactitem}
    \item $\mathit{seq}' \in \rho_t(v')$, where $\alpha_t(v') = \top$ and $\mathit{seq}' = v_1 \to ... \to v_x \to v \to v_1' \to ... \to v_y', (x \geq 0)$, and
    \item for every $v_i \in \{v_1, ..., v_x\}$, $v_i \in V(t)$ and $\alpha_t(v_i) = \bot$, and
    \item for every $v_i \in \{v_1, ..., v_x\}$, $\mathit{seq}_i \in \rho_t(v_i)$, where $\mathit{seq}_i = v_{i + 1} \to ... \to v_x \to v \to v_1' \to ... \to v_y'$.
\end{compactitem}
Hence, \Cref{definition:auxiliary_view} is satisfied and $v$ is an auxiliary view for $v'$ at time $t$.
Moreover, views from the $\{v_1, ..., v_x\}$ set are not installable at time $t$, which concludes the proof.
\end{proof}

The next lemma shows that, if $\mathit{seq}_1, \mathit{seq}_2 \in \beta_t(v)$, then $\mathit{seq}_1$ and $\mathit{seq}_2$ are comparable.

\begin{lemma} \label{lemma:comparable_beta}
Let $v \in V(t)$.
If $\mathit{seq}_1, \mathit{seq}_2 \in \beta_t(v)$, then either $\mathit{seq}_1 \subseteq \mathit{seq}_2$ or $\mathit{seq}_1 \supset \mathit{seq}_2$.
\end{lemma}
\begin{proof}
Let $\alpha_t(v) = \top$.
For every $\mathit{seq} \in \beta_t(v)$, where $\mathit{seq} = v_1' \to ... \to v_y', (y \geq 1)$, $\mathit{seq}' \in \rho_t(v')$, where $\mathit{seq}' = ... \to v \to v_1' \to ... \to v_y'$ and $v'$ leads to $v$ at time $t$ (according to \Cref{lemma:invariant_4_installable}).
By \Cref{lemma:leads_from_one}, only view $v'$ leads to $v$ at time $t$.
Thus, the lemma follows from \Cref{lemma:weak_accuracy}.

Let $\alpha_t(v) = \bot$.
For every $\mathit{seq} \in \beta_t(v)$, where $\mathit{seq} = v_1' \to ... \to v_y', (y \geq 1)$, $\mathit{seq}' \in \rho_t(v')$, where $\mathit{seq}' = ... \to v \to v_1' \to ... \to v_y'$ and $v$ is an auxiliary view for $v'$ at time $t$ (according to \Cref{lemma:invariant_4_auxiliary}).
According to \Cref{lemma:auxiliary_for_one_at_most}, $v$ is an auxiliary view only for view $v'$ at time $t$.
Thus, the lemma follows from \Cref{lemma:weak_accuracy}.
\end{proof}

Let $[\mathtt{INSTALL}, v, \mathit{set}, \omega]$ be the first $\mathtt{INSTALL}$ message obtained by a (correct or faulty) process after time $t$ such that $v \in V(t)$.
We prove that either $\mathit{set} \in \beta_t(v)$ or $\alpha_t(v) = \top$.
Note that we do not consider $\mathtt{INSTALL}$ messages with $v \notin V(t)$ since, in that case, $v$ is not a valid view (by \Cref{lemma:not_well_founded_not_valid}).

Importantly, if multiple new $\mathtt{INSTALL}$ messages are obtained at the same time $t' > t$, the notion of the ``first'' $\mathtt{INSTALL}$ message is not defined.
However, in that case, we order the messages in any arbitrary way.
In other words, we ``artificially'' select the ``first'' $\mathtt{INSTALL}$ message and observe its properties, as well as the properties after this message has been processed (i.e., added to the set of all obtained $\mathtt{INSTALL}$ messages).
Importantly, we process all messages from time $t'$ before moving on to processing messages from some time $t'' > t'$.
We are allowed to ``separately'' process the $\mathtt{INSTALL}$ messages from time $t'$ because, after all these $\mathtt{INSTALL}$ messages are processed, we obtain the exact same system state as we would obtain if all the messages were processed ``together''.

\begin{lemma} \label{lemma:produced_then_accepted}
Let $[\mathtt{INSTALL}, v, \mathit{set}, \omega] \notin \mathcal{I}^*(t)$ be the first $\mathtt{INSTALL}$ message obtained by a (correct or faulty) process after time $t$ such that $v \in V(t)$; let that time be $t' > t$.
Then, $\mathit{set} \in \beta_t(v)$ or $\alpha_t(v) = \top$.
\end{lemma}
\begin{proof}
By \Cref{invariant:valid}, $v$ is a valid view.
Moreover, $\mathit{preconditions}(v, t' - 1) = \top$ (since $t'$ is the first time after time $t$ at which an $\mathtt{INSTALL}$ message is obtained and by \cref{lemma:subset_of_beta,lemma:beta_sequence_follows,lemma:comparable_beta}).
By safety of $\mathit{vg}(v)$, we have that $(v^*, \emptyset, \epsilon^*) \in \Lambda_{t' - 1}^v$ or $(v^*, \mathit{set}, \epsilon^*) \in \Lambda_{t' - 1}^v$.
If $(v^*, \emptyset, \epsilon^*) \in \Lambda_{t' - 1}^v$, then $\alpha_t(v) = \top$ (by \Cref{lemma:subset_of_beta_2}).
If $(v^*, \mathit{set}, \epsilon^*) \in \Lambda_{t' - 1}^v$, then $\mathit{set} \in \beta_t(v)$ (by \Cref{lemma:subset_of_beta}).
Therefore, the lemma holds.
\end{proof}

Next, we prove some properties of $\mathit{set}$.

\begin{lemma} \label{lemma:set_sequence}
Let $[\mathtt{INSTALL}, v, \mathit{set}, \omega] \notin \mathcal{I}^*(t)$ be the first $\mathtt{INSTALL}$ message obtained by a (correct or faulty) process after time $t$ such that $v \in V(t)$; let that time be $t' > t$.
Then, (1) $\mathit{set} \neq \emptyset$, (2) $\mathit{set}$ is a sequence, and (3) $\mathit{set}.\mathtt{follows}(v) = \top$.
\end{lemma}
\begin{proof}
By \Cref{invariant:valid}, $v$ is a valid view.
Moreover, $\mathit{preconditions}(v, t' - 1) = \top$ (since $t'$ is the first time after time $t$ at which an $\mathtt{INSTALL}$ message is obtained and by \cref{lemma:subset_of_beta,lemma:beta_sequence_follows,lemma:comparable_beta}).
Hence, the lemma follows from validity of $\mathit{vg}(v)$.
\end{proof}

Finally, we prove that all well-founded views at time $t$ are comparable.
This lemma is the crucial ingredient for proving the view comparability property (see \Cref{lst:reconfiguration_properties_new}).

\begin{lemma} \label{lemma:all_comparable}
Let $v, v' \in V(t)$.
Then, either $v \subseteq v'$ or $v \supset v'$.
\end{lemma}
\begin{proof}
If $v = v'$, the lemma holds.
Otherwise, we consider four possible cases:
\begin{compactenum}
    \item $\alpha_t(v) = \alpha_t(v') = \top$: In this case, the lemma follows from \Cref{lemma:installable_sequence}.
    
    \item $\alpha_t(v) = \top$ and $\alpha_t(v') = \bot$:
    Let $v'$ be the auxiliary view for a view $v_i$ at time $t$ (such view $v_i$ exists due to \Cref{lemma:auxiliary_for_one}).
    According to \Cref{lemma:installable_sequence}, $v \subseteq v_i$ or $v \supset v_i$.
    We analyze all three possibilities:
    \begin{compactitem}
        \item $v = v_i$: In this case, the lemma holds since $v \subset v'$ (by \Cref{lemma:auxiliary_bigger}).
        
        \item $v \subset v_i$: We have that $v_i \subset v'$ (by \Cref{lemma:auxiliary_bigger}).
        Hence, $v \subset v'$ and the lemma holds.
        
        \item $v \supset v_i$: Let $v_i$ lead to a view $v''$ at time $t$ (such view $v''$ exists due to \cref{lemma:n-1,lemma:greatest_does_not_lead}).
        Because of \Cref{lemma:installable_sequence} and \Cref{lemma:installable_between_leading}, we have that $v'' \subseteq v$.
        We conclude that $v' \subset v$ since $v' \subset v''$ (by \Cref{lemma:auxiliary_smaller}).
    \end{compactitem}
    
    \item $\alpha_t(v) = \bot$ and $\alpha_t(v') = \top$: Analogous to the previous case.
    
    \item $\alpha_t(v) = \alpha_t(v') = \bot$: Let $v$ be an auxiliary view for $v_i$ and let $v'$ be an auxiliary view for $v_i'$ (by \Cref{lemma:auxiliary_for_one}).
    According to \Cref{lemma:installable_sequence}, we have three possibilities to examine:
    \begin{compactitem}
        \item $v_i = v_i'$: By \Cref{definition:auxiliary_view}, we have that $\mathit{seq}, \mathit{seq}' \in \rho_t(v_i)$, where $\mathit{seq} = ... \to v \to ...$ and $\mathit{seq}' = ... \to v' \to ...$.
        The lemma follows from \Cref{lemma:weak_accuracy}, \Cref{invariant:rhos} and the definition of a sequence.
        
        \item $v_i \subset v_i'$: Let $v_i$ lead to a view $v''$ at time $t$ (such view $v''$ exists by \cref{lemma:n-1,lemma:greatest_does_not_lead}).
        Because of \Cref{lemma:installable_sequence} and \Cref{lemma:installable_between_leading}, we have that $v'' \subseteq v_i'$.
        We conclude that $v \subset v'$ since $v \subset v''$ (by \Cref{lemma:auxiliary_smaller}).
        
        \item $v_i \supset v_i'$: Analogous to the case $v_i \subset v_i'$.
    \end{compactitem}
\end{compactenum}
The lemma holds since it holds in all four possibilities.
\end{proof}

Let $[\mathtt{INSTALL}, v, \mathit{set}, \omega]$ be the first $\mathtt{INSTALL}$ message obtained by a (correct or faulty) process after time $t$ such that $v \in V(t)$.
Let that happen at time $t'$.
The following lemma proves that $\mathit{set}$ comes from a view that is installable at time $t'$.

\begin{lemma} \label{lemma:pre_final}
Let $[\mathtt{INSTALL}, v, \mathit{seq} = v_1' \to ... \to v_y', \omega] \notin \mathcal{I}^*(t)$, where $y \geq 1$, be the first $\mathtt{INSTALL}$ message obtained by a (correct or faulty) process after time $t$ such that $v \in V(t)$; let that be at time $t' > t$.
Then:
\begin{compactitem}
    \item $\mathit{seq}' \in \rho_{t'}(v')$, where $v' \in V(t')$, $\alpha_{t'}(v') = \top$ and $\mathit{seq}' = v_1 \to ... \to v_x \to v_1' \to ... \to v_y', (x \geq 0)$, and
    \item for every $v_i \in \{v_1, ..., v_x\}$, $v_i \in V(t')$ and $\alpha_{t'}(v_i) = \bot$, and
    \item for every $v_i \in \{v_1, ..., v_x\}$, $\mathit{seq}_i \in \rho_{t'}(v_i)$, where $\mathit{seq}_i = v_{i+1} \to ... \to v_x \to v_1' \to ... \to v_y'$.
\end{compactitem}
\end{lemma}
\begin{proof}
If $\alpha_t(v) = \top$, then $\alpha_{t'}(v) = \top$ (since $\mathcal{I}^*(t) \subseteq \mathcal{I}^*(t')$).
Moreover, we have that $\mathit{seq} \in \rho_{t'}(v)$, which concludes the lemma.

Let $\alpha_t(v) = \bot$.
According to \Cref{lemma:produced_then_accepted}, $\mathit{seq} \in \beta_t(v)$.
Since \Cref{invariant:accepted_produced_installable} holds at time $t$, we conclude that:
\begin{compactitem}
    \item $\mathit{seq}' \in \rho_t(v')$, where $v' \in V(t)$, $\alpha_t(v') = \top$ and $\mathit{seq}' = v_1 \to ... \to v_x \to v \to v_1' \to ... \to v_y', (x \geq 0)$, and
    \item for every $v_i \in \{v_1, ..., v_x\}$, $v_i \in V(t)$ and $\alpha_t(v_i) = \bot$, and
    \item for every $v_i \in \{v_1, ..., v_x\}$, $\mathit{seq}_i \in \rho_t(v_i)$, where $\mathit{seq}_i = v_{i+1} \to ... \to v_x \to v \to v_1' \to ... \to v_y'$.
\end{compactitem}
At time $t'$, the previous statements still hold (since $\mathcal{I}^*(t) \subseteq \mathcal{I}^*(t')$).
Moreover, $\mathit{seq} \in \rho_{t'}(v)$, which concludes the lemma.
\end{proof}

Finally, we show that all invariants are preserved after a new $\mathtt{INSTALL}$ message is obtained.

\begin{theorem} \label{theorem:final_1}
Let $[\mathtt{INSTALL}, v, \mathit{seq} = v', \omega] \notin \mathcal{I}^*(t)$ be the first $\mathtt{INSTALL}$ message obtained by a (correct or faulty) process after time $t$ such that $v \in V(t)$; let that be at time $t' > t$.
Let $v' \notin V(t)$.
Then, all invariants hold at time $t'$.
\end{theorem}
\begin{proof}
We prove the preservation of invariants one by one.

\smallskip
\noindent \underline{\Cref{invariant:rhos}:}
\Cref{invariant:rhos} holds for $v'$ at time $t'$ since $\rho_{t'}(v') = \emptyset$.
Let us prove this statement.
By contradiction, suppose that a process has obtained an $\mathtt{INSTALL}$ message associated with view $v'$ before time $t'$.
Since $v' \notin V(t)$, we know that $v'$ is not a valid view (by \Cref{lemma:not_well_founded_not_valid}).
However, since $v' \in V(t')$, $v'$ is a valid view according to \Cref{invariant:valid}.
Thus, contradiction.

The invariant holds for all views $v_t \in V(t)$, where $v_t \neq v$, because $\rho_{t'}(v_t) = \rho_t(v_t)$ and the invariant holds at time $t$.
Finally, the invariant holds for $v$ at time $t'$ because of \Cref{lemma:set_sequence} and the fact that the invariant holds at time $t$.

\smallskip
\noindent \underline{\Cref{invariant:creation_installable}:}
We conclude that $\alpha_{t'}(v') = \top$.
By \Cref{lemma:pre_final}, there exists a view $v^* \in V(t')$ such that:
\begin{compactitem}
    \item $\alpha_{t'}(v^*) = \top$, and
    \item $\mathit{seq}' \in \rho_{t'}(v^*)$, where $\mathit{seq}' = v_1 \to ... \to v_x \to v', (x \geq 0)$, and
    \item for every $v_i \in \{v_1, ..., v_x\}$, $v_i \in V(t')$ and $\alpha_{t'}(v_i) = \bot$, and
    \item for every $v_i \in \{v_1, ..., v_x\}$, $\mathit{seq}_i \in \rho_{t'}(v_i)$, where $\mathit{seq}_i = v_{i + 1} \to ... \to v_x \to v'$.
\end{compactitem}
Hence, \Cref{invariant:creation_installable} holds for $v'$ at time $t'$.

Since $v' \notin V(t)$ and \Cref{invariant:creation_installable} holds at time $t$, we conclude that \Cref{invariant:creation_installable} holds for all views $v_t \in V(t)$.
Thus, \Cref{invariant:creation_installable} holds at time $t'$.

\smallskip
\noindent \underline{\Cref{invariant:initial_view}:}
Since $v^* \in V(t)$ and \Cref{invariant:initial_view} holds at time $t$, we know that $\mathit{genesis} \subseteq v^*$.
We know that $\mathit{seq}' \in \rho_{t'}(v^*)$.
Since \Cref{invariant:rhos} holds at time $t'$, we know that $v^* \subset v'$.
Hence, $\mathit{genesis} \subset v'$, which proves the invariant preservation at time $t'$.

\smallskip
\noindent \underline{\Cref{invariant:creation_non_installable}:}
Since $v' \notin V(t)$, $\alpha_{t'}(v') = \top$ and \Cref{invariant:creation_non_installable} holds at time $t$, \Cref{invariant:creation_non_installable} holds at time $t'$.

\smallskip
\noindent \underline{\Cref{invariant:accepted_produced_installable}:}
Since $v' \notin V(t)$, $\beta_{t'}(v') = \emptyset$ and \Cref{invariant:accepted_produced_installable} holds at time $t$, \Cref{invariant:accepted_produced_installable} holds at time $t'$.
\end{proof}

Now, we analyze the second possibility.

\begin{theorem} \label{theorem:final_2}
Let $[\mathtt{INSTALL}, v, \mathit{seq} = v', \omega] \notin \mathcal{I}^*(t)$ be the first $\mathtt{INSTALL}$ message obtained by a (correct or faulty) process after time $t$ such that $v \in V(t)$; let that be at time $t' > t$.
Let $v' \in V(t)$.
Then, all invariants hold at time $t'$.
\end{theorem}
\begin{proof}
We separate the proof into two cases.

\medskip
\noindent(1) Let $\alpha_t(v') = \top$.

\smallskip
\noindent \underline{\Cref{invariant:rhos}:}
\Cref{invariant:rhos} holds for all views $v_t \in V(t) \setminus\{v\}$ at time $t'$ since it holds at time $t$.
The invariant holds for $v$ at time $t'$ because of \Cref{lemma:set_sequence} and the fact that the invariant holds at time $t$.

\smallskip
\noindent \underline{\Cref{invariant:creation_installable}:}
Since $\alpha_t(v') = \top$, all views installable at time $t$ are installable at time $t'$.
The fact that \Cref{invariant:creation_installable} holds at time $t$ implies that the invariant holds at time $t'$.

\smallskip
\noindent \underline{\Cref{invariant:initial_view}:}
Since $\alpha_t(v') = \top$, all views installable at time $t$ are installable at time $t'$.
Hence, \Cref{invariant:initial_view} holds at time $t'$ since it holds at time $t$.

\smallskip
\noindent \underline{\Cref{invariant:creation_non_installable}:}
Since $\alpha_t(v') = \top$, all views installable at time $t$ are installable at time $t'$.
Hence, the invariant holds at time $t'$ since it holds at time $t$.

\smallskip
\noindent \underline{\Cref{invariant:accepted_produced_installable}:}
Since $\alpha_t(v') = \top$, all views installable at time $t$ are installable at time $t'$.
Moreover, $\beta_{t'}(v_t) = \beta_t(v_t)$, for every view $v_t \in V(t)$ (note that $V(t') = V(t)$).
Hence, the invariant holds at time $t'$ since it holds at time $t$.

\medskip
\noindent (2) Let $\alpha_t(v') = \bot$.

\smallskip
\noindent \underline{\Cref{invariant:rhos}:}
\Cref{invariant:rhos} holds for all views $v_t \in V(t) \setminus\{v\}$ at time $t'$ since it holds at time $t$.
The invariant holds for $v$ at time $t'$ because of \Cref{lemma:set_sequence} and the fact that the invariant holds at time $t$.

\smallskip
\noindent \underline{\Cref{invariant:creation_installable}:}
By \Cref{lemma:pre_final}, there exists a view $v^* \in V(t')$ such that:
\begin{compactitem}
    \item $\alpha_{t'}(v^*) = \top$, and
    \item $\mathit{seq}' \in \rho_{t'}(v^*)$, where $\mathit{seq}' = v_1 \to ... \to v_x \to v', (x \geq 0)$, and
    \item for every $v_i \in \{v_1, ..., v_x\}$, $v_i \in V(t')$ and $\alpha_{t'}(v_i) = \bot$, and
    \item for every $v_i \in \{v_1, ..., v_x\}$, $\mathit{seq}_i \in \rho_{t'}(v_i)$, where $\mathit{seq}_i = v_{i + 1} \to ... \to v_x \to v'$.
\end{compactitem}
Hence, \Cref{invariant:creation_installable} holds for $v'$ at time $t'$.

Consider now a view $v_{\mathit{ins}} \in V(t)$ such that (1) $v_{\mathit{ins}} \neq \mathit{genesis}$, and (2) $\alpha_t(v_{\mathit{ins}}) = \top$.
Since \Cref{invariant:creation_installable} holds at time $t$ and $v_{\mathit{ins}} \neq \mathit{genesis}$, there is a view $v_{\mathit{lead}}$ that leads to $v_{\mathit{ins}}$ at time $t$ (by \Cref{lemma:exists_leads_to}).
By \Cref{definition:leads}, we have:
\begin{compactitem}
    \item $\alpha_t(v_{\mathit{lead}}) = \top$, and
    
    \item $\mathit{seq}_{\mathit{ins}} \in \rho_t(v_{\mathit{lead}})$, where $\mathit{seq}_{\mathit{ins}} = v_1' \to ... \to v_y' \to v_{\mathit{ins}}$ and $y \geq 0$, and
    
    \item for every $v_i \in \{v_1', ..., v_y'\}$, $v_i \in V(t)$ and $\alpha_t(v_i) = \bot$, and
    
    \item for every $v_i \in \{v_1', ..., v_y'\}$, $\mathit{seq}_i \in \rho_t(v_i)$, where $\mathit{seq}_i = v_{i + 1}' \to ... \to v_y' \to v_{\mathit{ins}}$.
\end{compactitem}
If $v' \notin \mathit{seq}_{\mathit{ins}}$, \Cref{invariant:creation_installable} is satisfied for $v_{\mathit{ins}}$ at time $t'$.

Otherwise, let $\mathit{seq}_{\mathit{ins}} = v_1'' \to ... \to v_z'' \to v' \to v_1''' \to ... \to v_q''' \to v_{\mathit{ins}}$, where $z \geq 0$ and $q \geq 0$.
The following holds at time $t'$:
\begin{compactitem}
    \item $\alpha_{t'}(v') = \top$, and
    
    \item $\mathit{seq}_v \in \rho_{t'}(v')$, where $\mathit{seq}_v = v_1''' \to ... \to v_q''' \to v_{\mathit{ins}}$ (recall that $\mathit{seq}_v \in \rho_t(v')$), and
    
    \item for every $v_i \in \{v_1''', ..., v_q'''\}$, $v_i \in V(t')$ and $\alpha_{t'}(v_i) = \bot$, and
    
    \item for every $v_i \in \{v_1''', ..., v_q'''\}$, $\mathit{seq}'_i \in \rho_{t'}(v_i')$, where $\mathit{seq}'_i = v_{i + 1}''' \to ... \to v_q''' \to v_{\mathit{ins}}$.
\end{compactitem}
Therefore, \Cref{invariant:creation_installable} is, even in this case, satisfied for $\mathit{v}_{\mathit{ins}}$ at time $t'$.

\smallskip
\noindent \underline{\Cref{invariant:initial_view}:}
Since $v^* \in V(t)$ and \Cref{invariant:initial_view} holds at time $t$, $\mathit{genesis} \subseteq v^*$.
Moreover, $\mathit{seq}' \in \rho_{t'}(v^*)$.
Since \Cref{invariant:rhos} holds at time $t'$, $v^* \subset v'$.
Hence, $\mathit{genesis} \subset v'$, which proves the invariant preservation at time $t'$.

\smallskip
\noindent \underline{\Cref{invariant:creation_non_installable}:}
Consider a view $v_{\mathit{aux}} \in V(t)$ such that (1) $\alpha_t(v_{\mathit{aux}}) = \bot$, and (2) $v_{\mathit{aux}} \neq v'$.
Since all invariants hold at time $t$, we have (by \Cref{definition:auxiliary_view}):
\begin{compactitem}
    \item $\alpha_t(v_{\mathit{ins}}) = \top$, for some view $v_{\mathit{ins}} \in V(t)$, and
    
    \item $\mathit{seq}_{\mathit{aux}} \in \rho_t(v_{\mathit{ins}})$, where $\mathit{seq}_{\mathit{aux}} = v_1 \to ... \to v_x \to v_{\mathit{aux}} \to v_1' \to ... \to v_y'$, $x \geq 0$ and $y \geq 1$, and
    
    \item for every $v_i \in \{v_1, ..., v_x\}$, $v_i \in V(t)$ and $\alpha_t(v_i) = \bot$, and
    
    \item for every $v_i \in \{v_1, ..., v_x\}$, $\mathit{seq}_i \in \rho_t(v_i)$, where $\mathit{seq}_i = v_{i + 1} \to ... \to v_x \to v_{\mathit{aux}} \to v_1' \to ... \to v_y'$.
\end{compactitem}
If $v' \notin \{v_1, ..., v_x\}$, \Cref{invariant:creation_non_installable} is satisfied for $v_{\mathit{aux}}$ at time $t'$.

Otherwise, let $\mathit{seq}_{\mathit{aux}} = v_1 \to ... \to v_{z - 1} \to v' \to v_{z + 1} \to ... \to v_x \to v_{\mathit{aux}} \to v_1' \to ... \to v_y'$, where $z \geq 0$.
The following holds at time $t'$:
\begin{compactitem}
    \item $\alpha_{t'}(v') = \top$, and
    
    \item $\mathit{seq}_{v'} \in \rho_{t'}(v')$, where $\mathit{seq}_{v'} = v_{z + 1} \to ... \to v_x \to v_{\mathit{aux}} \to v_1' \to ... \to v_y'$ (recall that $\mathit{seq}_{v'} \in \rho_t(v')$), and
    
    \item for every $v_i \in \{v_{z + 1}, ..., v_x\}$, $v_i \in V(t')$ and $\alpha_{t'}(v_i) = \bot$, and
    
    \item for every $v_i \in \{v_{z + 1}, ..., v_x\}$, $\mathit{seq}_i \in \rho_{t'}(v_i)$, where $\mathit{seq}'_i = v_{i + 1} \to ... \to v_x \to v_{\mathit{aux}} \to v_1' \to ... \to v_y'$.
\end{compactitem}
Therefore, \Cref{invariant:creation_non_installable} is, even in this case, satisfied for $\mathit{v}_{\mathit{aux}}$ at time $t'$.

\smallskip
\noindent \underline{\Cref{invariant:accepted_produced_installable}:}
We have that $\beta_{t'}(v_t) = \beta_t(v_t)$, for every view $v_t \in V(t')$ (recall that $V(t') = V(t)$).
Consider a view $v^{**} \in V(t')$ such that $\mathit{set} \in \beta_{t}(v^{**})$ and $\mathit{set} = \{v_1', ..., v_y'\}$, where $y \geq 1$.
We know that $\mathit{set} \in \beta_{t'}(v^{**})$.
Since \Cref{invariant:accepted_produced_installable} holds at time $t$, we have the following:
\begin{compactitem}
    \item $\mathit{seq}^{**} \in \rho_t(v_{\mathit{ins}})$, where $\alpha_t(v_{\mathit{ins}}) = \top$ and $\mathit{seq}^{**} = v_1 \to ... \to v_x \to v^{**} \to v_1' \to ... \to v_y', (x \geq 0)$, and
    \item for every $v_i \in \{v_1, ..., v_x\}$, $v_i \in V(t)$ and $\alpha_t(v_i) = \bot$, and
    \item for every $v_i \in \{v_1, ..., v_x\}$, $\mathit{seq}_i \in \rho_t(v_i)$, where $\mathit{seq}_i = v_{i + 1} \to ... \to v_x \to v^{**} \to v_1' \to ... \to v_y'$.
\end{compactitem}
Now, if $v' \notin \{v_1, ..., v_x\}$, \Cref{invariant:accepted_produced_installable} is satisfied for $v^{**}$ at time $t'$.

Otherwise, let $\mathit{seq}^{**} = v_1 \to ... \to v_{z - 1} \to v' \to v_{z + 1} \to ... \to v_x \to v^{**} \to v_1' \to ... \to v_y'$, where $z \geq 0$.
The following holds at time $t'$:
\begin{compactitem}
    \item $\alpha_{t'}(v') = \top$, and
    
    \item $\mathit{seq}_{v'} \in \rho_{t'}(v')$, where $\mathit{seq}_{v'} = v_{z + 1} \to ... \to v_x \to v^{**} \to v_1' \to ... \to v_y'$ (recall that $\mathit{seq}_{v'} \in \rho_t(v')$), and
    
    \item for every $v_i \in \{v_{z + 1}, ..., v_x\}$, $v_i \in V(t')$ and $\alpha_{t'}(v_i) = \bot$, and
    
    \item for every $v_i \in \{v_{z + 1}, ..., v_x\}$, $\mathit{seq}_i \in \rho_{t'}(v_i)$, where $\mathit{seq}'_i = v_{i + 1} \to ... \to v_x \to v^{**} \to v_1' \to ... \to v_y'$.
\end{compactitem}
Therefore, \Cref{invariant:accepted_produced_installable} is, even in this case, satisfied for $\mathit{v}^{**}$ at time $t'$.
\end{proof}

Next, we analyze the third possible case.

\begin{theorem} \label{theorem:final_3}
Let $[\mathtt{INSTALL}, v, \mathit{seq} = v' \to v_1' \to ... \to v_y', \omega] \notin \mathcal{I}^*(t)$ be the first $\mathtt{INSTALL}$ message obtained by a (correct or faulty) process after time $t$ such that $v \in V(t)$; let that be at time $t' > t$.
Let $v' \notin V(t)$\footnote{Observe that $v'$ is the destination of the message since $\mathit{seq}$ is a sequence (by \Cref{lemma:set_sequence}) and $v'$ is the smallest view of the sequence.} and $y \geq 1$.
Then, all invariants hold at time $t'$.
\end{theorem}
\begin{proof}
We prove the preservation of the invariants one by one.

\smallskip
\noindent \underline{\Cref{invariant:rhos}:}
\Cref{invariant:rhos} holds for $v'$ at time $t'$ since $\rho_{t'}(v') = \emptyset$.
Let us prove this statement.
By contradiction, suppose that a process has obtained an $\mathtt{INSTALL}$ message associated with view $v'$ before time $t'$.
Since $v' \notin V(t)$, we know that $v'$ is not a valid view (by \Cref{lemma:not_well_founded_not_valid}).
However, since $v' \in V(t')$, $v'$ is a valid view according to \Cref{invariant:valid}.
Thus, contradiction.

The invariant holds for all views $v_t \in V(t)$, where $v_t \neq v$, because $\rho_{t'}(v_t) = \rho_t(v_t)$ and the invariant holds at time $t$.
Finally, the invariant holds for $v$ at time $t'$ because of \Cref{lemma:set_sequence} and the fact that the invariant holds at time $t$.

\smallskip
\noindent \underline{\Cref{invariant:creation_installable}:}
The set of installable views at time $t'$ is identical to the set of installable views at time $t$.
Thus, the invariant is preserved since it is satisfied at time $t$.

\smallskip
\noindent \underline{\Cref{invariant:initial_view}:}
The set of installable views at time $t'$ is identical to the set of installable views at time $t$.
Thus, the invariant is preserved since it is satisfied at time $t$.

\smallskip
\noindent \underline{\Cref{invariant:creation_non_installable}:}
By \Cref{lemma:pre_final}, there exists a view $v^* \in V(t')$ such that:
\begin{compactitem}
    \item $\alpha_{t'}(v^*) = \top$, and
    \item $\mathit{seq}' \in \rho_{t'}(v^*)$, where $\mathit{seq}' = v_1 \to ... \to v_x \to v' \to v_1' \to ... \to v_y', (x \geq 0)$, and
    \item for every $v_i \in \{v_1, ..., v_x\}$, $v_i \in V(t')$ and $\alpha_{t'}(v_i) = \bot$, and
    \item for every $v_i \in \{v_1, ..., v_x\}$, $\mathit{seq}_i \in \rho_{t'}(v_i)$, where $\mathit{seq}_i = v_{i + 1} \to ... \to v_x \to v' \to v_1' \to ... \to v_y'$.
\end{compactitem}
Hence, \Cref{invariant:creation_non_installable} is satisfied at time $t'$ for $v'$.
Moreover, since the set of installable views at time $t'$ is identical to the set of installable views at time $t$ and \Cref{invariant:creation_non_installable} holds at time $t$, the invariant is preserved at time $t'$.

\smallskip
\noindent \underline{\Cref{invariant:accepted_produced_installable}:}
We have that $\beta_{t'}(v') = \{v_1' \to ... \to v_y'\}$.
By \Cref{lemma:pre_final}, there exists a view $v^* \in V(t)$ such that:
\begin{compactitem}
    \item $\alpha_{t'}(v^*) = \top$, and
    \item $\mathit{seq}' \in \rho_{t'}(v^*)$, where $\mathit{seq}' = v_1 \to ... \to v_x \to v' \to v_1' \to ... \to v_y', (x \geq 0)$, and
    \item for every $v_i \in \{v_1, ..., v_x\}$, $v_i \in V(t')$ and $\alpha_{t'}(v_i) = \bot$, and
    \item for every $v_i \in \{v_1, ..., v_x\}$, $\mathit{seq}_i \in \rho_{t'}(v_i)$, where $\mathit{seq}_i = v_{i + 1} \to ... \to v_x \to v' \to v_1' \to ... \to v_y'$.
\end{compactitem}
Hence, \Cref{invariant:accepted_produced_installable} is satisfied for $v'$ at time $t'$.

For every view $v_t \in V(t)$, $\beta_{t'}(v_t) = \beta_t(v_t)$.
Therefore, \Cref{invariant:accepted_produced_installable} is satisfied for $v_t$ at time $t'$ since it holds at time $t$ and the set of installable views at time $t'$ is identical to the set of installable views at time $t$.
\end{proof}

Lastly, we analyze the fourth possible case.

\begin{theorem} \label{theorem:final_4}
Let $[\mathtt{INSTALL}, v, \mathit{seq} = v' \to v_1' \to ... \to v_y', \omega] \notin \mathcal{I}^*(t)$ be the first $\mathtt{INSTALL}$ message obtained by a (correct or faulty) process after time $t$ such that $v \in V(t)$; let that be at time $t' > t$.
Let $v' \in V(t)$\footnote{Note that $v'$ is the destination of the message since $\mathit{seq}$ is a sequence (by \Cref{lemma:set_sequence}) and $v'$ is the smallest view of the sequence.}  and $y \geq 1$.
Then, all invariants hold at time $t'$.
\end{theorem}
\begin{proof}
We distinguish two cases:

\medskip
\noindent(1) Let $\alpha_t(v') = \top$.

\smallskip
\noindent \underline{\Cref{invariant:rhos}:}
\Cref{invariant:rhos} holds for all views $v_t \in V(t) \setminus\{v\}$ at time $t'$ since it holds at time $t$.
The invariant holds for $v$ at time $t'$ because of \Cref{lemma:set_sequence} and the fact that the invariant holds at time $t$.

\smallskip
\noindent \underline{\Cref{invariant:creation_installable}:}
Since $\alpha_t(v') = \top$, all views installable at time $t$ are installable at time $t'$.
The fact that \Cref{invariant:creation_installable} holds at time $t$ implies that the invariant holds at time $t'$.

\smallskip
\noindent \underline{\Cref{invariant:initial_view}:}
Since $\alpha_t(v') = \top$, all views installable at time $t$ are installable at time $t'$.
Hence, \Cref{invariant:initial_view} holds at time $t'$ since it holds at time $t$.

\smallskip
\noindent \underline{\Cref{invariant:creation_non_installable}:}
Since $\alpha_t(v') = \top$, all views installable at time $t$ are installable at time $t'$.
Hence, the invariant holds at time $t'$ since it holds at time $t$.

\smallskip
\noindent \underline{\Cref{invariant:accepted_produced_installable}:}
Since $\alpha_t(v') = \top$, all views installable at time $t$ are installable at time $t'$.
Moreover, $\beta_{t'}(v_t) = \beta_t(v_t)$, for every view $v_t \in V(t) \setminus{\{v'\}}$ (note that $V(t') = V(t)$).
Hence, the invariant holds at time $t'$ for $v_t$ since it holds at time $t$.
Furthermore, \Cref{invariant:accepted_produced_installable} holds for $v'$ at time $t'$ because of \Cref{lemma:pre_final} and the fact that the invariant holds at time $t$.
Thus, the invariant is preserved at time $t'$.

\medskip
\noindent (2) Let $\alpha_t(v') = \bot$.

\smallskip
\noindent \underline{\Cref{invariant:rhos}:}
\Cref{invariant:rhos} holds for all views $v_t \in V(t) \setminus\{v\}$ at time $t'$ since it holds at time $t$.
The invariant holds for $v$ at time $t'$ because of \Cref{lemma:set_sequence} and the fact that the invariant holds at time $t$.

\smallskip
\noindent \underline{\Cref{invariant:creation_installable}:}
The set of views installable at time $t$ is identical to the set of views installable at time $t'$.
Given that the invariant holds at time $t$, the invariant is preserved at time $t'$.

\smallskip
\noindent \underline{\Cref{invariant:initial_view}:}
The set of views installable at time $t$ is identical to the set of views installable at time $t'$.
Given that the invariant holds at time $t$, the invariant is preserved at time $t'$.

\smallskip
\noindent \underline{\Cref{invariant:creation_non_installable}:}
The set of views installable at time $t$ is identical to the set of views installable at time $t'$.
Given that the invariant holds at time $t$, the invariant is preserved at time $t'$.

\smallskip
\noindent \underline{\Cref{invariant:accepted_produced_installable}:}
We have that $\beta_{t'}(v_t) = \beta_t(v_t)$, for every view $v_t \in V(t') \setminus{\{v'\}}$ (note that $V(t') = V(t)$).
Therefore, the invariant is ensured at time $t'$ for $v_t$ since it holds at time $t$.
Moreover, \Cref{invariant:accepted_produced_installable} is ensured for $v'$ at time $t'$ because of \Cref{lemma:pre_final} and the fact that the invariant holds at time $t$.
Thus, the invariant is preserved at time $t'$.
\end{proof}

\para{Epilogue}
We have shown that, at all times, all the invariants we have introduced hold.
Now, we give some results that represent consequences of the introduced invariants.

We start by proving that all valid views belong to $V(\infty)$.

\begin{lemma} [Valid Views are in $V(\infty)$] \label{lemma:valid_in_v}
Let a view $v$ be valid.
Then, $v \in V(\infty)$.
\end{lemma}
\begin{proof}
Follows from \Cref{definition:valid_view} and the definition of $V(\infty)$.
\end{proof}

Since all invariants hold at all times, we have the following for each valid view $v \in V(\infty)$:
\begin{compactitem}
    \item If $\mathit{set} \in \rho_{\infty}(v)$, then (1) $\mathit{set} \neq \emptyset$, (2) $\mathit{set}$ is a sequence, and (3) $\mathit{set}.\mathtt{follows}(v) = \top$ (due to the fact that \Cref{invariant:rhos} holds at time $\infty$).
    
    \item If $\alpha_{\infty}(v) = \bot$ and $\mathit{seq} \in \rho_{\infty}(v)$, where $\mathit{set} = v_1 \to ... \to v_x$ ($x \geq 1$), then $\mathit{seq}' \in \rho_{\infty}(v')$, where $\mathit{seq}' = ... \to v \to v_1 \to ... \to v_x$ and $v$ is an auxiliary view for $v'$ at time $\infty$ (due to \cref{lemma:invariant_4_auxiliary,lemma:produced_then_accepted}).
    
    \item $\mathit{preconditions}(v, t) = \top$, for any time $t$ (this follows from the fact that $\mathit{preconditions}(v, \infty) = \top$, which is a result of \cref{lemma:subset_of_beta,lemma:beta_sequence_follows,lemma:comparable_beta}).
\end{compactitem}

Next, any obtained $\mathtt{INSTALL}$ message associated with a valid view ``carries'' a sequence that follows the associated view.

\begin{lemma} \label{lemma:install_sequence_follows}
Let a (correct or faulty) process obtain the $m = [\mathtt{INSTALL}, \mathit{source}, \mathit{set}, \omega]$ message, where $\mathit{source} = m.\mathtt{source()}$ is a valid view.
Then, (1) $\mathit{set} \neq \emptyset$, (2) $\mathit{set}$ is a sequence, and (3) $\mathit{set}.\mathtt{follows}(\mathit{source}) = \top$.
\end{lemma}
\begin{proof}
Since $\mathit{source}$ is a valid view, $\mathit{source} \in V(\infty)$ (by \Cref{lemma:valid_in_v}).
Moreover, $m \in \mathcal{I}^*(\infty)$.
Hence, $\mathit{set} \in \rho_{\infty}(\mathit{source})$.
Therefore, the lemma holds since \Cref{invariant:rhos} holds at time $\infty$.
\end{proof}

Let $m = [\mathtt{INSTALL}, v, \mathit{seq}, \omega]$, where $v$ is a valid view.
According to \Cref{lemma:install_sequence_follows}, $\mathit{seq} \neq \emptyset$, $\mathit{seq}$ is a sequence and $\mathit{seq}.\mathtt{follows}(v) = \top$.
Moreover, $m.\mathtt{destination()} = \mathit{seq}.\mathtt{first()}$ and $m.\mathtt{tail()} = \mathit{seq} \setminus\{\mathit{seq}.\mathtt{first()}\}$.

The next lemma proves that a view $v$ leads to a view $v'$ at time $\infty$ if and only if $v'$ is the smallest installable view greater than $v$.

\begin{lemma} \label{lemma:installable_leads_to_first_biggest}
Let $v, v' \in V(\infty)$ such that $\alpha_{\infty}(v) = \alpha_{\infty}(v') = \top$ and $v \subset v'$.
Then, $v$ leads to $v'$ at time $\infty$ if and only if there does not exist a view $v'' \in V(\infty)$ such that $\alpha_{\infty}(v'') = \top$ and $v \subset v'' \subset v'$.
\end{lemma}
\begin{proof}
In order to prove the lemma, we prove the both directions of the statement.

First, if $v$ leads to $v'$, then there does not exist a view $v''$ such that $\alpha_{\infty}(v'') = \top$ and $v \subset v'' \subset v'$.
This statement is true due to \Cref{lemma:installable_between_leading}.

Now, suppose that such view $v''$ does not exist.
We need to prove that $v$ leads to $v'$ at time $\infty$.
By contradiction, suppose that $v$ does not lead to $v'$ at time $\infty$.
Since $v \subset v'$, $v$ leads to a view $v_i \neq v'$ at time $\infty$ (by \cref{lemma:n-1,lemma:greatest_does_not_lead}).
By \Cref{lemma:leads_to_bigger}, $v \subset v_i$.
By \cref{lemma:installable_sequence,lemma:installable_between_leading}, $v_i \subseteq v'$.
Since $v_i \neq v'$, $v_i \subset v'$.
However, this contradicts the fact that an installable view greater than $v$ and smaller than $v'$ does not exist.
Thus, the lemma holds.
\end{proof}

Similarly to \Cref{lemma:installable_leads_to_first_biggest}, a view $v'$ is an auxiliary view for a view $v$ at time $\infty$ if and only if $v$ is the greatest installable view smaller than $v'$. 

\begin{lemma} \label{lemma:auxiliary_first_smallest}
Let $v, v' \in V(\infty)$ such that $\alpha_{\infty}(v) = \top$, $\alpha_{\infty}(v') = \bot$ and $v \subset v'$.
Then, $v'$ is an auxiliary view for $v$ at time $\infty$ if and only if there does not exist a view $v'' \in V(\infty)$ such that $\alpha_{\infty}(v'') = \top$ and $v \subset v'' \subset v'$.
\end{lemma}
\begin{proof}
In order to prove the lemma, we prove the both directions of the statement.

First, if $v'$ is an auxiliary view for $v$ at time $\infty$, then there does not exist a view $v''$ such that $\alpha_{\infty}(v'') = \top$ and $v \subset v'' \subset v'$.
By contradiction, let there exist such view $v''$.
By \cref{lemma:n-1,lemma:greatest_does_not_lead}, $v$ leads to a view $v_i$ at time $\infty$.
Moreover, $v_i \subseteq v''$ (by \cref{lemma:installable_sequence,lemma:installable_between_leading}).
Since $v'$ is an auxiliary view for $v$ at time $\infty$, $v' \subset v_i$ (by \Cref{lemma:auxiliary_smaller}).
Hence, $v' \subset v''$, which contradicts that $v'' \subset v'$.
This direction of the statement is satisfied.

Now, suppose that such view $v''$ does not exist.
We need to prove that $v'$ is an auxiliary view for $v$ at time $\infty$.
By contradiction, suppose that $v'$ is not an auxiliary view for $v$ at time $\infty$.
Let $v'$ be an auxiliary view for a view $v^* \neq v$ (such view $v^*$ exists due to \Cref{lemma:auxiliary_for_one}).
By \Cref{lemma:auxiliary_bigger}, $v^* \subset v'$.
By \Cref{lemma:installable_sequence}, either $v \subseteq v^*$ or $v \supset v^*$.
Since $v \neq v^*$ and no installable view ``between'' $v$ and $v'$ exists, $v^* \subset v$.
Let $v^*$ lead to a view $v_i$ (by \cref{lemma:n-1,lemma:greatest_does_not_lead}).
By \cref{lemma:installable_sequence,lemma:installable_between_leading}, $v_i \subseteq v$.
By \Cref{lemma:auxiliary_smaller}, $v' \subset v_i$.
Hence, $v' \subset v$, which represents contradiction with $v \subset v'$.
Thus, the lemma holds.
\end{proof}

The next lemma shows that all valid views are greater than or equal to $\mathit{genesis}$.

\begin{lemma} \label{lemma:genesis_smallest_valid}
Let a view $v$ be valid.
Then, $\mathit{genesis} \subseteq v$.
\end{lemma}
\begin{proof}
By \Cref{lemma:valid_in_v}, we know that $v, \mathit{genesis} \in V(\infty)$.
By \Cref{lemma:all_comparable}, we know that either $\mathit{genesis} \subseteq v$ or $\mathit{genesis} \supset v$.
For the sake of contradiction, let $v \subset \mathit{genesis}$.

Since \Cref{invariant:initial_view} holds at time $\infty$, $\alpha_{\infty}(v) = \bot$.
Therefore, $v$ is an auxiliary view for some view $v^*$ at time $\infty$ (by \Cref{lemma:auxiliary_for_one}).
By \Cref{lemma:auxiliary_bigger}, $v^* \subset v$, which implies that $v^* \subset \mathit{genesis}$.
However, this is impossible due to \Cref{invariant:initial_view}.
The lemma holds.
\end{proof}

The next lemma shows that views that ``belong'' to a view-path are ordered.

\begin{lemma} \label{lemma:path_sequence}
Let $\mathit{path} = [m_1, ..., m_k]$, where $k \geq 1$, be a view-path.
Then:
\begin{compactenum}
    \item $m_{i - 1}.\mathtt{source()} \subset m_i.\mathtt{source()}$, for every $i \in [2, k]$, and
    
    \item $m_k.\mathtt{source()} \subset m_k.\mathtt{destination()}$.
\end{compactenum}
\end{lemma}
\begin{proof}
Note that $\mathit{path}.\mathtt{views()} \subseteq V(\infty)$ (by \Cref{lemma:valid_in_v}).
Since $m_i.\mathtt{source()} = m_{i - 1}.\mathtt{destination()}$, for every $i \in [2, k]$ (see \Cref{lst:view_path}), the first claim of the lemma follows from \Cref{invariant:rhos}.
Moreover, \Cref{invariant:rhos} ensures that the second claim of the lemma is true, as well.
\end{proof}

Next, we show that it is impossible that $(+, r) \notin v$ and $(-, r) \in v$, where $v$ is a valid view and $r$ is a server.

\begin{lemma} \label{lemma:minus_implies_plus_view}
Let $v$ be a valid view such that $(-, r) \in v$, where $r$ is a server.
Then, $(+, r) \in v$.
\end{lemma}
\begin{proof}
By contradiction, let $(-, r) \in v$ and $(+, r) \notin v$.
First, $v \neq \mathit{genesis}$ due to the fact that $\mathit{genesis}$ does not contain ``negative'' changes (see \Cref{appx:model}, paragraph ``Constants'').
Thus, there exists a view-path $\mathit{path} \subseteq \mathcal{I}^*(\infty)$ such that $\mathit{path}.\mathtt{destination()} = v$ (by \Cref{definition:valid_view}).

By \Cref{lemma:path_sequence}, $v' \subset v$, where $v' = m.\mathtt{source()}$, $v = m.\mathtt{destination()}$ and $m$ is the last message of $\mathit{path}$.
Thus, $(+, r) \notin v'$, which means that the membership validity property of $\mathit{vg}(v')$ is violated (recall that $v'$ is a valid view, by \Cref{lemma:knows_view_system}, and $\mathit{preconditions}(v', \infty) = \top$).
Hence, the lemma holds.
\end{proof}

We now show that there are finitely many valid views (see \Cref{lst:reconfiguration_properties_new}).

\begin{theorem} [Finitely Many Valid Views] \label{lemma:finite_valid_views}
There exist only finitely many valid views.
\end{theorem}
\begin{proof}
Consider any valid view $v$.
We prove that $|C_v^+ = \{(+, r) \,|\, (+, r) \in v\}| < \infty$.
We distinguish two possible cases:
\begin{compactitem}
    \item Let $v = \mathit{genesis}$.
    In this case, $|C_v^+| < \infty$ due to the definition of $\mathit{genesis}$ (see \Cref{appx:model}, paragraph ``Constants'').
    
    \item Let $v \neq \mathit{genesis}$.
    By \Cref{definition:valid_view}, there exists a view-path $\mathit{path} \subseteq \mathcal{I}^*(\infty)$ such that $\mathit{path}.\mathtt{destination()} = v$.
    Let $m \in \mathit{path}$ be the $\mathtt{INSTALL}$ message such that $m.\mathtt{destination()} = v$; note that $m.\mathtt{source()}$ is a valid view.
    Now, for every $(+, r) \in C_v^+$, the $\mathtt{JOIN}$ message is sent by $r$ (see \Cref{lst:view_set_evidence_valid}); hence, $r$ has previously started.
    Due to the assumption that only finitely many servers start, $|C_v^+| < \infty$.
\end{compactitem}

Now, let $C_v^- = \{(-, r) \,|\, (-, r) \in v\}$.
Because of \Cref{lemma:minus_implies_plus_view} and $|C_v^+| < \infty$, we conclude that $|C_v^-| < \infty$.
Therefore, $|v| < \infty$.
Recall that every valid view belongs to $V(\infty)$ (by \Cref{lemma:valid_in_v}).
Since all valid views have finitely many changes, all valid views belong to $V(\infty)$ and all views from $V(\infty)$ are comparable (by \Cref{lemma:all_comparable}), the theorem holds.
\end{proof}

Moreover, we prove that there are finitely many sequences in $\rho_{\infty}(v)$, where $v$ is a valid view.

\begin{lemma} [Finitely Many Sequences] \label{lemma:finite_sequences}
Let view $v$ be valid.
Then, $|\rho_{\infty}(v)| < \infty$.
\end{lemma}
\begin{proof}
Follows directly from the bounded decisions property of $\mathit{vg}(v)$.
\end{proof}

For the sake of brevity, we might use the result of \Cref{lemma:finite_sequences} without explicitly stating it in the following.

\subsubsection{Safety Properties} \label{subsection:safety_properties}

We start by proving the view comparability property specified in \Cref{lst:reconfiguration_properties_new}.

\begin{theorem} [View Comparability] \label{theorem:view_comparability}
View comparability is satisfied.
\end{theorem}
\begin{proof}
All valid views belong to $V(\infty)$ (by \Cref{lemma:valid_in_v}).
Hence, the theorem follows from \Cref{lemma:all_comparable}.
\end{proof}

The direct consequence of the view comparability property is that the $\mathit{history}$ variable (initialized at line~\ref{line:init_history} of \Cref{lst:storage_initialization}) of a correct server is a sequence.
The indirect consequences of the property are:
\begin{compactenum}
    \item The $\mathit{sequence}$ variable (initialized at line~\ref{line:sequence_init} of \Cref{lst:storage_initialization}) of a correct server contains sequences.
    
    \item The $\mathit{reconfiguration}.\mathit{sequence}$ variable (initialized at line~\ref{line:reconfiguration_sequence_init} of \Cref{lst:reconfiguration_initialization}) of a correct server is a sequence (or $\bot$).
\end{compactenum}

We are ready to prove the join safety property (see \Cref{subsection:carbon_properties}, paragraph ``Properties of \sysname'').

\begin{theorem} [Join Safety] \label{theorem:join_safety}
Join safety is satisfied.
\end{theorem}
\begin{proof}
Let a server $r$ join and let $r \notin \mathit{genesis}.\mathtt{members()}$.
We prove that $r$ is voted in, i.e., that a process obtains a voting proof $\sigma_v$ such that $\mathtt{verify\_voting}(\text{``add server } r\text{''}, \sigma_v) = \top$.

Since $r$ joins and $r \notin \mathit{genesis}.\mathtt{members()}$, there exists a correct server that executes line~\ref{line:joined} of \Cref{lst:reconfiguration_view_transition} (if $r$ is correct, then $r$ itself executes this line; see \Cref{subsection:carbon_properties}, paragraph ``Validators'') or line~\ref{line:trigger_for_others_2} of \Cref{lst:reconfiguration_view_transition} (if $r$ is faulty, a correct server executes this line).
Hence, there exists a valid view $v \neq \mathit{genesis}$ such that $r \in v.\mathtt{members()}$.

Since $v \neq \mathit{genesis}$ is valid, there exists a view-path $\mathit{path} \subseteq \mathcal{I}^*(\infty)$ such that $\mathit{path}.\mathtt{destination()} = v$ (by \Cref{definition:valid_view}).
All views that belong to $\mathit{path}.\mathtt{views()}$ are valid (by \Cref{lemma:knows_view_system}).
Hence, there exists an $\mathtt{INSTALL}$ message $m$ such that $m.\mathtt{source()}$ is a valid view and $r \in m.\mathtt{destination()}.\mathtt{members()}$.
Finally, the theorem holds because the evidence handed to the view generator primitive contains $\mathtt{JOIN}$ messages with the voting proofs for every $(+, r')$ ``proposed'' change (by line~\ref{line:check_join_1} or line~\ref{line:check_join_2} of \Cref{lst:view_set_evidence_valid}), where $r'$ is a server.
\end{proof}

The next theorem proves the leave safety property of \sysname (see \Cref{subsection:carbon_properties}, paragraph ``Properties of \sysname'').

\begin{theorem} [Leave Safety] \label{theorem:leave_safety}
Leave safety is satisfied.
\end{theorem}
\begin{proof}
Let a correct server $r$ leave.
Hence, $r$ executes line~\ref{line:left} of \Cref{lst:reconfiguration_view_transition}.
Therefore, there exist valid views $v_1$ and $v_2$ such that (1) $r \in v_1.\mathtt{members()}$, and (2) $r \notin v_2.\mathtt{members()}$ (by line~\ref{line:updated_discovery_exists} of \Cref{lst:reconfiguration_view_transition}).

Since $r \in v_1.\mathtt{members()}$, $(+, r) \in v_1$ (see \Cref{lst:view}).
Moreover, since $v_1 \subset v_2$ and $r \notin v_2.\mathtt{members()}$, $(+, r), (-, r) \in v_2$.
Hence, $v_2 \neq \mathit{genesis}$ (since $\mathit{genesis}$ does not contain ``negative'' changes; see \Cref{appx:model}, paragraph ``Constants''). 
By \Cref{definition:valid_view}, there exists a view-path $\mathit{path} \subseteq \mathcal{I}^*(\infty)$ such that $\mathit{path}.\mathtt{destination()} = v_2$.
Similarly to the proof of join safety (\Cref{theorem:join_safety}), there exists an $\mathtt{INSTALL}$ message $m$ such that $m.\mathtt{source()}$ is a valid view and $(-, r) \in m.\mathtt{destination()}$.
Furthermore, the evidence handed to the view generator primitive contains either voting proofs or $\mathtt{LEAVE}$ messages for every $(-, r')$ ``proposed'' change (by line~\ref{line:check_leave_1} or line~\ref{line:check_leave_2} of \Cref{lst:view_set_evidence_valid}), where $r'$ is a server.
Therefore, $r$ either sent the $\mathtt{LEAVE}$ message or is voted out.
Finally, if $r$ sent the $\mathtt{LEAVE}$ message, then $r$ has previously requested to leave (because the $\mathtt{LEAVE}$ message is sent at line~\ref{line:gossip_leave} of \Cref{lst:reconfiguration_leaving}).
The theorem holds.
\end{proof}

\subsubsection{Liveness Properties} \label{subsection:reconfiguration_liveness}

We now prove the finality property specified in \Cref{lst:reconfiguration_properties_new}.
Then, we show the join liveness, leave liveness and removal liveness properties of \sysname.

We start by showing that a correct server always hand ``valid'' proposals to a view generator instance.

\begin{lemma} \label{lemma:propose_only_valid}
Let a correct server $r$ propose $(v', \mathit{set}', \epsilon')$ to $\mathit{vg}(v)$.
Then, $\mathtt{valid}((v', \mathit{set}', \epsilon'), v) = \top$.
\end{lemma}
\begin{proof}
There are two places for $r$ to propose to $\mathit{vg}(v)$.
We take a look at both places:
\begin{compactitem}
    \item line~\ref{line:propose_installed} of \Cref{lst:reconfiguration_view_generation}:
    We separate two possibilities:
    \begin{compactitem}
        \item Let $v = \mathit{genesis}$.
        First, $v \subset v'$ (ensured by the $\mathtt{to\_propose}$ function; line~\ref{line:to_propose_function} of \Cref{lst:reconfiguration_view_generation}).
        Moreover, $\epsilon'.\mathit{path} = \bot$ (by line~\ref{line:current_view_proof_init} of \Cref{lst:reconfiguration_initialization}).
        The check at line~\ref{line:check_minus_r_1} of \Cref{lst:view_set_evidence_valid} passes due to the $\mathtt{to\_propose}$ function.
        Finally, for each change that belongs to the $\mathit{reconfiguration}.\mathit{requested}$ variable, server $r$ stores the $\mathtt{JOIN}$ (at line~\ref{line:store_join_message} of \Cref{lst:reconfiguration_joining}) or the $\mathtt{LEAVE}$ message (at line~\ref{line:store_leave_message} of \Cref{lst:reconfiguration_leaving}) or the corresponding voting proof (at line~\ref{line:store_passing_proof} of \Cref{lst:reconfiguration_leaving}) in the $\mathit{reconfiguration}.\mathit{requests}$ or $\mathit{reconfiguration}.\mathit{voting\_proofs}$ variable.
        Thus, the lemma holds in this scenario.
        
        \item Let $v \neq \mathit{genesis}$.
        First, $v \subset v'$ (ensured by the $\mathtt{to\_propose}$ function; line~\ref{line:to_propose_function} of \Cref{lst:reconfiguration_view_generation}).
        Then, the check at line~\ref{line:check_minus_r_2} of \Cref{lst:view_set_evidence_valid} passes due to the $\mathtt{to\_propose}$ function.
        The check at line~\ref{line:check_destination_tail} of \Cref{lst:view_set_evidence_valid} passes since $\epsilon'.\mathit{path} = \mathit{current\_view}.\mathit{proof}$.
        Finally, for each change that belongs to the $\mathit{reconfiguration}.\mathit{requested}$ variable, server $r$ stores the $\mathtt{JOIN}$ (at line~\ref{line:store_join_message} of \Cref{lst:reconfiguration_joining}) or the $\mathtt{LEAVE}$ message (at line~\ref{line:store_leave_message} of \Cref{lst:reconfiguration_leaving}) or the corresponding voting proof (at line~\ref{line:store_passing_proof} of \Cref{lst:reconfiguration_leaving}) in the $\mathit{reconfiguration}.\mathit{requests}$ or $\mathit{reconfiguration}.\mathit{voting\_proofs}$ variable.
        Therefore, the lemma holds in this scenario.
    \end{compactitem}

    \item line~\ref{line:start_with_proposal} of \Cref{lst:reconfiguration_view_transition}:
    We know that $\epsilon'.\mathit{path}.\mathtt{tail()} = \mathit{set}'$.
    Furthermore, $\mathit{set}' \in \beta_{\infty}(v)$.
    By \Cref{lemma:beta_genesis_empty}, $v \neq \mathit{genesis}$.
    
    The check at line~\ref{line:view_in_seq} of \Cref{lst:view_set_evidence_valid} passes since $v' \in \mathit{set}'$.
    Next, the check at line~\ref{line:set_follows} of \Cref{lst:view_set_evidence_valid} passes because of \Cref{lemma:beta_sequence_follows}.
    
    Now, if $(+, r) \notin v$, then no view $v^* \in \mathit{set}'$ exists such that $(-, r) \in v^*$.
    Let us prove this claim.
    Since $\mathit{set}' \in \beta_{\infty}(v)$, $\mathit{seq}'' \in \rho_{\infty}(v'')$, where $\alpha_{\infty}(v'') = \top$, $\mathit{seq}'' = ... \to v \to v_1 \to ... \to v_x, (x \geq 1)$ and $\mathit{set}' = v_1 \to ... \to v_x$ (by \Cref{invariant:accepted_produced_installable}).
    Since $\mathit{seq}''.\mathtt{follows}(v'') = \top$ (by \Cref{invariant:rhos}), $v'' \subset v$.
    Hence, $(+, r) \notin v''$.
    Therefore, no view that belongs to $\mathit{set}'$ contains $(-, r)$ due to membership validity of $\mathit{vg}(v'')$ (since $v''$ is a valid view and $\mathit{preconditions}(v'', \infty) = \top$).
    Thus, the checks at lines~\ref{line:check_valid_view_set_proof} and~\ref{line:check_minus_r_2} of \Cref{lst:view_set_evidence_valid} pass.
    
    The check at line~\ref{line:check_destination_tail} of \Cref{lst:view_set_evidence_valid} passes since $\epsilon'.\mathit{path} = \mathit{current\_view}.\mathit{proof}$.
    Finally, for each change that belongs to a view of $\mathit{set}'$, server $r$ stores the $\mathtt{JOIN}$ or the $\mathtt{LEAVE}$ message or the corresponding voting proof in the $\mathit{reconfiguration}.\mathit{requests}$ or $\mathit{reconfiguration}.\mathit{voting\_proofs}$ variable because of the $\mathtt{extract\_requests\_and\_voting\_proofs}$ function invoked after an $\mathtt{INSTALL}$ message is processed (line~\ref{line:extract_invocation} of \Cref{lst:storage_processing}).
\end{compactitem}
The lemma is satisfied since it holds in both possible scenarios.
\end{proof}

The next lemma shows that installable views cannot be ``skipped''.
Specifically, we show that any view-path to a view $v'$ ``contains'' a view $v$, where $v$ is the greatest installable view smaller than $v'$; observe that such installable view always exists since any view-path contains at least a single message and the source of the first message is $\mathit{genesis}$, which is the smallest valid view (by \Cref{lemma:genesis_smallest_valid}).

\begin{lemma} \label{lemma:installable_in_path}
Let $\mathit{path} \subseteq \mathcal{I}^*(\infty)$ such that $\mathit{path}.\mathtt{destination()} = v'$.\footnote{Note that $v' \supset \mathit{genesis}$ since view-paths have at least a single message (see \Cref{lst:view_path}) and \Cref{invariant:rhos} holds.}
Let $v$ be the greatest view smaller than $v'$ with $\alpha_{\infty}(v) = \top$.
Then, $v \in \mathit{path}.\mathtt{views()}$.
\end{lemma}
\begin{proof}
Let $m^*$ be the last message of $\mathit{path}$; note that $m^*.\mathtt{destination()} = v'$.
Let $v^* = m^*.\mathtt{source()}$.
If $v^* = v$, the lemma holds.
Otherwise, we prove that $v \subset v^*$.

By contradiction, suppose that $v^* \subset v$.
We distinguish two possibilities:
\begin{compactitem}
    \item Let $\alpha_{\infty}(v^*) = \top$.
    Since $v^* \subset v$, $v^*$ leads to a view $v_i$ at time $\infty$ (by \cref{lemma:n-1,lemma:greatest_does_not_lead}).
    By \cref{lemma:installable_sequence,lemma:installable_between_leading}, $v_i \subseteq v$.
    
    Since $m^* \in \mathit{path}$, $\mathit{seq} \in \rho_{\infty}(v^*)$, where $\mathit{seq} = v' \to ...$.
    By \Cref{definition:leads}, $\mathit{seq}_i \in \rho_{\infty}(v^*)$, where $\mathit{seq}_i = ... \to v_i$.
    By \Cref{lemma:weak_accuracy}, either $\mathit{seq} \subseteq \mathit{seq}_i$ or $\mathit{seq} \supset \mathit{seq}_i$.
    Therefore, $v' \subseteq v_i$, which implies that $v' \subseteq v$.
    However, this contradicts the fact that $v' \supset v$.
    In this case, $v \subset v^*$.
    
    \item Let $\alpha_{\infty}(v^*) = \bot$.
    By \Cref{lemma:auxiliary_for_one}, $v^*$ is an auxiliary view for a view $v_i$.
    By \Cref{lemma:auxiliary_bigger}, $v_i \subset v^*$.
    Since $v^* \subset v$, $v_i \subset v$.
    Therefore, $v_i$ leads to a view $v_i'$ at time $\infty$ (by \cref{lemma:n-1,lemma:greatest_does_not_lead}).
    \Cref{lemma:installable_sequence,lemma:installable_between_leading} show that $v_i' \subseteq v$.
    
    Since $m^* \in \mathit{path}$, $\mathit{seq} \in \rho_{\infty}(v^*)$, where $\mathit{seq} = v' \to ...$.
    Moreover, $\mathit{seq}' \in \rho_{\infty}(v_i)$, where $\mathit{seq}' = ... \to v^* \to v' \to ...$ (by \cref{lemma:produced_then_accepted,lemma:invariant_4_auxiliary} and the fact that $v^*$ is an auxiliary view for $v_i$).
    On the other hand, $\mathit{seq}_i' \in \rho_{\infty}(v_i)$, where $\mathit{seq}_i' = ... \to v_i'$ (since $v_i$ leads to $v_i'$ at time $\infty$ and \Cref{definition:valid_view}).
    By \Cref{lemma:weak_accuracy}, either $\mathit{seq}' \subseteq \mathit{seq}_i'$ or $\mathit{seq}' \supset \mathit{seq}_i'$.
    Therefore, $v' \subseteq v_i'$, which implies that $v' \subseteq v$.
    This statement contradicts the fact that $v' \supset v$.
    Thus, $v \subset v^*$ in this case, as well.
\end{compactitem}

Since $v \subset v^*$, we reach the starting point of the lemma.
The recursion eventually stops since $\mathit{path}$ ``contains'' finitely many views (by \Cref{lemma:finite_valid_views} and the fact that all views from $\mathit{path}$ are valid), the smallest view that belongs to $\mathit{path}.\mathtt{views()}$ is $\mathit{genesis}$ (see \Cref{lst:view_path} and \Cref{lemma:path_sequence}) and $v \supseteq \mathit{genesis}$ (by \Cref{lemma:genesis_smallest_valid}).
\end{proof}

A direct consequence of \Cref{lemma:installable_in_path} is that all installable views smaller than the destination of a view-path are ``contained'' in the view-path.

\begin{lemma} \label{lemma:all_installable_in_path}
Let $\mathit{path} \subseteq \mathcal{I}^*(\infty)$ such that $\mathit{path}.\mathtt{destination()} = v'$.
Consider a view $v \subset v'$ such that $\alpha_{\infty}(v) = \top$.
Then, $v \in \mathit{path}.\mathtt{views()}$.
\end{lemma}
\begin{proof}
Follows directly from \Cref{lemma:installable_in_path}.
\end{proof}

A direct consequence of \Cref{lemma:all_installable_in_path} is that there does not exist an $\mathtt{INSTALL}$ message $m$ such that (1) $m.\mathtt{source()}$ is a valid view, and (2) $m.\mathtt{source()} \subset v \subset m.\mathtt{destination()}$, where $v$ is a view installable at time $\infty$.
We prove this result below.

\begin{lemma} \label{lemma:install_skip_installable}
Let a view $v$ be installable at time $\infty$ (i.e., $\alpha_{\infty}(v) = \top$).
Then, $[\mathtt{INSTALL}, v', \mathit{seq}' = v'' \to ..., \omega] \notin \mathcal{I}^*(\infty)$, where (1) $v'$ is a valid view, and (2) $v' \subset v \subset v''$.
\end{lemma}
\begin{proof}
Let $m = [\mathtt{INSTALL}, v', \mathit{seq}' = v'' \to ..., \omega]$.
By contradiction, let $m \in \mathcal{I}^*(\infty)$.
Since $v'$ is a valid view, $v' = \mathit{genesis}$ or there exists a view-path $\mathit{path} \subseteq \mathcal{I}^*(\infty)$ such that $\mathit{path}.\mathtt{destination()} = v'$ (by \Cref{definition:valid_view}).
Let us consider both cases:
\begin{compactitem}
    \item Let $v' = \mathit{genesis}$.
    Let $\mathit{path}' = [m]$.
    Hence, $\mathit{path}'$ is a view-path (see \Cref{lst:view_path}) and $\mathit{path}'.\mathtt{destination()} = v''$.
    Moreover, $\mathit{path}' \subseteq \mathcal{I}^*(\infty)$.
    
    \item Let $v' \neq \mathit{genesis}$.
    Let $\mathit{path}' = \mathit{path} || m$.
    Hence, $\mathit{path}'$ is a view-path (see \Cref{lst:view_path}) and $\mathit{path}'.\mathtt{destination()} = v''$.
    Moreover, $\mathit{path}' \subseteq \mathcal{I}^*(\infty)$.
\end{compactitem}
Hence, $\mathit{path}' \subseteq \mathcal{I}^*(\infty)$ exists such that $\mathit{path}'.\mathtt{destination()} = v'' \supset v$ and $v \notin \mathit{path}'.\mathtt{views()}$.
This contradicts \Cref{lemma:all_installable_in_path}, which concludes the proof.
\end{proof}

Next, we prove that $\mathit{install\_messages}$ variables of two forever-correct servers are eventually identical.

\begin{lemma} \label{lemma:identical_install}
Let servers $r$ and $r'$ be forever-correct.
Then, $\mathit{install\_messages}_r = \mathit{install\_messages}_{r'}$, where $\mathit{install\_messages}_r$ (resp., $\mathit{install\_messages}_{r'}$) denotes the value of the $\mathit{install\_messages}$ variable at server $r$ (resp., $r'$) at time $\infty$.
\end{lemma}
\begin{proof}
First, $|\mathit{install\_messages}_r| < \infty$ and $|\mathit{install\_messages}_{r'}| < \infty$, which follows from the fact that only finitely many valid views exist (by \Cref{lemma:finite_valid_views}) and from the fact that $|\rho_{\infty}(v)| < \infty$, for every valid view $v$ (by \Cref{lemma:finite_sequences}).
Then, the lemma follows from the fact that all messages inserted into the $\mathit{install\_messages}$ variable at a correct server are gossiped (at line~\ref{line:gossip_install}) of \Cref{lst:storage_processing}).
\end{proof}

Now, we prove that a correct server eventually gets ``prepared'' for a view, i.e., a correct server eventually sets its $\mathit{reconfiguration}.\mathit{prepared}$ variable to $\top$.
The next lemma assumes that all correct servers send the ``valid'' state according to the $\mathtt{verify}$ function (see line~\ref{line:state_update_received} of \Cref{lst:reconfiguration_state_transfer}).
We prove that this is indeed the case in \Cref{lemma:always_verified}.

\begin{lemma} [Eventual Preparedness] \label{lemma:eventually_prepared}
Let a correct server $r$ enter the pseudocode at line~\ref{line:updated_discovery_exists} of \Cref{lst:reconfiguration_view_transition}.
Eventually,\\$\mathit{reconfiguration}.\mathit{prepared} = \top$ at server $r$.
\end{lemma}
\begin{proof}
By contradiction, suppose that $\mathit{reconfiguration}.\mathit{prepared}$ is never set to $\top$ at server $r$.
Hence, $r$ gossips the\\$\mathtt{STATE-REQUEST}$ message at line~\ref{line:gossip_state_request} of \Cref{lst:reconfiguration_view_transition}.
Moreover, since $r$ has entered the pseudocode at line~\ref{line:updated_discovery_exists} of \Cref{lst:reconfiguration_view_transition}, there exists a valid view $v_r$ such that $r \in v_r.\mathtt{members()}$.
Hence, the $\mathtt{JOIN}$ message was sent by $r$ (by checks at lines~\ref{line:check_join_1} and~\ref{line:check_join_2} of \Cref{lst:view_set_evidence_valid}), which implies that $r$ has requested to join.
Since $r$ is correct, $r$ does not halt before leaving (see \Cref{subsection:carbon_properties}, paragraph ``Server's rules'').
Because $\mathit{reconfiguration}.\mathit{prepared}$ is never set to $\top$ by server $r$ (by the assumption), $r$ never leaves.
Therefore, $r$ never halts, which means that $r$ is forever-correct.
Now, suppose that $\mathit{reconfiguration}.\mathit{source} = v$ and $\mathit{reconfiguration}.\mathit{sequence} = \mathit{seq}$ at server $r$.

Eventually, the rule at line~\ref{line:state-request_active} of \Cref{lst:reconfiguration_state_transfer} is active at every forever-correct server (follows from \Cref{lemma:identical_install}).
We now separate two possibilities:
\begin{compactitem}
    \item If $v.\mathtt{quorum()}$ correct members of $v$ are forever-correct, then these servers send the $\mathtt{STATE-UPDATE}$ message to $r$ at line~\ref{line:send_state_update} of \Cref{lst:reconfiguration_state_transfer} and these messages are received by $r$.
    Therefore, the rule at line~\ref{line:prepared_from_source} of \Cref{lst:reconfiguration_state_transfer} is eventually active at server $r$ and $r$ sets its $\mathit{reconfiguration}.\mathit{prepared}$ variable to $\top$ in this case.
    
    \item Otherwise, a correct member of $v$ halts; let $r' \in v.\mathtt{members()}$ be a correct server that halts.
    Since $r' \in v.\mathtt{members()}$, $r'$ has requested to join (by checks at lines~\ref{line:check_join_1} and~\ref{line:check_join_2} of \Cref{lst:view_set_evidence_valid}).
    Therefore, $r'$ leaves (see \Cref{subsection:carbon_properties}, paragraph ``Server's rules''), which implies that $r'$ executes line~\ref{line:left} of \Cref{lst:reconfiguration_view_transition}.
    Prior to leaving, $\mathit{reconfiguration}.\mathit{discharged} = \top$ (line~\ref{line:wait_for} of \Cref{lst:reconfiguration_view_transition}) and $\mathit{dischargement}.\mathit{dischargement\_view} = v'$, for some valid view $v' \supset v$, at server $r'$.
    This means that $r'$ has, before leaving, received $\mathtt{DISCHARGEMENT-CONFIRM}$ messages associated with a view $v^*$, where $v^* \supseteq v'$ and $v^*$ is a valid view, from a quorum of members of $v^*$ (line~\ref{line:discharged_rule} of \Cref{lst:reconfiguration_view_dischargement}).
    Let $C$ denote the subset of servers from which $r'$ has received the $[\mathtt{DISCHARGEMENT-CONFIRM}, v^*]$ messages such that, for every server $c \in C$, $c$ is correct; note that $|C| \geq v^*.\mathtt{plurality()}$.
    We separate two cases:
    \begin{compactitem}
        \item At least $v^*.\mathtt{plurality}()$ of servers that belong to $C$ are forever-correct.
        In this case, these servers eventually send the $\mathtt{STATE-UPDATE}$ message while their $\mathit{current\_view}.\mathit{view}$ variable is greater than or equal to $v^*$ (this holds due to the fact that these servers eventually set their $\mathit{current\_view}.\mathit{view}$ variable to $v^*$ and the $\mathit{current\_view}.\mathit{view}$ variable of a correct server is only updated with greater values).
        Moreover, $v^* \in \mathit{history}$ at server $r$ eventually (due to \Cref{lemma:identical_install}).
        Hence, the rule at line~\ref{line:prepared_from_bigger} of \Cref{lst:reconfiguration_state_transfer} is eventually active at server $r$, which ensures that $r$ sets its $\mathit{reconfiguration}.\mathit{prepared}$ variable to $\top$.
    
        \item Otherwise, there exists a server $c \in C$ that leaves (i.e., executes line~\ref{line:discharged_rule} of \Cref{lst:reconfiguration_view_dischargement}).
        Prior to leaving, \\$\mathit{reconfiguration}.\mathit{discharged} = \top$ (line~\ref{line:wait_for} of \Cref{lst:reconfiguration_view_transition}) and $\mathit{dischargement}.\mathit{dischargement\_view} = v''$, for some valid view $v'' \supset v^{*}$.
        Since $v^* \supset v$, $v'' \supset v$.
        Hence, we apply the same reasoning as done towards $v'$.
        Such recursion eventually stops due to the fact that there exist only finitely many valid views (by \Cref{lemma:finite_valid_views}), which means that $r$ eventually sets its $\mathit{reconfiguration}.\mathit{prepared}$ variable to $\top$.
    \end{compactitem}
\end{compactitem}
Since we reach contradiction, $\mathit{reconfiguration}.\mathit{prepared} = \top$ at server $r$ eventually and the lemma holds.
\end{proof}

The next two results show that some specific views are forever-alive (\Cref{definition:forever_alive_view}).
Recall that, given \cref{definition:valid_view,definition:forever_alive_view}, all forever-alive views are valid.

\begin{lemma} \label{lemma:current_view_forever_alive}
Let a correct server $r \in v.\mathtt{members}()$ set its $\mathit{current\_view}.\mathit{view}$ variable to $v$.
Then, a view $v' \supseteq v$ is forever-alive.
\end{lemma}
\begin{proof}
If $v = \mathit{genesis}$, $v$ is forever-alive by \Cref{definition:forever_alive_view}.
Hence, the lemma trivially holds in this case.

Let $v \neq \mathit{genesis}$.
Since $r$ sets its $\mathit{current\_view}.\mathit{view}$ variable to $v$, $v \in \mathit{history}$ at server $r$ (by \Cref{lemma:current_view_view_path}).
By \Cref{lemma:history_path}, there exists a view-path $\mathit{path} \subseteq \mathit{install\_messages}$ (at server $r$) such that $\mathit{path}.\mathtt{destination()} = v$.
Moreover, all messages that belong to $\mathit{path}$ are gossiped at line~\ref{line:gossip_install} of \Cref{lst:storage_processing}.
If $r$ is forever-correct, \Cref{definition:forever_alive_view} is satisfied for $v$.
Thus, the lemma holds in this case.

Otherwise, server $r$ halts.
Since $r \in v.\mathtt{members()}$ and $v$ is a valid view (since $v \in \mathit{history}$ at server $r$, $v$ is valid according to \Cref{lemma:history_valid}), server $r$ leaves (see \Cref{subsection:carbon_properties}, paragraph ``Server's rules'') at line~\ref{line:left} of \Cref{lst:reconfiguration_view_transition}.
Moreover, $v \supset \mathit{genesis}$ (by \Cref{lemma:genesis_smallest_valid}).
This means that $r$ learns from a correct server $r'$ that $r'$'s $\mathit{current\_view}.\mathit{view}$ variable is set to some valid view $v' \supset v$ (see \cref{lst:reconfiguration_view_transition,lst:reconfiguration_view_dischargement}), where $r \notin v'.\mathtt{members()}$ and $r' \in v'.\mathtt{members()}$.

By \Cref{lemma:current_view_view_path}, $v' \in \mathit{history}$ at server $r'$.
By \Cref{lemma:history_path}, there exists a view-path $\mathit{path}' \subseteq \mathit{install\_messages}$ (at server $r'$) such that $\mathit{path}'.\mathtt{destination()} = v'$.
Moreover, all messages that belong to $\mathit{path}'$ are gossiped at line~\ref{line:gossip_install} of \Cref{lst:storage_processing} by server $r'$.
If $r'$ is forever-correct, \Cref{definition:forever_alive_view} is satisfied for $v' \supset v$ and the lemma holds.
Otherwise, server $r'$ halts and we apply the same reasoning as done towards server $r$.
The recursion eventually stops due to the fact that only finitely many valid views exist (by \Cref{lemma:finite_valid_views}), which implies that a view greater than $v$ is forever-alive if server $r$ halts (and leaves).
Thus, the lemma holds.
\end{proof}

\begin{lemma} \label{lemma:install_message_forever_alive}
Let a correct server $r$ enter the pseudocode at line~\ref{line:updated_discovery_exists} of \Cref{lst:reconfiguration_view_transition} with $\mathit{joined} = \top$.
Then, a view $v' \supseteq v$ is forever-alive\footnote{View $v$ is the view specified at line~\ref{line:updated_discovery_exists} of \Cref{lst:reconfiguration_view_transition}.}.
\end{lemma}
\begin{proof}
We start by observing that $v \neq \mathit{genesis}$.
This statement follows from the fact that $\mathit{genesis}$ is the smallest valid view (by \Cref{lemma:genesis_smallest_valid}).
If $r \in v.\mathtt{members}()$, then eventually $\mathit{current\_view}.\mathit{view} = v$ at server $r$ (by \Cref{lemma:eventually_prepared}).
Therefore, the lemma follows from \Cref{lemma:current_view_forever_alive}.

Otherwise, we distinguish two possibilities:
\begin{compactitem}
    \item If $r$ never leaves, then $r$ is forever-correct (see \Cref{subsection:carbon_properties}, paragraph ``Server's rules'').
    Since $v \in \mathit{history}$ at server $r$ (due to the check at line~\ref{line:updated_discovery_exists} of \Cref{lst:reconfiguration_view_transition}), \Cref{lemma:history_path} shows (since $v \neq \mathit{genesis}$) that there exists a view-path $\mathit{path} \subseteq \mathit{install\_messages}$ (at server $r$) such that $\mathit{path}.\mathtt{destination()} = v$.
    Since all messages that belong to $\mathit{path}$ are gossiped (at line~\ref{line:gossip_install} of \Cref{lst:storage_processing}) by server $r$, $v$ is forever-alive (by \Cref{definition:forever_alive_view}).
    Therefore, the lemma holds in this case.
    
    \item If $r$ leaves, there exists a correct server $r' \in v'.\mathtt{members()}$, where $v' \supseteq v$, that sets its $\mathit{current\_view}.\mathit{view}$ variable to $v'$ (see \cref{lst:reconfiguration_view_transition,lst:reconfiguration_view_dischargement}).
    Hence, the lemma follows from \Cref{lemma:current_view_forever_alive}.
\end{compactitem}
The proof is concluded since the lemma holds in all possible cases.
\end{proof}

The next lemma proves that a correct member of a forever-alive view sets its $\mathit{current\_view}.\mathit{view}$ variable to that view or there exists a greater view which is forever-alive.

\begin{lemma} \label{lemma:forever_alive_current_view_or_bigger}
Let $v$ be a forever-alive view.
Then, a correct server $r \in v.\mathtt{members()}$ eventually sets its $\mathit{current\_view}.\mathit{view}$ variable to $v$ or a view $v' \supset v$ is forever-alive.
\end{lemma}
\begin{proof}
If $v = \mathit{genesis}$, the lemma trivially holds since all correct members of $v = \mathit{genesis}$ set their $\mathit{current\_view}.\mathit{view}$ variable to $\mathit{genesis}$ (at line~\ref{line:set_current_to_genesis} of \Cref{lst:reconfiguration_initialization}).

Let $v \neq \mathit{genesis}$.
Since $v$ is a forever-alive view, $v$ is a valid view (by \cref{definition:valid_view,definition:forever_alive_view}).
Because $r \in v.\mathtt{members()}$, $r$ has requested to join (by checks at lines~\ref{line:check_join_1} and~\ref{line:check_join_2} of \Cref{lst:view_set_evidence_valid}).
Therefore, $r$ does not halt unless it has previously left (see \Cref{subsection:carbon_properties}, paragraph ``Server's rules'').
Finally, since $v$ is a forever-alive view, every forever-correct process eventually obtains a view-path to $v$ (due to the validity property of the gossip primitive).

The lemma is satisfied if one of the following statements is true:
\begin{compactenum}
    \item Server $r$ sets its $\mathit{current\_view}.\mathit{view}$ variable to $v$.
    
    \item Server $r$ sets its $\mathit{current\_view}.\mathit{view}$ variable to $v' \supset v$.
    In this case, we know that a view $v^* \supseteq v'$ is forever-alive (by \Cref{lemma:current_view_forever_alive}).
    Since $v^* \supset v$, the lemma is indeed satisfied.
    
    \item Server $r$ enters the pseudocode at line~\ref{line:updated_discovery_exists} of \Cref{lst:reconfiguration_view_transition} with $\mathit{joined} = \top$, where the ``destination'' view $v^*$ (i.e., view $v$ from the line~\ref{line:updated_discovery_exists} of \Cref{lst:reconfiguration_view_transition}; note that $v^* \neq v$) is greater than $v$, i.e., $v \subset v^*$.
    By \Cref{lemma:install_message_forever_alive}, a view greater than or equal to $v^*$ is forever-alive.
    Since $v \subset v^*$, the lemma is satisfied.
\end{compactenum}
By contradiction, suppose that none of the three statements is true.
Hence, $r$ never halts (since it never leaves due to the fact that the third statement is not true).
Hence, $r$ is forever-correct and $r$ eventually obtains a view-path to $v$.

First, we prove that $r$ eventually joins.
If $r \in \mathit{genesis}.\mathtt{members}()$, $r$ joins at line~\ref{line:join_init} of \Cref{lst:reconfiguration_initialization}.
Otherwise, since $r$ eventually receives a view-path to $v$ and $r \notin \mathit{genesis}.\mathtt{members()}$, the rule at line~\ref{line:updated_discovery_exists} of \Cref{lst:reconfiguration_view_transition} becomes eventually active and $r$ joins (because of \Cref{lemma:eventually_prepared}).
Once $r$ joins, $r$ sets its $\mathit{current\_view}.\mathit{view}$ variable to $v_r$, where $r \in v_r.\mathtt{members()}$.
Moreover, $v_r \subset v$ (since the aforementioned statements are assumed to not be true).

Next, server $r$ eventually enters the pseudocode at line~\ref{line:updated_discovery_exists} of \Cref{lst:reconfiguration_view_transition} since $r$ eventually obtains a view-path to $v$ and $v_r \subset v$.
Let $v_r'$ be the ``destination'' view (i.e., $\mathit{reconfiguration}.\mathit{destination} = v_r'$ upon entering the pseudocode; set at line~\ref{line:set_reconfiguration_destination} of \Cref{lst:reconfiguration_view_transition}).
Since the third statement is not true, $v_r' \subseteq v$.
Suppose that $v_r' = v$.
By \Cref{lemma:eventually_prepared}, $r$ eventually sets its $\mathit{current\_view}.\mathit{view}$ variable to $v$.
This is impossible due to the fact that the first statement is not true.
Hence, $v_r' \subset v$.

We know that $v_r \subset v_r'$ (ensured by the check at line~\ref{line:updated_discovery_exists} of \Cref{lst:reconfiguration_view_transition}).
Since $r \in v_r.\mathtt{members()}$, $r \in v.\mathtt{members()}$ and $v_r \subset v_r' \subset v$, $r \in v_r'.\mathtt{members()}$.
Hence, $r$ eventually sets its $\mathit{current\_view}.\mathit{view}$ variable to $v_r' \subset v$ (by \Cref{lemma:eventually_prepared}).

We reach the same point as we have reached after $r$ has joined (and set its $\mathit{current\_view}.\mathit{view}$ variable to $v_r$).
Thus, we apply the same argument.
Such recursion eventually stops due to the fact that only finitely many valid views exist (by \Cref{lemma:finite_valid_views}), which implies that only finitely many valid views smaller than $v$ exist.
Therefore, $r$ eventually enters the pseudocode at line~\ref{line:updated_discovery_exists} of \Cref{lst:reconfiguration_view_transition} with the ``destination'' view greater than or equal to $v$.
Once this happens, either the first of the third statement is proven to be true, thus concluding the proof of the lemma.
\end{proof}

The next lemma shows that all correct members of a forever-alive and installable view eventually ``transit'' to that view.

\begin{lemma} \label{lemma:unskippable_forever_alive_all_go}
Let $v$ be a forever-alive view such that $\alpha_{\infty}(v) = \top$.
Then, a correct server $r \in v.\mathtt{members()}$ eventually sets its $\mathit{current\_view}.\mathit{view}$ variable to $v$.
\end{lemma}
\begin{proof}
If $v = \mathit{genesis}$, the lemma trivially holds since all correct members of $v = \mathit{genesis}$ set their $\mathit{current\_view}.\mathit{view}$ variable to $\mathit{genesis}$ (at line~\ref{line:set_current_to_genesis} of \Cref{lst:reconfiguration_initialization}).

Let $v \neq \mathit{genesis}$.
Since $v$ is forever-alive, $v$ is a valid view (by \cref{definition:valid_view,definition:forever_alive_view}).
Because $r \in v.\mathtt{members()}$, $r$ has requested to join (by checks at lines~\ref{line:check_join_1} and~\ref{line:check_join_2} of \Cref{lst:view_set_evidence_valid}).
Hence, $r$ does not halt unless it has previously left (see \Cref{subsection:carbon_properties}, paragraph ``Server's rules'').
Finally, since $v$ is forever-alive, every forever-correct process eventually obtains a view-path to $v$ (due to the validity property of the gossip primitive).

First, we prove that $r$ eventually joins.
If $r \in \mathit{genesis}.\mathtt{members()}$, $r$ joins at line~\ref{line:join_init} of \Cref{lst:reconfiguration_initialization} in this case.
Let $r \notin \mathit{genesis}.\mathtt{members()}$.
By contradiction, suppose that $r$ does not join.
Hence, $r$ does not leave, which implies that $r$ does not halt.
Therefore, $r$ is forever-correct, which means that $r$ eventually obtains a view-path to $v$.
Furthermore, the rule at line~\ref{line:updated_discovery_exists} of \Cref{lst:reconfiguration_view_transition} becomes eventually active (since $r \notin \mathit{genesis}.\mathtt{members()}$) and $r$ joins (by \Cref{lemma:eventually_prepared}).
Once $r$ joins, $r$ sets its $\mathit{current\_view}.\mathit{view}$ variable to $v_r$, where $r \in v_r.\mathtt{members()}$.

Next, we prove that $v_r \subseteq v$.
By contradiction, suppose that $v_r \supset v$ (note that both $v_r$ and $v$ are valid and, thus, comparable; \Cref{theorem:view_comparability}).
Let $\mathit{vr}' = \mathit{source}[v_r]$ at the moment of entering the pseudocode at line~\ref{line:updated_discovery_exists} of \Cref{lst:reconfiguration_view_transition}.
Since $v_r' \subset v_r$, $r \in v_r.\mathtt{members()}$ and $r \notin v_r'.\mathtt{members()}$, $v_r' \subset v$.
Hence, \Cref{lemma:install_skip_installable} is contradicted.
Therefore, $v_r \subseteq v$.

If $v_r = v$, the lemma holds.
Otherwise, $v_r \subset v$.
We prove that $r$ eventually enters the pseudocode at line~\ref{line:updated_discovery_exists} of \Cref{lst:reconfiguration_view_transition}.
By contradiction, suppose that this does not happen.
Hence, $r$ does not leave, which implies that $r$ does not halt (i.e., $r$ is forever-correct).
Therefore, $r$ eventually obtains a view-path to $v$, which means that $r$ eventually enters the pseudocode at line~\ref{line:updated_discovery_exists} of \Cref{lst:reconfiguration_view_transition}, which represents contradiction.

Let $v_r'$ be the ``destination'' view once $r$ enters the pseudocode at line~\ref{line:updated_discovery_exists} of \Cref{lst:reconfiguration_view_transition}.
By \Cref{lemma:install_skip_installable}, $v_r' \subseteq v$.
Since $r \in v_r.\mathtt{members()}$, $r \in v.\mathtt{members()}$ and $v_r \subset v_r' \subseteq v$, $r \in v_r'.\mathtt{members()}$.
By \Cref{lemma:eventually_prepared}, $r$ eventually sets its $\mathit{current\_view}.\mathit{view}$ variable to $v_r'$.

If $v_r' = v$, the lemma holds.
Otherwise, we reach the same point as we have reached after $r$ has joined (and set its $\mathit{current\_view}.\mathit{view}$ variable to $v_r$).
Thus, we apply the same argument.
Such recursion eventually stops due to the fact that only finitely many valid views exist (by \Cref{lemma:finite_valid_views}), which implies that only finitely many valid views smaller than $v$ exist.
Therefore, $r$ eventually sets its $\mathit{current\_view}.\mathit{view}$ variable to $v$, thus concluding the proof.
\end{proof}

Recall that $v_{\mathit{final}}$ denotes the greatest forever-alive view.
The greatest forever-alive view is well-defined since (1) $\mathit{genesis}$ is a forever-alive view (by \Cref{definition:forever_alive_view}; hence, $v_{\mathit{final}} \neq \bot$), (2) all forever-alive views are valid (by \cref{definition:valid_view,definition:forever_alive_view}), (3) all valid views are comparable (by the view comparability property; \Cref{theorem:view_comparability}), and (4) there exist only finitely many valid views (by \Cref{lemma:finite_valid_views}).
We now show that $\alpha_{\infty}(v_{\mathit{final}}) = \top$.

\begin{lemma} \label{lemma:v_final_installable}
$\alpha_{\infty}(v_{\mathit{final}}) = \top$.
\end{lemma}
\begin{proof}
Since $v_{\mathit{final}}$ is valid, $v_{\mathit{final}} \in V(\infty)$ (by \Cref{lemma:valid_in_v}).
By contradiction, let $\alpha_{\infty}(v_{\mathit{final}}) = \bot$.
Hence, $\mathit{genesis} \neq v_{\mathit{final}}$ (since $\alpha_{\infty}(\mathit{genesis}) = \top$).

By \Cref{lemma:forever_alive_current_view_or_bigger}, all correct members of $v_{\mathit{final}}$ eventually set their $\mathit{current\_view}.\mathit{view}$ variable to $v_{\mathit{final}}$ (since $v_{\mathit{final}}$ is the greatest forever-alive view).
Since $v_{\mathit{final}} \neq \mathit{genesis}$, this is done at line~\ref{line:update_current_view} of \Cref{lst:reconfiguration_view_transition}.
Furthermore, every correct member of $v_{\mathit{final}}$ starts $\mathit{vg}(v_{\mathit{final}})$ with a proposal at line~\ref{line:start_with_proposal} of \Cref{lst:reconfiguration_view_transition} (since $v_{\mathit{final}}$ is not installable at time $\infty$; similarly to \Cref{lemma:knows_view_system}).
Moreover, the proposal is valid (by \Cref{lemma:propose_only_valid}).

Suppose that a correct member of $v_{\mathit{final}}$ enters the pseudocode at line~\ref{line:updated_discovery_exists} of \Cref{lst:reconfiguration_view_transition} after setting its $\mathit{current\_view}.\mathit{view}$ variable to $v_{\mathit{final}}$.
By \Cref{lemma:install_message_forever_alive}, $v_{\mathit{final}}$ is not the greatest forever-alive view.
Thus, contradiction and no correct member of $v_{\mathit{final}}$ enters the pseudocode at line~\ref{line:updated_discovery_exists} of \Cref{lst:reconfiguration_view_transition} after setting its $\mathit{current\_view}.\mathit{view}$ variable to $v_{\mathit{final}}$.

Therefore, no correct member of $v_{\mathit{final}}$ stops $\mathit{vg}(v_{\mathit{final}})$ (since $\mathit{vg}(v_{\mathit{final}})$ can only be stopped at line~\ref{line:stop_vg_reconfiguration} of \Cref{lst:reconfiguration_view_transition}).
Because of the liveness property of $\mathit{vg}(v_{\mathit{final}})$, a correct member of $v_{\mathit{final}}$ decides $\mathit{set}_{\mathit{final}}$ from $\mathit{vg}(v_{\mathit{final}})$ (line~\ref{line:decide_from_vg} of \Cref{lst:reconfiguration_view_generation}).
By the validity property of $\mathit{vg}(v_{\mathit{final}})$, $\mathit{set}_{\mathit{final}} \neq \emptyset$, $\mathit{set}_{\mathit{final}}$ is a sequence and $\mathit{set}_{\mathit{final}}.\mathtt{follows}(v_{\mathit{final}}) = \top$ (since $\mathit{preconditions}(v_{final}, \infty) = \top$).
Thus, a correct member of $v_{\mathit{final}}$ enters the pseudocode at line~\ref{line:updated_discovery_exists} of \Cref{lst:reconfiguration_view_transition} (because of the $\mathtt{INSTALL}$ message ``created'' using $\mathit{set}_{\mathit{final}}$).
Therefore, we reach contradiction, which implies that $\alpha_{\infty}(v_{\mathit{final}}) = \top$.
\end{proof}

Next, we show that, if a view $v$ is forever-alive, then all installable views smaller than or equal to $v$ are forever-alive.

\begin{lemma} \label{lemma:forever_alive_all_unskippable}
For every view $v^*$ such that $\alpha_{\infty}(v^*) = \top$ and $v^* \subseteq v_{\mathit{final}}$, $v^*$ is forever-alive.
\end{lemma}
\begin{proof}
Follows directly from \Cref{lemma:all_installable_in_path} and the definition of $v_{\mathit{final}}$.
\end{proof}

Finally, we are able to prove the finality property.

\begin{theorem} [Finality] \label{theorem:finality}
Finality is satisfied.
\end{theorem}
\begin{proof}
By \Cref{lemma:forever_alive_current_view_or_bigger}, all correct members of $v_{\mathit{final}}$ eventually set their $\mathit{current\_view}.\mathit{view}$ variable to $v_{\mathit{final}}$ (since $v_{\mathit{final}}$ is the greatest forever-alive view).
Hence, all correct members of $v_{\mathit{final}}$ update their current view to $v_{\mathit{final}}$.

Next, we prove that no correct member of $v_{\mathit{final}}$ enters the pseudocode at line~\ref{line:updated_discovery_exists} of \Cref{lst:reconfiguration_view_transition}
after setting its $\mathit{current\_view}.\mathit{view}$ variable to $v_{\mathit{final}}$ (which is already proven in the proof of \Cref{lemma:v_final_installable}).
If a correct member of $v_{\mathit{final}}$ does enter the pseudocode, $v_{\mathit{final}}$ is not the greatest forever-alive view (by \Cref{lemma:install_message_forever_alive}).
Therefore, no correct member of $v_{\mathit{final}}$ enters the pseudocode at line~\ref{line:updated_discovery_exists} of \Cref{lst:reconfiguration_view_transition} after setting its $\mathit{current\_view}.\mathit{view}$ variable to $v_{\mathit{final}}$.

Now, consider a correct server $r \in v_{\mathit{final}}.\mathtt{members()}$.
We aim to prove that $r$ installs $v_{\mathit{final}}$.
If $v_{\mathit{final}} = \mathit{genesis}$, $r$ installs $v_{\mathit{final}}$ at line~\ref{line:install_init} of \Cref{lst:reconfiguration_initialization}.
Suppose that $v_{\mathit{final}} \neq \mathit{genesis}$.
By contradiction, assume that $r$ does not install $v_{\mathit{final}}$.
Since $r$ does not install $v_{\mathit{final}}$, $r$ starts $\mathit{vg}(v_{\mathit{final}})$ with a proposal at line~\ref{line:start_with_proposal} of \Cref{lst:reconfiguration_view_transition}.
Moreover, the proposal is valid (by \Cref{lemma:propose_only_valid}).
Since no correct member of $v_{\mathit{final}}$ enters the pseudocode at line~\ref{line:updated_discovery_exists} of \Cref{lst:reconfiguration_view_transition} after setting its $\mathit{current\_view}.\mathit{view}$ variable to $v_{\mathit{final}}$, no correct member of $v_{\mathit{final}}$ stops $\mathit{vg}(v_{\mathit{final}})$ (since $\mathit{vg}(v_{\mathit{final}})$ can only be stopped at line~\ref{line:stop_vg_reconfiguration} of \Cref{lst:reconfiguration_view_transition}).
Because of the liveness property of $\mathit{vg}(v_{\mathit{final}})$, $r$ decides $\mathit{set}_{\mathit{final}}$ from $\mathit{vg}(v_{\mathit{final}})$ (line~\ref{line:decide_from_vg} of \Cref{lst:reconfiguration_view_generation}).
By the validity property of $\mathit{vg}(v_{\mathit{final}})$, $\mathit{set}_{\mathit{final}} \neq \emptyset$, $\mathit{set}_{\mathit{final}}$ is a sequence and $\mathit{set}_{\mathit{final}}.\mathtt{follows}(v_{\mathit{final}}) = \top$ (since $\mathit{preconditions}(v_{final}, \infty) = \top$).
Thus, $r$ enters the pseudocode at line~\ref{line:updated_discovery_exists} of \Cref{lst:reconfiguration_view_transition} (because of the $\mathtt{INSTALL}$ message ``created'' using $\mathit{set}_{\mathit{final}}$).
Therefore, we reach contradiction, which implies that $r$ indeed installs $v_{\mathit{final}}$.

Since $r$ never enters the pseudocode at line~\ref{line:updated_discovery_exists} of \Cref{lst:reconfiguration_view_transition} after setting its $\mathit{current\_view}.\mathit{view}$ variable to $v_{\mathit{final}}$, $r$ does not update its current view after updating it to $v_{\mathit{final}}$.
Similarly, $r$ does not leave.

Finally, by contradiction, suppose that $r$ stops processing in $v_{\mathit{final}}$ (i.e., $r$ executes line~\ref{line:stop_processing} of \Cref{lst:reconfiguration_state_transfer} for some view $v_{\mathit{stop}} \supseteq v_{\mathit{final}}$).
By the check at line~\ref{line:state-request_active} of \Cref{lst:reconfiguration_state_transfer}, $v_{\mathit{stop}} \in \mathit{history}$ at server $r$.
By \Cref{lemma:history_valid}, $v_{\mathit{stop}}$ is a valid view.
Moreover, $\mathit{set}' \neq \emptyset$\footnote{$\mathit{set}'$ is the set of views specified at line~\ref{line:state-request_active} of \Cref{lst:reconfiguration_state_transfer}.} is a sequence and $\mathit{set}'.\mathtt{follows}(v_{\mathit{stop}}) = \top$ (since $\mathit{set}' \in \rho_{\infty}(v_{\mathit{stop}})$ and \Cref{invariant:rhos}).
Since $v_{\mathit{final}} \subseteq v_{\mathit{stop}}$ and $\mathit{set}'.\mathtt{follows}(v_{\mathit{stop}}) = \top$, $v_{\mathit{final}} \subset \mathit{set}'.\mathtt{first()}$.
Furthermore, $\mathit{set}'.\mathtt{first()} \in \mathit{history}$ (see \cref{lst:install_messages,lst:storage_processing}).
Hence, by \Cref{lemma:history_path}, $r$ has obtained a view-path $\mathit{path} \subseteq \mathit{install\_messages}$ (at server $r$) to $\mathit{set}'.\mathtt{first()}$ (note that $\mathit{set}'.\mathtt{first()} \neq \mathit{genesis}$ since $\mathit{set}'.\mathtt{first()} \supset v_{\mathit{final}}$ and $v_{\mathit{final}} \supseteq \mathit{genesis}$).
Finally, since $r$ never leaves and $r$ has gossiped all messages that belong to $\mathit{path}$ (at line~\ref{line:gossip_install} of \Cref{lst:storage_processing}), \Cref{definition:forever_alive_view} is satisfied for $\mathit{set}'.\mathtt{first()}$.
Hence, $\mathit{set}'.\mathtt{first()} \supset v_{\mathit{final}}$ is forever-alive.
Thus, we reach contradiction with the fact that $v_{\mathit{final}}$ is the greatest forever-alive view, which implies that $r$ does not stop processing in $v_{\mathit{final}}$.
Therefore, the theorem holds.
\end{proof}

The next important intermediate result we prove is that there exists a view $v$ such that $r \in v.\mathtt{members()}$ (resp., $r \notin v.\mathtt{members()}$) and $\alpha_{\infty}(v) = \top$ if a correct server $r$ requests to join (resp., leave).
First, we introduce the next two lemmas that help us prove the aforementioned result.

\begin{lemma} \label{lemma:plus_in_final}
Let $r \in v_{\mathit{final}}.\mathtt{members()}$ be a correct server.
Moreover, let $(+, r') \in \mathit{reconfiguration}.\mathit{requested}$ at server $r$.
Then, $(+, r') \in v_{\mathit{final}}$.
\end{lemma}
\begin{proof}
By contradiction, suppose that $(+, r') \notin v_{\mathit{final}}$.
By \Cref{lemma:minus_implies_plus_view}, $(-, r') \notin v_{\mathit{final}}$.

Consider any correct server $r^* \in v_{\mathit{final}}.\mathtt{members()}$.
By the finality property:
\begin{compactitem}
    \item $r^*$ eventually updates its current view to $v_{\mathit{final}}$,
    
    \item $r^*$ never updates its current view afterwards,
    
    \item $r^*$ installs $v_{\mathit{final}}$, and
    
    \item $r^*$ does not leave.
\end{compactitem}

Since $(+, r') \notin v_{\mathit{final}}$, the rule at line~\ref{line:new_proposal} of \Cref{lst:reconfiguration_view_generation} is eventually active at server $r$.
Hence, $r$ proposes to $\mathit{vg}(v_{\mathit{final}})$ at line~\ref{line:propose_installed} of \Cref{lst:reconfiguration_view_generation}.
By \Cref{lemma:propose_only_valid}, the proposal is valid.
Every correct member of $v_{\mathit{final}}$ starts $\mathit{vg}(v_{\mathit{final}})$ and no correct member of $v_{\mathit{final}}$ stops $\mathit{vg}(v_{\mathit{final}})$ (due to the finality property).
Because of the liveness property of $\mathit{vg}(v_{\mathit{final}})$, a correct member of $v_{\mathit{final}}$ decides $\mathit{set}_{\mathit{final}}$ from $\mathit{vg}(v_{\mathit{final}})$ (line~\ref{line:decide_from_vg} of \Cref{lst:reconfiguration_view_generation}).
By the validity property of $\mathit{vg}(v_{\mathit{final}})$, $\mathit{set}_{\mathit{final}} \neq \emptyset$, $\mathit{set}_{\mathit{final}}$ is a sequence and $\mathit{set}_{\mathit{final}}.\mathtt{follows}(v_{\mathit{final}}) = \top$ (since $\mathit{preconditions}(v_{final}, \infty) = \top$).
Thus, a correct member of $v_{\mathit{final}}$ enters the pseudocode at line~\ref{line:updated_discovery_exists} of \Cref{lst:reconfiguration_view_transition} (because of the $\mathtt{INSTALL}$ message ``created'' using $\mathit{set}_{\mathit{final}}$).
By \Cref{lemma:install_message_forever_alive}, $v_{\mathit{final}}$ is not the greatest forever-alive view.
Thus, we reach contradiction, which implies that $(+, r') \in v_{\mathit{final}}$.
\end{proof}

\begin{lemma} \label{lemma:minus_in_final}
Let $r \in v_{\mathit{final}}.\mathtt{members()}$ be a correct server.
Moreover, let $(-, r') \in \mathit{reconfiguration}.\mathit{requested}$ at server $r$ and let $(+, r') \in v_{\mathit{final}}$.
Then, $(-, r') \in v_{\mathit{final}}$.
\end{lemma}
\begin{proof}
By contradiction, suppose that $(-, r') \notin v_{\mathit{final}}$.

Consider any correct server $r^* \in v_{\mathit{final}}.\mathtt{members()}$.
By the finality property:
\begin{compactitem}
    \item $r^*$ eventually updates its current view to $v_{\mathit{final}}$,
    
    \item $r^*$ never updates its current view afterwards,
    
    \item $r^*$ installs $v_{\mathit{final}}$, and
    
    \item $r^*$ does not leave.
\end{compactitem}

Since $(-, r') \notin v_{\mathit{final}}$ and $(+, r') \in v_{\mathit{final}}$, the rule at line~\ref{line:new_proposal} of \Cref{lst:reconfiguration_view_generation} is eventually active at server $r$.
Hence, $r$ proposes to $\mathit{vg}(v_{\mathit{final}})$ at line~\ref{line:propose_installed} of \Cref{lst:reconfiguration_view_generation}.
By \Cref{lemma:propose_only_valid}, the proposal is valid.
Every correct member of $v_{\mathit{final}}$ starts $\mathit{vg}(v_{\mathit{final}})$ and no correct member of $v_{\mathit{final}}$ stops $\mathit{vg}(v_{\mathit{final}})$.
Because of the liveness property of $\mathit{vg}(v_{\mathit{final}})$, a correct member of $v_{\mathit{final}}$ decides $\mathit{set}_{\mathit{final}}$ from $\mathit{vg}(v_{\mathit{final}})$ (line~\ref{line:decide_from_vg} of \Cref{lst:reconfiguration_view_generation}).
By the validity property of $\mathit{vg}(v_{\mathit{final}})$, $\mathit{set}_{\mathit{final}} \neq \emptyset$, $\mathit{set}_{\mathit{final}}$ is a sequence and $\mathit{set}_{\mathit{final}}.\mathtt{follows}(v_{\mathit{final}}) = \top$ (since $\mathit{preconditions}(v_{final}, \infty) = \top$).
Thus, a correct member of $v_{\mathit{final}}$ enters the pseudocode at line~\ref{line:updated_discovery_exists} of \Cref{lst:reconfiguration_view_transition} (because of the $\mathtt{INSTALL}$ message ``created'' using $\mathit{set}_{\mathit{final}}$).
By \Cref{lemma:install_message_forever_alive}, $v_{\mathit{final}}$ is not the greatest forever-alive view.
Thus, we reach contradiction, which implies that $(-, r') \in v_{\mathit{final}}$.
\end{proof}

The next lemma proves that there cannot exist an $\mathtt{INSTALL}$ message $m$ such that $(-, r) \in m.\mathtt{destination()}$ and $(+, r) \notin m.\mathtt{source()}$, where $m.\mathtt{source()}$ is a valid view and $m$ is a server.

\begin{lemma} \label{lemma:empty_to_both}
Let $m \in \mathcal{I}^*(\infty)$, where $m.\mathtt{source()}$ is a valid view.
Let $(+, r) \notin m.\mathtt{source()}$, for some server $r$.
Then, $(-, r) \notin m.\mathtt{destination()}$.
\end{lemma}
\begin{proof}
Follows from the membership validity property of $\mathit{vg}(m.\mathtt{source()})$ (since $m.\mathtt{source()}$ is a valid view and \\$\mathit{preconditions}(m.\mathtt{source()}, \infty) = \top$).
\end{proof}

Finally, we prove that there exists an installable view $v_r \subseteq v_{\mathit{final}}$ such that $r \in v_r.\mathtt{members()}$, where $r$ is a correct server that requested to join.

\begin{lemma} \label{lemma:join_validity}
Let a correct server $r$ request to join.
Then, there exists a view $v_r \subseteq v_{\mathit{final}}$ such that $\alpha_{\infty}(v_r) = \top$ and $r \in v_r.\mathtt{members()}$.
\end{lemma}
\begin{proof}
Since $r$ requests to join, $r \notin \mathit{genesis}.\mathtt{members()}$ (see \Cref{subsection:carbon_properties}, paragraph ``Server's rules'').
Therefore, $(+, r) \notin \mathit{genesis}$.
First, we prove that $(+, r) \in v_{\mathit{final}}$.
By contradiction, assume that $(+, r) \notin v_{\mathit{final}}$.
Hence, $r$ does not join (by \Cref{lemma:current_view_forever_alive} and the fact that $v_{\mathit{final}}$ is the greatest forever-alive view), which implies that $r$ does not leave.
Hence, $r$ never halts (see \Cref{subsection:carbon_properties}, paragraph ``Server's rules''), which means that $r$ is forever-correct.

Since $r$ is forever-correct, the $\mathtt{JOIN}$ message gossiped by $r$ (at line~\ref{line:gossip_join} of \Cref{lst:reconfiguration_joining}) is received by a correct member of $v_\mathit{final}$ (due to the validity property of the gossip primitive).
Hence, $(+, r)$ is included into the $\mathit{reconfiguration}.\mathit{requested}$ variable of that server (due to the fact that $r$ broadcasts the voting proof; see \Cref{subsection:carbon_properties}, paragraph ``Server's rules'').
By \Cref{lemma:plus_in_final}, $(+, r) \in v_{\mathit{final}}$, which represents the contradiction with $(+, r) \notin v_{\mathit{final}}$.
Therefore, our starting assumption was not correct, hence $(+, r) \in v_{\mathit{final}}$.

If $(-, r) \notin v_{\mathit{final}}$, the lemma holds since $\alpha_{\infty}(v_{\mathit{final}}) = \top$ (by \Cref{lemma:v_final_installable}).

Hence, let $(-, r) \in v_{\mathit{final}}$.
In this case, $v_{\mathit{final}} \neq \mathit{genesis}$ (since $\mathit{genesis}$ does not contain ``negative'' changes; see \Cref{appx:model}, paragraph ``Constants'').
Therefore, there exists a view $v_i$ such that $v_i$ is the greatest view smaller than $v_{\mathit{final}}$ that is installable at time $\infty$.
By \Cref{lemma:installable_leads_to_first_biggest}, $v_i$ leads to $v_{\mathit{final}}$ at time $\infty$.
Furthermore, $\mathit{seq}_i \in \rho_{\infty}(v_i)$, where $\mathit{seq}_i = ... \to v_{\mathit{final}}$ (by \Cref{definition:leads}).
By the membership validity property of $\mathit{vg}(v_i)$ (since $\mathit{preconditions}(v_i, \infty) = \top$), $(+, r) \in v_i$.
Therefore, we reach the point we reached with $v_{\mathit{final}}$: (1) $\alpha_{\infty}(v_i) = \top$, and (2) $(+, r) \in v_i$.
The recursion eventually stops due to (1) the fact that there are only finitely many valid views (by \Cref{lemma:finite_valid_views}), which implies existence of only finitely many views installable at time $\infty$, (2) $\mathit{genesis}$ is the smallest installable view (since $\mathit{genesis}$ is installable by definition and $\mathit{genesis}$ is the smallest valid view, by \Cref{lemma:genesis_smallest_valid}), and (3) $(+, r) \notin \mathit{genesis}$.
Therefore, the lemma holds.
\end{proof}

Next, we prove that $(+, r), (-, r) \in v_{\mathit{final}}$, where $r$ is a correct server that requested to leave.

\begin{lemma} \label{lemma:leave_validity}
Let a correct server $r$ request to leave.
Then, $(+, r), (-, r) \in v_{\mathit{final}}$.
\end{lemma}
\begin{proof}
Since $r$ requests to leave, $r$ has previously joined (see \Cref{subsection:carbon_properties}, paragraph ``Server's rules'').
Hence, $(+, r) \in v_{\mathit{final}}$ (by \Cref{lemma:current_view_forever_alive} and the fact that $v_{\mathit{final}}$ is the greatest forever-alive view).
By contradiction, let $(-, r) \notin v_{\mathit{final}}$.
Thus, $r$ never leaves (by \Cref{lemma:install_message_forever_alive} and the fact that $v_{\mathit{final}}$ is the greatest forever-alive view), which means that $r$ is forever-correct.

Since $r$ is forever-correct, the $\mathtt{LEAVE}$ message gossiped by $r$ (at line~\ref{line:gossip_leave} of \Cref{lst:reconfiguration_leaving}) is received by a correct member of $v_{\mathit{final}}$ (due to the validity property of the gossip primitive; see \Cref{appx:model}).
Hence, $(-, r)$ is included into the $\mathit{reconfiguration}.\mathit{requested}$ variable of that server.
By \Cref{lemma:minus_in_final}, $(-, r) \in v_{\mathit{final}}$, which represents the contradiction with $(-, r) \notin v_{\mathit{final}}$.
Therefore, our starting assumption was not correct, hence $(-, r) \in v_{\mathit{final}}$.
The lemma holds.
\end{proof}

Before we are able to prove the join and leave liveness defined in \Cref{subsection:carbon_properties}, we show that a correct server eventually ``discharges'' a view.

\begin{lemma} [Eventual Dischargement] \label{lemma:eventual_dischargement} 
Let a correct server $r$ enter the pseudocode at line~\ref{line:updated_discovery_exists} of \Cref{lst:reconfiguration_view_transition}.
Eventually,\\$\mathit{reconfiguration}.\mathit{discharged} = \top$ at server $r$.
\end{lemma}
\begin{proof}
By contradiction, suppose that the $\mathit{reconfiguration}.\mathit{discharged}$ variable is not set to $\top$ at server $r$.
Hence, $r$ never halts (since it never leaves), which makes $r$ forever-correct.
Then, $[\mathtt{INSTALL}, v', \mathit{seq} = v \to ..., \omega] \in \mathit{install\_messages}$ at server $r$, where $r \notin v.\mathtt{members}()$.
By \Cref{lemma:install_message_forever_alive}, a view $v^* \supseteq v$ is forever-alive.
Since $v^* \subseteq v_{\mathit{final}}$, $v \subseteq v_{\mathit{final}}$.

Since $r$ is forever-correct, the $\mathtt{DISCHARGEMENT-REQUEST}$ gossiped by $r$ (at line~\ref{line:gossip_dischargement_request} of \Cref{lst:reconfiguration_view_transition}) is received by every member of $v_{\mathit{final}}$ while their $\mathit{current\_view}.\mathit{view}$ variable is equal to $v_{\mathit{final}}$ (holds because of the validity property of the gossip primitive and the finality property).
Moreover, $r$ eventually includes $v_{\mathit{final}}$ into its $\mathit{history}$ variable (since $r$ is forever-correct).
Therefore, $r$ eventually receives $[\mathtt{DISCHARGEMENT-CONFIRM}, v_{\mathit{final}}]$ messages from all correct members of $v_{\mathit{final}}$.
Once that happens, the rule at line~\ref{line:discharged_rule} of \Cref{lst:reconfiguration_view_dischargement} is active and $r$ sets its $\mathit{reconfiguration}.\mathit{discharged}$ variable to $\top$.
Thus, the lemma holds.
\end{proof}

Finally, we are ready to prove the join liveness property defined in \Cref{subsection:carbon_properties}.

\begin{theorem} [Join Liveness] \label{theorem:join_liveness}
Join liveness is satisfied.
\end{theorem}
\begin{proof}
Let $r$ be a correct server that requests to join.
By \Cref{lemma:join_validity}, there exists a view $v_r \subseteq v_{\mathit{final}}$ such that (1) $\alpha_{\infty}(v_r) = \top$, and (2) $r \in v_r.\mathtt{members()}$.
By \Cref{lemma:forever_alive_all_unskippable}, $v_r$ is forever-alive.
Then, $r$ joins due to \Cref{lemma:unskippable_forever_alive_all_go}.
\end{proof}

Next, we prove the leave liveness property defined in \Cref{subsection:carbon_properties}.

\begin{theorem} [Leave Liveness] \label{theorem:leave_liveness}
Leave liveness is satisfied.
\end{theorem}
\begin{proof}
Let $r$ be a correct server that requests to leave.
Therefore, $r$ has previously joined (see \Cref{subsection:carbon_properties}, paragraph ``Server's rules'').
By \Cref{lemma:leave_validity}, $(+, r), (-, r) \in v_{\mathit{final}}$; hence, $r \notin v_{\mathit{final}}.\mathtt{members()}$.

By contradiction, suppose that $r$ never leaves.
Hence, $r$ never halts, which means that $r$ is forever-correct.
Since $r$ has joined, the $\mathit{current\_view}.\mathit{view}$ variable at $r$ is not equal to $\bot$; let $\mathit{current\_view}.\mathit{view}$ variable have its final value $v_r$ (the final value is indeed defined since $\mathit{current\_view}.\mathit{view}$ variable of a correct server is set only to valid views, it only ``grows'' and there exist finitely many valid views, by \Cref{lemma:finite_valid_views}).
By \Cref{lemma:current_view_forever_alive}, $v_r \subseteq v_{\mathit{final}}$.
Since $r \in v_r.\mathtt{members()}$ and $r \notin v_{\mathit{final}}.\mathtt{members()}$, $v_r \subset v_{\mathit{final}}$. 

Since $r$ is forever-correct, $r$ eventually enters the pseudocode at line~\ref{line:updated_discovery_exists} of \Cref{lst:reconfiguration_view_transition} due to the fact that $r$ eventually obtains a view-path to $v_{\mathit{final}}$ and $v_r \subset v_{\mathit{final}}$.
Let the ``destination'' view be $v_r'$.
If $r \notin v_r'.\mathtt{members()}$, then $r$ eventually leaves (by \Cref{lemma:eventual_dischargement}).
This is impossible due to the fact that $r$ never leaves.
Hence, $r \in v_r'.\mathtt{members()}$.
Therefore, $r$ updates its $\mathit{current\_view}.\mathit{view}$ variable to $v_r' \supset v_r$ (by \Cref{lemma:eventually_prepared}).
Thus, we reach contradiction with the fact that $v_r$ is the final value of $\mathit{current\_view}.\mathit{view}$ at server $r$, which means that our starting assumption was not correct and $r$ leaves.
\end{proof}

Lastly, we prove the removal liveness property defined in \Cref{subsection:carbon_properties}.
In order to prove the removal liveness property, we prove that a correct forever-validator is a member of $v_{\mathit{final}}$.

\begin{lemma} \label{lemma:forever_validator_in_final}
Let $r$ be a correct forever-validator.
Then, $r \in v_{\mathit{final}}.\mathtt{members()}$.
\end{lemma}
\begin{proof}
Since $r$ is correct forever-validator, $r$ has joined (see \Cref{subsection:carbon_properties}, paragraph ``Validators'').
By \Cref{lemma:current_view_forever_alive}, $(+, r) \in v_{\mathit{final}}$.
By contradiction, suppose that $(-, r) \in v_{\mathit{final}}$.

Since $r$ has joined, the $\mathit{current\_view}.\mathit{view}$ variable at $r$ is not equal to $\bot$; let $\mathit{current\_view}.\mathit{view}$ variable have its final value $v_r$.
By \Cref{lemma:current_view_forever_alive}, $v_r \subseteq v_{\mathit{final}}$.
Since $r \in v_r.\mathtt{members()}$ and $r \notin v_{\mathit{final}}.\mathtt{members()}$, $v_r \subset v_{\mathit{final}}$. 

Since $r$ is forever-correct (because it does not leave), $r$ eventually enters the pseudocode at line~\ref{line:updated_discovery_exists} of \Cref{lst:reconfiguration_view_transition} due to the fact that $r$ eventually obtains a view-path to $v_{\mathit{final}}$ and $v_r \subset v_{\mathit{final}}$.
Let the ``destination'' view be $v_r'$.
If $r \notin v_r'.\mathtt{members()}$, then $r$ eventually leaves (by \Cref{lemma:eventual_dischargement}).
This is impossible due to the fact that $r$ never leaves.
Hence, $r \in v_r'.\mathtt{members()}$.
Therefore, $r$ updates its $\mathit{current\_view}.\mathit{view}$ variable to $v_r' \supset v_r$ (by \Cref{lemma:eventually_prepared}).
Thus, we reach contradiction with the fact that $v_r$ is the final value of $\mathit{current\_view}.\mathit{view}$ at server $r$, which means that our starting assumption was not correct and $(-, r) \notin v_{\mathit{final}}$.
Hence, $r \in v_{\mathit{final}}.\mathtt{members()}$.
\end{proof}

The next theorem proves the removal liveness property.

\begin{theorem} [Removal Liveness] \label{theorem:removal_liveness}
Removal liveness is satisfied.
\end{theorem}
\begin{proof}
Let $r^*$ be a correct forever-validator that obtains a voting proof $\sigma_v$ that a server $r$ is voted out.
By \Cref{lemma:forever_validator_in_final}, $r^* \in v_{\mathit{final}}.\mathtt{members()}$.
Since $r$ is a validator at time $t$, that means that $r$ joined by time $t$.
We aim to prove that $(+, r) \in v_{\mathit{final}}$.
Let us consider both possible cases:
\begin{compactitem}
    \item Let $r$ be a correct server.
    Hence, $(+, r) \in v_{\mathit{final}}$ due to \Cref{lemma:current_view_forever_alive}.
    
    \item Let $r$ be a faulty server.
    Hence, a correct server triggered the special $r \text{ } \mathtt{joined}$ event (at line~\ref{line:trigger_for_others_init} of \Cref{lst:reconfiguration_initialization} or at line~\ref{line:trigger_for_others_2} of \Cref{lst:reconfiguration_view_transition}).
    Hence, $(+, r) \in v_{\mathit{final}}$ due to \Cref{lemma:current_view_forever_alive}.
\end{compactitem}
Therefore, $(+, r) \in v_{\mathit{final}}$.
By \Cref{lemma:minus_in_final}, $(-, r) \in v_{\mathit{final}}$, as well.
Since $(-, r) \in v_{\mathit{final}}$, $\mathit{genesis} \neq v_{\mathit{final}}$ (see \Cref{appx:model}, paragraph ``Constants'').

Again, we separate two possible cases:
\begin{compactitem}
    \item Let $r$ be a correct server.
    By contradiction, suppose that $r$ does not leave.
    Since $r$ has joined, the $\mathit{current\_view}.\mathit{view}$ variable at $r$ is not equal to $\bot$; let $\mathit{current\_view}.\mathit{view}$ variable have its final value $v_r$ (the final value of the $\mathit{current\_view}.\mathit{view}$ variable exists since there are only finitely many valid views, by \Cref{lemma:finite_valid_views}, and a correct server updates the variable only with greater values, by the check at line~\ref{line:updated_discovery_exists} of \Cref{lst:reconfiguration_view_transition}).
    By \Cref{lemma:current_view_forever_alive}, $v_r \subseteq v_{\mathit{final}}$.
    Since $r \in v_r.\mathtt{members()}$ and $r \notin v_{\mathit{final}}.\mathtt{members()}$, $v_r \subset v_{\mathit{final}}$. 

    Since $r$ is forever-correct (because it does not leave), $r$ eventually enters the pseudocode at line~\ref{line:updated_discovery_exists} of \Cref{lst:reconfiguration_view_transition} due to the fact that $r$ eventually obtains a view-path to $v_{\mathit{final}}$ and $v_r \subset v_{\mathit{final}}$.
    Let the ``destination'' view be $v_r'$.
    If $r \notin v_r'.\mathtt{members()}$, then $r$ eventually leaves (by \Cref{lemma:eventual_dischargement}).
    This is impossible due to the fact that $r$ never leaves.
    Hence, $r \in v_r'.\mathtt{members()}$.
    Therefore, $r$ updates its $\mathit{current\_view}.\mathit{view}$ variable to $v_r' \supset v_r$ (by \Cref{lemma:eventually_prepared}).
    Thus, we reach contradiction with the fact that $v_r$ is the final value of $\mathit{current\_view}.\mathit{view}$ at server $r$, which means that our starting assumption was not correct and $r$ leaves.
    
    \item Let $r$ be a faulty server.
    By the finality property, $r^*$ eventually sets its $\mathit{current\_view}.\mathit{view}$ variable to $v_{\mathit{final}} \neq \mathit{genesis}$.
    Hence, $r^*$ enters the pseudocode at line~\ref{line:updated_discovery_exists} of \Cref{lst:reconfiguration_view_transition} with the ``source'' view $v_s$ and with the ``destination'' view $v_{\mathit{final}}$.
    We know that $(+, r), (-, r) \in v_{\mathit{final}}$.
    By the membership validity property of $\mathit{vg}(v_s)$ (since $v_s$ is a valid view and $\mathit{preconditions}(v_s, \infty) = \top$), $(+, r) \in v_s$.
    If $(-, r) \notin v_s$, then $r$ leaves since $r^*$ triggers the special $r \text{ } \mathtt{left}$ event (at line~\ref{line:leave_others} of \Cref{lst:reconfiguration_view_transition}).
    
    Otherwise, $(+, r), (-, r) \in v_s$ and $v_s \neq \mathit{genesis}$ (since $\mathit{genesis}$ does not contain ``negative'' changes; see \Cref{appx:model}, paragraph ``Constants'').
    Now, we consider a correct server $r_s \in v_s.\mathtt{members()}$ that sets its $\mathit{current\_view}.\mathit{view}$ variable to $v_s$ at line~\ref{line:update_current_view} of \Cref{lst:reconfiguration_view_transition}; such server exists due to the fact that the $\mathtt{INSTALL}$ message $m$ is obtained, where $m.\mathtt{source()} = v_s$ and $m.\mathtt{destination()} = v_{\mathit{final}}$, and due to the decision permission property of $\mathit{vg}(v_s)$.
    Hence, we apply the same reasoning as with $v_{\mathit{final}}$.
    Such recursion eventually stops due to (1) the fact that there are only finitely many valid views (by \Cref{lemma:finite_valid_views}), (2) $\mathit{genesis}$ is the smallest valid view (by \Cref{lemma:genesis_smallest_valid}), and (3) $(-, r) \notin \mathit{genesis}$.
    Therefore, $r$ leaves.
\end{compactitem}
The theorem holds since $r$ leaves in all possible cases.
\end{proof}

%% file: appendix/transactions.tex
\section{Client's Implementation} \label{appx:client}

This section is devoted to the implementation of the client.
As we have already mentioned in \Cref{appx:problem_definition}, clients are the users of \sysname, i.e., clients are the entities that are able to issue transactions.
Before we give the implementation of a client, we define the commitment proof of a transaction.
Specifically, we define when $\mathtt{verify\_commit}(\mathit{tx} \in \mathcal{T}, \sigma_c \in \Sigma_c) = \top$.

\begin{lstlisting}[
  caption={The $\mathtt{verify\_commit}$ function},
  label={lst:commitment_proof_verification},
  escapechar=?]
?\textbf{function}? verify_commit(Transaction ?$\mathit{tx}$?, Commitment_Proof ?$\sigma_c$?):
    if ?$\sigma_c$? is not Set(Message):
        ?\textbf{return}? ?$\bot$?
    if does not exist View ?$v$? such that ?$m = [$?COMMITTED, ?$\mathit{tx}$?, ?$v$?, View_Path ?$\mathit{path}]$? and ?$v = \mathit{path}$?.destination(), for every Message ?$m \in \sigma_c$?:
        ?\textbf{return}? ?$\bot$?
        
    Set(Server) ?$\mathit{senders} = \{m\text{.sender} \,|\, m \in \sigma_c\}$?
    if ?$\mathit{senders} \not\subseteq v$?.members():
        ?\textbf{return}? ?$\bot$?
    ?\textbf{return}? ?$|\mathit{senders}| \geq v$?.plurality()
    
?\textbf{function}? (Commitment_Proof ?$\sigma_c$?).signers:
    ?\textbf{return}? ?$\{m\text{.sender} \,|\, m \in \sigma_c\}$?
\end{lstlisting}

Intuitively, $\mathtt{verify\_commit}(\mathit{tx}, \sigma_c) = \top$ if and only if $\sigma_c$ contains $v.\mathtt{plurality()}$ $\mathtt{COMMITTED}$ messages, where $v$ is a valid view. 
Finally, we give the implementation of a client.
We assume that each transaction is signed by its issuer, i.e., a transaction $\mathit{tx}$ is signed by $\mathit{tx}.\mathit{issuer}$.
Hence, only $\mathit{tx}.\mathit{issuer}$ can issue $\mathit{tx}$.
For the sake of brevity, we omit this logic in the pseudocode given below.

\begin{lstlisting}[
  caption={Client - implementation},
  label={lst:client},
  escapechar = ?]
?\textbf{Client:}?
    ?\textcolor{plainorange}{Implementation:}?
        upon start: // executed as soon as the ?\textcolor{gray}{$\mathtt{start}$}? event is triggered
            Set(Message) ?$\mathit{waiting\_messages} = \emptyset$?
            
            // reconfiguration variables
            View ?$\mathit{current\_view} = \mathit{genesis}$?
            Sequence ?$\mathit{history} = \mathit{genesis}$?
            Set(Install_Message) ?$\mathit{install\_messages} = \emptyset$?
            
            String ?$state = $? "idle"
        
            // transaction variables
            Transaction ?$\mathit{current\_transaction} = \bot$?
            Map(View ?$\to$? Set(Server)) ?$\mathit{acks\_from} = \{v \to \emptyset\text{, for every View } v\}$?
            Map(View ?$\to$? Set(Message)) ?$\mathit{acks} = \{v \to \emptyset\text{, for every View } v\}$?
            Set(Message) ?$\mathit{certificate} = \bot$?
            
            // query variables
            Integer ?$\mathit{query\_id} = 0$?
            Map(View ?$\to$? Set(Server)) ?$\mathit{responses\_from} = \{v \to \emptyset\text{, for every View } v\}$?
            Map(View ?$\to$? Set(Integer)) ?$\mathit{responses} = \{v \to \emptyset\text{, for every View } v\}$?
            
        upon receipt of Message ?$m$?:
            ?$\mathit{waiting\_messages} = \mathit{waiting\_messages} \cup \{m\}$?
            
        upon ?exists? Install_Message ?$m \in \mathit{waiting\_messages}$? such that ?$m = [$?INSTALL, ?$v$?, ?$\mathit{set}$?, ?$\omega]$? and ?$v \in \mathit{history}$? and ?$m \notin \mathit{install\_messages}$?: ?\label{line:new_view_in_history}?
            ?$\mathit{waiting\_messages} = \mathit{waiting\_messages} \setminus{\{m\}}$?
            ?$\mathit{install\_messages} = \mathit{install\_messages} \cup \{m\}$?
            gossip ?$m$?
            View ?$v' = m$?.destination()
            if ?$v' \notin \mathit{history}$?:
                ?$\mathit{history} = \mathit{history} \cup \{v'\}$?
            
        upon ?exists? View ?$v \in \mathit{history}$? such that ?$\mathit{current\_view} \subset v$?:
            ?$\mathit{current\_view} = v$?
            if ?$\mathit{state} = $? "query":
                broadcast ?$[$?QUERY, ?$\mathit{query\_id}$?, ?$\mathit{current\_view}.\mathit{view}]$? to ?$\mathit{current\_view}.\mathit{view}$?.members() ?\label{line:query_1}?
            if ?$\mathit{state} = $? "certificate collection":
                broadcast ?$[$PREPARE?, ?$\mathit{current\_transaction}$?, ?$\mathit{current\_view}]$? to ?$\mathit{current\_view}$?.members() ?\label{line:prepare_1_c}?
            if ?$\mathit{state} = $? "commitment proof collection":
                broadcast ?$[$COMMIT?, ?$\mathit{current\_transaction}$?, ?$\mathit{certificate}$?, ?$\mathit{current\_view}]$? to ?$\mathit{current\_view}$?.members() ?\label{line:commit_c_1}?

        upon issue Transaction ?$\mathit{tx}$?: // the client issues a transaction ?\textcolor{gray}{$\mathit{tx}$}? ?\label{line:issue_tx}?
            // reset transaction variables
            ?$\mathit{acks\_from} = \{v \to \emptyset\text{, for every View } v\}$? 
            ?$\mathit{acks} = \{v \to \emptyset\text{, for every View } v\}$?
            ?$\mathit{certificate} = \bot$?
            
            ?$\mathit{state} = $? "certificate collection" // set state to "certificate collection"
            ?$\mathit{current\_transaction} = \mathit{tx}$? // remember the transaction
            // broadcast the prepare message to all members of ?\textcolor{gray}{$\mathit{current\_view}$}?
            broadcast ?$[$PREPARE?, ?$\mathit{tx}$?, ?$\mathit{current\_view}]$? to ?$\mathit{current\_view}.$?members() ?\label{line:prepare_2_c}?
            
        upon query market: // the client requests to learn the total amount of money
            // reset transaction variables
            ?$\mathit{responses\_from} = \{v \to \emptyset\text{, for every View } v\}$? 
            ?$\mathit{responses} = \{v \to \emptyset\text{, for every View } v\}$?
            
            ?$\mathit{state} = $? "query" // set state to "query"
            ?$\mathit{query\_id} = \mathit{query\_id} + 1$?
            // broadcast the query message to all members of ?\textcolor{gray}{$\mathit{current\_view}$}?
            broadcast ?$[$?QUERY, ?$\mathit{query\_id}$?, ?$\mathit{current\_view}]$? to ?$\mathit{current\_view}.$?members() ?\label{line:query_2}?
            
        upon ?exists? Message ?$m \in \mathit{waiting\_messages}$? such that ?$m = [$?ACK, Transaction ?$\mathit{tx}$?, View ?$v$?, View_Path ?$\mathit{path}]$? and ?$\mathit{tx} = \mathit{current\_transaction}$? and ?$\mathit{path}$?.destination() ?$= v$? and ?$m$?.sender ?$\in v$?.members() and?\\$state = $? "certificate collection":    
            ?$\mathit{waiting\_messages} = \mathit{waiting\_messages} \setminus{\{m\}}$?
            Server ?$r = m$?.sender
            if ?$r \notin \mathit{acks\_from}[v]$?:
                ?$\mathit{acks\_from}[v] = \mathit{acks\_from}[v] \cup \{r\}$?
                ?$\mathit{acks}[v] = \mathit{acks}[v] \cup \{m\}$?
                
        upon ?exists? View ?$v$? such that ?$|\mathit{acks\_from}[v]| \geq v$?.quorum() and ?$\mathit{state} = $? "certificate collection":
            ?$\mathit{state} = $? "commitment proof collection" // set state to "commitment proof collection"
            ?$\mathit{certificate} = \mathit{acks}[v]$? // update the certificate
            // broadcast the commit message to all members of ?\textcolor{gray}{$\mathit{current\_view}$}?
            broadcast ?$[$COMMIT?, ?$\mathit{tx}$?, ?$\mathit{certificate}$?, ?$\mathit{current\_view}]$? to ?$\mathit{current\_view}$?.members() ?\label{line:commit_c_2}?
            
        upon obtaining Commitment_Proof ?$\sigma_c$? such that verify_commit(Transaction ?$\mathit{tx}$?, ?$\sigma_c$?) ?$ = \top$?: ?\label{line:obtain_proof_c}?
            ?\textbf{trigger}? committed ?$\mathit{tx}$? // the client learns that ?\textcolor{gray}{$\mathit{tx}$}? is committed 
            if ?$\mathit{tx} = \mathit{current\_transaction}$?:
                ?$\mathit{state} = $? "idle"
                ?$\mathit{current\_transaction} = \bot$?
                
        upon ?exists? Message ?$m \in \mathit{waiting\_messages}$? such that ?$m = [$?QUERY-RESPONSE, Integer ?$\mathit{response}$?, Integer ?$\mathit{qid}$?, View ?$v$?, View_Path ?$\mathit{path}]$? and ?$\mathit{query\_id} = \mathit{qid}$? and ?$\mathit{state} = $? "query" and ?$\mathit{path}$?.destination() ?$= v$? and ?$m$?.sender ?$\in v$?.members(): 
            ?$\mathit{waiting\_messages} = \mathit{waiting\_messages} \setminus{\{m\}}$?
            Server ?$r = m$?.sender
            if ?$r \notin \mathit{responses\_from}[v]$?:
                ?$\mathit{responses\_from}[v] = \mathit{responses\_from}[v] \cup \{r\}$?
                ?$\mathit{responses}[v] = \mathit{responses}[v] \cup \{\mathit{response}\}$?
        
        upon ?exists? View ?$v$? such that ?$|\mathit{responses\_from}[v]| \geq v$?.quorum() and ?$\mathit{state} = $? "query": ?\label{line:query_rule}?
            ?$\mathit{state} = $? "idle"
            ?\textbf{trigger}? total money median(?$\mathit{responses}[v]$?) ?\label{line:median}?
\end{lstlisting}
Let us emphasize that the $\mathit{history}$ variable of a client is indeed a sequence since all views that belong to $\mathit{history}$ are valid (follows from line~\ref{line:new_view_in_history} of \Cref{lst:client}).

\section{Server's Modules: Transaction Module} \label{appx:transaction_module}

This section is devoted to the \emph{transaction module} of a server.
The transaction module contains the logic for processing transactions issued by clients.

We start by specifying what constitutes as a ``valid'' transaction certificate (see \Cref{lst:transaction_certificate_verification}).
Then, we give the implementation used for representing the \emph{state} of a server, i.e., all transactions processed by the server (paragraph ``State representation'').
Lastly, we present the implementation of the transaction module (paragraph ``Transaction module - implementation'').

\begin{lstlisting}[
  caption={Transaction certificate - verification},
  label={lst:transaction_certificate_verification},
  escapechar=?]
?\textbf{function}? verify_transaction_certificate(Transaction ?$\mathit{tx}$?, Set(Message) ?$\mathit{certificate}$?):
    if does not exist View ?$v$? such that ?$m = [$?ACK, ?$\mathit{tx}$?, ?$v$?, View_Path ?$\mathit{path}]$? and ?$v = \mathit{path}$?.destination(), for every Message ?$m \in \mathit{certificate}$?: ?\label{line:check_certificate}?
        ?\textbf{return}? ?$\bot$?
    ?\textbf{return}? ?$|\{m.\mathtt{sender} \,|\, m \in \mathit{certificate} \text{ and } m\text{.sender} \in v\text{.members()}\}| \geq v$?.quorum()
\end{lstlisting}
A transaction certificate $\mathit{certificate}$ is ``valid'' for a transaction $\mathit{tx}$ if and only if $\mathtt{verify\_transaction\_certificate}(\mathit{tx}, \mathit{certificate}) = \top$.
Intuitively, this is the case if $\mathit{certificate}$ contains a quorum of $\mathtt{ACK}$ messages for $\mathit{tx}$.

\para{State representation}
We give below the implementation for representing the state of a server.
This part of the pseudocode is important since servers exchange their state representations during the state transfer periods.

\begin{lstlisting}[
  caption={State representation},
  label={lst:state_representation},
     escapechar=?]
?\textbf{State\_Representation:}?
    instance ?$\mathit{state} = $? {
        Map(Transaction ?$\to$? Message) ?$\mathit{log}$?
        Map(Client, Integer ?$\to$? Set(Transaction)) ?$\mathit{allowed\_acks}$?
    }

    ?\textbf{function}? verify_transaction_proof(Transaction ?$\mathit{tx}$?, Message ?$\mathit{proof}$?): ?\label{line:verify_transaction_proof}?
        ?\textbf{return}? ?$\mathit{proof} = [$?COMMIT, ?$\mathit{tx}$?, Set(Message) ?$\mathit{certificate}$?, View ?$v]$? such that ?\\?verify_transaction_certificate(?$\mathit{tx}$?, ?$\mathit{certificate}$?) ?$ = \top$? // see ?\textcolor{gray}{\Cref{lst:transaction_certificate_verification}}?
            
    ?\textbf{function}? (State_Representation ?$\mathit{state}$?).extract_log():
        Log ?$\mathit{log} = \emptyset$?
        for each Transaction ?$\mathit{tx}$? such that ?$\mathit{state}.\mathit{log}[\mathit{tx}] \neq \bot$?:
            if verify_transaction_proof(?$\mathit{tx}$?, ?$\mathit{state}.\mathit{log}[\mathit{tx}]$?) ?$ = \top$?:
                ?$\mathit{log} = \mathit{log} \cup \{\mathit{tx}\}$?
            else:
                ?\textbf{return}? ?$\bot$?
        ?\textbf{return}? ?$\mathit{log}$?
        
    ?\textbf{function}? (State_Representation ?$\mathit{state}$?).verify_log():
        Log ?$\mathit{log} = \mathit{state}$?.extract_log()
        if ?$\mathit{log} = \bot$?: ?\label{line:check_extract_log_verify_log}?
            ?\textbf{return}? ?$\bot$?
        ?\textbf{return}? admissible(?$\mathit{log}$?) // see ?\textcolor{gray}{\Cref{lst:admissible_transaction_set}}? ?\label{line:check_admissible_state_representation}?
        
    ?\textbf{function}? (State_Representation ?$\mathit{state}$?).verify_allowed_acks():
        for each Client ?$c$? and Integer ?$i$? such that ?$\mathit{state}.\mathit{allowed\_acks}[c][i] \neq \bot$?:
            for each Transaction ?$\mathit{tx} \in \mathit{state}.\mathit{allowed\_acks}[c][i]$?:
                if ?$\mathit{tx}.\mathit{issuer} \neq c$? or ?$\mathit{tx}.\mathit{sn} \neq i$?:
                    ?\textbf{return}? ?$\bot$?
        ?\textbf{return}? ?$\top$?
        
    ?\textbf{function}? (State_Representation ?$\mathit{state}$?).verify():
        ?\textbf{return}? ?$\mathit{state}$?.verify_log() and ?$\mathit{state}$?.verify_allowed_acks()
          
    ?\textbf{function}? merge_logs(State_Representation ?$\mathit{state}_1$?, State_Representation ?$\mathit{state}_2$?):
        if ?$\mathit{state}_1$?.verify() ?$ = \bot$? or ?$\mathit{state}_2$?.verify() ?$ = \bot$?:
            ?\textbf{return}? ?$\bot$?
            
        Map(Transaction ?$\to$? Message) ?$\mathit{new\_log} = \{\mathit{tx} \to \bot, \text{ for every Transaction } \mathit{tx}\}$?
        ?$\mathit{new\_log} = \mathit{state}_1.\mathit{log}$?
        for each Transaction ?$\mathit{tx}$? such that ?$\mathit{state}_2.\mathit{log}[\mathit{tx}] \neq \bot$?:
            ?$\mathit{new\_log}[\mathit{tx}] = \mathit{state}_2.\mathit{log}[\mathit{tx}]$?
        
        ?\textbf{return}? ?$\mathit{new\_log}$?        
\end{lstlisting}

\para{Transaction module - implementation}
Finally, we give the implementation of the transaction module of a server.

\begin{lstlisting}[
  caption={Transaction module - implementation},
  label={lst:transaction_module},
  escapechar = ?]
?\textbf{Transaction Module:}?
    ?\textcolor{plainorange}{Implementation:}?
        upon start: // executed as soon as the ?\textcolor{gray}{$\mathtt{start}$}? event is triggered
            Log ?$\mathit{log} = \emptyset$? ?\label{line:log_init}?
            Map(Transaction ?$\to$? Bool) ?$\mathit{quasi\_committed} = \{\mathit{tx} \to \bot\text{, for every Transaction } \mathit{tx}\}$?
            Map(Transaction ?$\to$? Bool) ?$\mathit{confirmed} = \{\mathit{tx} \to \bot\text{, for every Transaction } \mathit{tx}\}$?
            Map(Client ?$\to$? Integer) ?$\mathit{log\_height} = \{c \to 0\text{, for every Client } c\}$?
            Map(Client ?$\to$? Integer) ?$\mathit{log\_quasi\_committed\_height} = \{c \to 0\text{, for every Client } c\}$?
            State_Representation ?$\mathit{state} = $? {
                ?$\{\mathit{tx} \to \bot$?, for every Transaction ?$\mathit{tx}\}$?, // ?\textcolor{gray}{$\mathit{log}$}?
                ?$\{c, i \to \emptyset$?, for every Client ?$c$? and every Integer ?$i\}$?, // ?\textcolor{gray}{$\mathit{allowed\_acks}$}?
            }
            
            Set(Client) ?$\mathit{faulty\_clients} = \emptyset$?
        
            Set(Transaction) ?$\mathit{quasi\_committed\_current\_view} = \emptyset$?
            Map(Transaction, View ?$\to$? Set(Server)) ?$\mathit{commit\_confirms\_from} = \{\mathit{tx}, v \to \emptyset\text{, for every Transaction } \mathit{tx} \text{ and every View } v\}$?
            Map(Transaction, View ?$\to$? Set(Message)) ?$\mathit{commit\_confirms} = \{\mathit{tx}, v \to \emptyset\text{, for every Transaction } \mathit{tx} \text{ and every View } v\}$?
            
            Map(Transaction, View ?$\to$? Set(Server)) ?$\mathit{committed\_from} = \{\mathit{tx}, v \to \emptyset\text{, for every Transaction } \mathit{tx} \text{ and every View } v\}$?
            Map(Transaction, View ?$\to$? Set(Message)) ?$\mathit{committeds} = \{\mathit{tx}, v \to \emptyset\text{, for every Transaction } \mathit{tx} \text{ and every View } v\}$?
            
        ?\textbf{function}? balance_after_height(Client ?$c$?, Integer ?$\mathit{height}$?):
            Integer ?$\mathit{balance} = $? initial_balance(?$c$?)
            Integer ?$current\_height = 1$?
            ?\textbf{while}? ?$\mathit{current\_height} \leq \mathit{height}$?:
                let ?$\mathit{tx}$? be Transaction such that ?$\mathit{tx} \in \mathit{log}$? and ?$\mathit{tx}.\mathit{issuer} = c$? and ?$\mathit{tx}.\mathit{sn} = \mathit{current\_height}$?
                if ?$\mathit{tx}$? is Withdrawal_Transaction:
                    ?$\mathit{balance} = \mathit{balance} - \mathit{tx}.\mathit{amount}$?
                else if ?$\mathit{tx}$? is Deposit_Transaction:
                    ?$\mathit{balance} = \mathit{balance} + \mathit{tx}.\mathit{withdrawal}.\mathit{amount}$?
                else if ?$\mathit{tx}$? is Minting_Transaction:
                    ?$\mathit{balance} = \mathit{balance} + \mathit{tx}.\mathit{amount}$?
                ?$\mathit{current\_height} = \mathit{current\_height} + 1$?
            ?\textbf{return}? ?$\mathit{balance}$?
    
        upon ?exists? Message ?$m \in \mathit{waiting\_messages}$? such that ?$m = [$?PREPARE, Transaction ?$\mathit{tx}$?, View ?$v]$? and ?$v = \mathit{current\_view}.\mathit{view}$? and ?$\mathit{current\_view}.\mathit{installed} = \top$? and ?$\mathit{tx}.\mathit{issuer} \notin \mathit{faulty\_clients}$? and ?$\mathit{current\_view}.\mathit{processing} = \top$? and ?$\mathit{tx}.\mathit{sn} = \mathit{log\_height}[\mathit{tx}.\mathit{issuer}] + 1$? and (if ?$\mathit{tx}$? is Deposit_Transaction then ?$\mathit{tx}.\mathit{withdrawal} \in \mathit{log}$?): ?\label{line:prepare_rule}?
            ?$\mathit{waiting\_messages} = \mathit{waiting\_messages} \setminus{\{m\}}$?
            
            // check whether the server is allowed to send ACK
            if ?$\mathit{state}.\mathit{allowed\_acks}[\mathit{tx}.\mathit{issuer}][\mathit{tx}.\mathit{sn}] = \emptyset$? or ?$\mathit{state}.\mathit{allowed\_acks}[\mathit{tx}.\mathit{issuer}][\mathit{tx}.\mathit{sn}] = \{\mathit{tx}\}$?: ?\label{line:allowed_acks_check}?
                if ?$\mathit{tx}$? is Withdrawal_Transaction: ?\label{line:check_start}?
                    if ?$\mathit{tx}.\mathit{amount} > $? balance_after_height(?$\mathit{tx}.\mathit{issuer}$?, ?$\mathit{tx}.\mathit{sn} - 1$?): // not enough money
                        ?\textbf{return}?
                else if ?$\mathit{tx}$? is Deposit_Transaction: // deposit transaction
                    if ?exists? Transaction ?$\mathit{tx}' \in \mathit{log}$? such that ?$\mathit{tx}'$? is Deposit_Transaction and ?$\mathit{tx}'.\mathit{issuer} = \mathit{tx}.\mathit{issuer}$? and ?$\mathit{tx}'.\mathit{withdrawal} = \mathit{tx}.\mathit{withdrawal}$?: // withdrawal already used
                        ?\textbf{return}?
                else if ?$\mathit{tx}$? is Vote_Transaction: // vote transaction
                    if ?exists? Transaction ?$\mathit{tx}' \in \mathit{log}$? such that ?$\mathit{tx}'$? is Vote_Transaction and ?$\mathit{tx}'.\mathit{issuer} = \mathit{tx}.\mathit{issuer}$? and ?$\mathit{tx}'.\mathit{motion} = \mathit{tx}.\mathit{motion}$?: // the motion is already being voted for ?\label{line:check_stop}?
                        ?\textbf{return}?
                ?$\mathit{state}.\mathit{allowed\_acks}[\mathit{tx}.\mathit{issuer}][\mathit{tx}.\mathit{sn}] = \{\mathit{tx}\}$? ?\label{line:remember_allowed_acks}?
                send ?$[$?ACK, ?$\mathit{tx}$?, ?$v$?, ?$\mathit{view\_path}[v]]$? to ?$\mathit{tx}.\mathit{issuer}$? ?\label{line:send_ack}?
            else:
                // new proof of misbehavior discovered
                ?$\mathit{state}.\mathit{allowed\_acks}[\mathit{tx}.\mathit{issuer}][\mathit{tx}.\mathit{sn}] = \mathit{state}.\mathit{allowed\_acks}[\mathit{tx}.\mathit{issuer}][\mathit{tx}.\mathit{sn}] \cup \{\mathit{tx}\}$? ?\label{line:multiple_allowed_acks}?
                ?$\mathit{faulty\_clients} = \mathit{faulty\_clients} \cup \{\mathit{tx}.\mathit{issuer}\}$? ?\label{line:update_faulty_clients}?
            
        ?\textbf{function}? allowed_to_commit_confirm(Transaction ?$\mathit{tx}$?):
            if ?$\mathit{tx} \in \mathit{log}$?: ?\label{line:tx_in_log}?
                ?\textbf{return}? ?$\top$?
            
            if ?$\mathit{tx}.\mathit{sn} \neq \mathit{log\_height}[\mathit{tx}.\mathit{issuer}] + 1$?: ?\label{line:ensure_no_gaps}?
                ?\textbf{return}? ?$\bot$?
            
            if ?$\mathit{tx}$? is Withdrawal_Transaction or ?$\mathit{tx}$? is Vote_Transaction or ?$\mathit{tx}$? is Minting_Transaction:
                ?\textbf{return}? ?$\top$?
            else:
                ?\textbf{return}? ?$\mathit{tx}.\mathit{withdrawal} \in \mathit{log}$? ?\label{line:check_withdrawal_when_including_deposit}?
                 
        upon ?exists? Message ?$m \in \mathit{waiting\_messages}$? such that ?$m = [$?COMMIT, Transaction ?$\mathit{tx}$?, Set(Message) ?$\mathit{certificate}$?, View ?$v]$? and verify_transaction_certificate(?$\mathit{tx}$?, ?$\mathit{certificate}$?) ?$= \top$? and ?$v = \mathit{current\_view}.\mathit{view}$? and ?$\mathit{current\_view}.\mathit{installed} = \top$? and ?$\mathit{current\_view}.\mathit{processing} = \top$? and allowed_to_commit_confirm(?$\mathit{tx}$?) ?$= \top$?: ?\label{line:commit_message_rule}?
            ?$\mathit{waiting\_messages} = \mathit{waiting\_messages} \setminus{\{m\}}$?
            // first sight of the commit message
            if ?$\mathit{tx} \notin \mathit{log}$?: ?\label{line:check_tx_in_log}?
                // update the state representation
                ?$\mathit{state}.\mathit{log}[\mathit{tx}] = m$? ?\label{line:update_log_state_representation}?
                
                ?$\mathit{log} = \mathit{log} \cup \{\mathit{tx}\}$? ?\label{line:update_log}?
                ?$\mathit{log\_height}[\mathit{tx}.\mathit{issuer}] = \mathit{tx}.\mathit{sn}$?
                
                broadcast ?$[$?COMMIT, ?$\mathit{tx}$?, ?$\mathit{certificate}$?, ?$v]$? to ?$v$?.members() ?\label{line:send_commit_1}?
            
            send ?$[$?COMMIT-CONFIRM, ?$\mathit{tx}$?, ?$v$?, ?$\mathit{view\_path}[v]]$? to ?$m$?.sender // confirm to the sender, not to the ?\textcolor{gray}{$\mathit{tx}.\mathit{issuer}$}? ?\label{line:send_commit_confirm}?
            
        upon ?exists? Message ?$m \in \mathit{waiting\_messages}$? such that ?$m = [$?COMMIT-CONFIRM, Transaction ?$\mathit{tx}$?, View ?$v$?, View_Path ?$\mathit{path}]$? and ?$\mathit{path}$?.destination() ?$= v$? and ?$m$?.sender ?$\in v$?.members() and ?$\mathit{tx} \in \mathit{log}$?:
            ?$\mathit{waiting\_messages} = \mathit{waiting\_messages} \setminus{\{m\}}$?
            if ?$m$?.sender ?$\notin \mathit{commit\_confirms\_from}[\mathit{tx}][v]$?:
                ?$\mathit{commit\_confirms\_from}[\mathit{tx}][v] = \mathit{commit\_confirms\_from}[\mathit{tx}][v] \cup \{m.$?sender?$\}$?
                ?$\mathit{commit\_confirms}[\mathit{tx}][v] = \mathit{commit\_confirms}[\mathit{tx}][v] \cup \{m\}$?
            
        ?\textbf{function}? allowed_to_quasi_commit(Transaction ?$\mathit{tx}$?): ?\label{line:allowed_to_quasi_commit}?
            if ?$\mathit{tx} \notin \mathit{log}$? or ?$\mathit{quasi\_committed}[\mathit{tx}] = \top$?: ?\label{line:allowed_to_quasi_log}?
                ?\textbf{return}? ?$\bot$?
            
            if ?$\mathit{log\_quasi\_committed\_height}[\mathit{tx}.\mathit{issuer}] + 1 \neq \mathit{tx}.\mathit{sn}$?:
                ?\textbf{return}? ?$\bot$?
                 
            if ?$\mathit{tx}$? is Deposit_Transaction and (?$\mathit{tx}.\mathit{withdrawal} \notin \mathit{log}$? or ?$\mathit{quasi\_committed}[\mathit{tx}.\mathit{withdrawal}] = \bot$?):
                ?\textbf{return}? ?$\bot$?
                
            if (?exists? View ?$v^*$? such that ?$|\mathit{commit\_confirms\_from}[\mathit{tx}][v^*]| \geq v^*$?.quorum()) ?$ = \bot$?:
                ?\textbf{return}? ?$\bot$?
            
            let ?$v$? be the smallest View such that ?$|\mathit{commit\_confirms\_from}[\mathit{tx}][v]| \geq v$?.quorum()
            if ?$v \subset \mathit{current\_view}.\mathit{view}$?:
                ?\textbf{return}? ?$\top$?
            else: // ?\textcolor{gray}{$v = \mathit{current\_view}.\mathit{view}$}?
                if ?$\mathit{quasi\_committed\_current\_view} = \emptyset$?:
                    ?\textbf{return}? ?$\top$?
                Set(Server) ?$\mathit{intersection} = \bigcap\limits_{\mathit{tx}' \in \mathit{quasi\_committed\_current\_view}} \mathit{commit\_confirms\_from}[\mathit{tx}'][v]$?
                ?\textbf{return}? ?$|\mathit{commit\_confirms\_from}[\mathit{tx}][v] \cap \mathit{intersection}| \geq v$?.quorum() ?\label{line:check_intersection_confirm}?
            
        upon ?exists? Transaction ?$\mathit{tx}$? such that allowed_to_quasi_commit(?$\mathit{tx}$?) ?$= \top$?:
            ?$\mathit{quasi\_committed}[\mathit{tx}] = \top$?
            ?$\mathit{log\_quasi\_committed\_height}[\mathit{tx}.\mathit{issuer}] = \mathit{tx}.\mathit{sn}$?
            
            let ?$v$? be the smallest View such that ?$\mathit{commit\_confirms\_from}[\mathit{tx}][v] \geq v$?.quorum()
            if ?$v = \mathit{current\_view}.\mathit{view}$?:
                ?$\mathit{quasi\_committed\_current\_view} = \mathit{quasi\_committed\_current\_view} \cup \{\mathit{tx}\}$?
            
        upon ?exists? Message ?$m \in \mathit{waiting\_messages}$? such that ?$m = [$?COMMITTED, Transaction ?$\mathit{tx}$?, View ?$v$?, View_Path ?$\mathit{path}]$? and ?$\mathit{path}$?.destination() ?$ = v$? and ?$m$?.sender ?$\in v$?.members():
            ?$\mathit{waiting\_messages} = \mathit{waiting\_messages} \setminus{\{m\}}$?
            if ?$m$?.sender ?$\notin \mathit{committed\_from}[\mathit{tx}][v]$?:
                ?$\mathit{committed\_from}[\mathit{tx}][v] = \mathit{committed\_from}[\mathit{tx}][v] \cup \{m.\mathtt{sender}\}$?
                ?$\mathit{committeds}[\mathit{tx}][v] = \mathit{committeds}[\mathit{tx}][v] \cup \{m\}$?
                
        ?\textbf{function}? allowed_to_broadcast_committed(Transaction ?$\mathit{tx}$?):
            if ?$\mathit{quasi\_committed}[\mathit{tx}] = \bot$?: ?\label{line:quasi_committed}?
                ?\textbf{return}? ?$\bot$?
            
            if ?$\mathit{tx}.\mathit{sn} = 1$? and ?$\mathit{tx}$? is not Deposit_Transaction:
                ?\textbf{return}? ?$\top$?
                
            if ?$\mathit{tx}.\mathit{sn} = 1$?:
                ?\textbf{return}? ?$\mathit{confirmed}[\mathit{tx}.\mathit{withdrawal}]$?
            
            if ?$\mathit{tx}$? is not Deposit_Transaction:
                ?\textbf{return}? ?exists? Transaction ?$\mathit{tx}'$? such that ?$\mathit{tx}'.\mathit{issuer} = \mathit{tx}.\mathit{issuer}$? and ?$\mathit{tx}'.\mathit{sn} = \mathit{tx}.\mathit{sn} - 1$? and ?$\mathit{confirmed}[\mathit{tx}'] = \top$? ?\label{line:check_previous_allowed_to_broadcast_1}?
            
            ?\textbf{return}? (?exists? Transaction ?$\mathit{tx}'$? such that ?$\mathit{tx}'.\mathit{issuer} = \mathit{tx}.\mathit{issuer}$? and ?$\mathit{tx}'.\mathit{sn} = \mathit{tx}.\mathit{sn} - 1$? and ?$\mathit{confirmed}[\mathit{tx}'] = \top$?) and ?$\mathit{confirmed}[\mathit{tx}.\mathit{withdrawal}] = \top$? ?\label{line:check_previous_allowed_to_broadcast_2}?
            
        upon ?exists? Transaction ?$\mathit{tx}$? such that allowed_to_broadcast_committed(?$\mathit{tx}$?) ?$ = \top$? and ?$\mathit{current\_view}.\mathit{installed} = \top$?:
            broadcast ?$[$?COMMITTED, ?$\mathit{tx}$?, ?$\mathit{current\_view}.\mathit{view}$?, ?$\mathit{view\_path}[\mathit{current\_view}.\mathit{view}]]$? to ?$\mathit{current\_view}.\mathit{view}$?.members() ?\label{line:broadcast_committed}?
                
        upon ?exist? Transaction ?$\mathit{tx}$? and View ?$v$? such that ?$\mathit{committed\_from}[\mathit{tx}][v] \geq v$?.plurality(): ?\label{line:committed}?
            ?$\mathit{confirmed}[\mathit{tx}] = \top$? ?\label{line:set_confirmed}?
            send ?$\mathit{committeds}[\mathit{tx}][v]$? to ?$\mathit{tx}.\mathit{issuer}$? ?\label{line:proof_to_issuer}?
            if ?$\mathit{tx}$? is Withdrawal_Transaction:
                send ?$\mathit{committeds}[\mathit{tx}][v]$? to ?$\mathit{tx}.\mathit{receiver}$? ?\label{line:proof_to_receiver}?
                
        upon ?exists? Message ?$m \in \mathit{waiting\_messages}$? such that ?$m =[$?QUERY, Integer ?$\mathit{qid}$?, View ?$v]$? and ?$v = \mathit{current\_view}.\mathit{view}$? and ?$\mathit{current\_view}.\mathit{installed} = \top$? and ?$\mathit{current\_view}.\mathit{processing} = \top$?: ?\label{line:receive_query}?
            ?$\mathit{waiting\_messages} = \mathit{waiting\_messages} \setminus{\{m\}}$?
            Log ?$\mathit{quasi\_committed\_log} = \{\mathit{tx} \,|\, \mathit{quasi\_committed}[\mathit{tx}] = \top\}$? ?\label{line:take_quasi_committed}?
            send ?$[$?QUERY-RESPONSE, ?$\mathit{quasi\_committed\_log}$?.total_money(), ?$\mathit{qid}$?, ?$v$?, ?$\mathit{view\_path}[v]]$? to ?$m$?.sender ?\label{line:send_query}?
                 
        ?\textbf{function}? refine_state(Set(State_Representation) ?$\mathit{states}$?): // ?\textcolor{gray}{$\mathit{states} \neq \emptyset$}? ?\label{line:refine_state}?
            for each State_Representation ?$\mathit{rec\_state} \in \mathit{states}$?:
                for each Client ?$c$? and Integer ?$i$? such that ?$\mathit{rec\_state}.\mathit{allowed\_acks}[c][i] \neq \bot$?:
                    ?$\mathit{state}.\mathit{allowed\_acks}[c][i] = \mathit{state}.\mathit{allowed\_acks}[c][i] \cup \mathit{rec\_state}.\mathit{allowed\_acks}[c][i]$? ?\label{line:state_allowed_acks_loop}?
            
            if ?$\mathit{state}$?.verify() ?$ = \top$?: ?\label{line:verify_state}?
                for each State_Representation ?$\mathit{rec\_state} \in \mathit{states}$?:
                    ?$\mathit{state}.\mathit{log} = $? merge_logs(?$\mathit{state}$?, ?$\mathit{rec\_state}$?) // from ?\textcolor{gray}{\Cref{lst:state_representation}}? ?\label{line:merge_states_all}?
            
            ?$\mathit{log} = \mathit{state}$?.extract_log() // ?\textcolor{gray}{$\mathit{committed}[\mathit{tx}] = \bot$, for every newly inserted transaction}? ?\label{line:log_update_extract_log}?
            for each Client ?$c$?:
                Log ?$\mathit{client\_log} = \{\mathit{tx} \,|\, \mathit{tx} \in \mathit{log} \text{ and } \mathit{tx}.\mathit{issuer} = c\}$?
                ?$\mathit{log\_height}[c] = |\mathit{client\_log}|$?
                // ?\textcolor{gray}{$\mathit{log\_quasi\_committed\_height}[c ]$ does not change}?
            ?$\mathit{faulty\_clients} = \{c \,|\, \text{exists Integer } i \text{ such that } |\mathit{state}.\mathit{allowed\_acks}[c][i]| \geq 2\}$?
        
        ?\textbf{function}? installed(View ?$v$?): ?\label{line:upon_install}?
            ?$\mathit{quasi\_committed\_current\_view} = \emptyset$?
            // try to quasi-commit non-quasi-committed transactions
            for each Transaction ?$\mathit{tx} \in \mathit{log}$? such that ?$\mathit{quasi\_committed}[\mathit{tx}] = \bot$?: ?\label{line:quasi_committed_check}?
                let ?$\mathit{state}.\mathit{log}[\mathit{tx}] = [$?COMMIT, ?$\mathit{tx}$?, Set(Message) ?$\mathit{certificate}$?, View ?$v']$?
                broadcast ?$[$?COMMIT, ?$\mathit{tx}$?, ?$\mathit{certificate}$?, ?$v]$? to ?$v$?.members() ?\label{line:broadcast_commit_2}?
\end{lstlisting}

\para{Proof of correctness}
We now prove the commitment validity, integrity, learning, admissibility and signing property of \sysname (see \Cref{subsection:carbon_properties}).
The first property we prove is the commitment integrity.

\begin{theorem} [Commitment Integrity] \label{theorem:commitment_integrity}
Commitment integrity is satisfied.
\end{theorem}
\begin{proof}
Let a transaction $\mathit{tx}$ be committed and let $\mathit{tx}.\mathit{issuer}$ be correct.
Since $\mathit{tx}$ is committed, a (correct or faulty) process has obtained a commitment proof $\sigma_c$ such that $\mathtt{verify\_commit}(\mathit{tx}, \sigma_c) = \top$ (see \Cref{subsection:core_problem}, paragraph ``Commitment proofs'').
Since $\sigma_c$ is a collection of $\mathtt{COMMITTED}$ messages that contain $\mathit{tx}$ signed by $\mathit{tx}.\mathit{issuer}$ (see \Cref{lst:commitment_proof_verification}), $\mathit{tx}.\mathit{issuer}$ has issued $\mathit{tx}$ (since $\mathit{tx}.\mathit{issuer}$ signs $\mathit{tx}$ upon issuing it; line~\ref{line:issue_tx} of \Cref{lst:client}).
Hence, the theorem holds.
\end{proof}

The next property we prove is the commitment admissibility property.
In order to prove the property, we first define \emph{certified} transactions.

\begin{definition} [Certified Transaction] \label{definition:certified_in_view}
A transaction $\mathit{tx}$ is \emph{certified} if and only if a (correct or faulty) process obtains a set of messages $\mathit{certificate}$ such that $\mathtt{verify\_transaction\_certificate}(\mathit{tx}, \mathit{certificate}) = \top$.
Moreover, if the view specified in the $\mathtt{ACK}$ messages of $\mathit{certificate}$ (see \Cref{lst:transaction_certificate_verification}) is $v$, then $\mathit{tx}$ is certified in $v$.
\end{definition}

Next lemma proves that a transaction can be certified only in a view that is installable at time $\infty$.

\begin{lemma} \label{lemma:certified_installable}
Let a transaction $\mathit{tx}$ be certified in a view $v$.
Then, $\alpha_{\infty}(v) = \top$.
\end{lemma}
\begin{proof}
Since $\mathit{tx}$ is certified in $v$, at least a single correct server $r \in v.\mathtt{members()}$ sends the $\mathtt{ACK}$ message associated with $v$ for $\mathit{tx}$ (see \Cref{lst:transaction_certificate_verification}); note that $v$ is a valid view since every $\mathtt{ACK}$ message is accompanied by a view-path to $v$ (by line~\ref{line:check_certificate} of \Cref{lst:transaction_certificate_verification}).
Since a correct member of $v$ sends the message at line~\ref{line:send_ack} of \Cref{lst:transaction_module}, $v$ is installed at that correct server (by the check at line~\ref{line:prepare_rule} of \Cref{lst:transaction_module}).
If $v = \mathit{genesis}$, $\alpha_{\infty}(\mathit{genesis}) = \top$ (by definition) and the lemma holds.
Otherwise, $\alpha_{\infty}(v) = \top$ (similarly to \Cref{lemma:knows_view_system}) and the lemma is concluded.
\end{proof}

The important result before proving the commitment admissibility property is that conflicting transactions cannot be certified.
In order to prove this statement, we first show that no transaction stored in the $\mathit{state}.\mathit{allowed\_acks}$ variable of a correct server is ever removed.

\begin{lemma} \label{lemma:remember_allowed_acks}
Consider any client $c$ and any integer $i$.
Let $\mathit{allowed\_acks}_t$ be the value of the $\mathit{state}.\mathit{allowed\_acks}[c][i]$ variable at a correct server $r$ at time $t$.
Let $\mathit{allowed\_acks}_{t'}$ be the value of the $\mathit{state}.\mathit{allowed\_acks}[c][i]$ variable at $r$ at time $t' > t$.
Then, $\mathit{allowed\_acks}_t \subseteq \mathit{allowed\_acks}_{t'}$.
\end{lemma}
\begin{proof}
The lemma follows from the fact that no transaction is ever removed from the $\mathit{state}.\mathit{allowed\_acks}$ variable at a correct server (see \Cref{lst:transaction_module}).
\end{proof}

The next lemma shows that two conflicting transactions cannot be certified in the same view.

\begin{lemma} \label{lemma:certified_in_the_same_view}
Let transactions $\mathit{tx}$ and $\mathit{tx}'$ be certified in a view $v$ such that (1) $\mathit{tx}.\mathit{issuer} = \mathit{tx}'.\mathit{issuer}$, and (2) $\mathit{tx}.\mathit{sn} = \mathit{tx}'.\mathit{sn}$.
Then, $\mathit{tx} = \mathit{tx}'$.
\end{lemma}
\begin{proof}
By contradiction, let $\mathit{tx} \neq \mathit{tx}'$.
By \Cref{lemma:certified_installable}, $v$ is installable at time $\infty$ (i.e., $\alpha_{\infty}(v) = \top$).
Moreover, $v$ is a valid view.
Hence, at least quorum of members of $v$ are correct (see \Cref{appx:model}, paragraph ``Failure model'').

Since $\mathit{tx}$ (resp., $\mathit{tx}'$) is certified in $v$, a quorum of members of $v$ have sent the $\mathtt{ACK}$ message associated with $v$ for $\mathit{tx}$ (resp., $\mathit{tx}'$), by \Cref{lst:transaction_certificate_verification}.
Because of the quorum intersection, there exists a correct server $r \in v.\mathtt{members()}$ that has sent $\mathtt{ACK}$ messages for both transactions.
Without loss of generality, suppose that $r$ has sent the $\mathtt{ACK}$ message for $\mathit{tx}$ before the $\mathtt{ACK}$ message for $\mathit{tx}'$.

At the moment of sending the $\mathtt{ACK}$ message for $\mathit{tx}$, $\mathit{state}.\mathit{allowed\_acks}[\mathit{tx}.\mathit{issuer}][\mathit{tx}.\mathit{sn}] = \{\mathit{tx}\}$ at server $r$ (by line~\ref{line:remember_allowed_acks} of \Cref{lst:transaction_module}).
Similarly, at the moment of sending the $\mathtt{ACK}$ message for $\mathit{tx}'$, $\mathit{state}.\mathit{allowed\_acks}[\mathit{tx}'.\mathit{issuer} = \mathit{tx}.\mathit{issuer}][\mathit{tx}'.\mathit{sn} = \mathit{tx}.\mathit{sn}] = \{\mathit{tx}'\}$ (by line~\ref{line:remember_allowed_acks} of \Cref{lst:transaction_module}).
This is impossible due to \Cref{lemma:remember_allowed_acks}, hence $r$ does not send both $\mathtt{ACK}$ messages.
Therefore, $\mathit{tx} \neq \mathit{tx}'$ and the lemma holds.
\end{proof}

Now, we prove that any correct member of a view $v' \supset v$ is aware of the fact that a transaction $\mathit{tx}$ is certified in a view $v$, if that is indeed the case.

\begin{lemma} \label{lemma:remember_ack}
Let a transaction $\mathit{tx}$ be certified in a view $v$.
Consider time $t$ and a correct server $r \in v'.\mathtt{members()}$ with $\mathit{current\_view}.\mathit{view} = v'$ at server $r$ at time $t$, where $v' \supset v$.
Then, $\mathit{tx} \in \mathit{state}.\mathit{allowed\_acks}[\mathit{tx}.\mathit{issuer}][\mathit{tx}.\mathit{sn}]$ at server $r$ at time $t$.
\end{lemma}
\begin{proof}
By \Cref{lemma:certified_installable}, $v$ is installable at time $\infty$ (i.e., $\alpha_{\infty}(v) = \top$).
Moreover, $v$ is a valid view.

Let $t^*$ be the first time at which a correct server sets its $\mathit{current\_view}.\mathit{view}$ variable to a view greater than $v$; let us denote that server by $r^*$ and the view by $v^*$.
We now prove that the statement of the lemma holds at time $t^*$ at server $r^*$.

Since $v \subset v^*$, $v^* \neq \mathit{genesis}$ (by \Cref{lemma:genesis_smallest_valid}).
Hence, $r^*$ has executed line~\ref{line:update_current_view} of \Cref{lst:reconfiguration_view_transition}.
This implies that \\$\mathit{reconfiguration}.\mathit{destination} = v^*$ at that time at server $r^*$.
Furthermore, $\mathit{reconfiguration}.\mathit{source} = v$.
Let us prove this claim.
If $\mathit{reconfiguration}.\mathit{source} \subset v$, this fact conflicts with \Cref{lemma:install_skip_installable}.
If $\mathit{reconfiguration}.\mathit{source} \supset v$, then $t^*$ would not be the first time at which a correct server sets its $\mathit{current\_view}.\mathit{view}$ variable to a view greater than $v$ (since the $\mathtt{INSTALL}$ message $m$ is obtained, where $m.\mathtt{source()} \supset v$).
Finally, $\mathit{reconfiguration}.\mathit{prepared} = \top$ at server $r^*$.

Since $r^* \in v^*.\mathtt{members}()$, $r^*$ has not executed line~\ref{line:prepared_free} of \Cref{lst:reconfiguration_view_transition}.
Let us investigate both places at which $r^*$ could have set its $\mathit{reconfiguration}.\mathit{prepared}$ variable to $\top$:
\begin{compactitem}
    \item line~\ref{line:prepared_1} of \Cref{lst:reconfiguration_state_transfer}:
    Server $r^*$ has received $\mathtt{STATE-UPDATE}$ messages from a quorum of members of $v$ (recall that $\mathit{reconfiguration}.\mathit{source} = v$ at that time).
    By the quorum intersection, there exists a correct server $R$ such that (1) $R$ has sent the $\mathtt{ACK}$ message for $\mathit{tx}$ in view $v$, and (2) $r^*$ has received the $\mathtt{STATE-UPDATE}$ message from $R$.
    
    First, we show that $R$ has sent the $\mathtt{ACK}$ message before the $\mathtt{STATE-UPDATE}$ message.
    By contradiction, suppose that $R$ has sent the $\mathtt{STATE-UPDATE}$ message to $r^*$ before the $\mathtt{ACK}$ message.
    At the moment of sending the $\mathtt{STATE-UPDATE}$ message (line~\ref{line:send_state_update} of \Cref{lst:reconfiguration_state_transfer}), $\mathit{current\_view}.\mathit{view}$ at server $R$ is smaller than or equal to $v$ or is equal to $\bot$ (otherwise, $t^*$ would not be the first time at which a correct server sets its $\mathit{current\_view}.\mathit{view}$ variable to a view greater than $v$).
    Hence, the $\mathit{stop\_processing\_until}$ variable (set at line~\ref{line:stop_processing_until} of \Cref{lst:reconfiguration_state_transfer}) at server $R$ is at least equal to $v$ at that time.
    Therefore, the check at line~\ref{line:prepare_rule} of \Cref{lst:transaction_module} does not pass while $\mathit{current\_view}.\mathit{view} = v$ (since $\mathit{current\_view}.\mathit{processing} = \bot$; see \Cref{lst:reconfiguration_view_transition}), which means that $R$ does not send the $\mathtt{ACK}$ message associated with view $v$.
    Hence, we reach contradiction.
    
    At the moment of sending the $\mathtt{ACK}$ message for $\mathit{tx}$ associated with view $v$, $\mathit{state}.\mathit{allowed\_acks}[\mathit{tx}.\mathit{issuer}][\mathit{tx}.\mathit{sn}] = \{\mathit{tx}\}$ at server $R$ (by line~\ref{line:remember_allowed_acks} of \Cref{lst:transaction_module}).
    Therefore, by \Cref{lemma:remember_allowed_acks}, at the moment of sending the $\mathtt{STATE-UPDATE}$ message to $r^*$, $\mathit{tx} \in \mathit{state}.\mathit{allowed\_acks}[\mathit{tx}.\mathit{issuer}][\mathit{tx}.\mathit{sn}]$ at server $R$.
    After executing the $\mathtt{refine\_state}$ function (see \Cref{lst:transaction_module}), $\mathit{tx} \in \mathit{state}.\mathit{allowed\_acks}[\mathit{tx}.\mathit{issuer}][\mathit{tx}.\mathit{sn}]$ at server $r^*$.
    By \Cref{lemma:remember_allowed_acks}, the statement of the lemma holds for $r^*$ at time $t^*$.

    \item line~\ref{line:prepared_2} of \Cref{lst:reconfiguration_state_transfer}:
    This scenario is impossible since it contradicts the fact that $t^*$ is the first time at which a correct server sets its $\mathit{current\_view}.\mathit{view}$ variable to a view greater than $v$.
\end{compactitem}
Note that, because of \Cref{lemma:remember_allowed_acks}, the statement of the lemma is satisfied at server $r^*$ at all times after time $t^*$.

Next, we introduce the following invariant.
Consider time $t^{**}$ and a server $r^{**}$ such that $\mathit{current\_view}.\mathit{view} = v^{**}$ at server $r^{**}$ at time $t^{**}$, where $v \subset v^{**}$.
Then, the statement of the lemma holds at server $r^{**}$ at time $t^{**}$.
Note that this invariant is satisfied at time $t^*$, which is the first time at which these conditions are satisfied.

Consider a server $r'$ that sets its $\mathit{current\_view}.\mathit{view}$ variable to $v' \supset v$.
Hence, $\mathit{reconfiguration}.\mathit{prepared} = \top$ at server $r'$.
Moreover, $\mathit{reconfiguration}.\mathit{source} \supseteq v$ (otherwise, \Cref{lemma:install_skip_installable} is contradicted).
Let us investigate both places at which $r'$ could have set its $\mathit{reconfiguration}.\mathit{prepared}$ variable to $\top$:
\begin{compactitem}
    \item line~\ref{line:prepared_1} of \Cref{lst:reconfiguration_state_transfer}:
    If $\mathit{reconfiguration}.\mathit{source} \supset v$, we know that a correct process from which $r'$ has received the $\mathtt{STATE-UPDATE}$ message has its $\mathit{current\_view}.\mathit{view}$ greater than $v$ at the moment of sending the $\mathtt{STATE-UPDATE}$ message (because of the obtained $\mathtt{INSTALL}$ message $m$, where $m.\mathtt{source()} \supset v$ and the quorum intersection).
    Thus, the invariant is preserved since the invariant holds for a correct server from which $r'$ receives the $\mathtt{STATE-UPDATE}$ message (see the $\mathtt{refine\_state}$ function in \Cref{lst:state_representation}).
    
    If $\mathit{reconfiguration}.\mathit{source} = v$, then $r'$ has received the $\mathtt{STATE-UPDATE}$ messages from at least $v.\mathtt{plurality}()$ of correct members of $v$.
    If at least a single such member had $\mathit{current\_view}.\mathit{view}$ variable equal to a view greater than $v$ at the moment of sending the $\mathtt{STATE-UPDATE}$ message, the invariant is preserved because of the invariant hypothesis (and the $\mathtt{refine\_state}$ function).
    Otherwise, the invariant is preserved because of the argument given in the base step.
    The statement of the lemma is preserved for server $r'$ forever by \Cref{lemma:remember_allowed_acks}.
    
    \item line~\ref{line:prepared_2} of \Cref{lst:reconfiguration_state_transfer}:
    Since $\mathit{reconfiguration}.\mathit{source} \supseteq v$, the invariant is preserved because of the invariant hypothesis (and the $\mathtt{state\_refine}$ function).
\end{compactitem}
Since the invariant is preserved, the lemma holds.
\end{proof}

The next lemma proves that two conflicting transactions cannot be certified in different views.

\begin{lemma} \label{lemma:no_conflicting_different_views}
Let a transaction $\mathit{tx}$ be certified in a view $v$.
Let a transaction $\mathit{tx}'$ be certified in a view $v' \neq v$.
Moreover, let (1) $\mathit{tx}.\mathit{issuer} = \mathit{tx}'.\mathit{issuer}$, and (2) $\mathit{tx}.\mathit{sn} = \mathit{tx}'.\mathit{sn}$.
Then, $\mathit{tx} = \mathit{tx}'$.
\end{lemma}
\begin{proof}
By contradiction, let $\mathit{tx} \neq \mathit{tx}'$.
By \Cref{lemma:certified_installable}, $v$ and $v'$ are installable views at time $\infty$.
Therefore, both views are valid.
By the view comparability property, $v$ and $v'$ are comparable.
Without loss of generality, let $v \subset v'$.

Since $\mathit{tx}'$ is certified in $v'$, there exists a correct server $r'$ that sends the $\mathtt{ACK}$ message for $\mathit{tx}'$ associated with view $v'$ (at line~\ref{line:send_ack} of \Cref{lst:transaction_module}).
By line~\ref{line:remember_allowed_acks} of \Cref{lst:transaction_module}, $\mathit{state}.\mathit{allowed\_acks}[\mathit{tx}'.\mathit{issuer}][\mathit{tx}'.\mathit{sn}] = \{\mathit{tx}'\}$ at server $r'$ at that time.

Let $r'$ set its $\mathit{current\_view}.\mathit{view}$ variable to $v'$ at some time $t_{v'}$; note that this indeed happens since $r'$ sends the $\mathtt{ACK}$ message associated with $v'$ (by the check at line~\ref{line:prepare_rule} of \Cref{lst:transaction_module}).
Because of \Cref{lemma:remember_ack}, $\mathit{tx} \in \mathit{state}.\mathit{allowed\_acks}[\mathit{tx}.\mathit{issuer} = \mathit{tx}'.\mathit{issuer}][\mathit{tx}.\mathit{sn} = \mathit{tx}'.\mathit{sn}]$ at server $r'$ at time $t_{v'}$.
Finally, $r'$ sends the $\mathtt{ACK}$ message for $\mathit{tx}'$ after time $t_{v'}$ (by the check at line~\ref{line:prepare_rule} of \Cref{lst:transaction_module}, $r'$ has already updated its $\mathit{current\_view}.\mathit{view}$ variable to $v'$).
Recall that $\mathit{state}.\mathit{allowed\_acks}[\mathit{tx}'.\mathit{issuer}][\mathit{tx}'.\mathit{sn}] = \{\mathit{tx}'\}$ at server $r'$ at that time.
This is impossible due to \Cref{lemma:remember_allowed_acks}, which implies that $\mathit{tx} = \mathit{tx}'$.
\end{proof}

Finally, we prove that conflicting transactions cannot be certified.

\begin{lemma} [Conflicting Transactions Cannot Be Certified] \label{lemma:no_conflicting}
Let transactions $\mathit{tx}$ and $\mathit{tx}'$ be certified such that (1) $\mathit{tx}.\mathit{issuer} = \mathit{tx}'.\mathit{issuer}$, and (2) $\mathit{tx}.\mathit{sn} = \mathit{tx}'.\mathit{sn}$.
Then, $\mathit{tx} = \mathit{tx}'$.
\end{lemma}
\begin{proof}
Let $v$ denote the smallest view in which $\mathit{tx}$ is certified and let $v'$ denote the smallest view in which $\mathit{tx}'$ is certified.
Since $v$ and $v'$ are valid views (follows from \Cref{lemma:certified_installable}), $v$ and $v'$ are comparable (by the view comparability property defined in \Cref{lst:reconfiguration_properties_new}; see \Cref{theorem:view_comparability}).
If $v = v'$, the lemma follows from \Cref{lemma:certified_in_the_same_view}.
Otherwise, the lemma follows from \Cref{lemma:no_conflicting_different_views}.
\end{proof}

The next lemma proves that the $\mathit{state}.\mathit{log}$ variable at a correct server only ``grows''.

\begin{lemma} \label{lemma:log_only_grows}
Consider a correct server $r$.
Let $\mathit{log}_t^r = \{\mathit{tx} \,|\, \mathit{state}.\mathit{log}[\mathit{tx}] \neq \bot \text{ at server } r  \text{ at some time } t\}$.
Let $\mathit{log}_{t'}^r = \{\mathit{tx} \,|\, \mathit{state}.\mathit{log}[\mathit{tx}] \neq \bot \text{ at server } r  \text{ at some time } t' > t\}$.
Then, $\mathit{log}_t^r \subseteq \mathit{log}_{t'}^r$.
\end{lemma}
\begin{proof}
In order to prove the lemma, we prove that $\mathit{state}.\mathit{log}[\mathit{tx}]$, for any transaction $\mathit{tx}$, is never reverted to $\bot$.
Therefore, we consider all the place at which $\mathit{state}.\mathit{log}$ variable at server $r$ changes:
\begin{compactitem}
    \item line~\ref{line:update_log_state_representation} of \Cref{lst:transaction_module}:
    The invariant is trivially preserved in this case.
    
    \item line~\ref{line:merge_states_all} of \Cref{lst:transaction_module}:
    In this case, the invariant is preserved because of the $\mathtt{merge\_logs}$ function (see \Cref{lst:state_representation}).
\end{compactitem}
Therefore, the lemma holds.
\end{proof}

The next important intermediate result we show is that $\mathit{state}.\mathtt{verify()} = \top$ at a correct server at all times.
In order to do so, we prove that the $\mathtt{merge\_logs}$ function (see \Cref{lst:state_representation}) invoked with two ``valid'' states (according to the $\mathtt{verify()}$ function) returns an admissible log.

\begin{lemma} \label{lemma:merge_verified}
Let $\mathit{state}_1$ (resp., $\mathit{state}_2$) be a state representation such that $\mathit{state}_1.\mathtt{verify()} = \top$ (resp., $\mathit{state}_2.\mathtt{verify()} = \top$).
Let $\mathit{log\_map} = \mathtt{merge\_logs}(\mathit{state}_1, \mathit{state}_2)$ and let $\mathit{log} = \{\mathit{tx} \,|\, \mathit{log\_map}[\mathit{tx}] \neq \bot\}$.
Then, $\mathit{log}$ is admissible.
\end{lemma}
\begin{proof}
Let $\mathit{log}_1 = \mathit{state}_1.\mathtt{extract\_log()}$ and $\mathit{log}_2 = \mathit{state}_2.\mathtt{extract\_log()}$.
Every transaction $\mathit{tx} \in \mathit{log}_1 \cup \mathit{log}_2$ is certified (follows from the $\mathtt{verify\_transaction\_proof}$ function and the fact that $\mathit{state}_1.\mathtt{verify()} = \mathit{state}_2.\mathtt{verify()} = \top$).
Furthermore, $\mathit{log} = \mathit{log}_1 \cup \mathit{log}_2$.

We now consider the $\mathtt{admissible}(\mathit{log})$ function (see \Cref{lst:admissible_transaction_set}):
\begin{compactenum}
    \item The check at line~\ref{line:check_withdrawal_admissible} of \Cref{lst:admissible_transaction_set} passes since $\mathit{log} = \mathit{log}_1 \cup \mathit{log}_2$ and $\mathtt{admissible}(\mathit{log}_1) = \mathtt{admissible}(\mathit{log}_2) = \top$ (since $\mathit{state}_1.\mathtt{verify()} = \mathit{state}_2.\mathtt{verify()} = \top$).
    
    \item Let $\mathit{log}_c = \{\mathit{tx} \,|\, \mathit{tx} \in \mathit{log} \text{ and } \mathit{tx}.\mathit{issuer} = c\}$, for some client $c$.
    Let $\mathit{log}_c^1 = \{\mathit{tx} \,|\, \mathit{tx} \in \mathit{log}_1 \text{ and } \mathit{tx}.\mathit{issuer} = c\}$ and $\mathit{log}_c^2 = \{\mathit{tx} \,|\, \mathit{tx} \in \mathit{log}_2 \text{ and } \mathit{tx}.\mathit{issuer} = c\}$.
    Since $\mathit{state}_1.\mathtt{verify()} = \mathit{state}_2.\mathtt{verify()} = \top$, $\mathit{log} = \mathit{log}_1 \cup \mathit{log}_2$ and \Cref{lemma:no_conflicting}, $\mathit{log}_c = \mathit{log}_c^1$ or $\mathit{log}_c = \mathit{log}_c^2$.
    Therefore, $\mathtt{admissible\_client\_log}(c, \mathit{log}_c) = \top$ (because $\mathit{state}_1.\mathtt{verify()} = \mathit{state}_2.\mathtt{verify()} = \top$).
\end{compactenum}
Therefore, $\mathit{log}$ is indeed admissible.
\end{proof}

The next lemma proves that, if $\mathit{state}.\mathit{log}[\mathit{tx}] \neq \bot$ at a correct server, then $\mathit{tx}$ is certified.

\begin{lemma} \label{lemma:in_state_certified}
At all times, if $\mathit{state}.\mathit{log}[\mathit{tx}] \neq \bot$ at a correct server, then $\mathit{tx}$ is certified.
\end{lemma}
\begin{proof}
Initially, the statement of the lemma holds (since $\mathit{state}$ does not ``contain'' any information).
We introduce the invariant that the statement of the lemma holds and we consider all the places at which $\mathit{state}.\mathit{log}$ is modified:
\begin{compactitem}
    \item line~\ref{line:update_log} of \Cref{lst:transaction_module}:
    The invariant is preserved since the newly inserted transaction is certified (by line~\ref{line:commit_message_rule} of \Cref{lst:transaction_module}).
    
    \item line~\ref{line:merge_states_all} of \Cref{lst:transaction_module}:
    The invariant is preserved since $\mathit{state}.\mathtt{verify()} = \top$ (by line~\ref{line:verify_state} of \Cref{lst:transaction_module}) and $\mathit{rec\_state}.\mathtt{verify()} = \top$ (ensured by line~\ref{line:state_update_received} of \Cref{lst:reconfiguration_state_transfer}) and the $\mathtt{merge\_logs}$ function (\Cref{lst:state_representation}).
\end{compactitem}
The lemma holds.
\end{proof}

The next lemma proves that $\mathit{state}.\mathtt{verify()} = \top$ at a correct server.

\begin{lemma} \label{lemma:always_verified}
At all times, $\mathit{state}.\mathtt{verify()} = \top$ at a correct server.
\end{lemma}
\begin{proof}
Initially, $\mathit{state}.\mathtt{verify()} = \top$ (since $\mathit{state}$ does not ``contain'' any information).
We introduce the invariant that $\mathit{state}.\mathtt{verify()} = \top$ and check the next modification of the variable.
Let us consider all the places at which such modification can occur:
\begin{compactitem}
    \item line~\ref{line:remember_allowed_acks} of \Cref{lst:transaction_module}:
    The only change introduced is that $\mathit{state}.\mathit{allowed\_acks}[\mathit{tx}.\mathit{issuer}][\mathit{tx}.\mathit{sn}] = \{\mathit{tx}\}$, for some transaction $\mathit{tx}$.
    The $\mathtt{verify\_allowed\_acks}$ function still returns $\top$, which means that the invariant is preserved.
    
    \item line~\ref{line:multiple_allowed_acks} of \Cref{lst:transaction_module}:
    Similarly to the previous case, the $\mathtt{verify\_allowed\_acks}$ function returns $\top$, which implies that the invariant is preserved.
    
    \item line~\ref{line:update_log_state_representation} of \Cref{lst:transaction_module}:
    Let $\mathit{state}_{\mathit{new}}$ be the the value of the $\mathit{state}$ variable \emph{after} the execution of this line.
    First, note that $\mathit{state}_{\mathit{new}}.\mathtt{extract\_log()} = \mathit{log} \neq \bot$.
    Moreover, $\mathit{state}_{\mathit{new}}.\mathtt{verify\_allowed\_acks()} = \top$.
    Therefore, it is left to show that $\mathit{log}$ is admissible.
    
    Let $\mathit{log}_{\mathit{before}} = \{\mathit{tx} \,|\, \mathit{state}_{\mathit{before}}.\mathit{log}[\mathit{tx}] \neq \bot\}$, where $\mathit{state}_{\mathit{before}}$ is the value of the $\mathit{state}$ variable \emph{before} the execution of this line.
    Hence, $\mathit{log} = \mathit{log}_{\mathit{before}} \cup \{\mathit{tx}\}$, where $\mathit{tx}$ is the transaction inserted at this line.
    Let $c = \mathit{tx}.\mathit{issuer}$.
    Let us consider the $\mathtt{admissible}(\mathit{log})$ function (see \Cref{lst:admissible_transaction_set}):
    \begin{compactenum}
        \item The check at line~\ref{line:no_cycles} of \Cref{lst:admissible_transaction_set} passes due to the construction of $\mathit{log}$.
    
        \item The check at line~\ref{line:check_withdrawal_admissible} of \Cref{lst:admissible_transaction_set} passes since $\mathit{log}_{\mathit{before}}$ is admissible and, if $\mathit{tx}$ is a deposit transaction, then $\mathit{tx}.\mathit{withdrawal} \in \mathit{log}_{\mathit{before}}$ (by the check at line~\ref{line:check_withdrawal_when_including_deposit} of \Cref{lst:transaction_module}).
        
        \item Consider any client $c' \neq c$.
        Let $\mathit{log}_{c'} = \{\mathit{tx} \,|\, \mathit{tx} \in \mathit{log} \text{ and } \mathit{tx}.\mathit{issuer} = c'\}$ and let $\mathit{log}_{c'}^{\mathit{before}} = \{\mathit{tx} \,|\, \mathit{tx} \in \mathit{log}_{\mathit{before}} \text{ and } \mathit{tx}.\mathit{issuer} = c'\}$.
        Since the only inserted transaction is $\mathit{tx}$ and $\mathit{tx}.\mathit{issuer} \neq c'$, $\mathit{log}_{c'} = \mathit{log}_{c'}^{\mathit{before}}$.
        Hence, $\mathtt{admissible\_client\_log}(c', \mathit{log}_{c'}) = \top$ (because of the invariant). 
        Therefore, we need to prove the same for $c$.
        
        Let $\mathit{log}_{c} = \{\mathit{tx} \,|\, \mathit{tx} \in \mathit{log} \text{ and } \mathit{tx}.\mathit{issuer} = c\}$ and $\mathit{log}_{c}^{\mathit{before}} = \{\mathit{tx} \,|\, \mathit{tx} \in \mathit{log}_{\mathit{before}} \text{ and } \mathit{tx}.\mathit{issuer} = c\}$.
        Because of the check at line~\ref{line:ensure_no_gaps} of \Cref{lst:transaction_module}, $\mathit{tx}.\mathit{sn} = |\mathit{log}_c^{\mathit{before}}| + 1$.
        Consider the $\mathtt{admissible\_client\_log}(c, \mathit{log}_c)$ function:
        \begin{compactenum}
            \item The check at line~\ref{line:check_conflicting_well_formed} of \Cref{lst:admissible_transaction_set} passes since all transactions from $\mathit{log}_c$ are certified (since $\mathit{state}_{\mathit{new}}.\mathtt{extract\_log()} \neq \bot$) and no two conflicting transactions are certified (by \Cref{lemma:no_conflicting}).
        
            \item The check at line~\ref{line:no_gaps_admissible} of \Cref{lst:admissible_transaction_set} passes since the invariant holds for $\mathit{state}_{\mathit{before}}$ and $\mathit{tx}.\mathit{sn} = |\mathit{log}_c^{\mathit{before}}| + 1$.
            
            \item  Since $\mathit{tx}$ (the only inserted transaction) is certified (by the check at line~\ref{line:commit_message_rule} of \Cref{lst:transaction_module}), there exists a correct server $R$ that has sent the $\mathtt{ACK}$ message for $\mathit{tx}$ (at line~\ref{line:send_ack} of \Cref{lst:transaction_module}).
            Therefore, the check from line~\ref{line:check_admissible_client_start} to~\ref{line:check_admissible_client_stop} of \Cref{lst:admissible_transaction_set} passes because of the fact that the invariant holds for $\mathit{state}_{\mathit{before}}$, because $R$ has executed the checks from line~\ref{line:check_start} to line~\ref{line:check_stop} of \Cref{lst:transaction_module} and by \cref{lemma:no_conflicting,lemma:in_state_certified}.
        \end{compactenum}
        Therefore, $\mathtt{admissible\_client\_log}(c, \mathit{log}_c) = \top$, which concludes the invariant preservation in this case.
    \end{compactenum}
    
    \item line~\ref{line:state_allowed_acks_loop} of \Cref{lst:transaction_module}:
    In this case, the invariant is preserved since $\mathit{rec\_state}.\mathtt{verify()} = \top$ (ensured by line~\ref{line:state_update_received} of \Cref{lst:reconfiguration_state_transfer}).
    
    \item line~\ref{line:merge_states_all} of \Cref{lst:transaction_module}:
    Let $\mathit{state}_{\mathit{new}}$ be the value of the $\mathit{state}$ variable \emph{after} executing this line.
    First, $\mathit{state}_{\mathit{new}}.\mathtt{extract\_log()} = \mathit{log} \neq \bot$.
    By \Cref{lemma:merge_verified}, $\mathit{log}$ is admissible.
    Therefore, the invariant is preserved.
\end{compactitem}
The invariant is always preserved, thus the lemma holds.
\end{proof}

Note that \Cref{lemma:always_verified} is crucial for the liveness of the reconfiguration module since the proof (see \Cref{subsection:reconfiguration_liveness}) assumes that correct servers always send the ``valid'' (according to the $\mathtt{verify}$ function) state.
\Cref{lemma:always_verified} proves that this is indeed the case.
Finally, we are ready to prove the commitment admissibility property.

\begin{theorem} [Commitment Admissibility] \label{theorem:commitment_admissibility}
Commitment admissibility is satisfied.
\end{theorem}
\begin{proof}
Consider $\mathit{log}_t$, the set of committed transactions at time $t$.
Consider a transaction $\mathit{tx} \in \mathit{log}_t$.
Since $\mathit{tx} \in \mathit{log}_t$ (i.e., $\mathit{tx}$ is committed), there exists a correct server that sends the $\mathtt{COMMITTED}$ message for $\mathit{tx}$ (see \Cref{lst:commitment_proof_verification}); the $\mathtt{COMMITTED}$ message is broadcast at line~\ref{line:broadcast_committed} of \Cref{lst:transaction_module}.
We start by proving the following three claims:
\begin{compactenum}
    \item If $\mathit{tx}.\mathit{sn} > 1$, then $\mathit{tx}' \in \mathit{log}_t$, where $\mathit{tx}'.\mathit{issuer} = \mathit{tx}.\mathit{issuer}$ and $\mathit{tx}'.\mathit{sn} = \mathit{tx}.\mathit{sn} - 1$.
    Before the $\mathtt{COMMITTED}$ message for $\mathit{tx}$ is sent, the server checks whether $\mathit{confirmed}[\mathit{tx}'] = \top$ (at lines~\ref{line:check_previous_allowed_to_broadcast_1} or~\ref{line:check_previous_allowed_to_broadcast_2} of \Cref{lst:transaction_module}).
    Therefore, $\mathit{tx}'$ is committed (by line~\ref{line:set_confirmed} of \Cref{lst:transaction_module} and \Cref{lst:commitment_proof_verification}).
    
    \item If $\mathit{tx}$ is a deposit transaction, then $\mathit{tx}.\mathit{withdrawal} \in \mathit{log}_t$ (similarly to the previous case). 
    
    \item $\mathit{tx}$ does not precede $\mathit{tx}$ in $\mathit{log}_t$.
    By contradiction, suppose that it does.
    Therefore, the entire cycle was included in the $\mathit{log}$ variable of the sender at the moment of sending the $\mathtt{COMMITTED}$ message.
    Hence, the entire cycle is stored in $\mathit{state}.\texttt{extract\_log()}$.
    Since $\mathit{state}.\mathtt{verify()} = \top$ (by \Cref{lemma:always_verified}), this is impossible (due to the check at line~\ref{line:no_cycles} of \Cref{lst:admissible_transaction_set}).
\end{compactenum}

Because of the $\mathtt{allowed\_to\_broadcast\_committed}$ function, $\mathit{quasi\_committed}[\mathit{tx}] = \top$ at the server (by line~\ref{line:quasi_committed} of \Cref{lst:transaction_module}).
Therefore, $\mathit{state}.\mathit{log}[\mathit{tx}] \neq \bot$ at the server at that time (because of the $\mathtt{allowed\_to\_quasi\_commit}$ function; by the check at line~\ref{line:allowed_to_quasi_log} of \Cref{lst:transaction_module}).
Since $\mathit{state}.\mathtt{verify()} = \top$ (by \Cref{lemma:always_verified}), $\mathit{tx}$ is certified.

Let us now consider the $\mathtt{admissible}(\mathit{log}_t)$ function (see \Cref{lst:admissible_transaction_set}):
\begin{compactenum}
    \item The check at line~\ref{line:no_cycles} of \Cref{lst:admissible_transaction_set} passes because of the third statement.

    \item The check at line~\ref{line:check_withdrawal_admissible} of \Cref{lst:admissible_transaction_set} passes because of the second statement.
    
    \item Consider a client $c$.
    Let $\mathit{log}_c = \{\mathit{tx} \,|\, \mathit{tx} \in \mathit{log}_t \text{ and } \mathit{tx}.\mathit{issuer} = c\}$.
    Since all transactions that belong to $\mathit{log}_t$ are certified, no conflicting transactions are in $\mathit{log}_c$ (by \Cref{lemma:no_conflicting}).
    Moreover, if a transaction $\mathit{tx} \in \mathit{log}_c$, where $\mathit{tx}.\mathit{sn} > 1$, then $\mathit{tx}' \in \mathit{log}_c$ with $\mathit{tx}'.\mathit{sn} = \mathit{tx}.\mathit{sn} - 1$ (by the first statement from the proof).
    Finally, $\mathtt{admissible\_client\_log}(c, \mathit{log}_c) = \top$ since, for every transaction $\mathit{tx} \in \mathit{log}_c$, there exists a correct server such that $\mathit{state}.\mathit{log}[\mathit{tx}] \neq \bot$ at some time, and \cref{lemma:no_conflicting,lemma:always_verified} hold.
\end{compactenum} 
Therefore, the theorem holds.
\end{proof}

Next, we define when a transaction is \emph{quasi-committed}.
In order to do so, we define the $\mathtt{verify\_quasi\_committed}$ function below.

\begin{lstlisting}[
  caption={Quasi-commitment - verification},
  label={lst:quasi_commitment_verification},
  escapechar=?]
?\textbf{function}? verify_quasi_committed(Transaction ?$\mathit{tx}$?, Set(Message) ?$\mathit{quasi\_committed\_certificate}$?):
    if does not exist View ?$v$? such that ?$m = [$?COMMIT-CONFIRM, ?$\mathit{tx}$?, ?$v$?, View_Path ?$\mathit{path}]$? and?\\$v = \mathit{path}$?.destination(), for every Message ?$m \in \mathit{quasi\_committed\_certificate}$?: ?\label{line:check_quasi}?
        ?\textbf{return}? ?$\bot$?
    ?\textbf{return}? ?$|m.\mathtt{sender} \,|\, m \in \mathit{quasi\_committed\_certificate} \text{ and } m\text{.sender} \in v\text{.members()}| \geq v$?.quorum()
\end{lstlisting}

\begin{definition} [Quasi-Committed Transaction]
We say that a transaction $\mathit{tx}$ is \emph{quasi-committed} if and only if a (correct or faulty) process obtains a set of messages $\mathit{quasi\_committed\_certificate}$ such that $\mathtt{verify\_quasi\_committed}(\mathit{tx}, \mathit{quasi\_committed\_certificate}) = \top$.
Moreover, if the view specified in the $\mathtt{COMMIT-CONFIRM}$ messages of $\mathit{quasi\_committed\_certificate}$ (see \Cref{lst:quasi_commitment_verification}) is $v$, then $\mathit{tx}$ is quasi-committed in $v$.
\end{definition}

Next, we prove that a transaction can be quasi-committed only in a view that is installable at time $\infty$ (similarly to \Cref{lemma:certified_installable}).

\begin{lemma} \label{lemma:quasi_committed_installable}
Let a transaction $\mathit{tx}$ be quasi-committed in a view $v$.
Then, $\alpha_{\infty}(v) = \top$.
\end{lemma}
\begin{proof}
Since $\mathit{tx}$ is quasi-committed in $v$, at least a single correct server $r \in v.\mathtt{members()}$ sends the $\mathtt{COMMIT-CONFIRM}$ message associated with $v$ for $\mathit{tx}$ (see \Cref{lst:quasi_commitment_verification}); note that $v$ is a valid view since every $\mathtt{COMMIT-CONFIRM}$ message is accompanied by a view-path to $v$ (by line~\ref{line:check_quasi} of \Cref{lst:quasi_commitment_verification}).
Since a correct member of $v$ sends the message at line~\ref{line:send_commit_confirm} of \Cref{lst:transaction_module}, $v$ is installed at that correct server (by the check at line~\ref{line:commit_message_rule} of \Cref{lst:transaction_module}).
If $v = \mathit{genesis}$, $\alpha_{\infty}(\mathit{genesis}) = \top$ (by definition) and the lemma holds.
Otherwise, $\alpha_{\infty}(v) = \top$ (similarly to \Cref{lemma:knows_view_system}) and the lemma is concluded.
\end{proof}

The next lemma proves that a quasi-committed transaction is ``carried'' by servers forever.

\begin{lemma} \label{lemma:committed_in_log}
Let a transaction $\mathit{tx}$ be quasi-committed in a view $v$.
Consider time $t$ and a correct server $r \in v'.\mathtt{members()}$ with $\mathit{current\_view}.\mathit{view} = v'$ at time $t$, where $v' \supset v$.
Then, $\mathit{state}.\mathit{log}[\mathit{tx}] \neq \bot$ at server $r$ at time $t$.
\end{lemma}
\begin{proof}
By \Cref{lemma:quasi_committed_installable}, $v$ is installable at time $\infty$ (i.e., $\alpha_{\infty}(v) = \top)$.
Moreover, $v$ is a valid view.

Let time $t^*$ be the first time at which a correct server sets its $\mathit{current\_view}.\mathit{view}$ variable to a view greater than $v$; let us denote that server by $r^*$ and the view by $v^*$.
We now prove that the statement of the lemma holds at server $r^*$ at time $t^*$.

Since $v \subset v^*$, $v^* \neq \mathit{genesis}$ (by \Cref{lemma:genesis_smallest_valid}).
Hence, $r^*$ has executed line~\ref{line:update_current_view} of \Cref{lst:reconfiguration_view_transition}.
This implies that \\$\mathit{reconfiguration}.\mathit{destination} = v^*$ at that time at server $r^*$.
Furthermore, $\mathit{reconfiguration}.\mathit{source} = v$.
Let us prove this claim.
If $\mathit{reconfiguration}.\mathit{source} \subset v$, then this fact conflicts with \Cref{lemma:install_skip_installable}.
If $\mathit{reconfiguration}.\mathit{source} \supset v$, then $t^*$ would not be the first time at which a correct server sets its $\mathit{current\_view}.\mathit{view}$ variable to a view greater than $v$ (since the $\mathtt{INSTALL}$ message $m$ is obtained, where $m.\mathtt{source()} \supset v$).
Finally, $\mathit{reconfiguration}.\mathit{prepared} = \top$.

Since $r^* \in v^*.\mathtt{members}()$, $r^*$ has not executed line~\ref{line:prepared_free} of \Cref{lst:reconfiguration_view_transition}.
Let us investigate both places at which $r^*$ could have set its $\mathit{reconfiguration}.\mathit{prepared}$ variable to $\top$:
\begin{compactitem}
    \item line~\ref{line:prepared_1} of \Cref{lst:reconfiguration_state_transfer}:
    Hence, $r^*$ has received $\mathtt{STATE-UPDATE}$ messages from a quorum of members of $v$ (recall that $\mathit{reconfiguration}.\mathit{source} = v$ at that time).
    By the quorum intersection, we know that there is a correct server $R$ such that (1) $R$ has sent the $\mathtt{COMMIT-CONFIRM}$ message for $\mathit{tx}$ in view $v$, and (2) $r^*$ has received the $\mathtt{STATE-UPDATE}$ message from $R$.
    
    First, we show that $R$ has sent the $\mathtt{COMMIT-CONFIRM}$ message before the $\mathtt{STATE-UPDATE}$ message.
    By contradiction, suppose that $R$ has sent the $\mathtt{STATE-UPDATE}$ message to $r^*$ before the $\mathtt{COMMIT-CONFIRM}$ message.
    At the moment of sending the $\mathtt{STATE-UPDATE}$ message (line~\ref{line:send_state_update} of \Cref{lst:reconfiguration_state_transfer}), we know that $\mathit{current\_view}.\mathit{view}$ at server $R$ is smaller than or equal to $v$ or is equal to $\bot$ (otherwise, $t^*$ would not be the first time at which a correct server sets its $\mathit{current\_view}.\mathit{view}$ variable to a view greater than $v$).
    Hence, the $\mathit{stop\_processing\_until}$ variable (set at line~\ref{line:stop_processing_until} of \Cref{lst:reconfiguration_state_transfer}) is at least equal to $v$ at that time.
    Therefore, the check at line~\ref{line:commit_message_rule} of \Cref{lst:transaction_module} does not pass while $\mathit{current\_view}.\mathit{view} = v$ (since $\mathit{current\_view}.\mathit{processing} = \bot$; see \Cref{lst:reconfiguration_view_transition}), which means that $R$ does not send the $\mathtt{COMMIT-CONFIRM}$ message associated with view $v$.
    Thus, we reach contradiction.
    
    At the moment of sending the $\mathtt{COMMIT-CONFIRM}$ message, $\mathit{state}.\mathit{log}[\mathit{tx}] \neq \bot$ at server $R$ (by lines~\ref{line:check_tx_in_log} and~\ref{line:update_log_state_representation} of \Cref{lst:transaction_module}).
    By \Cref{lemma:log_only_grows}, $\mathit{state}.\mathit{log}[\mathit{tx}] \neq \bot$ at server $R$ at the moment of sending the $\mathtt{STATE-UPDATE}$ message.
    Since $\mathit{state}.\mathtt{verify()} = \top$ at server $r^*$ (by \Cref{lemma:always_verified}), the check at line~\ref{line:verify_state} of \Cref{lst:transaction_module} passes at server $r^*$.
    Therefore, because of the $\mathtt{merge\_logs}$ function, $\mathit{state}.\mathit{log}[\mathit{tx}] \neq \bot$ at server $r^*$ at time $t^*$.
    
    \item line~\ref{line:prepared_2} of \Cref{lst:reconfiguration_state_transfer}:
    This scenario is impossible since it contradicts the fact that $t^*$ is the first time at which a correct server sets its $\mathit{current\_view}.\mathit{view}$ variable to a view greater than $v$.
\end{compactitem}
By \Cref{lemma:log_only_grows}, the statement of the lemma is satisfied at all times after time $t^*$ at server $r^*$.

Next, we introduce the following invariant.
Consider time $t^{**}$ and a server $r^{**}$ such that $\mathit{current\_view}.\mathit{view} = v^{**}$ at server $r^{**}$ at time $t^{**}$, where $v \subset v^{**}$.
Then, the statement of the lemma holds at server $r^{**}$ at time $t^{**}$.
Note that this invariant is satisfied at time $t^*$, which is the first time at which these conditions are satisfied.

Consider a server $r'$ that sets its $\mathit{current\_view}.\mathit{view}$ variable to $v' \supset v$.
Hence, $\mathit{reconfiguration}.\mathit{prepared} = \top$ at server $r'$.
Moreover, $\mathit{reconfiguration}.\mathit{source} \supseteq v$ (otherwise, \Cref{lemma:install_skip_installable} is contradicted).
Let us investigate both places at which $r'$ could have set its $\mathit{reconfiguration}.\mathit{prepared}$ variable to $\top$:
\begin{compactitem}
    \item line~\ref{line:prepared_1} of \Cref{lst:reconfiguration_state_transfer}:
    If $\mathit{reconfiguration}.\mathit{source} \supset v$, we know that a correct process from which $r'$ has received $\mathtt{STATE-UPDATE}$ message has its $\mathit{current\_view}.\mathit{view}$ greater than $v$ at the moment of sending the $\mathtt{STATE-UPDATE}$ message (because of the obtained $\mathtt{INSTALL}$ message $m$, where $m.\mathtt{source()} \supset v$ and the quorum intersection).
    Hence, the invariant is preserved since $\mathit{state}.\mathtt{verify()} = \top$ (by \Cref{lemma:always_verified}), which means that the check at line~\ref{line:verify_state} of \Cref{lst:transaction_module} passes, and the invariant holds for a server from which $r'$ has received the $\mathtt{STATE-UPDATE}$ message (see the $\mathtt{merge\_logs}$ function).
    
    If $\mathit{reconfiguration}.\mathit{source} = v$, then $r'$ has received the $\mathtt{STATE-UPDATE}$ messages from at least $v.\mathtt{plurality}()$ of correct members of $v$.
    If at least a single such member had $\mathit{current\_view}.\mathit{view}$ variable equal to a view greater than $v$ at the moment of sending the $\mathtt{STATE-UPDATE}$ message, the invariant is preserved because the check at line~\ref{line:verify_state} of \Cref{lst:transaction_module} passes at $r'$ (by \Cref{lemma:always_verified}) and the invariant hypothesis.
    Otherwise, the invariant is preserved because of the argument given in the base step.
    
    \item line~\ref{line:prepared_2} of \Cref{lst:reconfiguration_state_transfer}:
    Since $\mathit{reconfiguration}.\mathit{source} \supseteq v$, the invariant is preserved because of the check at line~\ref{line:verify_state} of \Cref{lst:transaction_module} passes at $r'$ (by \Cref{lemma:always_verified}) and the invariant hypothesis.
\end{compactitem}
Since the invariant is always preserved, the lemma holds.
\end{proof}

We say that a transaction $\mathit{tx}$ \emph{depends on} a transaction $\mathit{tx}'$ if and only if:
\begin{compactitem}
    \item $\mathit{tx}.\mathit{issuer} = \mathit{tx}'.\mathit{issuer}$ and $\mathit{tx}.\mathit{sn} = \mathit{tx}'.\mathit{sn} + 1$; or
    
    \item $\mathit{tx}$ is a deposit transaction, $\mathit{tx}'$ is a withdrawal transaction and $\mathit{tx}.\mathit{withdrawal} = \mathit{tx}'$.
\end{compactitem}
Now, given an admissible log, we specify the rank of every transaction that belongs to the log; we denote by $\mathit{rank}_L(\mathit{tx})$ the rank of a transaction $\mathit{tx}$ in an admissible log $L$.
We define the rank in the following manner:
\begin{compactitem}
    \item If $\mathit{tx}$ does not depend on any transaction that belongs to $L$, then $\mathit{rank}_L(\mathit{tx}) = 0$.
    
    \item Otherwise, $\mathit{rank}_L(\mathit{tx}) = \mathit{max}(\mathit{rank}_L(\mathit{tx}_1), \mathit{rank}_L(\mathit{tx}_2)) + 1$, where $\mathit{tx}$ depends on $\mathit{tx}_1 \in L$ and $\mathit{tx}$ depends on $\mathit{tx}_2 \in L$.
\end{compactitem}

Recall that $v_{\mathit{final}}$ is the greatest forever-alive view (see \Cref{definition:forever_alive_view}).
Moreover, we say that a transaction $\mathit{tx}$ is quasi-committed at a correct server $r$ if and only if $\mathit{quasi\_committed}[\mathit{tx}] = \top$ at server $r$.
The next lemma shows that any transaction that belongs to the $\mathit{log}$ variable (recall that $\mathit{log} = \mathit{state}.\mathtt{extract\_log()}$) of a correct member of $v_{\mathit{final}}$ eventually belongs to the $\mathit{log}$ variable of all other correct members of $v_{\mathit{final}}$.

\begin{lemma} \label{lemma:same_logs}
Consider a correct server $r \in v_{\mathit{final}}.\mathtt{members}()$.
Let $\mathit{state}.\mathit{log}[\mathit{tx}] \neq \bot$ at server $r$.
Eventually, $\mathit{state}.\mathit{log}[\mathit{tx}] \neq \bot$ at every correct server $r' \in v_{\mathit{final}}.\mathtt{members()}$.
\end{lemma}
\begin{proof}
Recall that, by the finality property (see \Cref{lst:reconfiguration_properties_new}), the following holds:
\begin{compactenum}
    \item all correct members of $v_{\mathit{final}}$ update their current view to $v_{\mathit{final}}$ (i.e., set their $\mathit{current\_view}.\mathit{view}$ variable to $v_{\mathit{final}}$),
    
    \item no correct member of $v_{\mathit{final}}$ updates its current view to any other view afterwards,
    
    \item all correct members of $v_{\mathit{final}}$ install $v_{\mathit{final}}$,
    
    \item no correct member of $v_{\mathit{final}}$ leaves, and
    
    \item no correct member of $v_{\mathit{final}}$ stops processing in $v_{\mathit{final}}$.
\end{compactenum}
Let $\mathit{log}_r = \{\mathit{tx} \,|\, \mathit{state}.\mathit{log}[\mathit{tx}] \neq \bot \text{ at server } r\}$.
By \Cref{lemma:always_verified}, $\mathit{log}_r$ is admissible (see \Cref{lst:state_representation}).
In the rest of the proof, we consider a particular server $r^* \neq r$, where $r^* \in v_{\mathit{final}}.\mathtt{members()}$.
We prove the lemma by induction.

\smallskip
\underline{Base step:}
We consider all the transactions that belong to $\mathit{log}_r$ with rank $0$.
Let us denote that set of transactions by $\mathit{TX}_0$.

If $r$ broadcasts the $\mathtt{COMMIT}$ message for $\mathit{tx} \in \mathit{TX}_0$ to all members of $v_{\mathit{final}}$ at line~\ref{line:send_commit_1} or line~\ref{line:broadcast_commit_2} of \Cref{lst:transaction_module}, $r^*$ eventually receives the message and sets $\mathit{state}.\mathit{log}[\mathit{tx}]$ to the received $\mathtt{COMMIT}$ message (by the finality property and the fact that $\mathit{tx}$ is certified, by \Cref{lemma:always_verified}).

Otherwise, $\mathit{tx}$ is quasi-committed at server $r$ (since $\mathit{quasi\_committed}[\mathit{tx}] = \top$; line~\ref{line:quasi_committed_check} of \Cref{lst:transaction_module}) at the moment of installing $v_{\mathit{final}}$ (line~\ref{line:upon_install} of \Cref{lst:transaction_module}).
Let $v_{\mathit{tx}}$ be the smallest view in which $\mathit{tx}$ is quasi-committed; note that such view is defined since all views in which $\mathit{tx}$ is quasi-committed are valid (by \Cref{lemma:quasi_committed_installable}) and all valid views are comparable (by the view comparability property; see \Cref{lst:reconfiguration_properties_new}).
By \Cref{lst:transaction_module}, $v_{\mathit{tx}} \subset v_{\mathit{final}}$.
Hence, the base case is satisfied by \Cref{lemma:committed_in_log} and the finality property.

\smallskip
\underline{Inductive step:}
We consider all transactions that belong to $\mathit{log}_r$ with rank $\mathit{rank}$.
Let $\mathit{TX}_{\mathit{rank}}$ denote that set of transactions.
We assume, for all transactions $\mathit{tx}'$ that belong to $\mathit{log}_r$ with the smaller rank, that $\mathit{state}.\mathit{log}[\mathit{tx}'] \neq \bot$ at every correct member of $v_{\mathit{final}}$.
We prove that the invariant is preserved for all transactions that belong to $\mathit{TX}_{\mathit{rank}}$.

If $r$ broadcasts the $\mathtt{COMMIT}$ message for $\mathit{tx} \in \mathit{TX}_{\mathit{rank}}$ to all members of $v_{\mathit{final}}$ at line~\ref{line:send_commit_1} or line~\ref{line:broadcast_commit_2} of \Cref{lst:transaction_module}, $r^*$ eventually receives the message and sets $\mathit{state}.\mathit{log}[\mathit{tx}]$ to the received $\mathtt{COMMIT}$ message (by the finality property, induction hypothesis and the fact that $\mathit{tx}$ is certified, by \Cref{lemma:always_verified}).

Otherwise, $\mathit{tx}$ is quasi-committed at server $r$ (since $\mathit{quasi\_committed}[\mathit{tx}] = \top$; line~\ref{line:quasi_committed_check} of \Cref{lst:transaction_module}) at the moment of installing $v_{\mathit{final}}$ (line~\ref{line:upon_install} of \Cref{lst:transaction_module}).
Let $v_{\mathit{tx}}$ be the smallest view in which $\mathit{tx}$ is quasi-committed; note that such view is defined since all views in which $\mathit{tx}$ is quasi-committed are valid (by \Cref{lemma:quasi_committed_installable}) and all valid views are comparable (by the view comparability property; see \Cref{lst:reconfiguration_properties_new}).
By \Cref{lst:transaction_module}, $v_{\mathit{tx}} \subset v_{\mathit{final}}$.
Hence, the inductive case is satisfied by \Cref{lemma:committed_in_log} and the finality property.
\end{proof}

The next lemma proves that a correct member of $v_{\mathit{final}}$ eventually receives the $\mathtt{COMMIT-CONFIRM}$ messages from all correct members of $v_{\mathit{final}}$.

\begin{lemma} \label{lemma:in_log_received_from_all}
Let a correct server $r \in v_{\mathit{final}}.\mathtt{members}()$ broadcast the $\mathtt{COMMIT}$ message for a transaction $\mathit{tx}$ to all members of $v_{\mathit{final}}$.
Eventually, $\mathit{r}$ receives $\mathtt{COMMIT-CONFIRM}$ message for $\mathit{tx}$ associated with $v_{\mathit{final}}$ from all correct members of $v_{\mathit{final}}$.
\end{lemma}
\begin{proof}
We consider all possible places at which $r$ could broadcast the $\mathtt{COMMIT}$ message:
\begin{compactitem}
    \item line~\ref{line:send_commit_1} of \Cref{lst:transaction_module}:
    Therefore, $\mathit{state}.\mathit{log}[\mathit{tx}] \neq \bot$ at server $r$ (by line~\ref{line:update_log_state_representation} of \Cref{lst:transaction_module}).
    
    \item line~\ref{line:broadcast_commit_2} of \Cref{lst:transaction_module}:
    Again, $\mathit{state}.\mathit{log}[\mathit{tx}] \neq \bot$ at server $r$ (by the check at line~\ref{line:quasi_committed_check} of \Cref{lst:transaction_module}).
\end{compactitem}
Since in both possible cases, $\mathit{state}.\mathit{log}[\mathit{tx}] \neq \bot$ at server $r$, $\mathit{state}.\mathit{log}[\mathit{tx}] \neq \bot$ at every correct server $r' \in v_{\mathit{final}}.\mathtt{members()}$ (by \Cref{lemma:same_logs}).

Once this happens at a correct server $r'$ and $r'$ receives the $\mathtt{COMMIT}$ message sent by $r$ (which happens because of the finality property), $r'$ sends the $\mathtt{COMMIT-CONFIRM}$ message for $\mathit{tx}$ associated with $v_{\mathit{final}}$ at line~\ref{line:send_commit_confirm} of \Cref{lst:transaction_module} (because of the $\mathtt{allowed\_to\_commit\_confirm}$ function; line~\ref{line:tx_in_log} of \Cref{lst:transaction_module}).
Furthermore, the finality property ensures that $r$ eventually receives this $\mathtt{COMMIT-CONFIRM}$ message sent by $r'$, which concludes the proof.
\end{proof}

Next, we prove that a transaction $\mathit{tx}$ is quasi-committed at a correct server $r \in v_{\mathit{final}}$, where $\mathit{state}.\mathit{log}[\mathit{tx}] \neq \bot$ at $r$.

\begin{lemma} \label{lemma:in_log_then_committed}
Consider a correct server $r \in v_{\mathit{final}}.\mathtt{members}()$.
Let $\mathit{state}.\mathit{log}[\mathit{tx}] \neq \bot$ at server $r$.
Then, $r$ eventually quasi-commits $\mathit{tx}$.
\end{lemma}
\begin{proof}
We prove the lemma by contradiction.
Hence, suppose that $r$ does not quasi-commit $\mathit{tx}$.
By \Cref{lemma:same_logs}, $\mathit{state}.\mathit{log}[\mathit{tx}] \neq \bot$ at every correct member of $v_{\mathit{final}}$.

Let $\mathit{log}_r = \{\mathit{tx} \,|\, \mathit{state}.\mathit{log}[\mathit{tx}] \neq \bot \text{ at server } r\}$.
Because of \Cref{lemma:always_verified}, $\mathit{log}_r$ is admissible.
Recall that $\mathit{tx} \in \mathit{log}_r$.

Let $\mathcal{D}(\mathit{tx})$ be the reflexive transitive closure of the ``depends-on'' relation for $\mathit{tx}$ given $\mathit{log}_r$.
Let $\mathit{quasi\_committed}_r$ denote the set of transactions quasi-committed by $r$; note that $\mathit{tx} \notin \mathit{quasi\_committed}_r$ (by the assumption).
Moreover, let $\mathit{quasi\_committed}_r^*$ denote the set of transactions quasi-committed at $r$ for which $r$ has sent the $\mathtt{COMMIT}$ message associated with $v_{\mathit{final}}$; observe that $\mathit{quasi\_committed}_r^* \subseteq \mathit{quasi\_committed}_r$.
By \Cref{lemma:in_log_received_from_all}, $r$ eventually receives $\mathtt{COMMIT-CONFIRM}$ messages for $\mathit{tx}'$ associated with $v_{\mathit{final}}$ from all correct members of $v_{\mathit{final}}$, for all transactions $\mathit{tx}' \in \mathit{quasi\_committed}_r^*$.

Let $\mathit{tx}^* \in \mathcal{D}(\mathit{tx})$ such that (1) $\mathit{tx}^* \notin \mathit{quasi\_committed}_r$, and (2) all dependencies of $\mathit{tx}^*$ are quasi-committed at $r$; note that $\mathit{tx}^*$ indeed exists since $\mathit{tx} \in \mathcal{D}(\mathit{tx})$, $\mathit{tx} \notin \mathit{quasi\_committed}_r$ and $\mathit{log}_r$ is a DAG (since it is admissible).

Since $\mathit{tx}^* \notin \mathit{quasi\_committed}_r$, $r$ broadcasts the $\mathtt{COMMIT}$ message for $\mathit{tx}^*$ to all members of $v_{\mathit{final}}$ (by the finality property; lines~\ref{line:send_commit_1} or~\ref{line:broadcast_commit_2} of \Cref{lst:transaction_module}).
Eventually, $r$ receives $\mathtt{COMMIT-CONFIRM}$ messages for $\mathit{tx}^*$ associated with $v_{\mathit{final}}$ from all correct members of $v_{\mathit{final}}$ (by \Cref{lemma:in_log_received_from_all}).

Recall that all dependencies of $\mathit{tx}^*$ are quasi-committed by $r$.
Therefore, once (1) all dependencies of $\mathit{tx}^*$ are quasi-committed by $r$, (2) $\mathtt{COMMIT-CONFIRM}$ messages for $\mathit{tx}^*$ are received from all correct members of $v_{\mathit{final}}$, and (3) $\mathtt{COMMIT-CONFIRM}$ messages for $\mathit{tx}'$ are received from all correct members of $v_{\mathit{final}}$, for every $\mathit{tx}' \in \mathit{quasi\_committed}_r^*$ (note that $\mathit{quasi\_committed\_current\_view} \subseteq \mathit{quasi\_committed}_r^*$), $r$ quasi-commits $\mathit{tx}^*$ (since the $\mathtt{allowed\_to\_quasi\_commit}$ function returns $\top$).
Thus, we reach contradiction and $\mathit{tx}$ is quasi-committed at $r$.
\end{proof}

The next lemma proves that a correct server $r \in v_{\mathit{final}}.\mathtt{members()}$ eventually obtains a commitment proof for $\mathit{tx}$ if $\mathit{tx}$ is quasi-committed at $r$.

\begin{lemma} \label{lemma:quasi_committed_proof}
Consider a correct server $r \in v_{\mathit{final}}.\mathtt{members()}$.
Let $\mathit{tx}$ be quasi-committed at $r$.
Then, $r$ eventually obtains a commitment proof $\sigma_c$ such that $\mathtt{verify\_commit}(\mathit{tx}, \sigma_c) = \top$.
\end{lemma}
\begin{proof}
Since $\mathit{tx}$ is quasi-committed at $r$, $\mathit{state}.\mathit{log}[\mathit{tx}] \neq \bot$ at $r$ (because of the $\mathtt{allowed\_to\_quasi\_commit}$ function and \Cref{lemma:log_only_grows}; line~\ref{line:allowed_to_quasi_log} of \Cref{lst:transaction_module}).
Let $\mathit{log}_r = \{\mathit{tx} \,|\, \mathit{state}.\mathit{log}[\mathit{tx}] \neq \bot \text{ at server } r\}$.
Because of \Cref{lemma:always_verified}, $\mathit{log}_r$ is admissible.
Recall that $\mathit{tx} \in \mathit{log}_r$.

By \Cref{lemma:same_logs}, $\mathit{state}.\mathit{log}[\mathit{tx}'] \neq \bot$ at every correct member of $v_{\mathit{final}}$, for every $\mathit{tx}' \in \mathit{log}_r$.
Finally, by \Cref{lemma:in_log_then_committed}, $\mathit{tx}'$ is quasi-committed at every correct member of $v_{\mathit{final}}$, for every $\mathit{tx}' \in \mathit{log}_r$.

Let $\mathcal{D}(\mathit{tx})$ be the reflexive transitive closure of the ``depends-on'' relation for $\mathit{tx}$ given $\mathit{log}_r$; note that $\mathcal{D}(\mathit{tx}) \subseteq \mathit{log}_r$.
We prove the lemma by induction.

\smallskip
\underline{Base step:}
We consider all the transactions that belong to $\mathcal{D}(\mathit{tx})$ with rank $0$.
Let us denote that set of transactions by $\mathit{TX}_0$.

Eventually, every correct member of $v_{\mathit{final}}$ broadcasts the $\mathtt{COMMITTED}$ message for $\mathit{tx} \in \mathit{TX}_0$ to members of $v_{\mathit{final}}$ (because of the $\mathtt{allowed\_to\_broadcast\_committed}$ function and the finality property).
The finality property ensures that every correct member of $v_{\mathit{final}}$ receives the set of $v_{\mathit{final}}.\mathtt{plurality()}$ messages (at line~\ref{line:committed} of \Cref{lst:transaction_module}), which constitutes a commitment proof $\sigma_c$ for $\mathit{tx}$ (see \Cref{lst:commitment_proof_verification}).
Moreover, $\mathit{confirmed}[\mathit{tx}] = \top$ at every correct member of $v_{\mathit{final}}$.
Therefore, the base step holds.

\smallskip
\underline{Inductive step:}
We consider all transactions that belong to $\mathcal{D}(\mathit{tx})$ with rank $\mathit{rank}$.
Let $\mathit{TX}_{\mathit{rank}}$ denote that set of transactions.
We assume, for all transactions $\mathit{tx}'$ that belong to $\mathcal{D}(\mathit{tx})$ with the smaller rank, that $\mathit{confirmed}[\mathit{tx}] = \top$ at every correct member of $v_{\mathit{final}}$.
We prove that the invariant is preserved for all transactions that belong to $\mathit{TX}_{\mathit{rank}}$.

Eventually, every correct member of $v_{\mathit{final}}$ broadcasts the $\mathtt{COMMITTED}$ message for $\mathit{tx} \in \mathit{TX}_{\mathit{rank}}$ to members of $v_{\mathit{final}}$ (because of the $\mathtt{allowed\_to\_broadcast\_committed}$ function, the inductive hypothesis and the finality property).
The finality property ensures that every correct member of $v_{\mathit{final}}$ receives the set of $v_{\mathit{final}}.\mathtt{plurality()}$ messages (at line~\ref{line:committed} of \Cref{lst:transaction_module}), which constitutes a commitment proof $\sigma_c$ for $\mathit{tx}$ (see \Cref{lst:commitment_proof_verification}).
Moreover, $\mathit{confirmed}[\mathit{tx}] = \top$ at every correct member of $v_{\mathit{final}}$.
Therefore, the inductive step holds, as well.
Thus, the proof is concluded.
\end{proof}

Finally, we are ready to prove the commitment validity property.

\begin{theorem} [Commitment Validity] \label{theorem:commitment_validity}
Commitment validity is satisfied.
\end{theorem}
\begin{proof}
Consider a transaction $\mathit{tx}$ issued by a forever-correct client which is not committed.
All transactions on which $\mathit{tx}$ depends are committed (see \Cref{subsection:core_problem}, paragraph ``Rules'').
Let $\mathit{tx}$ be dependent on $\mathit{tx}'$; as already mentioned, $\mathit{tx}'$ is committed.
Therefore, $\mathit{tx}'$ is quasi-committed (by the $\mathtt{allowed\_to\_broadcast\_committed}$ function; line~\ref{line:quasi_committed} of \Cref{lst:transaction_module}).
Let $v_{\mathit{tx}'}$ be the smallest view in which $\mathit{tx}'$ is quasi-committed.
We distinguish two cases:
\begin{compactitem}
    \item Let $v_{\mathit{tx}'} \subset v_{\mathit{final}}$.
    Therefore, $\mathit{state}.\mathit{log}[\mathit{tx}'] \neq \bot$ at every correct member of $v_{\mathit{final}}$ (by \Cref{lemma:committed_in_log}).
    
    \item Let $v_{\mathit{tx}'} = v_{\mathit{final}}$.
    Hence, there exists a correct member of $v_{\mathit{final}}$ that sends the $\mathtt{COMMIT-CONFIRM}$ message for $\mathit{tx}'$ associated with $v_{\mathit{final}}$.
    Thus, $\mathit{state}.\mathit{log}[\mathit{tx}'] \neq \bot$ at that server (by line~\ref{line:update_log_state_representation} of \Cref{lst:transaction_module}).
    By \Cref{lemma:same_logs}, $\mathit{state}.\mathit{log}[\mathit{tx}'] \neq \bot$ at every correct member of $v_{\mathit{final}}$.
\end{compactitem}
In both cases, $\mathit{state}.\mathit{log}[\mathit{tx}'] \neq \bot$ at every correct member of $v_{\mathit{final}}$.
By \Cref{lemma:in_log_then_committed}, $\mathit{tx}'$ is quasi-committed at every correct member of $v_{\mathit{final}}$.
Finally, by \Cref{lemma:quasi_committed_proof}, every correct member of $v_{\mathit{final}}$ obtains a commitment proof for $\mathit{tx}'$ (and sets $\mathit{confirmed}[\mathit{tx}'] = \top$).

Since $\mathit{tx}.\mathit{issuer}$ is forever-correct, the client eventually learns about $v_{\mathit{final}}$ (i.e., $v_{\mathit{final}} \in \mathit{history}$ at the client).
Once that happens, the client sends the $\mathtt{PREPARE}$ message for $\mathit{tx}$ to all members of $v_{\mathit{final}}$ (at lines~\ref{line:prepare_1_c} or~\ref{line:prepare_2_c} of \Cref{lst:client}).
Since $\mathit{state}.\mathit{log}[\mathit{tx}'] \neq \bot$ at every correct member of $v_{\mathit{final}}$, for every dependency $\mathit{tx}'$ of $\mathit{tx}$, all correct members of $v_{\mathit{final}}$ send the $\mathtt{ACK}$ message to the client (at line~\ref{line:send_ack} of \Cref{lst:transaction_module}).
Hence, the client eventually obtains a transaction certificate for $\mathit{tx}$.

Similarly, the client sends the $\mathtt{COMMIT}$ message to all members of $v_{\mathit{final}}$ (at lines~\ref{line:commit_c_1} or~\ref{line:commit_c_2} of \Cref{lst:client}).
Therefore, every correct member of $v_{\mathit{final}}$ eventually receives the $\mathtt{COMMIT}$ message (since $\mathit{state}.\mathit{log}[\mathit{tx}'] \neq \bot$ at every correct member of $v_{\mathit{final}}$, for every dependency $\mathit{tx}'$ of $\mathit{tx}$).
Finally, this means that $\mathit{log}.\mathit{state}[\mathit{tx}] \neq \bot$ at every correct member of $v_{\mathit{final}}$ (by line~\ref{line:update_log_state_representation} of \Cref{lst:transaction_module}).
By \Cref{lemma:in_log_then_committed}, $\mathit{tx}$ is quasi-committed at every correct member of $v_{\mathit{final}}$.
Furthermore, every correct member of $v_{\mathit{final}}$ eventually obtains a commitment proof for $\mathit{tx}$ (by \Cref{lemma:quasi_committed_proof}), which implies that $\mathit{tx}$ is committed.
Since we reach contradiction, our starting assumption was not correct, which means that $\mathit{tx}$ is committed.
The theorem holds.
\end{proof}

Next, we prove the commitment learning property.

\begin{theorem} [Commitment Learning] \label{theorem:commitment_learning}
Commitment learning is satisfied.
\end{theorem}
\begin{proof}
Let a transaction $\mathit{tx}$ be committed.
Therefore, $\mathit{tx}$ is quasi-committed (by the $\mathtt{allowed\_to\_broadcast\_committed}$ function; line~\ref{line:quasi_committed} of \Cref{lst:transaction_module}).
Let $v_{\mathit{tx}}$ be the smallest view in which $\mathit{tx}'$ is quasi-committed.
We distinguish two cases:
\begin{compactitem}
    \item Let $v_{\mathit{tx}} \subset v_{\mathit{final}}$.
    Therefore, $\mathit{state}.\mathit{log}[\mathit{tx}] \neq \bot$ at every correct member of $v_{\mathit{final}}$ (by \Cref{lemma:committed_in_log}).
    
    \item Let $v_{\mathit{tx}} = v_{\mathit{final}}$.
    Hence, there exists a correct member of $v_{\mathit{final}}$ that sends the $\mathtt{COMMIT-CONFIRM}$ message for $\mathit{tx}$ associated with $v_{\mathit{final}}$.
    Thus, $\mathit{state}.\mathit{log}[\mathit{tx}] \neq \bot$ at that server (by line~\ref{line:update_log_state_representation} of \Cref{lst:transaction_module}).
    By \Cref{lemma:same_logs}, $\mathit{state}.\mathit{log}[\mathit{tx}] \neq \bot$ at every correct member of $v_{\mathit{final}}$.
\end{compactitem}
In both cases, $\mathit{state}.\mathit{log}[\mathit{tx}] \neq \bot$ at every correct member of $v_{\mathit{final}}$.
By \Cref{lemma:in_log_then_committed}, $\mathit{tx}$ is quasi-committed at every correct member of $v_{\mathit{final}}$.
Finally, by \Cref{lemma:quasi_committed_proof}, every correct member of $v_{\mathit{final}}$ obtains a commitment proof for $\mathit{tx}$ (and sets $\mathit{confirmed}[\mathit{tx}'] = \top$).
Once that happens, a correct member of $v_{\mathit{final}}$ sends the commitment proof to $\mathit{tx}.\mathit{issuer}$ (at line~\ref{line:proof_to_issuer} of \Cref{lst:transaction_module}) and, if $\mathit{tx}$ is a withdrawal transaction, to $\mathit{tx}.\mathit{receiver}$ (at line~\ref{line:proof_to_receiver} of \Cref{lst:transaction_module}).
Since no correct member of $v_{\mathit{final}}$ leaves (by the finality property), the appropriate clients eventually obtain commitment proofs (at line~\ref{line:obtain_proof_c} of \Cref{lst:client}) and the theorem holds.
\end{proof}

The last property of \sysname we need to prove is the commitment signing property (see \Cref{subsection:carbon_properties}, paragraph ``Properties of \sysname'').
In order to show that this property holds, we first prove that an issued transaction cannot be quasi-committed in ``stale'' views.

\begin{lemma} \label{lemma:no_quasi_committed_below_max}
Let a transaction $\mathit{tx}$ be issued at time $t$.
Let $\mathcal{V}(t) = \{v \,|\, \mathit{current\_view}.\mathit{view} = v \neq \bot \text{ at a correct server at time } t\}$.
If $\mathcal{V}(t) \neq \emptyset$ and $v_{\mathit{max}}$ is the greatest view of $\mathcal{V}(t)$, then $\mathit{tx}$ is not quasi-committed in a view $v$, for any view $v \subset v_{\mathit{max}}$.
\end{lemma}
\begin{proof}
Note that $\mathcal{V}(t)$ might contain values of the $\mathit{current\_view}.\mathit{view}$ variable of correct servers that halted by time $t$ (note that $\mathit{current\_view}.\mathit{view}$ is not modified upon leaving; see \Cref{lst:reconfiguration_view_transition}).
We prove the lemma by contradiction.
Therefore, let $\mathit{tx}$ be quasi-committed in a view $v \subset v_{\mathit{max}}$.

First, note that no correct server obtains $\mathit{tx}$ before time $t$ (due to the definition of ``issued at time $t$''; see \Cref{subsection:carbon_properties}, paragraph ``Properties of \sysname'').
Therefore, no correct server sends the $\mathtt{COMMIT-CONFIRM}$ message for $\mathit{tx}$ before $t$.
Moreover, since $\mathit{tx}$ is quasi-committed in $v$ and \Cref{lemma:quasi_committed_installable} holds, $v$ is a view installable at time $\infty$.

Let $r_{\mathit{first}}$ be the first correct server to set its $\mathit{current\_view}.\mathit{view}$ variable to a view greater than $v$; let that view be $v_{\mathit{first}} \supset v$.
Observe that $r_{\mathit{first}}$ sets its $\mathit{current\_view}.\mathit{view}$ variable to $v_{\mathit{first}}$ by time $t$.
Since $v \subset v_{\mathit{first}}$, $v_{\mathit{first}} \neq \mathit{genesis}$ (by \Cref{lemma:genesis_smallest_valid}).
Therefore, before updating its $\mathit{current\_view}.\mathit{view}$ variable to $v_{\mathit{first}}$, $r_{\mathit{first}}$ has received $\mathtt{STATE-UPDATE}$ messages from the quorum of members of $v$ (the rule at line~\ref{line:prepared_from_source} of \Cref{lst:reconfiguration_state_transfer} becomes active and $\mathit{reconfiguration}.\mathit{source} = v$ at that time at $r_{\mathit{first}}$, by \Cref{lemma:install_skip_installable}).
Each such $\mathtt{STATE-UPDATE}$ message is sent before time $t$ at line~\ref{line:send_state_update} of \Cref{lst:reconfiguration_state_transfer}.
Therefore, before sending the $\mathtt{STATE-UPDATE}$ message (i.e., before time $t$), the $\mathit{stop\_processing\_until}$ variable is equal to (at least) $v$ at each correct server that sends the $\mathtt{STATE-UPDATE}$ message received by $r_{\mathit{first}}$.

Furthermore, since $\mathit{tx}$ is quasi-committed in $v$, at least $v.\mathtt{plurality()}$ of correct members of $v$ have sent the $\mathtt{COMMIT-CONFIRM}$ message for $\mathit{tx}$ associated with $v$ (at line~\ref{line:send_commit_confirm} of \Cref{lst:transaction_module}); each such message is sent at some time greater than or equal to $t$.
Therefore, there exists a correct member of $v$ that (1) sets its $\mathit{stop\_processing\_until}$ variable to (at least) $v$ before time $t$, and (2) sends the $\mathtt{COMMIT-CONFIRM}$ message for $\mathit{tx}$ associated with $v$ at time $t$ (or later).
Such behavior is not a correct one (see \Cref{lst:reconfiguration_view_transition}).
Hence, we reach contradiction and $\mathit{tx}$ is not quasi-committed in $v$.
The lemma holds.
\end{proof}

The direct consequence of \Cref{lemma:no_quasi_committed_below_max} is that no correct server sends the $\mathtt{COMMITTED}$ message associated with a view smaller than $v_{\mathit{max}}$.
Let us prove this claim.

\begin{lemma} \label{lemma:no_committed_below_max}
Let a transaction $\mathit{tx}$ be issued at time $t$.
Let $\mathcal{V}(t) = \{v \,|\, \mathit{current\_view}.\mathit{view} = v \neq \bot \text{ at a correct server at time } t\}$.
If $\mathcal{V}(t) \neq \emptyset$ and $v_{\mathit{max}}$ is the greatest view of $\mathcal{V}(t)$, then no correct server sends the $\mathtt{COMMITTED}$ message for $\mathit{tx}$ associated with a view $v$, for any view $v \subset v_{\mathit{max}}$.
\end{lemma}
\begin{proof}
By contradiction, let a correct server $r$ send the $\mathtt{COMMITTED}$ message for $\mathit{tx}$ associated with a view $v \subset v_{\mathit{max}}$.
That means that $\mathit{tx}$ is quasi-committed at server $r$ (by the $\mathtt{allowed\_to\_broadcast\_committed}$ function; line~\ref{line:quasi_committed} of \Cref{lst:transaction_module}).
Since $\mathit{current\_view}.\mathit{view} = v$ at the moment of sending the $\mathtt{COMMITTED}$ message (by line~\ref{line:broadcast_committed} of \Cref{lst:transaction_module}), $\mathit{tx}$ is quasi-committed in a view smaller than or equal to $v$.
Since $v \subset v_{\mathit{max}}$, this is not possible due to \Cref{lemma:no_quasi_committed_below_max}.
Therefore, $r$ does not send the $\mathtt{COMMITTED}$ message for $\mathit{tx}$ associated with $v \subset v_{\mathit{max}}$ and the lemma holds.
\end{proof}

The next theorem shows that the commitment signing property is satisfied.

\begin{theorem} [Commitment Signing] \label{theorem:commitment_signing}
Commitment signing is satisfied.
\end{theorem}
\begin{proof}
Let a transaction $\mathit{tx}$ be issued at time $t$ and let a commitment proof $\sigma_c$ be obtained at time $t' \geq t$, where $\mathtt{verify\_commit}(\mathit{tx}, \sigma_c) = \top$.
Since $\mathit{tx}$ is issued at time $t$, no correct server sends the $\mathtt{COMMITTED}$ message for $\mathit{tx}$ before time $t$ (follows from the definition of ``issued at time $t$''; see \Cref{subsection:carbon_properties}, paragraph ``Properties of \sysname'').

Consider a correct server $r$.
We distinguish two cases:
\begin{compactenum}
    \item Let $r$ leave before time $t$.
    Since $r$ does not send any $\mathtt{COMMITTED}$ message for $\mathit{tx}$ before time $t$ and $r$ halts immediately after leaving (which happens before time $t$), $r \notin \sigma_c.\mathtt{signers}$.
    
    \item Let $r$ join after time $t'$.
    All the $\mathtt{COMMITTED}$ messages that belong to $\sigma_c$ are sent by time $t'$.
    Since $r$ does not send any $\mathtt{COMMITTED}$ messages before joining (see \Cref{lst:transaction_module}; $r$ joins after time $t'$), $r \notin \sigma_c.\mathtt{signers}$.
\end{compactenum}

Finally, consider a faulty server $r$.
Again, we distinguish two cases:
\begin{compactenum}
    \item Let $r$ leave before time $t$.
    Therefore, there exists a correct server $r^*$ that triggers the special $r \text{ } \mathtt{left}$ event (at line~\ref{line:leave_others} of \Cref{lst:reconfiguration_view_transition}) before time $t$.
    
    Let $\mathcal{V}(t) = \{v \,|\, \mathit{current\_view}.\mathit{view} = v \neq \bot \text{ at a correct server at time } t\}$.
    First, $\mathcal{V}(t) \neq \emptyset$ because of $r^*$ and $r \notin v_{\mathit{max}}.\mathtt{members()}$, where $v_{\mathit{max}}$ is the greatest view of $\mathcal{V}(t)$.
    Moreover, for any view $v$ such that $r \in v.\mathtt{members()}$, $v \subset v_{\mathit{max}}$.
    By \Cref{lemma:no_committed_below_max}, no correct server sends the $\mathtt{COMMITTED}$ message for $\mathit{tx}$ associated with any view smaller than $v_{\mathit{max}}$.
    Hence, $r \notin \sigma_c.\mathtt{signers}$.
    
    \item Let $r$ join after time $t'$.
    All the $\mathtt{COMMITTED}$ messages that belong to $\sigma_c$ are sent by time $t'$.
    Let $r^*$ be a correct server that sends the $\mathtt{COMMITTED}$ message $m \in \sigma_c$; such correct server indeed exists due to \Cref{lst:commitment_proof_verification}.
    Let that message be associated with view $v$.
    
    Assume that $r \in v.\mathtt{members()}$.
    If $v = \mathit{genesis}$, then $r$ would have joined by time $t'$ because $r^*$ would have triggered the special $r \text{ } \mathtt{joined}$ event at line~\ref{line:trigger_for_others_init} of \Cref{lst:reconfiguration_initialization}.
    Hence, $v \neq \mathit{genesis}$.
    Therefore, $r^*$ have entered the pseudocode at line~\ref{line:updated_discovery_exists} of \Cref{lst:reconfiguration_view_transition} with the ``source'' view being $v_s$.
    If $r \notin v_s.\mathtt{members()}$, $r$ would have joined by time $t'$ (by line~\ref{line:trigger_for_others_2} of \Cref{lst:reconfiguration_view_transition} executed by $r^*$).
    Hence, $r \in v_s.\mathtt{members()}$ and we reach the same point as we reached with view $v$.
    Eventually, the recursion stops since $\mathit{genesis}$ is the smallest valid view (by \Cref{lemma:genesis_smallest_valid}).
    Therefore, $r \in \mathit{genesis}.\mathtt{members()}$, which means that $r$ joins before time $t'$.
    Thus, we reach contradiction and $r \notin v.\mathtt{members()}$, which implies that $r \notin \sigma_c.\mathtt{signers}$.
\end{compactenum}
The theorem holds and commitment signing is satisfied.
\end{proof}

Next, we prove the query validity property.

\begin{theorem} [Query Validity] \label{theorem:query_validity}
Query validity is satisfied.
\end{theorem}
\begin{proof}
By contradiction, suppose that the query validity property is not satisfied.
Since the client is forever-correct, the client eventually learns about $v_{\mathit{final}}$ (i.e., $v_{\mathit{final}} \in \mathit{history}$ at the client) and it broadcasts the $\mathtt{QUERY}$ message to members of $v_{\mathit{final}}$ (line~\ref{line:query_1} or~\ref{line:query_2} of \Cref{lst:client}).
By the finality property (see \Cref{lst:reconfiguration_properties_new}), every correct member of $v_{\mathit{final}}$ eventually installs $v_{\mathit{final}}$, does not update its current view afterwards and does not leave.
Therefore, every correct member of $v_{\mathit{final}}$ eventually receives the $\mathtt{QUERY}$ message (at line~\ref{line:receive_query} of \Cref{lst:transaction_module}) and responds with the $\mathtt{QUERY-RESPONSE}$ message (at line~\ref{line:send_query} of \Cref{lst:transaction_module}).
Eventually, the client receives the $\mathtt{QUERY-RESPONSE}$ messages from all correct members of $v_{\mathit{final}}$ (i.e., the rule at line~\ref{line:query_rule} of \Cref{lst:client} becomes active) and the client learns the total amount of money in \sysname.
Thus, the theorem holds.
\end{proof}

The next theorem proves the query safety property.

\begin{theorem} [Query Safety] \label{theorem:query_safety}
Query safety is satisfied.
\end{theorem}
\begin{proof}
Since a correct client learns the total amount of money, the client has received the $\mathtt{QUERY-RESPONSE}$ messages from $v.\mathtt{quorum()}$ of members of a valid view $v$ (by the rule at line~\ref{line:query_rule} of \Cref{lst:client}).
Let the client learn that the total amount of money is $X$.
Therefore, there exists a correct server that has sent the $\mathtt{QUERY-RESPONSE}$ message for a value $X' \geq X$.
Let that server be $r \in v.\mathtt{members()}$.

Since $r$ ``calculates'' the total amount of money by ``looking into'' all the transactions it has quasi-committed (by lines~\ref{line:take_quasi_committed} and~\ref{line:send_query} of \Cref{lst:transaction_module}), there exists a set $\mathit{mints}$ of quasi-committed minting transactions such that $X' = \sum\limits_{\mathit{tx} \in \mathit{mints}} \mathit{tx}.\mathit{amount}$.
For each transaction $\mathit{tx} \in \mathit{mints}$, let $v_{\mathit{tx}}$ be the smallest view in which $\mathit{tx}$ is quasi-committed.
We distinguish two cases:
\begin{compactitem}
    \item Let $v_{\mathit{tx}} \subset v_{\mathit{final}}$.
    Therefore, $\mathit{state}.\mathit{log}[\mathit{tx}] \neq \bot$ at every correct member of $v_{\mathit{final}}$ (by \Cref{lemma:committed_in_log}).
    
    \item Let $v_{\mathit{tx}} = v_{\mathit{final}}$.
    Hence, there exists a correct member of $v_{\mathit{final}}$ that sends the $\mathtt{COMMIT-CONFIRM}$ message for $\mathit{tx}$ associated with $v_{\mathit{final}}$.
    Thus, $\mathit{state}.\mathit{log}[\mathit{tx}] \neq \bot$ at that server (by line~\ref{line:update_log_state_representation} of \Cref{lst:transaction_module}).
    By \Cref{lemma:same_logs}, $\mathit{state}.\mathit{log}[\mathit{tx}] \neq \bot$ at every correct member of $v_{\mathit{final}}$.
\end{compactitem}
In both cases, $\mathit{state}.\mathit{log}[\mathit{tx}] \neq \bot$ at every correct member of $v_{\mathit{final}}$, for every transaction $\mathit{tx} \in \mathit{mints}$.
By \Cref{lemma:in_log_then_committed}, $\mathit{tx}$ is quasi-committed at every correct member of $v_{\mathit{final}}$, for every $\mathit{tx} \in \mathit{mints}$.
Finally, by \Cref{lemma:quasi_committed_proof}, every correct member of $v_{\mathit{final}}$ obtains a commitment proof for $\mathit{tx}$, which means that $\mathit{tx}$ is committed, for every $\mathit{tx} \in \mathit{mints}$.
Hence, the theorem holds.
\end{proof}

Finally, we prove the query liveness property.

\begin{theorem} [Query Liveness] \label{theorem:query_liveness}
Query liveness is satisfied.
\end{theorem}
\begin{proof}
Let $X = \sum\limits_{\mathit{tx} \in \mathit{mints}} \mathit{tx}.\mathit{amount}$.
Since $\mathit{tx}$ is committed, for every $\mathit{tx} \in \mathit{mints}$, $\mathit{tx}$ is quasi-committed (by the \\$\mathtt{allowed\_to\_broadcast\_committed}$ function; line~\ref{line:quasi_committed} of \Cref{lst:transaction_module}).
For each transaction $\mathit{tx} \in \mathit{mints}$, let $v_{\mathit{tx}}$ be the smallest view in which $\mathit{tx}$ is quasi-committed.
We distinguish two cases:
\begin{compactitem}
    \item Let $v_{\mathit{tx}} \subset v_{\mathit{final}}$.
    Therefore, $\mathit{state}.\mathit{log}[\mathit{tx}] \neq \bot$ at every correct member of $v_{\mathit{final}}$ (by \Cref{lemma:committed_in_log}).
    
    \item Let $v_{\mathit{tx}} = v_{\mathit{final}}$.
    Hence, there exists a correct member of $v_{\mathit{final}}$ that sends the $\mathtt{COMMIT-CONFIRM}$ message for $\mathit{tx}$ associated with $v_{\mathit{final}}$.
    Thus, $\mathit{state}.\mathit{log}[\mathit{tx}] \neq \bot$ at that server (by line~\ref{line:update_log_state_representation} of \Cref{lst:transaction_module}).
    By \Cref{lemma:same_logs}, $\mathit{state}.\mathit{log}[\mathit{tx}] \neq \bot$ at every correct member of $v_{\mathit{final}}$.
\end{compactitem}
In both cases, $\mathit{state}.\mathit{log}[\mathit{tx}] \neq \bot$ at every correct member of $v_{\mathit{final}}$, for every transaction $\mathit{tx} \in \mathit{mints}$.
By \Cref{lemma:in_log_then_committed}, $\mathit{tx}$ is quasi-committed at every correct member of $v_{\mathit{final}}$, for every $\mathit{tx} \in \mathit{mints}$.
Therefore, eventually all correct members of $v_{\mathit{final}}$ would send the $\mathtt{QUERY-RESPONSE}$ message (at line~\ref{line:send_query} of \Cref{lst:transaction_module}) for a value $X' \geq X$.

Let the client request to learn the total amount of money after this happens and after all correct members of $v_{\mathit{final}}$ have set their $\mathit{current\_view}.\mathit{view}$ variable to $v_{\mathit{final}}$ (which does happen due to the finality property).
Therefore, the client eventually receives (at most) $v_{\mathit{final}}.\mathtt{plurality}() - 1$ values (through the $\mathtt{QUERY-RESPONSE}$ messages) smaller than $X$.
Hence, the $\mathtt{median}$ function (at line~\ref{line:median} of \Cref{lst:client}) returns a value greater than or equal to $X$.
Thus, the query liveness property is satisfied.
\end{proof}

Hence, all properties of \sysname specified in \Cref{subsection:carbon_properties} are satisfied.

\begin{corollary}
All properties of \sysname specified in \Cref{subsection:carbon_properties} are satisfied.
\end{corollary}
\begin{proof}
We list all the properties of \sysname and their corresponding proofs:
\begin{compactitem}
    \item Commitment validity follows from \Cref{theorem:commitment_validity}.
    
    \item Commitment integrity follows from \Cref{theorem:commitment_integrity}.
    
    \item Commitment learning follows from \Cref{theorem:commitment_learning}.
    
    \item Commitment admissibility follows from \Cref{theorem:commitment_admissibility}.
    
    \item Commitment signing follows from \Cref{theorem:commitment_signing}.
    
    \item Query validity follows from \Cref{theorem:query_validity}.
    
    \item Query safety follows from \Cref{theorem:query_safety}.
    
    \item Query liveness follows from \Cref{theorem:query_liveness}.
    
    \item Join safety follows from \Cref{theorem:join_safety}.
    
    \item Leave safety follows from \Cref{theorem:leave_safety}.
    
    \item Join liveness follows from \Cref{theorem:join_liveness}.
    
    \item Leave liveness follows from \Cref{theorem:leave_liveness}.
    
    \item Removal liveness follows from \Cref{theorem:removal_liveness}.
\end{compactitem}
\end{proof}

%% file: appendix/voting.tex
\section{Server's Modules: Voting Module} \label{appx:voting}

Finally, we present the last module of a server: the \emph{voting module}.
This module has a responsibility of ensuring the properties of the asynchronous stake-based voting (see \Cref{subsection:balance_based_voting}).

First, we define when $\mathtt{verify\_voting}(\mathit{mot}, \sigma_v)$ returns $\top$ (see \Cref{lst:voting_verification}), where $\sigma_v$ is a voting proof.
Then, we present the implementation of the module.
Lastly, we prove that the voting liveness and voting safety properties are satisfied.

\begin{lstlisting}[
  caption={The $\mathtt{verify\_voting}$ function},
  label={lst:voting_verification},
  escapechar=?]
?\textbf{function}? verify_voting(Motion ?$\mathit{mot}$?, Voting_Proof ?$\sigma_v$?):
    if ?$\sigma_v$? is not Set(Message):
        ?\textbf{return}? ?$\bot$?
    
    if does not exist View ?$v$? such that ?$m = [$?SUPPORT, ?$\mathit{mot}$?, ?$v$?, View_Path ?$\mathit{path}]$? and ?$v = \mathit{path}$?.destination(), for every Message ?$m \in \sigma_v$?:
        ?\textbf{return}? ?$\bot$?
        
    ?\textbf{return}? ?$|m.\mathtt{sender} \,|\, m \in \sigma_v \text{ and } m\text{.sender} \in v\text{.members()}| \geq v$?.quorum()
\end{lstlisting}
Intuitively, a voting proof is ``valid'' for a motion if and only if a quorum of members of a view claim that they ``support'' the motion to pass.

\para{Voting module - implementation}
Next, we present the implementation of the voting module of a server.

\begin{lstlisting}[
  caption={Voting module - implementation},
  label={lst:voting_module},
  escapechar = ?]
?\textbf{Voting Module:}?
    ?\textcolor{plainorange}{Implementation:}?
        upon start: // initialization of the module; executed as soon as the ?\textcolor{gray}{$\mathtt{start}$}? event is triggered
            Map(Motion, View ?$\to$? Set(Server)) ?$\mathit{support\_from} = \{\mathit{mot}, v \to \emptyset\text{, for every Motion } \mathit{mot} \text{ and every View } v\}$?
            Map(Motion, View ?$\to$? Set(Message)) ?$\mathit{supports} = \{\mathit{mot}, v \to \emptyset\text{, for every Motion } \mathit{mot} \text{ and every View } v\}$?
    
        upon ?exists? Motion ?$\mathit{mot}$? such that ?$\mathit{current\_view}.\mathit{installed} = \top$? and ?$\mathit{current\_view}.\mathit{processing} = \top$? and support(?$\mathit{log}$?, ?$\mathit{mot}$?) ?$ = \top$?: // see ?\textcolor{gray}{\Cref{subsection:balance_based_voting}}? ?\label{line:rule_support}?
            broadcast ?$[$?SUPPORT, ?$\mathit{mot}$?, ?$\mathit{current\_view}.\mathit{view}$?, ?$\mathit{view\_path}[\mathit{current\_view}.\mathit{view}]]$??\\?to ?$\mathit{current\_view}.\mathit{view}$?.members() ?\label{line:broadcast_support}?
    
        upon ?exists? Message ?$m \in \mathit{waiting\_messages}$? such that ?$m = [$?SUPPORT, Motion ?$\mathit{mot}$?, View ?$v$?, View_Path ?$\mathit{path}]$? such that ?$v = \mathit{path}$?.destination() and ?$m$?.sender ?$\in v$?.members():
            ?$\mathit{waiting\_messages} = \mathit{waiting\_messages} \setminus{\{m\}}$?
            if ?$m$?.sender ?$\notin \mathit{support\_from}[\mathit{mot}][v]$?:
                ?$\mathit{support\_from}[\mathit{mot}][v] = \mathit{support\_from}[\mathit{mot}][v] \cup \{m$?.sender?$\}$?
                ?$\mathit{supports}[\mathit{mot}][v] = \mathit{supports}[\mathit{mot}][v] \cup \{m\}$?
                
        upon ?exist? Motion ?$\mathit{mot}$? and View ?$v$? such that ?$\mathit{|support\_from}[\mathit{mot}][v]| \geq v$?.quorum(): ?\label{line:rule_motion_passes}?
            // the motion passes since the server obtains the voting proof ?\label{line:motion_passes}?
\end{lstlisting}

\para{Proof of correctness}
We start by proving the voting safety property (see \Cref{subsection:balance_based_voting}, paragraph ``Properties'').
First, we show that any transaction quasi-committed at a correct server is seen by a quorum of a correct members of the latest installable view at the time.
We say that a transaction $\mathit{tx}$ is quasi-committed at server $r$ at time $t$ if and only if $\mathit{quasi\_committed}[\mathit{tx}] = \top$ at server $r$ at time $t$ (note that the $\mathit{quasi\_committed}$ variable is never reset to $\bot$; see \Cref{lst:transaction_module}).
Moreover, the time at which a correct server $r$ updates its current view to $v$ (i.e., executes line~\ref{line:formal_update_view_1} of \Cref{lst:reconfiguration_initialization} or line~\ref{line:update_view_formal_2} of \Cref{lst:reconfiguration_view_transition}) is denoted by $\mathit{update\_time}(r, v)$.

\begin{lemma} \label{lemma:voting_time}
Consider a time $t$ and a correct server $r$.
Let $\mathit{current\_view}.\mathit{view} = v \neq \bot$ at server $r$ at time $t$.
Moreover, let $\mathit{quasi\_committed}_r(t)$ denote the set of transactions that are quasi-committed at server $r$ at time $t$ and let there exist a transaction $\mathit{tx} \in \mathit{quasi\_committed}_r(t)$ that is quasi-committed in $v$ by time $t$.

There exists a set $C_v$ of servers, where $C_v \subseteq v.\mathtt{members()}$ and $|C_v| \geq v.\mathtt{plurality()}$, such that, for every server $r_v \in C_v$, the following holds:
\begin{compactenum}
    \item $r_v$ is correct, and 
    
    \item $\mathit{quasi\_committed}_r(t) \subseteq \mathit{log}_{r_v}(t)$, where $\mathit{log}_{r_v}(t) = \{\mathit{tx} \,|\, \mathit{state}.\mathit{log}[\mathit{tx}] \neq \bot \text{ at server } r_v \text{ at time } t\}$.
\end{compactenum}
\end{lemma}
\begin{proof}
Every transaction $\mathit{tx} \in \mathit{quasi\_committed}_r(t)$ is quasi-committed in a view smaller than or equal to $v$ (see \Cref{lst:transaction_module}).
Recall that there is a transaction $\mathit{tx} \in \mathit{quasi\_committed}_r(t)$ that is quasi-committed in $v$.

Let $\mathit{quasi\_committed}_r^*(t)$ denote the set of transactions that are quasi-committed at server $r$ at time $t$ and which are quasi-committed in view $v$ (and not in a smaller view); note that $\mathit{quasi\_committed}_r^*(t) \subseteq \mathit{quasi\_committed}_r(t)$.
According to the $\mathtt{allowed\_to\_quasi\_commit}$ function (line~\ref{line:allowed_to_quasi_commit} of \Cref{lst:transaction_module}), a set $R_v \subseteq v.\mathtt{members()}$ of servers exists, where $|R_v| \geq v.\mathtt{plurality()}$, every server from $R_v$ is correct and every server from the $R_v$ set has sent the $\mathtt{COMMIT-CONFIRM}$ associated with view $v$ for every transaction from $\mathit{quasi\_committed}_r^*(t)$.
Moreover, all these servers have updated their current view to $v$ before time $t$ (because of the check at line~\ref{line:commit_message_rule} of \Cref{lst:transaction_module}).
    
Consider any server $r_v \in R_v$ at time $t$.
For every transaction $\mathit{tx} \in \mathit{quasi\_committed}_r(t) \setminus{\mathit{quasi\_committed}_r^*(t)}$, $\mathit{state}.\mathit{log}[\mathit{tx}] \neq \bot$ at time $t$ at server $r_v$ (follows from \cref{lemma:committed_in_log,lemma:log_only_grows}).
Moreover, for every transaction $\mathit{tx}^* \in \mathit{quasi\_committed}_r^*(t)$, $\mathit{state}.\mathit{log}[\mathit{tx}] \neq \bot$ at time $t$ at server $r_v$ (follows from line~\ref{line:update_log_state_representation} of \Cref{lst:transaction_module} and \Cref{lemma:log_only_grows}).
Hence, the $R_v$ set satisfies the statement of the lemma and the lemma holds.
\end{proof}

Next, we prove that (at most) $v.\mathtt{plurality()}$ of correct members of a ``stale'' view $v$ send a $\mathtt{SUPPORT}$ message associated with $v$.
This lemma is similar to \Cref{lemma:no_quasi_committed_below_max}.

\begin{lemma} \label{lemma:no_support_below_max}
Let a motion $\mathit{mot}$ be proposed at time $t$.
Let $\mathcal{V}(t) = \{v \,|\, \mathit{current\_view}.\mathit{view} = v \neq \bot \text{ at a correct server at time } t\}$.
If $\mathcal{V}(t) \neq \emptyset$ and $v_{\mathit{max}}$ is the greatest view of $\mathcal{V}(t)$, then at most $v.\mathtt{plurality()} - 1$ of correct members of a view $v \subset v_{\mathit{max}}$ send the $\mathtt{SUPPORT}$ message for $\mathit{mot}$ associated with view $v$.
\end{lemma}
\begin{proof}
The proof of the lemma is similar to the proof of \Cref{lemma:no_quasi_committed_below_max}.
Note that $\mathcal{V}(t)$ might contain values of the $\mathit{current\_view}.\mathit{view}$ variable of correct servers that halted by time $t$ (note that $\mathit{current\_view}.\mathit{view}$ is not modified upon leaving; see \Cref{lst:reconfiguration_view_transition}).
We prove the lemma by contradiction.
Without loss of generality, suppose that $v.\mathtt{plurality()}$ of correct members of a view $v \subset v_{\mathit{max}}$ send the $\mathtt{SUPPORT}$ message for $\mathit{mot}$ associated with $v$.

First, note that no correct server obtains a vote transaction for $\mathit{mot}$ before time $t$ (follows from the definition ``motion proposed at time $t$''; see \Cref{subsection:balance_based_voting}, paragraph ``Motions'').
Therefore, no correct server sends the $\mathtt{SUPPORT}$ message for $\mathit{mot}$ before $t$ (due to the $\mathtt{support}$ function from \Cref{lst:log_supports_motion} and the check at line~\ref{line:rule_support} of \Cref{lst:voting_module}).
Moreover, since a correct server sends the $\mathtt{SUPPORT}$ message associated with $v$, $v$ is installable at time $\infty$ (due to the check at line~\ref{line:rule_support} of \Cref{lst:voting_module}).

Let $r_{\mathit{first}}$ be the first correct server to set its $\mathit{current\_view}.\mathit{view}$ variable to a view greater than $v$; let that view be $v_{\mathit{first}} \supset v$.
Observe that $r_{\mathit{first}}$ sets its $\mathit{current\_view}.\mathit{view}$ variable to $v_{\mathit{first}}$ by time $t$.
Since $v \subset v_{\mathit{first}}$, $v_{\mathit{first}} \neq \mathit{genesis}$ (by \Cref{lemma:genesis_smallest_valid}).
Therefore, before updating its $\mathit{current\_view}.\mathit{view}$ variable to $v_{\mathit{first}}$, $r_{\mathit{first}}$ has received $\mathtt{STATE-UPDATE}$ messages from the quorum of members of $v$ (the rule at line~\ref{line:prepared_from_source} of \Cref{lst:reconfiguration_state_transfer} becomes active and $\mathit{reconfiguration}.\mathit{source} = v$ at that time at $r_{\mathit{first}}$, by \Cref{lemma:install_skip_installable}).
Each such $\mathtt{STATE-UPDATE}$ message is sent before time $t$ at line~\ref{line:send_state_update} of \Cref{lst:reconfiguration_state_transfer}.
Therefore, before sending the $\mathtt{STATE-UPDATE}$ message (i.e., before time $t$), the $\mathit{stop\_processing\_until}$ variable is equal to (at least) $v$ at each correct server that sends the $\mathtt{STATE-UPDATE}$ message received by $r_{\mathit{first}}$.

Furthermore, by the assumption, $v.\mathtt{plurality()}$ of correct members of $v \subset v_{\mathit{max}}$ send the $\mathtt{SUPPORT}$ message for $\mathit{mot}$ associated with $v$; each such message is sent at some time greater than or equal to $t$.
Therefore, there exists a correct member of $v$ that (1) sets its $\mathit{stop\_processing\_until}$ variable to (at least) $v$ before time $t$, and (2) sends the $\mathtt{SUPPORT}$ message for $\mathit{mot}$ associated with $v$ at time $t$ (or later).
Such behavior is not a correct one (see \cref{lst:reconfiguration_view_transition,lst:voting_module}).
Hence, we reach contradiction and the lemma holds.
\end{proof}

Before we prove the voting safety property, we add the assumption that member of $v_{\mathit{final}}$ store in their $\mathit{log}$ variable every transaction ever stored in the $\mathit{log}$ variable of a correct server.

\begin{assumption} \label{assumption:final}
Consider a correct server $r$ such that $\mathit{log}.\mathit{state}[\mathit{tx}] \neq \bot$ at server $r$, for some transaction $\mathit{tx}$.
Eventually, $\mathit{log}.\mathit{state}[\mathit{tx}] \neq \bot$ at every correct member of $v_{\mathit{final}}$.
\end{assumption}

\Cref{assumption:final} can easily be implemented in the following manner.
If a correct server $r$ leaves, then $r \notin v_{\mathit{final}}.\mathtt{members()}$ (follows from \Cref{lemma:install_message_forever_alive} and the fact that $v_{\mathit{final}}$ is the greatest forever-alive view).
Hence, once a correct server leaves, it ``pushes'' all transactions that are in its $\mathit{log}$ variable until it receives a quorum of confirmations (i.e., the $\mathtt{COMMIT-CONFIRM}$ messages) with respect to a view.
This eventually happens because of the fact that $v_{\mathit{final}}$ exists.
Only once the server has ensured that all of its transactions are ``preserved'', the server halts.

Finally, we are ready to prove the voting safety property.

\begin{theorem} [Voting Safety]
Voting safety is satisfied.
\end{theorem}
\begin{proof}
Let a motion $\mathit{mot}$ be proposed at time $t$ and let $\mathit{mot}$ pass.
Since the client learns that the amount of money in \sysname is $\mathit{money}$, the client has received a quorum of $\mathtt{QUERY-RESPONSE}$ messages associated with some view $v_{\mathit{query}}$ (at line~\ref{line:query_rule} of \Cref{lst:client}).
Hence, there exists a correct server $r \in v_{\mathit{query}}.\mathtt{members()}$ such that the total amount of money in its ``log of quasi-committed'' transactions is at least $\mathit{money}$ (since quasi-committed transactions are never ``reset''; see \Cref{lst:transaction_module}).
Let $\mathit{log}_r = \{\mathit{tx} \,|\, \mathit{quasi\_committed}[r] = \top \text{ at server } r \text{ at the moment of sending the } \mathtt{QUERY-RESPONSE} \text{ message}\}$; $\mathit{log}_r.\mathtt{total\_money()} \geq \mathit{money}$.
Finally, note that every transaction $\mathit{tx} \in \mathit{log}_r$ is quasi-committed in a view smaller than or equal to $v_{\mathit{query}}$ (by \Cref{lst:transaction_module}).

Let $\mathcal{V}(t) = \{v \,|\, \mathit{current\_view}.\mathit{view} = v \neq \bot \text{ at a correct server at time } t\}$.
Since the client learns the amount of money in \sysname at time $t$, $\mathcal{V}(t) \neq \emptyset$.
Let $v_{\mathit{max}}$ be the greatest view of $\mathcal{V}(t)$; observe that $v_{\mathit{query}} \subseteq v_{\mathit{max}}$.

Since $\mathit{mot}$ passes, there exists a valid (and installable at time $\infty$) view $v_{\mathit{pass}}$ such that (at least) $v_{\mathit{pass}}.\mathtt{plurality()}$ of correct members of $v_{\mathit{pass}}$ send the $\mathtt{SUPPORT}$ message for $\mathit{mot}$ associated with $v_{\mathit{pass}}$ (by \Cref{lst:voting_verification}); let us denote this set of servers by $R_{\mathit{pass}}$.
No server from the $R_{\mathit{pass}}$ sends the $\mathtt{SUPPORT}$ message for $\mathit{mot}$ before time $t$ (follows from the definition of motion being proposed at time $t$ and \Cref{lst:log_supports_motion}).
By \Cref{lemma:no_support_below_max}, $v_{\mathit{max}} \subseteq v_{\mathit{pass}}$.
We separate two cases:
\begin{compactitem}
    \item Let $v_{\mathit{max}} \subset v_{\mathit{pass}}$.
    For each server $r_{\mathit{pass}} \in R_{\mathit{pass}}$, at the moment of updating its current view to $v_{\mathit{pass}}$, $\mathit{state}.\mathit{log}[\mathit{tx}] \neq \bot$, for every transaction $\mathit{tx} \in \mathit{log}_r$ (by \Cref{lemma:committed_in_log}).
    \\Let $\mathit{log}_{r_{\mathit{pass}}} = \{\mathit{tx} \,|\, \mathit{state}.\mathit{log}[\mathit{tx}] \neq \bot \text{ at server } r_{\mathit{pass}} \text{ at the moment of sending the } \mathtt{SUPPORT} \text{ message}\}$.
    By \Cref{lemma:log_only_grows}, $\mathit{log}_r \subseteq \mathit{log}_{r_{\mathit{pass}}}$, for every server $r_{\mathit{pass}} \in R_{\mathit{pass}}$.
    By \Cref{lemma:always_verified}, $\mathit{log}_{r_{\mathit{pass}}}$ is admissible, for every server $r_{\mathit{pass}} \in R_{\mathit{pass}}$.
    Since $\mathit{log}_r \subseteq \mathit{log}_{r_{\mathit{pass}}}$, $\mathit{money} \leq \mathit{log}_r.\mathtt{total\_money()} \leq \mathit{log}_{r_{\mathit{pass}}}.\mathtt{total\_money()}$, for every server $r_{\mathit{pass}} \in R_{\mathit{pass}}$.
    Finally, by \Cref{assumption:final}, eventually $\mathit{state}.\mathit{log}[\mathit{tx}] \neq \bot$ at a correct member of $v_{\mathit{final}}$, which implies that the correct member of $v_{\mathit{final}}$ quasi-commits $\mathit{tx}$ (by \Cref{lemma:in_log_then_committed}) and obtains a commitment proof for $\mathit{tx}$ (by \Cref{lemma:quasi_committed_proof}), for every $\mathit{tx} \in \mathit{log}_{r_{\mathit{pass}}}$ and every server $r_{\mathit{pass}} \in R_{\mathit{pass}}$.
    Hence, $\mathit{log}_{r_{\mathit{pass}}} \subseteq \mathit{log}_{\infty}$, for every server $r_{\mathit{pass}} \in R_{\mathit{pass}}$, which concludes the proof in this case.
    
    \item Let $v_{\mathit{max}} = v_{\mathit{pass}}$.
    If all transactions that belong to $\mathit{log}_r$ are quasi-committed in a view smaller than $v_{\mathit{max}} = v_{\mathit{pass}}$, the theorem holds because of the argument given in the previous case.
    Therefore, we assume the opposite in the rest of the proof.
    
    By \Cref{lemma:voting_time}, there exists a set $R_{\mathit{pass}}$ of correct members of $v_{\mathit{pass}}$, where $R_{\mathit{pass}} \subseteq v_{\mathit{pass}}.\mathtt{members()}$ and $|R_{\mathit{pass}}| \geq v_{\mathit{pass}}.\mathtt{plurality()}$, such that, for every server $r_{\mathit{pass}} \in R_{\mathit{pass}}$, the following holds:
    \begin{compactenum}
        \item $r_{\mathit{pass}}$ is correct, and 
    
        \item $\mathit{log}_r \subseteq \mathit{log}_{r_{\mathit{pass}}}(t)$, where $\mathit{log}_{r_{\mathit{pass}}}(t) = \{\mathit{tx} \,|\, \mathit{state}.\mathit{log}[\mathit{tx}] \neq \bot \text{ at server } r_{\mathit{pass}} \text{ at time } t\}$.
    \end{compactenum}
    Therefore, there exists a server $r^* \in R_{\mathit{pass}}$ that has sent the $\mathtt{SUPPORT}$ message for $\mathit{mot}$ associated with $v_{\mathit{pass}}$ (after time $t$).
    Let $\mathit{log}_{r^*} = \{\mathit{tx} \,|\, \mathit{state}.\mathit{log}[\mathit{tx}] \neq \bot \text{ at server } r^* \text{ at the moment of sending the } \mathtt{SUPPORT} \text{ message}\}$.
    By \Cref{lemma:log_only_grows}, $\mathit{log}_r \subseteq \mathit{log}_{r^*}$.
    By \Cref{lemma:always_verified}, $\mathit{log}_{r^*}$ is admissible.
    Since $\mathit{log}_r \subseteq \mathit{log}_{r^*}$, $\mathit{money} \leq \mathit{log}_r.\mathtt{total\_money()} \leq \mathit{log}_{r^*}.\mathtt{total\_money()}$.
    Finally, by \Cref{assumption:final}, eventually $\mathit{state}.\mathit{log}[\mathit{tx}] \neq \bot$ at a correct member of $v_{\mathit{final}}$, which implies that the correct member of $v_{\mathit{final}}$ quasi-commits $\mathit{tx}$ (by \Cref{lemma:in_log_then_committed}) and obtains a commitment proof for $\mathit{tx}$ (by \Cref{lemma:quasi_committed_proof}), for every $\mathit{tx} \in \mathit{log}_{r^*}$.
    Hence, $\mathit{log}_{r^*} \subseteq \mathit{log}_{\infty}$, which concludes the proof in this case.
\end{compactitem}
Since the proof is verified in both possible cases, the theorem holds.
\end{proof}

Lastly, we prove the voting liveness property.

\begin{theorem} [Voting Liveness]
Voting liveness is satisfied.
\end{theorem}
\begin{proof}
Let $\mathit{log}_{\mathit{pass}} \subseteq \mathit{log}_{\infty}$ exist.
Every transaction $\mathit{tx} \in \mathit{log}_{\mathit{greater}}$ is quasi-committed (by the check at line~\ref{line:quasi_committed} of \Cref{lst:transaction_module}), for any admissible log $\mathit{log}_\mathit{greater}$ with $\mathit{log}_{\mathit{pass}} \subseteq \mathit{log}_{\mathit{greater}} \subseteq \mathit{log}_{\infty}$.
Let $v_{\mathit{tx}}$ be the smallest view in which $\mathit{tx} \in \mathit{log}_{\mathit{greater}}$ is quasi-committed, for every $\mathit{tx} \in \mathit{log}_{\mathit{greater}}$.
We distinguish two cases:
\begin{compactitem}
    \item Let $v_{\mathit{tx}} \subset v_{\mathit{final}}$.
    Therefore, $\mathit{state}.\mathit{log}[\mathit{tx}] \neq \bot$ at every correct member of $v_{\mathit{final}}$ (by \Cref{lemma:committed_in_log}).
    
    \item Let $v_{\mathit{tx}} = v_{\mathit{final}}$.
    Hence, there exists a correct member of $v_{\mathit{final}}$ that sends the $\mathtt{COMMIT-CONFIRM}$ message for $\mathit{tx}$ associated with $v_{\mathit{final}}$.
    Thus, $\mathit{state}.\mathit{log}[\mathit{tx}] \neq \bot$ at that server (by line~\ref{line:update_log_state_representation} of \Cref{lst:transaction_module}).
    By \Cref{lemma:same_logs}, $\mathit{state}.\mathit{log}[\mathit{tx}] \neq \bot$ at every correct member of $v_{\mathit{final}}$.
\end{compactitem}
In both cases, $\mathit{state}.\mathit{log}[\mathit{tx}] \neq \bot$ at every correct member of $v_{\mathit{final}}$, for every $\mathit{tx} \in \mathit{log}_{\mathit{greater}}$.

Eventually, every correct member of $v_{\mathit{final}}$ broadcasts the $\mathtt{SUPPORT}$ message for the motion to all members of $v_{\mathit{final}}$ (at line~\ref{line:broadcast_support} of \Cref{lst:voting_module}) since it obtains all transactions from $\mathit{log}_{\mathit{greater}}$ (and any ``greater'' log supports the motion) and all correct members of $v_{\mathit{final}}$ receive this message (by the finality property).
Therefore, the rule at line~\ref{line:rule_motion_passes} of \Cref{lst:voting_module} is eventually active at every correct member of $v_{\mathit{final}}$, which means that the motion passes (at line~\ref{line:motion_passes} of \Cref{lst:voting_module}).
\end{proof}